\documentclass[11pt]{article} 
\usepackage{fullpage}
\usepackage{graphicx}
\usepackage{latexsym}

\usepackage{multirow}
\usepackage{mdframed}
\usepackage{tikz}
\usepackage{color} 
\usepackage[linesnumbered,ruled]{algorithm2e} 

\usepackage[toc,page]{appendix}

\usepackage{amsmath, amsthm, amssymb}           
\usepackage[english]{babel}                                  
\usepackage[T1]{fontenc}                                      
\usepackage[utf8]{inputenc}
\usepackage[round]{natbib}
\usepackage{array}           
\usepackage{geometry}
\usepackage{bbm}
\usepackage{hyperref}
\usepackage{subcaption}

\usepackage{dsfont}
\usepackage{mathtools} 
\usepackage{comment}
\usepackage{stackrel}
\usepackage{pdfpages}
\usepackage{extarrows} 
\usepackage{listings} 
\usepackage{ulem} 
\usepackage{hyphenat} 

\definecolor{upmaroon}{rgb}{0.48, 0.07, 0.07}
\definecolor{royalazure}{rgb}{0.0, 0.22, 0.66}
\definecolor{pakistangreen}{rgb}{0.0, 0.4, 0.0}

\theoremstyle{plain}                                
\newtheorem{theorem}{Theorem}[section]
\newtheorem{definition}[theorem]{Definition}
\newtheorem{proposition}[theorem]{Proposition}

\newtheorem{example}[theorem]{Example}   

\newtheorem{lemma}[theorem]{Lemma}
\newtheorem{assumptions}[theorem]{Assumptions}
\newtheorem{remark}[theorem]{Remark}

\usepackage{enumitem}
\setlist[itemize]{noitemsep, topsep=0pt}

\def\iid{iid}
\def\DMLone{DML1}
\def\DMLtwo{DML2}
\def\regsDML{\nohyphens{regsDML}}
\def\regDML{\nohyphens{regDML}}

\newcommand{\eps}{\varepsilon}

\newcommand{\R}{\mathbb{R}} 
\newcommand{\bo}{\boldsymbol{0}}
\newcommand{\one}{\mathds{1}}

\newcommand\independent{\protect\mathpalette{\protect\independenT}{\perp}}
\def\independenT#1#2{\mathrel{\rlap{$#1#2$}\mkern2mu{#1#2}}}

\newcommand{\norm}[1]{\lVert #1 \rVert}
\newcommand{\normbig}[1]{\big\lVert #1 \big\rVert}
\newcommand{\normBig}[1]{\Big\lVert #1 \Big\rVert}
\newcommand{\normbigg}[1]{\bigg\lVert #1 \bigg\rVert}
\newcommand{\normP}[2]{\norm{#1}_{\PP, #2}}
\newcommand{\normPbig}[2]{\normbig{#1}_{\PP, #2}}
\newcommand{\normPBig}[2]{\normBig{#1}_{\PP, #2}}

\newcommand{\normone}[1]{\lvert #1\rvert}
\newcommand{\normonebig}[1]{\big\lvert #1\big\rvert}
\newcommand{\normoneBig}[1]{\Big\lvert #1\Big\rvert}
\newcommand{\normonebigg}[1]{\bigg\lvert #1\bigg\rvert}

\DeclarePairedDelimiter{\ceil}{\lceil}{\rceil}

\DeclareMathOperator*{\argmin}{arg\,min}
\DeclareMathOperator{\E}{\mathbb{E}}
\DeclareMathOperator{\id}{\mathrm{Id}}

\DeclareMathOperator{\Prob}{\mathbb{P}}
\DeclareMathOperator{\Var}{\textrm{Var}}

\newcommand{\Pra}{P_{R_A}}
\newcommand{\Rx}{R_X}
\newcommand{\Ry}{R_Y}
\newcommand{\Ra}{R_A}

\newcommand{\hbRak}{\widehat{\boldsymbol{R}}^{\Ik}_{\boldsymbol{A}}}
\newcommand{\hRaki}{\widehat R_{A,i}^{\Ik}}

\newcommand{\hbRxk}{\widehat{\boldsymbol{R}}^{\Ik}_{\boldsymbol{X}}}
\newcommand{\hbRxtilk}{\widehat{\boldsymbol{R}}^{\Ik}_{\widetilde{\boldsymbol{X}}}}

\newcommand{\hbRxki}{\widehat R_{X,i}^{\Ik}}
\newcommand{\hbRyk}{\widehat{\boldsymbol{R}}^{\Ik}_{\boldsymbol{Y}}}
\newcommand{\hbRytilk}{\widehat{\boldsymbol{R}}^{\Ik}_{\widetilde{\boldsymbol{Y}}}}

\newcommand{\hbRyki}{\widehat R_{Y,i}^{\Ik}}
\newcommand{\hbRyIi}{\widehat R_{Y,i}^{\I}}
\newcommand{\hbRxIi}{\widehat R_{X,i}^{\I}}

\newcommand{\I}{I}
\newcommand{\Ic}{I^c}

\newcommand{\Ik}{I_k}
\newcommand{\Ikc}{I_k^c}
\newcommand{\PiIkcIk}{\Pi_{\hbRak}}
\newcommand{\NN}{N}
\newcommand{\nn}{n}
\newcommand{\KK}{K}
\newcommand{\Lone}{L_1}
\newcommand{\Ltwo}{L_2}
\newcommand{\Lthree}{L_3}
\newcommand{\Lfour}{L_4}
\newcommand{\kk}{k}
\newcommand{\TauN}{\mathcal{T}}

\newcommand{\EpsN}{\mathcal{E}_{\NN}}
\newcommand{\PP}{P}
\newcommand{\indset}[1]{[#1]}
\newcommand{\losstest}{\varphi}
\newcommand{\loss}{\psi}
\newcommand{\lossA}{\psi^A}
\newcommand{\lossB}{\psi^B}
\newcommand{\lossAfull}{\psi^A_{\textrm{full}}}
\newcommand{\lossBfull}{\psi^B_{\textrm{full}}}
\newcommand{\hlossA}{\hat\psi^A}
\newcommand{\hlossB}{\hat\psi^B}
\newcommand{\hlossAfull}{\hat\psi^A_{\textrm{full}}}
\newcommand{\hlossBfull}{\hat\psi^B_{\textrm{full}}}
\newcommand{\lossoverline}{\overline\psi}
\newcommand{\lossoverlinep}{\overline\psi'}
\newcommand{\hlossoverlinep}{\widehat{\overline\psi}'}
\newcommand{\losstilde}{\widetilde\psi}
\newcommand{\lossone}{\psi_1}
\newcommand{\losstwo}{\psi_2}
\newcommand{\lossthree}{\psi_3}

\newcommand{\rNpnumber}{r_{\NN}}
\newcommand{\lambdaNpnumber}{\lambda_{\NN}}
\newcommand{\deltaNnumber}{\delta_{\NN}}
\newcommand{\SIkc}{\{S_i\}_{i\in\Ikc}}

\newcommand{\etazero}{\eta^0}
\newcommand{\hetaIkc}{\hat\eta^{\Ikc}}

\newcommand{\Jzero}{J_0}
\newcommand{\tilJzero}{\tilde\Jzero}
\newcommand{\Jzerop}{\Jzero'}
\newcommand{\Jzerodp}{\Jzero''}
\newcommand{\hJzerok}{\hat J_{k,0}}
\newcommand{\hJzero}{\hat J_0}
\newcommand{\RN}{R_{\NN, \kk}}
\newcommand{\FIk}{\PP_{\Ik}}
\newcommand{\EP}{\E_{\PP}}
\newcommand{\EPN}{\E_{\PPN}}

\newcommand{\PcalN}{\mathcal{P}_{\NN}}
\newcommand{\PPN}{\PP_{\NN}}
\newcommand{\betazero}{\beta_0}
\newcommand{\hbetaN}{\hat\beta}
\newcommand{\hbetaNi}{\hat\beta_{s}}
\newcommand{\hbetaNDMLone}{\hbetaN^{\textrm{\DMLone}}}
\newcommand{\hbetaNDMLtwo}{\hbetaN^{\textrm{\DMLtwo}}}
\newcommand{\hbetaNkDMLone}{\hat\beta^{\Ik}}
\newcommand{\matA}{D_{1}}
\newcommand{\matB}{D_{2}}
\newcommand{\matC}{D_{3}}
\newcommand{\matD}{D_{4}}
\newcommand{\matE}{D_{5}}
\newcommand{\hmatA}{\hat D_{1}}
\newcommand{\hmatB}{\hat D_{2}}
\newcommand{\hmatD}{\hat D_{4}}
\newcommand{\hmatAk}{\hat D_1^k}
\newcommand{\hmatBk}{\hat D_2^k}
\newcommand{\hmatCk}{\hat D_3^k}
\newcommand{\hmatDk}{\hat D_4^k}
\newcommand{\hmatEk}{\hat D_5^k}
\newcommand{\Icalk}{\mathcal{I}_{k}}
\newcommand{\IcalkA}{\mathcal{I}_{k,A}}
\newcommand{\IcalkB}{\mathcal{I}_{k,B}}
\newcommand{\Icalone}{\mathcal{I}_1}
\newcommand{\Icaltwo}{\mathcal{I}_2}
\newcommand{\Ione}{\mathcal{I}_{1}}
\newcommand{\Itwo}{\mathcal{I}_{2}}
\newcommand{\Ithree}{\mathcal{I}_{3}}
\newcommand{\Ifour}{\mathcal{I}_{4}}
\newcommand{\Ifive}{\mathcal{I}_{5}}
\newcommand{\Isix}{\mathcal{I}_{6}}
\newcommand{\Iseven}{\mathcal{I}_{7}}
\newcommand{\Ieight}{\mathcal{I}_{8}}
\newcommand{\Inine}{\mathcal{I}_{9}}
\newcommand{\Iten}{\mathcal{I}_{10}}
\newcommand{\Ieleven}{\mathcal{I}_{11}}
\newcommand{\Itwelve}{\mathcal{I}_{12}}
\newcommand{\Ithirteen}{\mathcal{I}_{13}}
\newcommand{\Ifourteen}{\mathcal{I}_{14}}
\newcommand{\Ififteen}{\mathcal{I}_{15}}
\newcommand{\Isixteen}{\mathcal{I}_{16}}

\newcommand{\CpnormRV}{C_1} 
\newcommand{\CpnormEta}{C_2} 
\newcommand{\Cbetazero}{C_3} 
\newcommand{\CnormLossboth}{C_5} 
\newcommand{\Cbg}{C_4} 

\newcommand{\cone}{c_1}
\newcommand{\ctwo}{c_2}
\newcommand{\cthree}{c_3}
\newcommand{\cfour}{c_4}
\newcommand{\DeltaN}{\Delta_{\NN}}
\newcommand{\DeltaNnumber}{\Delta_{\NN}}
\newcommand{\deltaN}{\delta_{\NN}}
\newcommand{\rhoN}{\rho_{\NN}}
\newcommand{\rhoNtilde}{\tilde\rho_{\NN}}

\newcommand{\X}{\boldsymbol{X}}
\newcommand{\A}{\boldsymbol{A}}
\newcommand{\Y}{\boldsymbol{Y}}
\newcommand{\W}{\boldsymbol{W}}
\newcommand{\Fmatone}{F_1}
\newcommand{\Fmattwo}{F_2}
\newcommand{\Gmatone}{G_1}
\newcommand{\Gmattwo}{G_2}
\newcommand{\Neps}{\NN_{\eps}}
\newcommand{\Ceps}{C_{\eps}}
\newcommand{\Ntilde}{\widetilde\NN}
\newcommand{\Nbar}{\overline\NN}
\newcommand{\hsigma}{\hat\sigma}
\newcommand{\hsigmai}{\hat\sigma_s}
\newcommand{\hsigmaMed}{\hat\sigma^{2,\mathrm{med}}}
\newcommand{\hsigmaMedReg}{\hat\sigma^{2,\mathrm{med}}_{\mathrm{reg}}}
\newcommand{\hbetaMed}{\hat\beta^{\mathrm{med}}}
\newcommand{\hbMedReg}{\hat b^{\mathrm{med}}_{\mathrm{reg}}}
\newcommand{\hgamma}{\hat{\gamma}}
\newcommand{\hgammai}{\hat{\gamma}_{s}}
\newcommand{\hgammap}{\hat{\gamma}'}
\newcommand{\hmu}{\hat{\mu}}
\newcommand{\hgammapi}{\hat{\gamma}_{s}'}
\newcommand{\bg}{b^{\gamma}}
\newcommand{\hbg}{\hat b^{\gamma}}
\newcommand{\hbgi}{\hat b^{\gamma}_{s}}

\newcommand{\hbhgp}{\hat b^{\hgammap}}
\newcommand{\hbhgpi}{\hat b^{\hgammapi}_s}
\newcommand{\hbgN}{\hat b^{\gammaN}}
\newcommand{\bhgp}{ b^{\hgammap}}
\newcommand{\bgN}{b^{\gammaN}}
\newcommand{\gammaN}{{\gamma}_{\NN}}
\newcommand{\GammaN}{{\Xi}_{\NN}}
\newcommand{\CN}{C_{\NN}}
\newcommand{\hbgDMLone}{\hat b^{\gamma, \textrm{\DMLone}}}
\newcommand{\hbgDMLtwo}{\hat b^{\gamma, \textrm{\DMLtwo}}}
\newcommand{\hbgkDMLone}{\hat b^{\gamma}_{\kk}}
\newcommand{\hb}{\hat b}
\newcommand{\bzero}{b^0}

\newcommand{\gX}{g_X}
\newcommand{\gY}{g_Y}
\newcommand{\gA}{g_A}
\newcommand{\gH}{g_H}
\newcommand{\hX}{h_X}
\newcommand{\hY}{h_Y}
\newcommand{\hW}{h_W}
\newcommand{\aX}{a_X}
\newcommand{\aW}{a_W}
\newcommand{\mX}{m_X}
\newcommand{\mA}{m_A}
\newcommand{\mY}{m_Y}
\newcommand{\mU}{m_U}
\newcommand{\mV}{m_V}
\newcommand{\mZ}{m_Z}
\newcommand{\hmX}{\hat m_X}
\newcommand{\hmA}{\hat m_A}
\newcommand{\hmY}{\hat m_Y}
\newcommand{\hmU}{\hat m_U}

\newcommand{\ttt}{u}
\newcommand{\sss}{v}
\newcommand{\Salg}{\mathcal{S}}

\definecolor{springGreen}{rgb}{0, 0.7, 0}
\definecolor{snowGray}{rgb}{0.93, 0.91, 0.91}
\definecolor{darkViolet}{rgb}{0.58, 0, 0.83}
\definecolor{darkOlive}{rgb}{0.64, 0.8, 0.35}
\definecolor{salmon}{rgb}{0.80, 0.51,0.38}

\begin{document}

\title{Regularizing Double Machine Learning in Partially Linear Endogenous Models}  
  
\author{Corinne Emmenegger and Peter B\"uhlmann\\
Seminar for Statistics, ETH Z\"urich}

\maketitle

\begin{abstract}
The linear coefficient in a partially linear model with confounding variables can be estimated using double machine learning (DML). 
However, this DML estimator has a two-stage least squares (TSLS) interpretation and may produce overly wide confidence intervals. 
To address this issue, we propose a
regularization and selection scheme, \emph{\regsDML}, 
which leads to narrower confidence intervals. 
It selects either the TSLS DML estimator or a regularization-only estimator depending on whose estimated variance is smaller.
The regularization-only estimator is tailored to have a low mean squared error. 
The \regsDML\ estimator is fully data driven. 
The \regsDML\ estimator converges at the parametric rate, is asymptotically Gaussian distributed, and asymptotically equivalent to the TSLS DML estimator, but \regsDML\  
exhibits
substantially better finite sample properties.
The \regsDML\ estimator uses the idea of 
k-class estimators, and 
we show how DML and k-class estimation can be combined to estimate the linear coefficient in a partially linear endogenous model. 
Empirical examples demonstrate our methodological and theoretical developments.
Software code for our \regsDML\ method is available in the \textsf{R}-package \texttt{dmlalg}.\end{abstract}


\section{Introduction}

Partially linear models (PLMs) combine the flexibility of nonparametric approaches with ease of interpretation of linear models. Allowing for nonparametric terms makes the estimation procedure robust to some model misspecifications. 
A plaguing issue is potential endogeneity. 
For instance, if a treatment is not randomly assigned in a clinical study,
subjects receiving different treatments  differ in other ways than only the
treatment~\citep{Okui2012}.  Another situation where an explanatory variable is correlated with the error term occurs if the explanatory variable is determined simultaneously with the response~\citep{Wooldridge2013}.
In such situations, 
employing estimation methods that do not account for endogeneity can lead to biased estimators~\citep{Fuller1987}.

Let us consider the PLM 
\begin{equation}\label{eq:PLM}
	Y=X^T\betazero + \gY(W) + \hY(H)+\eps_Y. 
\end{equation}
The covariates $X$ and $W$ and the response $Y$ are observed whereas the variable
$H$ is not observed and acts as a potential confounder. 
It can cause endogeneity in the model when it is correlated with $X$, $W$, and $Y$. 
The variable $\eps_Y$ denotes a random error. 
An overview of PLMs is presented in~\citet{Haerdle2000}. Semiparametric methods are summarized in~\citet{Ruppert2003} and \citet{Haerdle2004}, for instance.
\\

\citet{Chernozhukov2018} introduce double machine learning (DML) to estimate the linear coefficient $\betazero$ in a model similar to~\eqref{eq:PLM}. The central ingredients are Neyman orthogonality and sample splitting with cross-fitting.
They allow estimates of so-called 
 nuisance terms to be plugged 
into the estimating equation of $\betazero$. The resulting estimator  converges at the parametric rate $\NN^{-\frac{1}{2}}$, 
with $\NN$ denoting the sample size, 
and is asymptotically Gaussian. 
\\

A common approach to cope with endogeneity uses instrumental variables (IVs). Consider
a random variable $A$ that typically satisfies the assumptions of a conditional instrument~\citep{Pearl2009}. 
The DML procedure first 
adjusts $A$, $X$, and $Y$ for $W$  by regressing 
out $W$ of them. 
Then the residual $Y-\E[Y|W]$ is regressed on $X-\E[X|W]$ using the instrument $A-\E[A|W]$. The population parameter is identified by
\begin{equation}\label{eq:TSLSinterpretation}
	\betazero = \frac{\E\big[(A-\E[A|W])(Y-\E[Y|W])\big]}{\E\big[(A-\E[A|W])(X-\E[X|W])\big]}
\end{equation}
if both $A$ and $X$ are 1-dimensional. The restriction to the 1-dimensional case is only for simplicity at this point. Below, we consider multivariate $A$ and $X$. 
In practice, we insert potentially biased machine learning (ML) estimates of the nuisance parameters $\E[A|W]$, $\E[X|W]$, and $\E[Y|W]$ into this equation for $\betazero$.  
Estimates of these nuisance parameters are typically biased if their complexity is regularized. 
Neyman orthogonal scores and sample splitting allow circumventing empirical process conditions to justify inserting ML estimators of nuisance parameters into estimating equations~\citep{Bickel1982, Chernozhukov2018}. 
\\

Equation~\eqref{eq:TSLSinterpretation} has a two-stage least squares (TSLS) interpretation~\citep{Theil1953a, Theil1953b, Basmann1957, Bowden1985,Angrist1996,Anderson2005}. 
As mentioned above, the residual term $Y-\E[Y|W]$ is regressed on $X-\E[X|W]$ using the instrument $A-\E[A|W]$. 
In entirely linear models, the following findings have been reported 
about TSLS and related procedures.
The TSLS estimator has been observed to 
be highly variable, 
leading to overly wide confidence intervals. 
For instance, although ordinary least squares (OLS) is biased in the presence of endogeneity, 
it has been observed to be less variable~\citep{Wagner1958, Nagar1960, Summers1965, Cragg1967, Lloyd1975}. 
The issue with large or nonexisting variance of TSLS (the order of existing moments of TSLS depends on the degree of overidentification~\citep{Mariano1972, Mariano1982, Mariano2001}) is also coupled with the strength of the instrument~\citep{Bound1995, Staiger1997, Stock2002, Crown2011, AndrewsForthcoming}. 
Reducing the variability 
is sometimes possible by using k-class estimators~\citep{Theil1961, Hill2011, Rothenhausler2018, Jakobsen2020}.

The k-class estimators have been developed for entirely linear models.
The TSLS estimator is a k-class estimator with a fixed value of $k=1$, and \citep{Anderson1986} recommend to not use fixed k-class estimators.
Three particularly well-established k-class estimators are the limited information maximum likelihood (LIML) estimator~\citep{Anderson-Rubin1949, Amemiya1985} and the Fuller(1) and Fuller(4) estimators~\citep{Fuller1977}. 
They have been developed for entirely linear models to overcome some deficiencies of  TSLS.  
If many instruments are present, LIML experiences some optimality properties~\citep{Anderson-Kunitomo-Matsushita2010}.
Furthermore, the normal approximation for the finite sample estimator may be suboptimal for TSLS but useful for LIML~\citep{Anderson-Sawa1979, Anderson-Kunitomo-Sawa1982, Anderson1983}. 
However, LIML has no moments~\cite{Mariano1982, Phillips1984, Phillips1985, Hillier1993}. The Fuller estimators overcome this problem. 
Having no moments can lead to poor squared error performance, 
especially in weak instrument situations~\citep{Hahn2004}.
On the other hand, the Fuller(1) estimator is approximately unbiased and Fuller(4) has particularly low mean squared error (MSE)~\citep{Fuller1977}. 
\citet{Takeuchi-Morimune1985} give further asymptotic optimality results of the Fuller estimators.
\\

We propose a regularization-selection DML method using the idea of 
k-class estimators.
We call our method \emph{regsDML}.
It is tailored to reduce variance and hence 
improve the MSE of the estimator of $\betazero$.
Nevertheless, 
\regsDML\ converges at the parametric rate, and
its coverage of confidence intervals for the linear coefficient $\betazero$ remains 
valid. 
Empirical simulations demonstrate that \regsDML\ typically leads 
to shorter confidence intervals than LIML, Fuller(1), and Fuller(4), while it still 
attains the nominal coverage level.

\subsection{Our Contribution}\label{sect:ourContribution}

Our contribution is twofold. First, we build on 
the work of~\citet{Chernozhukov2018} to estimate $\betazero$ in the 
endogenous PLM~\eqref{eq:PLM} with multidimensional $A$ and $X$ such that its estimator $\hbetaN$ converges at the parametric rate, $\NN^{-\frac{1}{2}}$,
and is asymptotically Gaussian. 
In contrast to~\citet{Chernozhukov2018}, 
we formulate the underlying model as a structural equation model (SEM) and allow $A$ and $X$ to be multidimensional. 
We directly specify an identifiability condition of $\betazero$ instead of giving additional conditional moment restrictions. The SEM may be overidentified in the sense that the dimension of $A$ can exceed the dimension of $X$. 
Overidentification can lead to more efficient estimators~\citep{Amemiya1974, Berndt1974, Hansen1985} and more robust estimators~\citep{Pearl2004}.
Considering SEMs and an identifiability condition allows us to apply DML to more general situations than in~\citet{Chernozhukov2018}. 

Second, we propose a DML method that employs regularization and selection. This method is called \regsDML, 
and we develop it in Section~\ref{sect:regularizedDML}. 
It reduces the potentially excessive estimated standard deviation 
of DML because it 
selects either the TSLS DML estimator or a regularization-only estimator called \regDML\ depending on whose estimated variance is smaller. 
The underlying idea
of the regularization-only estimator \regDML\ 
 is similar to k-class estimation~\citep{Theil1961} and anchor regression~\citep{Rothenhausler2018, Buehlmann2018}. Both k-class estimation and anchor regression are designed for linear models and may require choosing a regularization parameter. Our approach is designed for PLMs, and the regularization parameter is data driven. Recently,~\citet{Jakobsen2020} have proposed a related strategy for linear (structural equation) models; whereas
they rely on testing for choosing the amount of regularization, we tailor our approach to reduce the MSE such that the coverage of confidence intervals for $\betazero$ remains 
valid. 
The \regsDML\ estimator converges at the parametric rate and is asymptotically Gaussian.
 In this sense, and in contrast to~\citet{Jakobsen2020}, 
\regsDML\ focuses on statistical inference beyond point estimation with coverage guarantees 
not only in linear models but also in potentially complex partially linear ones.
The \regsDML\ estimator is asymptotically equivalent to the TSLS-type DML estimator, but \regsDML\ may exhibit substantially better finite sample properties.
Furthermore, our developments show how DML and k-class estimation can be combined to estimate the linear coefficient in an endogenous PLM. 

Our approach allows flexible model specification.
We only require that 
$X$ enters linearly in~\eqref{eq:PLM} and that the other terms are additive.
In particular, the form of the effect of $W$ on $A$ or of $A$ on $W$ is not constrained. This is partly 
similar to TSLS, which is 
robust to model misspecifications in its first stage because  it does not rely on a correct specification of the instrument effect on the covariate~\citep{Robins2005}. 
The detailed assumptions on how the variables $A$, $X$, $W$, $H$, and $Y$ interact are given in Section~\ref{sect:identifiabilityConditionAndDML}: the variable $A$ needs to satisfy an assumption similar to that for a conditional instrument, but there is some flexibility. 
\\

We consider a motivating example to illustrate some of the points mentioned above.  
Figure~\ref{fig:introSEM} gives the SEM we generate data from and its associated causal graph~\citep{Lauritzen1996, Pearl1998, Pearl2009, Pearl2010, Peters2017,Maathuis2019}. 
By convention, we omit error variables in a causal graph if they are mutually independent~\citep{Pearl2009}. 
The variable $A$ is similar to a conditional instrument given $W$.

\begin{figure}[h!]
	\centering
	\caption[]{\label{fig:introSEM}An SEM and its associated causal graph.}
	\begin{tabular}{cc}
	\begin{tabular}{l}
	$\displaystyle 
	\begin{array}{r}
		(\eps_A, \eps_H, \eps_X, \eps_Y)\sim\mathcal{N}_4(\bo,\one)\\	
	\end{array}$
	\\
	$\displaystyle
	\begin{array}{lcl} 
		W &\sim& \pi\cdot\textrm{Unif}([-1, 1])\\
			A &\leftarrow& 3\cdot\tanh(2W) + \eps_A\\
			H &\leftarrow& 2\cdot\sin(W) + \eps_H\\
			X &\leftarrow& -\normone{A}-2\cdot\tanh(W) -H + \eps_X\\
			Y &\leftarrow& X + 0.5W^2 -3\cdot\cos(0.25\pi H)+\eps_Y
	\end{array}$ 
	\end{tabular}
	& 
	\begin{tabular}{c}
          \includegraphics[width=0.35\textwidth]{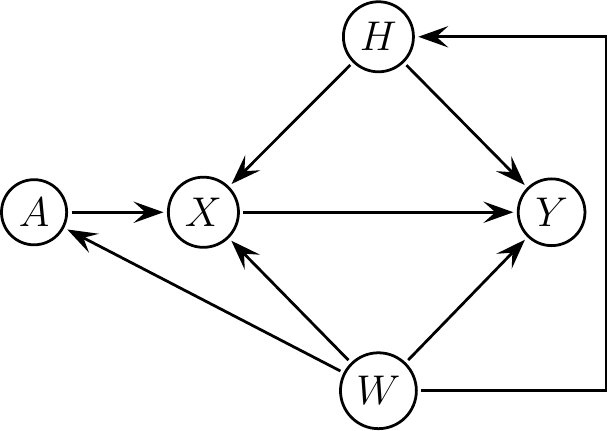}
	\end{tabular}
\end{tabular}
\end{figure}

We simulate $M=1000$ datasets each for a range of sample sizes $\NN$. The nuisance parameters $\E[A|W]$, $\E[X|W]$, and $\E[Y|W]$ are estimated with additive cubic B-splines with $\ceil[\big]{\NN^{\frac{1}{5}}}+2$ degrees of freedom. 
The simulation results are displayed in Figure~\ref{fig:COVERsimulationIntro}. This figure displays the 
coverage, power, and relative length of the  $95\%$ confidence intervals for $\beta_0$ 
using  ``standard'' DML (red) 
and the newly proposed methods
\regDML\ (blue) and \regsDML\ (green). 
The \regDML\ method is a version of \regsDML\ with regularization only but no selection. 
If the blue curve is not visible in Figure~\ref{fig:COVERsimulationIntro}, it coincides with the green curve. The dashed lines in the coverage and power plots indicate $95\%$ confidence regions with respect to uncertainties in the $M$ simulation runs. 

The \regsDML\ method succeeds in producing much narrower confidence intervals than DML 
although it maintains good coverage. 
The power of \regsDML\ is close to $1$ for all considered sample sizes. For small sample sizes, \regsDML\ leads to confidence intervals whose length is around $10\%-20\%$ the length of DML's. 
As the sample size increases, \regsDML\ starts to resemble the behavior of the DML estimator but continues to produce substantially shorter confidence intervals. Thus, the regularization-selection \regsDML\ (and also its version with regularization only) is a highly effective method to increase the power and sharpness of statistical inference whereas keeping the type I error and coverage under control. 

\begin{figure}[h!]
	\centering
	\caption[]{\label{fig:COVERsimulationIntro} 
	The results come from $M =1000$ simulation runs each from the SEM in Figure~\ref{fig:introSEM} for a range of sample sizes $\NN$ and with $\KK=2$ and $\Salg=100$ in Algorithm~\ref{algo:Summary}.
The nuisance functions are estimated with additive splines.  
	The figure displays the coverage of two-sided confidence intervals for $\betazero$, power for two-sided testing of the 
	hypothesis $H_0:\ \betazero = 0$, and scaled lengths of two-sided confidence intervals of DML (red),  \regDML\ (blue), and  \regsDML\  (green), 
	where all results are at level $95\%$.	
	At each $\NN$, the lengths of the confidence intervals are scaled with the median length from DML.
	The shaded regions in the coverage and power plots represent $95\%$ confidence bands with respect to the $M$ simulation runs. 
	The blue and green lines are indistinguishable in the left panel.
	}
	\includegraphics[width=\textwidth]{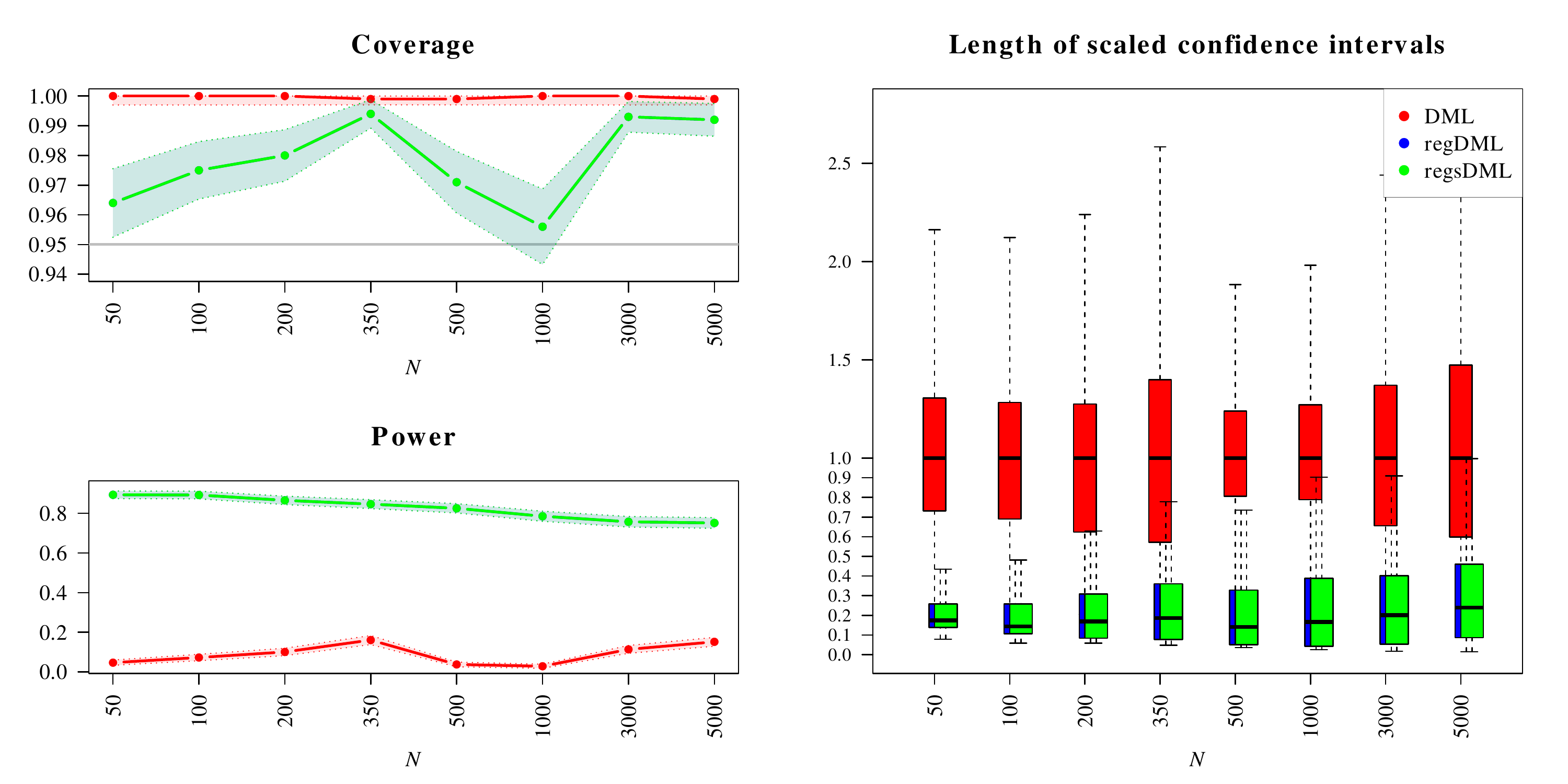}
\end{figure}

Simulation results with $\betazero=0$ in the SEM of Figure~\ref{fig:COVERsimulationIntro} are presented in Figure~\ref{fig:COVERsimulationIntroBeta0} in Section~\ref{sect:additioinalSimulation} in the appendix. 
Further numerical results are given in Section~\ref{sect:simulation}.

\subsection{Additional Literature}

PLMs have received considerable interest. 
\citet{Haerdle2000} present an overview of estimation methods in purely exogenous PLMs, and many references are given there. 
The remaining part of this paragraph refers to literature investigating endogenous PLMs. 
\citet{Ai2003} consider semiparametric estimation with  a sieve estimator. \citet{Ma2006} introduce a parametric model for the latent variable. \citet{Yao2012} considers a heteroskedastic error term and 
 a partialling-out scheme~\citep{Robinson1988, Speckman1988}.
\citet{Florens2012} propose to solve an ill-posed integral equation.
\citet{Su2016} investigate a partially linear dynamic panel data model with fixed effects and lagged variables and consider sieve IV estimators
as well as an approach with solving integral equations. 
\citet{Horowitz2011} compares inference and other properties of nonparametric and parametric  estimation if instruments are employed. 
\\

Combining Neyman orthogonality and sample splitting (with cross-fitting) 
allows a diverse range of estimators and machine learning algorithms to be used to estimate nuisance parameters. 
This procedure has alternatively been considered in~\citet{Newey1994}, \citet{Laan2003}, and~\citet{Chernozhukov2018}.
DML methods have been applied in various situations. \citet{Chen2019} 
consider instrumental variables quantile regression. 
\citet{Liu2020} apply DML in logistic partially linear models. 
\citet{Colangelo2020} employ doubly debiased machine learning methods to a fully nonparametric equation of the response with a continuous treatment. 
\citet{Knaus2020} presents an overview of DML methods in unconfounded models. 
\citet{Farbmacher2020} decompose the causal effect of a binary treatment by a mediation analysis and estimate it by DML. 
\citet{Lewis2020} extend DML to estimate dynamic effects of treatments. 
\citet{Chiang2020} apply DML under multiway clustered sampling environments. 
\citet{Tchetgen2020}  
propose a technique to 
reduce the bias of DML estimators. 

Nonparametric components can be estimated without sample splitting and cross-fitting if the underlying function class satisfies some entropy conditions; 
see for instance~\citet{Geer-Mammen1997}.
Alternatively, \citet{Chen2016} partial out the nonparametric component using a kernel method and employ the generalized method of moments principle~\citep{Hasen1982}. The mentioned entropy regularity conditions limit the complexity of the function class, and ML algorithms do usually not satisfy them. Particularly, these conditions fail to hold if the dimension of the nonparametric variables increases with the sample size~\citep{Chernozhukov2018}. 
\\

Double robustness and orthogonality arguments have also been considered in the following works.
\citet{Okui2012} consider doubly robust estimation of the parametric part. Their estimator is consistent if either the model for the effect of the measured confounders on the outcome or the model of the effect of the measured confounders on the instrument is correctly specified. \citet{Rotnitzky2019} consider doubly robust estimation of scalar parameters where the nuisance functions are $\ell_1$-constrained.
Targeted minimum loss based estimators and G-estimators also feature an orthogonality property; an overview is given in~\citet{DiazOrdaz2019}. 
\\

The literature 
presented in this subsection
is
related to but rather 
distinct from our work with the only exception of~\citet{Chernozhukov2018}. The difference to this 
latter contribution is highlighted in Section~\ref{sect:identifiabilityConditionAndDML} and Section~\ref{sect:strictlyWeakerCondition} in the appendix.
\\

\textit{Outline of the Paper.}
Sections~\ref{sect:identifiabilityConditionAndDML} and~\ref{sect:DML} describe the DML estimator. The former section introduces an identifiability condition, and the latter investigates asymptotic properties. 
Section~\ref{sect:regularizedDML} introduces the regularized
regularization-selection estimator \regDML\ and its regularization-only version \regDML\   
and investigates their asymptotic properties. 
Section~\ref{sect:simulation} presents numerical experiments and an empirical  
real data example. Section~\ref{sect:conclusion} concludes our work. 
Proofs and additional definitions and material are given in the appendix. 
\\

\textit{Notation.}
We denote by $\indset{\NN}$ the set $\{1,2,\ldots,\NN\}$. We add the probability law as a subscript to the probability operator $\Prob$ and the expectation operator $\E$ whenever we want to emphasize the corresponding dependence.
We denote the $L^p(\PP)$ norm by $\normP{\cdot}{p}$ and the Euclidean or  operator
norm  by $\norm{\cdot}$, depending on the context. 
We implicitly assume that given expectations and conditional expectations exist. We denote by $\stackrel{d}{\rightarrow}$ convergence in distribution. 
Furthermore, we denote by $\one_{d\times d}\in\R^{d\times d}$ the $d\times d$ identity matrix and write $\one$ if we do not want to underline its dimension.

\section{An Identifiability Condition and the DML Estimator}\label{sect:identifiabilityConditionAndDML}

Before we introduce \regsDML\ in Section~\ref{sect:regularizedDML}, 
we present our TSLS-type DML estimator of $\betazero$ because 
we require it to formulate \regsDML.
The DML estimator estimates 
the linear coefficient in an endogenous and potentially overidentified PLM where $A$ and $X$ may me multidimensional. 
Our work builds on~\citet{Chernozhukov2018}, but they only consider univariate $A$ and $X$ and restrict conditional moments to identify the linear coefficient. 
We impose an unconditional moment restriction below.
However, our results recover theirs if $A$ and $X$ are univariate 
and the additional conditional moment restrictions are satisfied.

Our PLM is cast as an SEM. The SEM specifies the generating mechanism of the random variables $A$, $W$, $H$, $X$, and $Y$ of dimensions $q$, $\sss$, $r$, $d$, and $1$, respectively. The structural equation of the response is given by
\begin{equation}\label{eq:SEM}
	Y \leftarrow X^T\betazero + \gY(W) + \hY(H) + \eps_Y
\end{equation} 
as in~\eqref{eq:PLM}, where $\betazero\in\R^d$ is a fixed unknown parameter vector, and where the functions 
$\gY$ and $\hY$ are unknown. The variable $H$ is hidden and causes endogeneity. The variable $\eps_Y$ denotes an unobserved error term. 
The model is potentially overidentified in the sense that the dimension of $A$ may exceed the dimension of $X$.
Observe that $A$ does not directly affect the response $Y$ in the sense that it does not appear on the right hand side of~\eqref{eq:SEM}.
The model is required to satisfy an indentifiability condition as in~\eqref{eq:identificationCondition} below.

Econometric models are often presented as a system of simultaneous structural equations.
Full information models consider all equations at once, and limited information models only consider equations of interest~\citep{Anderson1983}.

\subsection{Identifiability Condition}\label{sect:identifiabilityCondition}

An identifiability condition is required to identify $\betazero$ in~\eqref{eq:SEM}.
We define the residual terms
\begin{equation}\label{eq:residuals}
	\Ra:=A-\E[A|W],\quad \Rx:=X-\E[X|W],\quad\textrm{and}\quad \Ry:=Y-\E[Y|W]
\end{equation}
that adjust $A$, $X$, and $Y$ for $W$.
Our DML estimator of $\betazero$ is obtained by performing TSLS of $\Ry$ on $\Rx$ using the instrument $\Ra$. 
This scheme requires the unconditional moment condition 
\begin{equation}\label{eq:identificationCondition}
	\E\big[\Ra(\Ry-\Rx^T\betazero)\big]=\bo
\end{equation}
to identify $\betazero$ in~\eqref{eq:SEM}.
For instance, 
this condition is satisfied if $A$ is independent of both $H$ and $\eps_Y$ given $W$ or if 
$A$ is independent of $H$, $\eps_Y$, and $W$. 
The identifiability condition~\eqref{eq:identificationCondition} is strictly weaker than the conditional moment conditions introduced in~\citet{Chernozhukov2018}; see Section~\ref{sect:strictlyWeakerCondition} in the appendix 
that presents an example where our identifiability condition holds but the conditional moment conditions do not. 
The subsequent theorem asserts identifiability of $\betazero$. 

\begin{theorem}\label{thm:identifiability}
	Let the dimensions $q=\mathrm{dim}(A)$ and $d = \mathrm{dim}(X)$, and assume $q\ge d$.
	Assume furthermore that the matrices $\E[\Rx \Ra^T]$ and $\E[\Ra\Ra^T]$ 
are of full rank, and assume the identifiability condition~\eqref{eq:identificationCondition}. 
	We then have
	\begin{displaymath}
		\betazero = \Big(\E\big[\Rx\Ra^T\big]\E\big[\Ra\Ra^T\big]^{-1}\E\big[\Ra\Rx^T\big] \Big)^{-1}\E\big[\Rx\Ra^T\big]\E\big[\Ra\Ra^T\big]^{-1}\E[\Ra\Ry]. 
	\end{displaymath}
\end{theorem}

Theorem~\ref{thm:identifiability} precludes 
underidentification.
The full rank condition of the matrix $\EP[\Rx \Ra^T]$  expresses that the correlation between $X$ and $A$ is strong enough after regressing out $W$. This is a typical TSLS assumption~\citep{Theil1953a, Theil1953b, Basmann1957, Bowden1985,Angrist1996,Anderson2005}.
The rank assumptions in Theorem~\ref{thm:identifiability} in particular require that $A$, $X$, and $Y$ are not deterministic functions of $W$.

The instrument $A$ instead of $\Ra$ can alternatively identify $\betazero$ in Theorem~\ref{thm:identifiability}. However, this procedure leads to a suboptimal convergence rate of the resulting estimator; see Section~\ref{sect:nonidentifyingProcedures}. 

The identifiability condition~\eqref{eq:identificationCondition} is central to Theorem~\ref{thm:identifiability}. 
Section~\ref{sect:discussionIdentifiabilityCondition} in the appendix presents examples illustrating SEMs where the identifiability condition holds and where it fails to hold.

\subsection{Alternative Interpretations of $\betazero$}\label{sect:alternativeInterpretations}

We present two alternative interpretations of $\betazero$ apart from performing TSLS of $\Ry$ on $\Rx$ using the instrument $\Ra$. 
The second representation will be used to formulate our regularization schemes in Section~\ref{sect:regularizedDML}.
To formulate these alternative representations, we introduce the linear projection operator $\Pra$ on $\Ra$ that maps a random variable $Z$ to its projection
\begin{displaymath}
	\Pra Z := \E\big[Z\Ra^T\big]\E\big[\Ra\Ra^T\big]^{-1}\Ra. 
\end{displaymath}

By Theorem~\ref{thm:identifiability},  the population parameter $\betazero$ solves the TSLS moment equation
\begin{displaymath}
	\bo =\E\big[\Rx\Ra^T\big]\E\big[\Ra\Ra^T\big]^{-1}\E\big[\Ra(\Ry-\Rx^T\betazero)\big]. 
\end{displaymath}
This motivates a generalized method of moments interpretation of $\betazero$ because we have
\begin{displaymath}
	\betazero = \argmin_{\beta\in\R^d}\E[\loss(S;\beta,\etazero)]\E\big[\Ra\Ra^T\big]^{-1}\E\big[\loss^T(S;\beta,\etazero)\big]
\end{displaymath}
for $\loss(S;\beta,\etazero) = \Ra(\Ry-\Rx^T\beta)$, where $\etazero=(\E[A|W], \E[X|W], \E[Y|W])$ denotes the nuisance parameter and $S=(A,W,X,Y)$ denotes the concatenation of the observable variables.

This leads to the second interpretation of $\betazero$. The coefficient $\betazero$ minimizes the squared projection of the residual $\Ry-\Rx^T\beta$ on $\Ra$, namely
\begin{equation}\label{eq:optimizeL2projection}
	\betazero = \argmin_{\beta\in\R^d} \E\Big[\big(\Pra(\Ry-\Rx^T\beta)\big)^2\Big].
\end{equation}
We employ the representation of $\betazero$ in~\eqref{eq:optimizeL2projection} 
to formulate our regularization schemes in Section~\ref{sect:regularizedDML}.

\section{Formulation of the DML Estimator and its Asymptotic Properties}\label{sect:DML}

In this section, we describe how to estimate $\betazero$ using the TSLS-type DML scheme, and we describe the asymptotic properties of this estimator.

Consider $\NN$ \iid\ realizations  $\{S_i=(A_i, X_i,W_i,Y_i)\}_{i\in\indset{\NN}}$ of $S=(A,X,W,Y)$ from the SEM in~\eqref{eq:SEM}. We concatenate the observations of $A$ row-wise to form an $(\NN\times q)$-dimensional matrix  $\A$. Analogously, we construct the matrices $\X\in\R^{\NN\times d}$ and  $\W\in\R^{\NN\times \sss}$
and the vector $\Y\in\R^{\NN}$ containing the respective observations. 

We construct a DML estimator of $\betazero$ as follows. 
First, we split the data into $\KK\ge 2$ disjoint sets $I_1,\ldots,I_{\KK}$. For simplicity, we assume that these sets are of equal cardinality $\nn=\frac{\NN}{\KK}$. In practice, their cardinality might differ due to rounding issues. 
  
For each $\kk\in\indset{\KK}$, we estimate the conditional expectations $\mA^0(W):=\E[A|W]$, $\mX^0(W):=\E[X|W]$, and $\mY^0(W):=\E[Y|W]$, which act as nuisance parameters, with data from $\Ikc$. We call the resulting estimators $\hmA^{\Ikc}$, $\hmX^{\Ikc}$, and $\hmY^{\Ikc}$, respectively. 
Then, the adjusted residual terms $\hRaki:=A_i-\hmA^{\Ikc}(W_i)$, $\hbRxki:=X_i-\hmX^{\Ikc}(W_i)$, and $\hbRyki:=Y_i-\hmY^{\Ikc}(W_i)$ for $i\in\Ik$ are evaluated on $\Ik$, the complement of $\Ikc$. 
We concatenate them row-wise to form the matrices $\hbRak\in\R^{\nn \times q}$ and $\hbRxk\in\R^{\nn\times d}$ and the vector $\hbRyk\in\R^{\nn}$. 

These $\KK$ iterates are assembled to form the DML estimator
\begin{equation}\label{eq:betaDMLtwo}
		\hbetaN := \bigg(\frac{1}{\KK}\sum_{\kk=1}^{\KK}\big(\hbRxk\big)^T\PiIkcIk\hbRxk\bigg)^{-1}\frac{1}{\KK}\sum_{\kk=1}^{\KK}\big(\hbRxk\big)^T\PiIkcIk\hbRyk
\end{equation}
of $\betazero$, 
where
\begin{equation}\label{eq:RaProjection}
	\PiIkcIk := \hbRak
			\Big(\big(\hbRak\big)^T\hbRak \Big)^{-1}
			\big(\hbRak\big)^T
\end{equation}
denotes the orthogonal projection matrix onto the space spanned by the columns of $\hbRak$.

To obtain $\hbetaN$ in~\eqref{eq:betaDMLtwo}, the individual matrices are first averaged before the final matrix is inverted. 
It is also possible to compute $\KK$ individual TSLS estimators on the $\KK$ iterates individually and average these.
Both schemes are asymptotically equivalent. \citet{Chernozhukov2018} call these two schemes
DML2 and DML1, respectively, where DML2 is as in~\eqref{eq:betaDMLtwo}.
The DML1 version of the coefficient estimator is given in the appendix  in Section~\ref{sect:DML1Binfty}. 
The advantage of DML2 over DML1 is that it enhances stability properties of the  estimator. 
To ensure stability of the DML1 estimator, every individual matrix that is inverted needs to be well conditioned. 
Stability of the DML2 estimator is ensured if the average of these matrices is well conditioned. 
\\

The $\KK$ sample splits are random. To reduce the effect of this randomness, we repeat the overall procedure $\Salg$ times and assemble the results as suggested in~\citet{Chernozhukov2018}. This procedure is described in Algorithm~\ref{algo:Summary} in Section~\ref{sect:gammaEst} below. 
\\

The following theorem establishes that $\hbetaN$ converges at the parametric rate and is asymptotically Gaussian. 
\begin{theorem}\label{thm:asymptNormal}
	Consider model~\eqref{eq:SEM}.
	Suppose that Assumption~\ref{assumpt:DMLboth} in the appendix in Section~\ref{sect:proofsOfDML}  holds and consider  $\lossoverline$ given in Definition~\ref{def:lossFunc}  in the appendix in Section~\ref{sect:proofsOfDML}. 
	Then $\hbetaN$ as in~\eqref{eq:betaDMLtwo}
	concentrates in a $\frac{1}{\sqrt{\NN}}$ neighborhood of $\betazero$. It is
	approximately linear and centered Gaussian, namely
	\begin{displaymath}
		\sqrt{\NN}\sigma^{-1}(\hbetaN-\betazero) = \frac{1}{\sqrt{\NN}}\sum_{i=1}^{\NN}\lossoverline( S_i;\betazero, \etazero) + o_{\PP}(1) \stackrel{d}{\rightarrow}\mathcal{N}(0,\one_{d\times d})\quad (\NN\rightarrow\infty),
	\end{displaymath}
	uniformly over the law $\PP$ of $S=(A,W,X,Y)$, 
	and where the variance-covariance matrix $\sigma^2$ is given by
		$\sigma^2 =\Jzero\tilJzero\Jzero^T$
	for the matrices $\tilJzero$ and $\Jzero$ given in Definition~\ref{def:lossFunc} in the appendix.
	\end{theorem}
	A similar result to Theorem~\ref{thm:asymptNormal} is presented by~\citet{Chernozhukov2018}. However, their result requires univariate $A$ and $X$, and it imposes conditional moment restrictions instead of the identifiability condition~\eqref{eq:identificationCondition}; see also Section~\ref{sect:strictlyWeakerCondition} in the appendix
	that presents an example where our identifiability condition holds but the conditional moment conditions do not.  
	If $A$ and $X$ are univariate and the respective conditional moment conditions hold, our result coincides with~\citet{Chernozhukov2018}. 
	
	Theorem~\ref{thm:asymptNormal} also holds for the DML1 version of $\hbetaN$ defined in the appendix in Section~\ref{sect:DML1Binfty}.
	Assumption~\ref{assumpt:DMLboth}  specifies regularity conditions and the convergence rate of the machine learners estimating the conditional expectations. 
	The machine learners are required to satisfy the product relations
		\begin{equation}\label{eq:multCondition}
			\begin{array}{l}
			\normP{\mA^0(W)-\hmA^{\Ikc}(W)}{2}^2 
			\ll\NN^{-\frac{1}{2}},\\
			\normP{\mA^0(W)-\hmA^{\Ikc}(W)}{2}\big(\normP{\mY^0(W)-\hmY^{\Ikc}(W)}{2}+\normP{\mX^0(W)-\hmX^{\Ikc}(W)}{2}\big)\ll\NN^{-\frac{1}{2}}
			\end{array}
		\end{equation}
	for $\kk\in\indset{\KK}$, which
        allows us to employ a broad range of ML estimators.         
          For instance, these convergence rates are satisfied by
          $\ell_1$-penalized and related methods in a variety of sparse, high-dimensional linear
        models~\citep{Candes-Tao2007, Bickel2009, Buehlmann2011, Belloni-Chernozhukov2013},  
        forward selection in sparse linear models \citep{Kozbur2020},
        high-dimensional additive models~\citep{Meier-Geer-Buehlmann2009, Koltchinskii-Yuan2010, Yuan-Zhou2016}, 
         or regression trees and random forests \citep{Wager2016, Athey-Tibshirani-Wager2018}.
         Please see~\citet{Chernozhukov2018} for additional references. 
In particular, the rate condition~\eqref{eq:multCondition} is satisfied if the individual ML estimators converge at rate $\NN^{-\frac{1}{4}}$. Therefore, the individual ML estimators are not required to converge at rate $\NN^{-\frac{1}{2}}$. 	
	
The asymptotic variance $\sigma^2$ can be consistently estimated by 
replacing the true $\beta_0$ by $\hat{\beta}$
or its \DMLone\ version.
The nuisance functions are estimated on subsampled datasets, 
and the estimator of $\sigma^2$ is obtained by cross-fitting. 
The formal definition, the consistency result, and its proof 
are given in Definition~\ref{def:lossFunc} and in 
Theorem~\ref{thm:estSD} in the appendix 
in Section~\ref{sect:proofsOfDML}.

For fixed $\PP$, the asymptotic variance-covariance matrix $\sigma^2$ is the same as if the conditional expectations $\mA^0(W)$, $\mX^0(W)$, and $\mY^0(W)$ and hence $\Ra$, $\Rx$, and $\Ry$ were known.

Theorem~\ref{thm:asymptNormal} holds uniformly over laws $\PP$. This uniformity guarantees some robustness of the asymptotic statement~\citep{Chernozhukov2018}.
The dimension $\sss$ of the covariate $W$ may grow as the sample size  increases. Thus, high-dimensional methods can be considered to estimate the conditional expectations $\E[A|W]$, $\E[X|W]$, and $\E[Y|W]$. 
\\

The estimator $\hbetaN$ solves the moment equations
\begin{displaymath}
	\bo=
	\frac{1}{\KK}\sum_{\kk=1}^{\KK}\bigg(\frac{1}{\nn}\sum_{i\in\Ik}\hbRxki \big(\hRaki\big)^T
	\Big(\frac{1}{\nn}\sum_{i\in\Ik}\hRaki\big(\hRaki\big)^T\Big)^{-1}
	\frac{1}{\nn}\sum_{i\in\Ik}\loss(S_i;\hbetaN, \hetaIkc)\bigg),
\end{displaymath} 
where the score function $\loss$ is given by 
\begin{equation}\label{eq:score_psi}
	\loss(S;\beta,\eta) := \big(A-\mA(W)\big)\Big(Y-\mY(W)-\big(X-\mX(W)\big)^T\beta\Big)
\end{equation}
for $\eta=(\mA,\mX,\mY)$, 
and where the estimated nuisance parameter is given by $\hetaIkc=(\hmA^{\Ikc}, \hmX^{\Ikc}, \hmY^{\Ikc})$. Observe that  $\loss(S;\betazero,\etazero)$ with $\etazero=(\mA^0,\mX^0,\mY^0)$ coincides with the term whose expectation is constrained to equal $\bo$ in the identifiability condition~\eqref{eq:identificationCondition}.
The crucial step to prove asymptotic normality of $\sqrt{\NN}(\hbetaN-\betazero)$ is to analyze the asymptotic behavior of $\frac{1}{\sqrt{\nn}}\sum_{i\in\Ik}\loss(S_i;\hbetaN, \hetaIkc)$ for $\kk\in\indset{\KK}$.

Apart from the identifiability condition, the first fundamental requirement to analyze these terms is the ML convergence rates in~\eqref{eq:multCondition}.
Second, we employ sample splitting and cross-fitting. Sample splitting ensures that the data used to estimate the nuisance parameters and the data on which these estimators are evaluated are independent. Cross-fitting enables us to regain full efficiency. 
The third requirement is that the underlying score function $\loss$ in~\eqref{eq:score_psi} is Neyman orthogonal, which we explain next.

Neyman orthogonality ensures that $\loss$ is insensitive to small changes in the nuisance parameter $\eta$ at the true unknown linear coefficient $\betazero$ and the true unknown nuisance parameter $\etazero$. This makes estimation of $\betazero$ robust to inserting biased ML estimators of the nuisance parameter in the estimation equation. The following definition formally introduces this concept. 

\begin{definition}\label{def:Neyman-orth}\citep[Definition 2.1]{Chernozhukov2018}.
	A score $\loss=\loss(S;\beta,\eta)$ is Neyman orthogonal at $(\betazero,\etazero)$
	if the pathwise derivative map
	\begin{displaymath}
		\frac{\partial}{\partial r}\EP\big[\loss\big(S;\betazero,\etazero+r(\eta-\etazero)\big)\big]
	\end{displaymath}
	exists for all $r\in[0,1)$ and nuisance parameters $\eta$ 
	and vanishes 
	at $r=0$. 
\end{definition}
Definition~\ref{def:Neyman-orth} does not entirely coincide with~\citet[Definition 2.1]{Chernozhukov2018} because the latter also 
includes an identifiability condition. We directly assume the identifiability condition~\eqref{eq:identificationCondition}. 

The subsequent proposition states that the score function $\loss$ in~\eqref{eq:score_psi} is indeed Neyman orthogonal. 

\begin{proposition}\label{prop:Neyman_orth}
    The score $\psi$ given in Equation~\eqref{eq:score_psi} is Neyman orthogonal. 
\end{proposition}
We would like to remark that Neyman orthogonality of $\loss$ neither depends on the distribution of $S$ nor on the value of the coefficients $\betazero$ and $\etazero$.
In addition to being Neyman orthogonal, $\loss$ is linear in $\beta$ in the sense that we have
\begin{equation}\label{eq:psi_linear}
		\loss(S;\beta,\eta)
		=\loss^b(S;\eta)-\loss^a(S;\eta)\beta
\end{equation}
for 
\begin{displaymath}	
	\loss^b(S;\eta):=\big(A-\mA(W)\big)\big(Y-\mY(W)\big)
\end{displaymath}
and
\begin{displaymath}
	\loss^a(S;\eta):=\big(A-\mA(W)\big)\big(X-\mX(W)\big)^T. 
\end{displaymath}
This linearity property is also employed in the proof of Theorem~\ref{thm:asymptNormal}.

\subsection{Suboptimal Estimation Procedure}\label{sect:nonidentifyingProcedures}

In general, we cannot employ $A$ as an instrument instead of $\Ra$ 
in our TSLS-type DML estimation procedure.
For simplicity, we assume $\KK=2$ in this 
subsection and consider disjoint index sets $\I$ and $\Ic$ of size $\nn=\frac{\NN}{2}$.
The term
\begin{equation}\label{eq:asymptRAIV}
	\frac{1}{\sqrt{\nn}}\sum_{i\in\I}A_i\big(\hbRyIi -(\hbRxIi)^T\betazero\big)
\end{equation}
can diverge as $\NN\rightarrow\infty$ because $\hmX^{\Ic}$ and $\hmY^{\Ic}$ can be biased estimators of $\mX^0$ and $\mY^0$. This in particular happens if the functions $\mX^0$ and $\mY^0$ are high-dimensional and need to be estimated by regularization techniques; see~\citet{Chernozhukov2018}. Even if sample splitting is employed, the term~\eqref{eq:asymptRAIV} is asymptotically not well behaved because the underlying score function 
\begin{displaymath}
	\varphi(S;\beta,\eta):=A\Big(Y- \mY(W) - \big(X-\mX(W)\big)^T\beta\Big)
\end{displaymath}
is not Neyman orthogonal. The issue is illustrated in Figure~\ref{fig:AvsRa}. 
The SEM used to generate the data is similar to the nonconfounded model used in~\citet[Figure 1]{Chernozhukov2018}. 
The centered and rescaled term $\frac{\hbetaN-\betazero}{\widehat{\Var}(\hbetaN)}$ using $A$ as an instrument is biased whereas it is not if the instrument $\Ra$ is used.
Here, $\widehat{\Var}(\hbetaN)$ denotes the empirically observed variance of $\hbetaN$ with respect to the performed simulation runs.

\begin{figure}[h!]
	\centering
	\caption[]{\label{fig:AvsRa} Histograms of $\frac{\hbetaN-\betazero}{\widehat{\Var}(\hbetaN)}$, 
	where $\widehat{\Var}(\hbetaN)$ denotes the empirically observed variance of $\hbetaN$ with respect to the simulation runs, 
	using $A$ as an instrument in the left plot and using $\Ra$ as an instrument in the right plot. 
	The orange curves represent the density of $\mathcal{N}(0,1)$. The results come from $5000$ simulation runs 
	of sample size $5000$ each
	from the SEM in the appendix in Section~\ref{sect:histogramRavsA} with $\KK=2$. 	
	The conditional expectations are estimated with random forests consisting of $500$ trees that have a minimal node size of $5$.}
	\includegraphics[width=\textwidth]{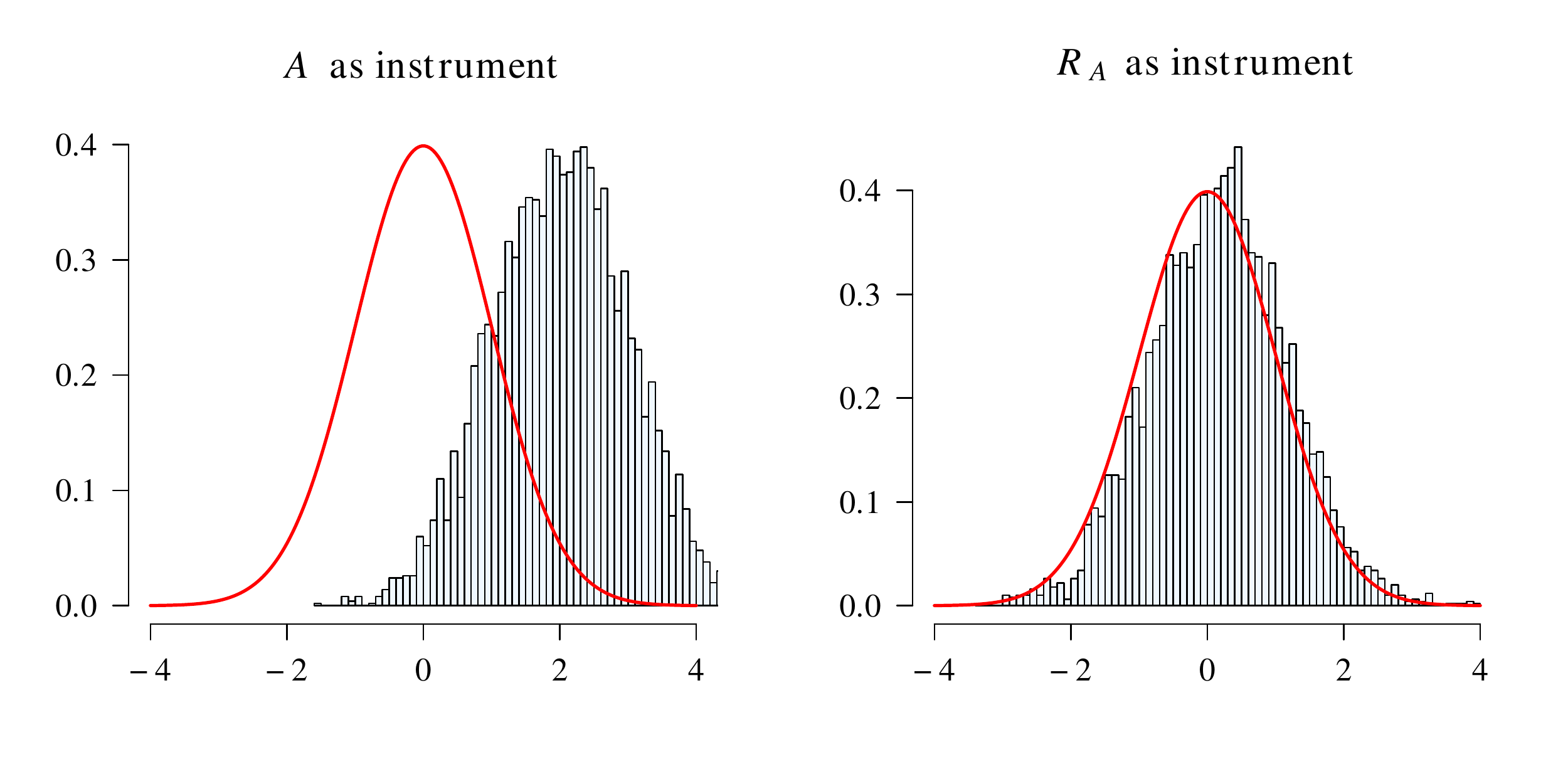}
\end{figure}

\section{Regularizing the DML Estimator: \regDML\ and \regsDML}\label{sect:regularizedDML}

We introduce a regularized estimator, \regsDML, whose estimated standard deviation is 
typically smaller and never worse than  
the one of 
the TSLS-type DML estimator described above. 
Supporting theory and simulations 
illustrate that the associated confidence intervals nevertheless reach 
valid andgood coverage. 
The reg{\bf s}DML estimator {\bf s}elects either the DML estimator or its regularization-only version \regDML, depending on which of the two estimators has a smaller estimated standard deviation. 
\\

Subsequently, we first introduce the regularization-only method \regDML.
The \regDML\ estimator is obtained by regularizing DML 
and choosing a data-dependent regularization parameter. Before we describe the choice of the regularization parameter, we introduce the regularization scheme for fixed regularization parameters.

Given a regularization parameter $\gamma\ge 0$, the population coefficient $\bg$ of the regularization scheme optimizes an objective function similar to the one used in k-class regression~\citep{Theil1961} or anchor regression~\citep{Rothenhausler2018, Buehlmann2018}. 
We established the representation
\begin{displaymath}
	\betazero = \argmin_{\beta\in\R^d} \E\Big[\big(\Pra(\Ry-\Rx^T\beta)\big)^2\Big]
\end{displaymath}
of $\betazero$ in~\eqref{eq:optimizeL2projection}. 
For a regularization parameter $\gamma\ge 0$,
we consider the regularized objective function and corresponding population coefficient
\begin{equation}\label{eq:regularizedObjective}
	\bg := \argmin_{\beta\in\R^d} \E\Big[\big((\id-\Pra)(\Ry-\Rx^T\beta)\big)^2\Big] + \gamma \E\Big[\big(\Pra(\Ry-\Rx^T\beta)\big)^2\Big].
\end{equation}
This regularized objective is form-wise
analogous to the objective function employed in anchor regression. The anchor regression estimator has been reformulated as a k-class estimator by~\citet{Jakobsen2020} for a linear model.

If $\gamma=1$, ordinary least squares regression of $\Ry$ 
on $\Rx$ is performed. If $\gamma=0$, 
we are partialling out or adjusting for
the variable $\Ra$.
If $\gamma=\infty$, we perform
TSLS regression of $\Ry$ on $\Rx$ using the instrument $\Ra$. In this case, $\bg$ coincides with $\betazero$. 
The coefficient $\bg$ interpolates between the OLS coefficient $b^{\gamma=1}$ and the TSLS coefficient $\betazero$ for general choices of $\gamma>1$.
For $\gamma>1$, there is a one to one correspondence between $\bg$ and the k-class estimator (based on $\Ra$, $\Rx$, and $\Ry$) with regularization parameter  $\kappa = \frac{\gamma-1}{\gamma} \in(0, 1)$; see~\citet{Jakobsen2020}.

\subsection{Estimation and Asymptotic Normality}\label{sect:bgammaEst}

In this section, we describe how to estimate $\bg$ in~\eqref{eq:regularizedObjective}
for fixed $\gamma\ge 0$
 using a DML scheme, and we describe the asymptotic properties of this estimator.
We consider the residual matrices $\hbRak\in\R^{\nn \times q}$ and  $\hbRxk\in\R^{\nn\times d}$ and the vector $\hbRyk\in\R^{\nn}$ introduced in Section~\ref{sect:DML} that adjust the data with respect to the nonparametric variables.
The estimator of $\bg$ is given by
\begin{displaymath}
	\hbg :=  \argmin_{b\in\R^d} \frac{1}{\KK}\sum_{\kk=1}^{\KK}\bigg( \normBig{ \big(\one-\PiIkcIk\big)\big(\hbRyk-\big(\hbRxk\big)^Tb\big)}_2^2+\gamma\normBig{ \PiIkcIk(\hbRyk-(\hbRxk)^Tb) }_2^2\bigg),
\end{displaymath} 
where $\PiIkcIk$ is as in~\eqref{eq:RaProjection}.
This estimator can be expressed in closed form by
\begin{equation}\label{eq:regularizedBetaDMLtwo}
	\hbg = \bigg(  \frac{1}{\KK}\sum_{\kk=1}^{\KK} \big(\hbRxtilk\big)^T\hbRxtilk\bigg)^{-1} \frac{1}{\KK}\sum_{\kk=1}^{\KK} \big(\hbRxtilk\big)^T\hbRytilk,
\end{equation}
where
\begin{equation}\label{eq:OLSresiduals}
	\hbRxtilk := \Big(   \one + (\sqrt{\gamma}-1)\PiIkcIk  \Big)\hbRxk 
	\quad\textrm{and}\quad
	\hbRytilk := \Big(   \one + (\sqrt{\gamma}-1)\PiIkcIk  \Big)\hbRyk.
\end{equation}
The computation of $\hbg$ is similar to an OLS scheme where $\hbRytilk$ is regressed on $\hbRxtilk$. 
To obtain $\hbg$, individual matrices are first averaged before the final matrix is inverted. 
It is also possible to directly carry out the $\KK$ OLS regressions of $\hbRytilk$ on $\hbRxtilk$ and average the resulting parameters.
Both schemes are asymptotically equivalent. We call the two schemes DML2 and DML1, respectively. This is analogous to~\citet{Chernozhukov2018}
as already mentioned in Section~\ref{sect:DML}.
The DML1 version is presented in the appendix in Section~\ref{sect:DML1gamma}. 
As mentioned in Section~\ref{sect:DML}, the
advantage of DML2 over DML1 is that it enhances stability properties of the coefficient estimator because the average of matrices needs to be well conditioned  but not every individual matrix. 
\\

\begin{theorem}\label{thm:asymptNormalgamma}
Let $\gamma\ge 0$. 
Suppose that Assumption~\ref{assumpt:DMLboth}  in the appendix in Section~\ref{sect:proofsOfDML}  (same as in Theorem~\ref{thm:asymptNormal}) 
except~\ref{assumpt:DMLboth1}
 holds, 
and consider the quantities
$\sigma^2(\gamma)$ and $\lossoverline$ introduced in Definition~\ref{def:asymptNormalgamma} in the appendix in Section~\ref{sect:proofsRegularizedDML}. 
The estimator $\hbg$ 
concentrates in a $\frac{1}{\sqrt{\NN}}$ neighborhood of $\bg$. It is approximately linear and centered Gaussian, namely
	\begin{displaymath}
		\sqrt{\NN}\sigma^{-1}(\gamma)(\hbg-\bg) = \frac{1}{\sqrt{\NN}}\sum_{i=1}^{\NN}\lossoverline(S_i;\bg, \etazero) + o_{\PP}(1) \stackrel{d}{\rightarrow}\mathcal{N}(0,\one_{d\times d})\quad (\NN\rightarrow\infty),
	\end{displaymath}
	uniformly over laws $\PP$ of $S=(A,W,X,Y)$. 
\end{theorem}

  Theorem~\ref{thm:asymptNormalgamma} also holds for the DML1 version of $\hbg$ defined in the appendix in Section~\ref{sect:DML1gamma}. The influence function is denoted by $\lossoverline$ in both Theorems~\ref{thm:asymptNormal} and~\ref{thm:asymptNormalgamma} but is defined differently.
Assumption~\ref{assumpt:DMLboth} specifies regularity conditions and the convergence rate of the machine learners of the conditional expectations. The machine learners are required to satisfy the product relations
	\begin{displaymath}
			\begin{array}{l}
			\normP{\mA^0(W)-\hmA^{\Ikc}(W)}{2}^2 
			\ll\NN^{-\frac{1}{2}},\\
			\normP{\mX^0(W)-\hmX^{\Ikc}(W)}{2}\big(\normP{\mY^0(W)-\hmY^{\Ikc}(W)}{2}+ \normP{\mX^0(W)-\hmX^{\Ikc}(W)}{2} \big)\ll \NN^{-\frac{1}{2}},\\
			\normP{\mA^0(W)-\hmA^{\Ikc}(W)}{2}\big(\normP{\mY^0(W)-\hmY^{\Ikc}(W)}{2}+\normP{\mX^0(W)-\hmX^{\Ikc}(W)}{2}\big)\ll\NN^{-\frac{1}{2}}
			\end{array}
		\end{displaymath}
	for $\kk\in\indset{\KK}$.
The main difference to Theorem~\ref{thm:asymptNormal} and quantity of interest is the asymptotic variance $\sigma^2(\gamma)$.
It can be consistently estimated with either $\hbg$ or its DML1 version as illustrated in Theorem~\ref{thm:estSDgamma} in the appendix in Section~\ref{sect:proofsRegularizedDML}. Typically, for $\gamma < \infty$, the asymptotic variance $\sigma^2(\gamma)$ is smaller than $\sigma^2$ in Theorem 3.1. Such a variance gain comes at the price of bias because $\hat{b}^{\gamma}$ estimates $b^{\gamma}$ and not the true parameter $\beta_0$.

The proof of Theorem~\ref{thm:asymptNormalgamma} uses Neyman orthogonality of the underlying score function. 
Recall that Neyman orthogonality neither depends on the distribution of $S$ nor on the value of the coefficients $\betazero$ and $\etazero$ as discussed in Section~\ref{sect:DML}.

For fixed $\gamma > 1$, Theorem~\ref{thm:asymptNormalgamma} furthermore implies that 
the k-class estimator corresponding to $\hbg$ converges at the parametric rate and follows a Gaussian distribution asymptotically.

\subsection{Estimating the Regularization Parameter $\gamma$}\label{sect:gammaEst}

For simplicity, we assume $d=1$ in this subsection. The results can be extended to $d>1$. 
\\

Subsequently, we introduce a data-driven method to choose the regularization parameter $\gamma$ in practice. 
This scheme first optimizes the estimated asymptotic MSE of $\hbg$. 
The estimated regularization for the parameter $\gamma$ leads to an estimate of $\betazero$ that 
asymptotically
has the same MSE behavior as the TSLS-type estimator $\hbetaN$ in~\eqref{eq:betaDMLtwo} but may exhibit substantially
better finite sample properties. 
\\

We consider the estimated regularization parameter
\begin{equation}\label{eq:hgamma}
	\hgamma := \argmin_{\gamma\ge 0} \frac{1}{\NN} \hat\sigma^2(\gamma) + \normone{\hbg-\hbetaN}^2.
\end{equation}
It optimizes an estimate of the asymptotic MSE of $\hbg$: 
the term
$\hat\sigma^2(\gamma)$ is the consistent estimator of $\sigma^2(\gamma)$ described in Theorem~\ref{thm:estSDgamma} in the appendix in Section~\ref{sect:proofsRegularizedDML},
and the term $\normone{\hbg-\hbetaN}^2$ is a plug-in estimator of the squared population bias $\normone{\bg - \betazero}^2$.
The estimated regularization parameter $\hgamma$ is random because it depends on the data. 
\\

First, we investigate the bias of the population parameter $\bgN$ for a nonrandom sequence of regularization parameters $\{\gammaN\}_{\NN\ge 1}$ as $\NN\to\infty$. Afterwards, we propose a modified estimator of the regularization parameter whose corresponding parameter estimate is denoted by \regDML, and we introduce the regularization-selection estimator \regsDML.
Finally, we and analyze the asymptotic properties of \regDML\ and \regsDML.  

Let us consider a deterministic sequence $\{\gammaN\}_{\NN\ge 1}$ of regularization parameters. By Proposition~\ref{prop:populationBias} below, the (scaled) population bias $\sqrt{\NN}\normone{\bgN-\betazero}$ vanishes as $\NN\rightarrow\infty$ if $\gammaN$ is of larger order than $\sqrt{\NN}$. 

\begin{proposition}\label{prop:populationBias}
	Suppose that~\ref{assumpt:DMLboth1}, \ref{assumpt:DMLboth3}, and~\ref{assumpt:DMLboth4} of Assumption~\ref{assumpt:DMLboth} in the appendix in Section~\ref{sect:proofsOfDML}  hold 
(subset of the assumptions in Theorem~\ref{thm:asymptNormal}).
	Assume $\{\gammaN\}_{\NN\ge 1}$ is sequence of non-negative real numbers. 
	Then we have 
	\begin{displaymath}
		\sqrt{\NN}\normone{\bgN-\betazero} \rightarrow \begin{cases}0, & \mbox{if }\gammaN \gg\sqrt{\NN}\\
		                                                                          C, & \mbox{if }\gammaN\sim\sqrt{\NN}\\
		                                                                          \infty, &\mbox{if } \gammaN \ll\sqrt{\NN}
		                                                                          \end{cases}
	\end{displaymath}
	as $\NN\rightarrow\infty$ for some non-negative finite real number $C$. 
\end{proposition}

Theorem~\ref{thm:stochasticOrderGammaN} below shows that the estimated regularization parameter $\hgamma$ is of equal or larger stochastic order than $\sqrt{\NN}$.
If it were not, choosing $\gamma=\infty$ in~\eqref{eq:hgamma}, and hence selecting the TSLS-type estimator $\hbetaN$, would lead to a smaller estimated asymptotic MSE. 

\begin{theorem}\label{thm:stochasticOrderGammaN}
	Let $\gammaN = o(\sqrt{\NN})$, and 
	suppose that Assumption~\ref{assumpt:DMLboth} in the appendix in Section~\ref{sect:proofsOfDML}  holds 
(same as in Theorem~\ref{thm:asymptNormal}). 
We then have
	\begin{displaymath}
		\lim_{\NN\rightarrow\infty} \PP\big(\hsigma^2(\gammaN) + \NN(\hbgN-\hbetaN)^2\le \hsigma^2\big)=0. 
	\end{displaymath}
\end{theorem}

If $\hgamma$ is multiplied by a deterministic scalar $a_{\NN}$ that diverges to $+\infty$ at an arbitrarily slow rate as $\NN\rightarrow\infty$, the modified regularization parameter $\hgammap:=a_{\NN}\hgamma$ is of stochastic order larger than $\sqrt{\NN}$. 
By default, we choose $a_{\NN}=\log(\sqrt{\NN})$.
Proposition~\ref{prop:populationBias} is formulated for deterministic regularization parameters, 
but the deterministic statements can be replaced by probabilistic ones. 
Proposition~\ref{prop:populationBias} then implies 
that the population bias term $\normone{\bhgp - \betazero}$ vanishes at rate 
$o_{\PP}(\NN^{-\frac{1}{2}})$. 
Thus, the two quantities $\sqrt{\NN}(\hbhgp-\bhgp)$ and $\sqrt{\NN}(\hbhgp-\betazero)$ are asymptotically equivalent due to Theorem~\ref{thm:gammaHatNormal} below, and we have 
\begin{displaymath}
    \sqrt{\NN}(\hbhgp-\betazero) \approx {\cal N} \big(0,\sigma^2(\hgammap)\big)
\end{displaymath}
whenever $\NN$ is sufficiently large 
(note that asymptotically as $N \to \infty$, the right-hand side has the same limit as described in Theorem~\ref{thm:gammaHatNormal}).

We call $\hbhgp$ the \regDML\ (regularized DML) estimator. 
The regularization-selection estimator $\hbhgp$ selects between DML and \regDML\ based on whose variance estimate is smaller. The ``s'' in \regsDML\ stands for selection. 

\begin{theorem}\label{thm:gammaHatNormal}
Suppose that Assumption~\ref{assumpt:DMLboth} in the appendix in Section~\ref{sect:proofsOfDML}  holds 
(same as in Theorem~\ref{thm:asymptNormal}).
Let $\{a_{j}\}_{j\ge 1}$ be a sequence of deterministic, non-negative real numbers that diverges to $\infty$ as $\NN\rightarrow\infty$. Furthermore, consider $\hgammap=a_{\NN}\hgamma$ as above. Then, we have
\begin{displaymath}
	\sqrt{\NN}\hat\sigma^{-1}(\hgammap)(\hbhgp-\bhgp) 
	= \sqrt{\NN} \sigma^{-1}(\hbetaN-\betazero) + o_{\PP}(1)
\end{displaymath}
uniformly over laws $\PP$ of $S=(A,W,X,Y)$, where $\hsigma(\cdot)$ ist the estimator from Theorem~\ref{thm:estSDgamma} in the appendix, which consistently estimates $\sigma(\cdot)$ from~\ref{thm:asymptNormalgamma}. 
\end{theorem}

Particularly, $\hbhgp$ and $\hbetaN$ are asymptotically equivalent. But $\hbhgp$ may  exhibit substantially better finite sample properties as we demonstrate in the subsequent section. 
Because $\hbhgp$ and $\hbetaN$ are asymptotically equivalent, 
the same result also holds for the selection estimator \regsDML.

The proof of Theorem~\ref{thm:gammaHatNormal} does not depend on the precise construction of $\hgammap$ and only uses that the random regularization parameter is of stochastic order larger than $\sqrt{\NN}$.
Thus,  Theorem~\ref{thm:gammaHatNormal} remains valid if the regularization parameter comes from k-class estimaton and is of the required stochastic order.  
The same stochastic order is also required to show that k-class estimators are asymptotically Gaussian~\citep{Nagar1959, Mariano2001}. 
\\

The $\KK$ sample splits are random. To reduce the effect of this randomness, we repeat the overall procedure $\Salg$ times and assemble the results as suggested in~\citet{Chernozhukov2018}. The assembled parameter estimate is given by the median of the individual parameter estimates; see Steps~\ref{algo:step9} and~\ref{algo:step10} of Algorithm~\ref{algo:Summary}.
The assembled variance  estimate is given by adding a correction term to the individual variances and subsequently taking the median of these corrected terms. The correction term 
measures the variability due to sample spitting across $s \in \indset{\Salg}$. 

It is possible that the assembled variance of \regDML\ is larger than the assembled variance of DML. In such a case, we do not use the \regDML\ estimator and select the DML estimator instead to ensure that the final estimator of $\betazero$ does not experience a larger estimated variance than DML. This is the  \regsDML\ scheme.
A summary of this procedure is given in Algorithm~\ref{algo:Summary}. 
\\

\begin{algorithm}[h!]
	\SetKwInOut{Input}{Input}
   	\SetKwInOut{Output}{Output}

 \Input{$\NN$ \iid\ realizations from the SEM~\eqref{eq:SEM}, a natural number $\Salg$,  a regularization parameter grid $\{\gamma_i\}_{i\in\indset{M}}$ for some natural number $M$, a non-negative diverging sequence $\{a_n\}_{n\ge 1}$.}
 \Output{An estimator of $\betazero$ in~\eqref{eq:SEM} together with its estimated asymptotic variance. }
 
 \For{$s\in\indset{\Salg}$}
 {
 Compute $\hbetaNi=\hbetaN$ and $\hsigmai^2=\hsigma^2$.
 
 Compute $\hat{b}_{s}^{\gamma_i}=\hat{b}^{\gamma_i}$ and $\hsigmai^2(\gamma_i)=\hsigma^2(\gamma_i)$ for $i\in\indset{M}$. 

Choose $\hgammai=\argmin_{\gamma\in\{\gamma_i\}_{i\in\indset{M}}} \big(\frac{1}{\NN}\hsigmai^2(\gamma) + \normone{\hbgi - \hbetaNi}^2\big)$ and let $\hgammapi=a_{\NN}\hgammai$.

Compute $\hbhgpi=\hat{b}^{\hgammapi}$ and $\hsigmai^2(\hgammapi)=\hsigma^2(\hgammapi)$. 
 }
 
 Compute $\hbetaMed = \mathrm{median}_{s\in\indset{\Salg}}(\hbetaNi)$.
 
 Compute $\hbMedReg = \mathrm{median}_{s\in\indset{\Salg}}(\hbhgpi)$.
 
 Compute $\hsigmaMed=\mathrm{median}_{s\in\indset{\Salg}}\big(\hsigmai^2 + (\hbetaNi-\hbetaMed)^2\big)$. \label{algo:step9}
 
 Compute $\hsigmaMedReg=\mathrm{median}_{s\in\indset{\Salg}}\big(\hsigmai^2(\hgammapi) + (\hbhgpi-\hbMedReg)^2\big)$.\label{algo:step10}

\eIf{$\hsigmaMedReg<\hsigmaMed$
      }{Take the parameter estimate $\hbMedReg$ together with its associated estimated asymptotic variance $\frac{1}{\NN}\hsigmaMedReg$.
      }{Take the parameter estimate $\hbetaMed$ together with its associated estimated asymptotic variance $\frac{1}{\NN}\hsigmaMed$.
      }

 \caption{\regsDML\ in a PLM with confounding variables.}\label{algo:Summary}
\end{algorithm}

\section{Numerical Experiments}\label{sect:simulation}

This section illustrates the performance of the DML, \regDML, and \regsDML\ estimators in a  simulation study and for an empirical dataset. 
Our implementation is available in the \textsf{R}-package \texttt{dmlalg}~\citep{dmlalg}. 
We employ the \DMLtwo\ method  
and $\KK=2$ and $\Salg=100$ in Algorithm~\ref{algo:Summary}. 
Furthermore, we compare our estimation schemes with 
the following three k-class estimators: LIML, Fuller(1), and Fuller(4). 
On each of the $\KK$ sample splits, we compute the regularization parameter of the respective k-class estimation procedure and average them. Then, we compute the corresponding $\gamma$-value and proceed as for the other regularized estimators according to Algorithm~\ref{algo:Summary}. 

The first example in Section~\ref{sect:exampleForest}  considers an overidentified model in which the dimension of $A$ is
larger than the dimension of $X$. 
The conditional expectations acting as nuisance parameters are estimated with random forests. 
The second example in Section~\ref{sect:empirical} considers justidentified real-world data. The conditional expectations are also estimated with random forests. 

An example where the conditional expectations are estimated with splines is given in Section~\ref{sect:ourContribution}. Additional empirical results are provided in the appendix in Sections~\ref{sect:additioinalSimulation}, \ref{sect:strong-weak}, and~\ref{sect:counterexamples}.
The latter section considers examples where DML, \regDML, and \regsDML\ do not work well in finite sample situations: we follow the NCP (No Cherry Picking) guideline~\citep{Buhlmann2018} to possibly enhance further insights into the finite sample behavior.
Section~\ref{sect:strong-weak} in the appendix
 presents examples where the link $A\rightarrow X$ is weak and examples illustrating the bias-variance tradeoff of the respective estimated quantities as a function of $\gamma$.

\subsection{Simulation Example with Random Forests}\label{sect:exampleForest}

We generate data from the SEM in Figure~\ref{fig:simulation2}. 
This SEM 
satisfies the identifiability condition~\eqref{eq:identificationCondition} because $A_1$ and $A_2$ are independent of $H$ given $W_1$ and $W_2$; a proof is given in the appendix in Section~\ref{appendix:proofRandomForestSEM}. 
The model is overidentified because the dimension of $A=(A_1,A_2)$ is larger than the dimension of $X$. The variable $A_1$ directly influences $A_2$ that in turn directly affects $W_1$. Both $W_1$ and $W_2$ directly influence $H$. Both $A_1$ and $A_2$ directly influence $X$. The variable $A_1$ is a source node. 
\\

\begin{figure}[h!]
	\centering
	\caption[]{\label{fig:simulation2}An SEM and its associated causal graph.}
	\begin{tabular}{cc}
	\begin{tabular}{l}
	$\displaystyle 
	\begin{array}{r}
		(\eps_{A_1}, \eps_{A_2,} \eps_{W_1}, \eps_{W_2},\eps_H, \eps_X, \eps_Y)\sim\mathcal{N}_7(\bo,\one)\\	
	\end{array}$
	\\
	$\displaystyle
	\begin{array}{lcl} 
		A_1 &\leftarrow& \one_{\{\eps_{A_1}\le 0\}}\\
		A_2 &\leftarrow& -4A_1+\eps_{A_2}\\
		W_1 &\leftarrow& 2A_2+\eps_{W_1}\\
		W_2&\leftarrow& \eps_{W_2}\\
		H &\leftarrow& 2\one_{\{\sin(\pi W_1)\cdot\tanh(W_2) \ge 0\}}+ \eps_H\\
		X&\leftarrow& 1.5A_1-0.5A_2+\tanh(H) \\
		&&\quad - 2\one_{\{W_1\ge 0\}}\one_{\{W_2\le 0\}} + \eps_X\\
		Y &\leftarrow&  X  + \one_{\{W_2\le 0\}} + \sin(\pi H)+ \eps_Y
	\end{array}$ 
	\end{tabular}
	& 
	\begin{tabular}{c}
          \includegraphics[width=0.35\textwidth]{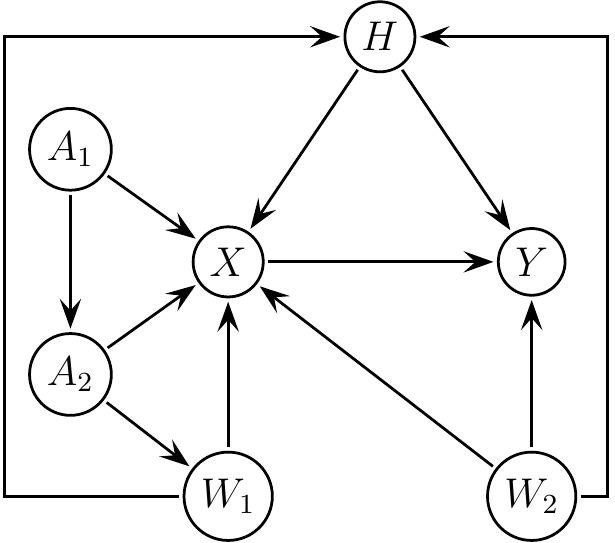}
	\end{tabular}
\end{tabular}
\end{figure} 

We simulate $M=1000$ datasets each from the SEM in Figure~\ref{fig:simulation2} for a range of sample sizes. 
For every dataset, we compute a parameter estimate and an associated confidence interval with DML, \regDML, and \regsDML. We choose $\KK=2$ and $\Salg=100$ in Algorithm~\ref{algo:Summary} and estimate the conditional expectations with random forests consisting of $500$ trees that have a minimal node size of $5$. 
\\

Figure~\ref{fig:COVERsimulation2} illustrates our findings. It gives  the coverage, power, and relative length of the $95\%$ confidence intervals for a range of sample sizes $\NN$ of the three methods. The blue and green curves correspond to \regDML\ and \regsDML, respectively.
If the blue curve is not visible in Figure~\ref{fig:COVERsimulation2}, it coincides with the green one. The two regularization methods perform similarly because
regularization can considerably improve DML. 
The red curves correspond to DML. 
If the red curves are not visible, they coincide with LIML, whose results are displayed in orange. 
The Fuller(1) and Fuller(4) estimators correspond to purple and cyan, respectively. 

The top left plot in Figure~\ref{fig:COVERsimulation2} displays the coverages as interconnected dots. The dashed lines represent $95\%$ confidence regions of the coverages. These confidence regions are computed with respect to uncertainties in the $M$ simulation runs.
No coverage region falls below the nominal $95\%$ level that is marked by the gray line. 

The bottom eft plot in Figure~\ref{fig:COVERsimulation2} shows that the power of DML, LIML, and Fuller(1) is lower for small sample sizes and increases gradually. 
The power of the other regularization methods remains approximately $1$.
The dashed lines represent $95\%$ confidence regions that are computed with respect to uncertainties in the $M$ simulation runs.

The right plot in Figure~\ref{fig:COVERsimulation2} displays boxplots of the scaled lengths of the confidence intervals. 
For each $\NN$, the  confidence interval lengths of all methods are divided by the median confidence interval lengths of DML. 
The length of the \regsDML\ confidence intervals is around $50\%-80\%$ the length of DML's.
Nevertheless, the coverage of \regsDML\ remains around $95\%$. 
The LIML, Fuller(1), and Fuller(4) confidence intervals are considerably longer than \regsDML's. 
Although the confidence intervals  of \regsDML\ are the shortest of all considered methods, its coverage remains valid. 

\begin{figure}[h!]
	\centering
	\caption[]{\label{fig:COVERsimulation2} 
	The results come from $M=1000$ simulation runs each from the SEM in Figure~\ref{fig:simulation2} for a range of sample sizes $\NN$ and with $\KK=2$ and $\Salg=100$ in Algorithm~\ref{algo:Summary}. 
	The nuisance functions are estimated with random forests. 
	The figure displays the coverage of two-sided confidence intervals for $\betazero$, power for two-sided testing of the 
	hypothesis $H_0:\ \betazero = 0$, and scaled lengths of two-sided confidence intervals of DML (red),  \regDML\ (blue),  \regsDML\  (green), 
	LIML (orange), Fuller(1) (purple), and Fuller(4) (cyan), 
	where all results are at level $95\%$.
	At each $\NN$, the lengths of the confidence intervals are scaled with the median length from DML. 
	The shaded regions in the coverage and the power plots represent $95\%$ confidence bands with respect to the $M$ simulation runs.
	The blue and green lines as 
	well as the red and orange ones 
	are indistinguishable in the left panel.
	}
	\includegraphics[width=\textwidth]{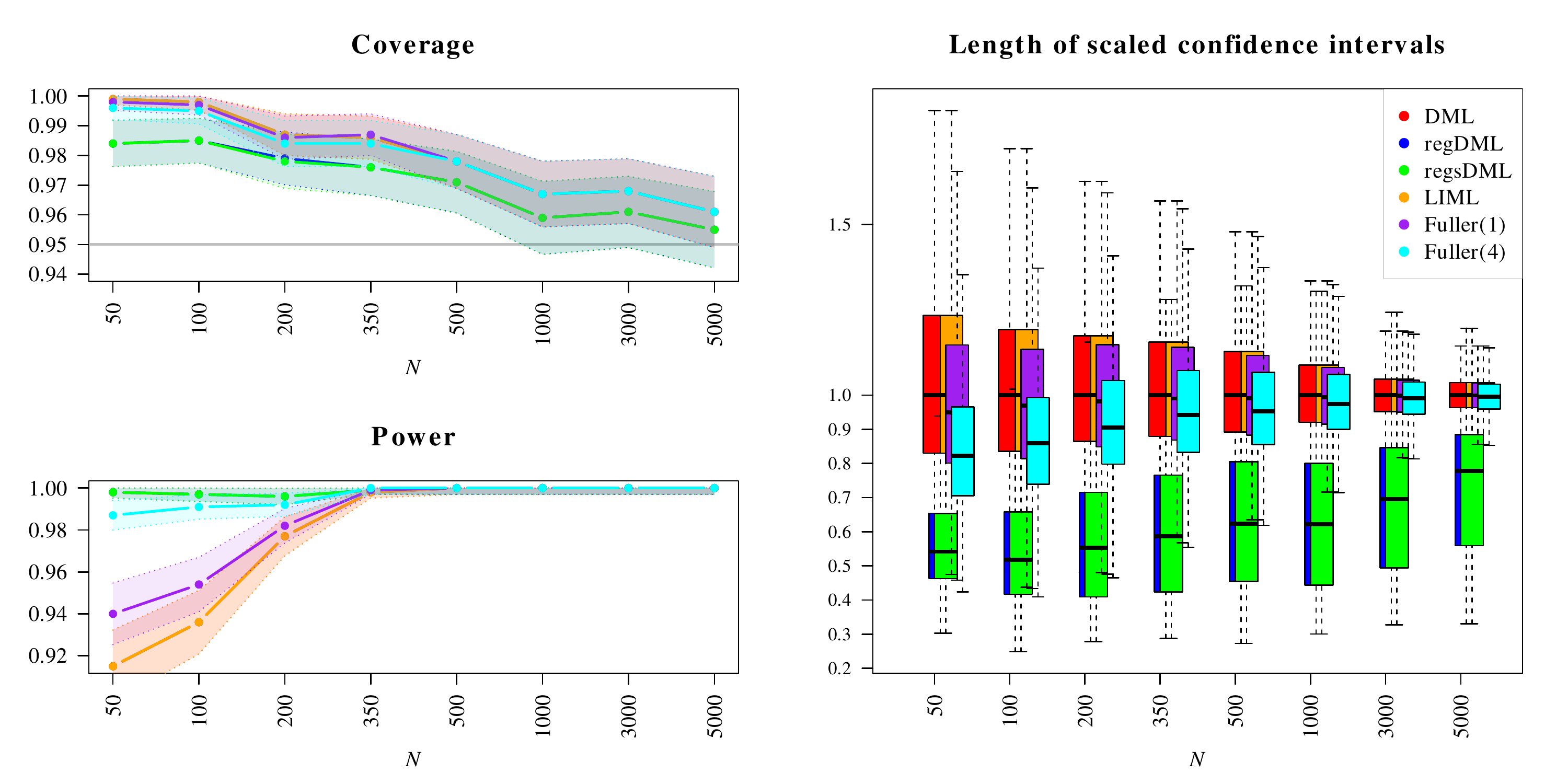}
\end{figure}

Simulation results with $\betazero=0$ in the SEM in Figure~\ref{fig:simulation2} are presented in Figure~\ref{fig:COVERsimulation2Beta0} in the appendix in Section~\ref{sect:additioinalSimulation}.

\subsection{Real Data Example}\label{sect:empirical}

We apply the DML and \regsDML\ methods to a real dataset. 
We estimate the linear effect $\betazero$ of institutions on economic performance following the work of~\citet{Acemoglu2001} and~\citet{Chernozhukov2018}. Countries with better institutions achieve a greater level of income per capita, and wealthy economies can afford better institutions. This may cause simultaneity. To overcome it, mortality rates of the first European settlers in colonies are considered as a source of exogenous variation in institutions. 
For further details, we refer to~\citet{Acemoglu2001} and~\cite{Chernozhukov2018}. 
The data is available in the \textsf{R}-package hdm~\citep{hdm2016} and is called AJR. In our notation, the response $Y$ is the GDP, the covariate $X$ the average protection against expropriation risk, the variable $A$ the logarithm of settler mortality, and the covariate $W$ consists of the latitude, the squared latitude, and the binary factors Africa, Asia, North America, and South America. 
That is, we adjust nonparametrically for the latitude and geographic information. 
\\

We choose $\KK=2$ and $\Salg=100$ in Algorithm~\ref{algo:Summary} and compute the conditional expectations with random forests with $1000$ trees that have a minimal node size of $5$. 
The estimation results are displayed in Table~\ref{tab:empirical}. 
This table gives the estimated linear coefficient, its standard deviation, and a confidence interval for $\betazero$ for DML and \regsDML. 
The coefficient estimate of DML  is not significant because the respective confidence interval includes $0$. The  \regsDML\ estimate  is significant because it has a smaller standard deviation than the DML estimate. Note that the coefficient estimate of \regsDML\ falls within the DML confidence interval. 
\\

\begingroup
\begin{table}
\centering
\setlength{\tabcolsep}{10pt} 
\renewcommand{\arraystretch}{1.25} 
\begin{tabular}{| l | c | c | c | }
\hline
& Estimate of $\betazero$ & Standard error & Confidence interval for $\betazero$\\
\hline
DML & $0.739$ & $0.459$ & $[-0.161,  1.639]$\\
\hline
\regsDML & $0.688$ & $0.229$ & $[0.239, 1.136]$\\
\hline
\end{tabular}
\caption{\label{tab:empirical}Coefficient estimate, its standard error, and a confidence interval with \regsDML\ and DML on the AJR dataset, where  $\KK=2$ and $\Salg=100$ in Algorithm~\ref{algo:Summary}, and where the conditional expectations are estimated with random forests consisting of $1000$ trees that have a minimal node size of $5$.}
\end{table}
\endgroup

The AJR dataset has also been analyzed in~\citet{Chernozhukov2018}. They also estimate  conditional expectations with random forests consisting of $1000$ trees that have a minimal node size of $5$ but implicitly assume 
an additional homoscedasticity condition for the errors $\Ry-\Rx^T\betazero$; see~\citet{gitChernozhukov}.
Such a homoscedastic error assumption is questionable though.
Their procedure leads to a smaller estimate of the standard deviation of DML than what we obtain.

\section{Conclusion}\label{sect:conclusion}

We  extended and regularized double machine learning (DML) in potentially overidentified partially linear models (PLMs) with hidden variables. 
Our goal was to estimate the linear coefficient $\betazero$ of the PLM.
Hidden variables confound the observables, which can cause endogeneity.
For instance, a clinical study may experience an endogeneity issue  if a treatment is not randomly assigned and subjects receiving different treatments differ in other ways than the treatment~\citep{Okui2012}. 
In such situations, employing estimation methods that do not account for endogeneity lead to biased estimators~\citep{Fuller1987}. 
\\

Our contribution was twofold. First, we formulated the 
PLM as a structural equation model (SEM) and imposed an identifiability condition on it to recover the population parameter $\betazero$. 
We estimated $\betazero$ using DML
similarly to~\citet{Chernozhukov2018}. However, our setting is more general than the one considered in~\citet{Chernozhukov2018}
because we allow the predictors to be multivariate, and we impose a moment condition instead of restricting conditional moments. 
The DML estimation procedure allows biased estimators of additional nuisance functions to be plugged into the estimating equation of $\betazero$.
The resulting estimator of $\betazero$ is asymptotically Gaussian and converges at the parametric rate of $\NN^{-\frac{1}{2}}$. 
However, DML has a two-stage least squares (TSLS) interpretation and may therefore 
lead to overly wide 
confidence intervals. 
 
Second, we proposed a regularization-only DML scheme, \regDML, and a regularization-selection DML scheme, \regsDML.
The latter has shorter confidence intervals by construction because it 
selects between DML and \regDML\ depending on whose estimated standard deviation is smaller. 
Although \regsDML\ and plain DML are asymptotically equivalent,
 \regsDML\ leads to drastically shorter confidence intervals  for finite sample sizes. Nevertheless, coverage guarantees for $\betazero$ remain.
The \regDML\ estimator is similar to k-class estimation~\citep{Theil1961} and anchor regression~\citep{Rothenhausler2018, Buehlmann2018, Jakobsen2020} but allows potentially complex partially linear models and chooses a data-driven regularization parameter. 
\\

Empirical examples demonstrated our methodological and theoretical developments. The results showed that \regsDML\ is a highly effective method to increase the power and sharpness of statistical inference. 
The DML estimator has a TSLS interpretation. Therefore, if the confounding is strong, the DML estimator leads to overly wide confidence intervals and can be substantially biased. In such a case, \regsDML\ drastically reduces the width of the confidence intervals but may inherit additional bias from DML. This effect can be particularly pronounced for small sample sizes.
Section~\ref{sect:counterexamples} in the appendix presents 
examples with strong and reduced confounding and demonstrates the coverage behavior of DML and \regsDML. 
Section~\ref{sect:strong-weak} in the appendix analyzes the performance of our methods if the strength of the link $A\rightarrow X$ varies, and investigates the bias-variance tradeoff of the respective estimated quantities for different values of the regularization parameter. 

Although a wide range of machine learners can be employed to estimate the nuisance functions, we observed that additive splines can estimate more precise results than random forests if the underlying structure is additive in good approximation. 
This effect is particularly pronounced if the sample size is small. 
If such a finding is to be expected, it may be worthwhile to use structured models rather than ``general'' machine learning algorithms, especially with small or moderate sample size. 
Our \regsDML\ methodology can be used with the implementation that
is available in the \textsf{R}-package \texttt{dmlalg}~\citep{dmlalg}.

\section*{Acknowledgements}

We thank Matthias L\"{o}ffler for constructive comments. 
\\

This project has received funding from the European Research Council (ERC) under the European Union’s Horizon 2020 research and innovation programme (grant agreement No. 786461).

\phantomsection
\addcontentsline{toc}{section}{References}
\bibliography{references}

\begin{thebibliography}{99}
\providecommand{\natexlab}[1]{#1}
\providecommand{\url}[1]{\texttt{#1}}
\expandafter\ifx\csname urlstyle\endcsname\relax
  \providecommand{\doi}[1]{doi: #1}\else
  \providecommand{\doi}{doi: \begingroup \urlstyle{rm}\Url}\fi

\bibitem[Acemoglu et~al.(2001)Acemoglu, Johnson, and Robinson]{Acemoglu2001}
D.~Acemoglu, S.~Johnson, and J.~A. Robinson.
\newblock The colonial origins of comparative development: An empirical
  investigation.
\newblock \emph{The American Economic Review}, 91\penalty0 (5):\penalty0
  1369--1401, 2001.

\bibitem[Ai and Chen(2003)]{Ai2003}
C.~Ai and X.~Chen.
\newblock Efficient estimation of models with conditional moment restrictions
  containing unknown functions.
\newblock \emph{Econometrica}, 71\penalty0 (6):\penalty0 1795--1843, 2003.

\bibitem[Amemiya(1974)]{Amemiya1974}
T.~Amemiya.
\newblock The nonlinear two-stage least-squares estimator.
\newblock \emph{Journal of Econometrics}, 2\penalty0 (2):\penalty0 105--110,
  1974.

\bibitem[Amemiya(1985)]{Amemiya1985}
T.~Amemiya.
\newblock \emph{Advanced Econometrics}.
\newblock Harvard University Press, Cambridge, Massachusetts, 1985.

\bibitem[Anderson et~al.(2010)Anderson, Kunitomo, and
  Matsushita]{Anderson-Kunitomo-Matsushita2010}
T.~Anderson, N.~Kunitomo, and Y.~Matsushita.
\newblock On the asymptotic optimality of the liml estimator with possibly many
  instruments.
\newblock \emph{Journal of Econometrics}, 157\penalty0 (2):\penalty0 191--204,
  2010.

\bibitem[Anderson(1983)]{Anderson1983}
T.~W. Anderson.
\newblock Some recent developments on the distributions of single-equation
  estimators.
\newblock In A.~Deaton, D.~McFadden, and H.~Sonnenschein, editors,
  \emph{Advances in econometrics}, Econometric Society Monographs in
  Quantitative Economics, chapter~4, pages 109--122. Cambridge University
  Press, Cambridge, 1983.

\bibitem[Anderson(2005)]{Anderson2005}
T.~W. Anderson.
\newblock Origins of the limited information maximum likelihood and two-stage
  least squares estimators.
\newblock \emph{Journal of Econometrics}, 127\penalty0 (1):\penalty0 1--16,
  2005.

\bibitem[Anderson and Rubin(1949)]{Anderson-Rubin1949}
T.~W. Anderson and H.~Rubin.
\newblock Estimation of the parameters of a single equation in a complete
  system of stochastic equations.
\newblock \emph{The Annals of Mathematical Statistics}, 20\penalty0
  (1):\penalty0 46--63, 1949.

\bibitem[Anderson and Sawa(1979)]{Anderson-Sawa1979}
T.~W. Anderson and T.~Sawa.
\newblock Evaluation of the distribution function of the two-stage least
  squares estimate.
\newblock \emph{Econometrica}, 47\penalty0 (1):\penalty0 163--182, 1979.

\bibitem[Anderson et~al.(1982)Anderson, Kunitomo, and
  Sawa]{Anderson-Kunitomo-Sawa1982}
T.~W. Anderson, N.~Kunitomo, and T.~Sawa.
\newblock Evaluation of the distribution function of the limited information
  maximum likelihood estimator.
\newblock \emph{Econometrica}, 50\penalty0 (4):\penalty0 1009--1027, 1982.

\bibitem[Anderson et~al.(1986)Anderson, Kunitomo, and Morimune]{Anderson1986}
T.~W. Anderson, N.~Kunitomo, and K.~Morimune.
\newblock Comparing single-equation estimators in a simultaneous equation
  system.
\newblock \emph{Econometric Theory}, 2\penalty0 (1):\penalty0 1--32, 1986.

\bibitem[Andrews et~al.(2019)Andrews, Stock, and Sun]{AndrewsForthcoming}
I.~Andrews, J.~Stock, and L.~Sun.
\newblock Weak instruments in {IV} regression: Theory and practice.
\newblock \emph{Annual Review of Economics}, 11:\penalty0 727--753, 2019.

\bibitem[Angrist et~al.(1996)Angrist, Imbens, and Rubin]{Angrist1996}
J.~D. Angrist, G.~W. Imbens, and D.~B. Rubin.
\newblock Identification of causal effects using instrumental variables.
\newblock \emph{Journal of the American Statistical Association}, 91\penalty0
  (434):\penalty0 444--455, 1996.

\bibitem[Athey et~al.(2019)Athey, Tibshirani, and
  Wager]{Athey-Tibshirani-Wager2018}
S.~Athey, J.~Tibshirani, and S.~Wager.
\newblock Generalized random forests.
\newblock \emph{The Annals of Statistics}, 47\penalty0 (2):\penalty0
  1148--1178, 2019.

\bibitem[Bang and Robins(2005)]{Robins2005}
H.~Bang and J.~M. Robins.
\newblock Doubly robust estimation in missing data and causal inference models.
\newblock \emph{Biometrics}, 61\penalty0 (4):\penalty0 962--972, 2005.

\bibitem[Basmann(1957)]{Basmann1957}
R.~L. Basmann.
\newblock A generalized classical method of linear estimation of coefficients
  in a structural equation.
\newblock \emph{Econometrica}, 25\penalty0 (1):\penalty0 77--83, 1957.

\bibitem[Belloni and Chernozhukov(2013)]{Belloni-Chernozhukov2013}
A.~Belloni and V.~Chernozhukov.
\newblock Least squares after model selection in high-dimensional sparse
  models.
\newblock \emph{Bernoulli}, 19\penalty0 (2):\penalty0 521--547, 2013.

\bibitem[Berndt et~al.(1974)Berndt, Hall, Hall, and Hausman]{Berndt1974}
E.~R. Berndt, B.~H. Hall, R.~E. Hall, and J.~A. Hausman.
\newblock Estimation and inference in nonlinear structural models.
\newblock \emph{Annals of Economic and Social Measurement}, 3\penalty0
  (4):\penalty0 653--665, 1974.

\bibitem[Bickel(1982)]{Bickel1982}
P.~J. Bickel.
\newblock On adaptive estimation.
\newblock \emph{The Annals of Statistics}, 10\penalty0 (3):\penalty0 647--671,
  1982.

\bibitem[Bickel et~al.(2009)Bickel, Ritov, and Tsybakov]{Bickel2009}
P.~J. Bickel, Y.~Ritov, and A.~B. Tsybakov.
\newblock Simultaneous analysis of lasso and dantzig selector.
\newblock \emph{The Annals of Statistics}, 37\penalty0 (4):\penalty0
  1705--1732, 2009.

\bibitem[Bound et~al.(1995)Bound, Jaeger, and Baker]{Bound1995}
J.~Bound, D.~A. Jaeger, and R.~M. Baker.
\newblock Problems with instrumental variables estimation when the correlation
  between the instruments and the endogenous explanatory variable is weak.
\newblock \emph{Journal of the American Statistical Association}, 90\penalty0
  (430):\penalty0 443--450, 1995.

\bibitem[Bowden and Turkington(1985)]{Bowden1985}
R.~J. Bowden and D.~A. Turkington.
\newblock \emph{Instrumental variables}.
\newblock Econometric Society Monographs. Cambridge University Press,
  Cambridge, 1985.

\bibitem[B\"{u}hlmann(2020)]{Buehlmann2018}
P.~B\"{u}hlmann.
\newblock Invariance, causality and robustness.
\newblock \emph{Statistical Science}, 35\penalty0 (3):\penalty0 404--426, 2020.

\bibitem[B\"{u}hlmann and van~de Geer(2011)]{Buehlmann2011}
P.~B\"{u}hlmann and S.~van~de Geer.
\newblock \emph{Statistics for High-Dimensional Data: Methods, Theory and
  Applications}.
\newblock Springer Series in Statistics. Springer, Heidelberg, 2011.

\bibitem[B\"{u}hlmann and {van de Geer}(2018)]{Buhlmann2018}
P.~B\"{u}hlmann and S.~{van de Geer}.
\newblock Statistics for big data: A perspective.
\newblock \emph{Statistics \& Probability Letters}, 136:\penalty0 37--41, 2018.

\bibitem[Candes and Tao(2007)]{Candes-Tao2007}
E.~Candes and T.~Tao.
\newblock The dantzig selector: Statistical estimation when $p$ is much larger
  than $n$.
\newblock \emph{The Annals of Statistics}, 35\penalty0 (6):\penalty0
  2313--2351, 2007.

\bibitem[Chen et~al.(2016)Chen, Liang, and Zhou]{Chen2016}
B.~Chen, H.~Liang, and Y.~Zhou.
\newblock {GMM} estimation in partial linear models with endogenous covariates
  causing an over-identified problem.
\newblock \emph{Communications in Statistics - Theory and Methods}, 45\penalty0
  (11):\penalty0 3168--3184, 2016.

\bibitem[Chen et~al.(2021)Chen, Huang, and Tien]{Chen2019}
J.~Chen, C.-H. Huang, and J.-J. Tien.
\newblock Debiased/double machine learning for instrumental variable quantile
  regressions.
\newblock \emph{Econometrics}, 9\penalty0 (2), 2021.

\bibitem[Chernozhukov et~al.(2016)Chernozhukov, Hansen, and Spindler]{hdm2016}
V.~Chernozhukov, C.~Hansen, and M.~Spindler.
\newblock {hdm}: High-dimensional metrics.
\newblock \emph{\textsf{R} Journal}, 8\penalty0 (2):\penalty0 185--199, 2016.

\bibitem[Chernozhukov et~al.(2017)Chernozhukov, Chetverikov, Demirer, Duflo,
  Newey, and Robins]{gitChernozhukov}
V.~Chernozhukov, D.~Chetverikov, M.~Demirer, E.~Duflo, W.~Newey, and J.~Robins.
\newblock Repo for the paper ``double/debiased machine learning for treatment
  and structural parameters''.
\newblock \url{https://github.com/VC2015/DMLonGitHub}, 2017.
\newblock Accessed: September 23, 2020.

\bibitem[Chernozhukov et~al.(2018)Chernozhukov, Chetverikov, Demirer, Duflo,
  Hansen, Newey, and Robins]{Chernozhukov2018}
V.~Chernozhukov, D.~Chetverikov, M.~Demirer, E.~Duflo, C.~Hansen, W.~Newey, and
  J.~Robins.
\newblock Double/debiased machine learning for treatment and structural
  parameters.
\newblock \emph{The Econometrics Journal}, 21\penalty0 (1):\penalty0 C1--C68,
  2018.

\bibitem[Chiang et~al.(2021)Chiang, Kato, Ma, and Sasaki]{Chiang2020}
H.~D. Chiang, K.~Kato, Y.~Ma, and Y.~Sasaki.
\newblock Multiway cluster robust double/debiased machine learning.
\newblock \emph{Journal of Business \& Economic Statistics}, 0\penalty0
  (0):\penalty0 1--11, 2021.

\bibitem[Colangelo and Lee(2020)]{Colangelo2020}
K.~Colangelo and Y.-Y. Lee.
\newblock Double debiased machine learning nonparametric inference with
  continuous treatments, 2020.
\newblock Preprint arXiv:2004.03036.

\bibitem[Cragg(1967)]{Cragg1967}
J.~G. Cragg.
\newblock On the relative small-sample properties of several
  structural-equation estimators.
\newblock \emph{Econometrica}, 35\penalty0 (1):\penalty0 89--110, 1967.

\bibitem[Crown et~al.(2011)Crown, Henk, and Vanness]{Crown2011}
W.~H. Crown, H.~J. Henk, and D.~J. Vanness.
\newblock Some cautions on the use of instrumental variables estimators in
  outcomes research: How bias in instrumental variables estimators is affected
  by instrument strength, instrument contamination, and sample size.
\newblock \emph{Value in Health}, 14\penalty0 (8):\penalty0 1078--1084, 2011.

\bibitem[Cui and {Tchetgen Tchetgen}(2020)]{Tchetgen2020}
Y.~Cui and E.~{Tchetgen Tchetgen}.
\newblock Selective machine learning of doubly robust functionals, 2020.
\newblock Preprint arXiv:1911.02029.

\bibitem[DasGupta(2008)]{DasGupta2008}
A.~DasGupta.
\newblock \emph{Asymptotic theory of statistics and probability}.
\newblock Springer Texts in Statistics. Springer, New York, 2008.

\bibitem[DiazOrdaz et~al.(2019)DiazOrdaz, Daniel, and Kreif]{DiazOrdaz2019}
K.~DiazOrdaz, R.~Daniel, and N.~Kreif.
\newblock Data-adaptive doubly robust instrumental variable methods for
  treatment effect heterogeneity, 2019.
\newblock Preprint arXiv:1802.02821.

\bibitem[Durrett(2010)]{Durrett2010}
R.~Durrett.
\newblock \emph{Probability: Theory and examples}.
\newblock Cambridge Series in Statistical and Probabilistic Mathematics.
  Cambridge University Press, Cambridge, 4 edition, 2010.

\bibitem[Emmenegger(2021)]{dmlalg}
C.~Emmenegger.
\newblock \emph{{dmlalg}: Double machine learning algorithms}, 2021.
\newblock URL \url{https://cran.r-project.org/web/packages/dmlalg/index.html}.
\newblock \textsf{R}-package available on CRAN.

\bibitem[Farbmacher et~al.(2020)Farbmacher, Huber, Laff\'{e}rs, Langen, and
  Spindler]{Farbmacher2020}
H.~Farbmacher, M.~Huber, L.~Laff\'{e}rs, H.~Langen, and M.~Spindler.
\newblock Causal mediation analysis with double machine learning, 2020.
\newblock Preprint arXiv:2002.12710.

\bibitem[Florens et~al.(2012)Florens, Johannes, and Van~Bellegem]{Florens2012}
J.-P. Florens, J.~Johannes, and S.~Van~Bellegem.
\newblock Instrumental regression in partially linear models.
\newblock \emph{The Econometrics Journal}, 15\penalty0 (2):\penalty0 304--324,
  2012.

\bibitem[Fuller(1977)]{Fuller1977}
W.~A. Fuller.
\newblock Some properties of a modification of the limited information
  estimator.
\newblock \emph{Econometrica}, 45\penalty0 (4):\penalty0 939--53, 1977.

\bibitem[Fuller(1987)]{Fuller1987}
W.~A. Fuller.
\newblock \emph{Measurement error models}.
\newblock Wiley series in probability and mathematical statistics. John Wiley
  \& Sons, New York, 1987.

\bibitem[Hahn et~al.(2004)Hahn, Hausman, and Kuersteiner]{Hahn2004}
J.~Hahn, J.~Hausman, and G.~Kuersteiner.
\newblock Estimation with weak instruments: Accuracy of higher-order bias and
  mse approximations.
\newblock \emph{The Econometrics Journal}, 7\penalty0 (1):\penalty0 272--306,
  2004.

\bibitem[Hansen(1982)]{Hasen1982}
L.~P. Hansen.
\newblock Large sample properties of generalized method of moments estimators.
\newblock \emph{Econometrica}, 50\penalty0 (4):\penalty0 1029--1054, 1982.

\bibitem[Hansen(1985)]{Hansen1985}
L.~P. Hansen.
\newblock A method for calculating bounds on the asymptotic covariance matrices
  of generalized method of moments estimators.
\newblock \emph{Journal of Econometrics}, 30\penalty0 (1):\penalty0 203--238,
  1985.

\bibitem[H\"{a}rdle et~al.(2000)H\"{a}rdle, Liang, and Gao]{Haerdle2000}
W.~H\"{a}rdle, H.~Liang, and J.~Gao.
\newblock \emph{Partially linear models}.
\newblock Contributions to Statistics. Springer, Berlin Heidelberg, 2000.

\bibitem[H\"{a}rdle et~al.(2004)H\"{a}rdle, M\"{u}ller, Sperlich, and
  Werwatz]{Haerdle2004}
W.~H\"{a}rdle, M.~M\"{u}ller, S.~Sperlich, and A.~Werwatz.
\newblock \emph{Nonparametric and semiparametric models}.
\newblock Springer series in statistics. Springer, Berlin, 2004.

\bibitem[Henderson and Searle(1981)]{Henderson1981}
H.~V. Henderson and S.~R. Searle.
\newblock On deriving the inverse of a sum of matrices.
\newblock \emph{SIAM Review}, 23\penalty0 (1):\penalty0 53--60, 1981.

\bibitem[Hill et~al.(2011)Hill, Griffiths, and Lim]{Hill2011}
R.~C. Hill, W.~E. Griffiths, and G.~C. Lim.
\newblock \emph{Principles of econometrics}.
\newblock Wiley, Hoboken, New Jersey, 4 edition, 2011.

\bibitem[Hillier and Skeels(1993)]{Hillier1993}
G.~H. Hillier and C.~L. Skeels.
\newblock Some further exact results for structural equation estimators.
\newblock In P.~C.~B. Phillips, editor, \emph{Models, Methods and Applications
  of Econometrics: essays in Honor of A. R. Bergstroms}, pages 117--139.
  Blackwell, Cambridge, Massachusetts, 1993.

\bibitem[Horowitz(2011)]{Horowitz2011}
J.~L. Horowitz.
\newblock Applied nonparametric instrumental variables estimation.
\newblock \emph{Econometrica}, 79\penalty0 (2):\penalty0 347--394, 2011.

\bibitem[Jakobsen and Peters(2020)]{Jakobsen2020}
M.~E. Jakobsen and J.~Peters.
\newblock Distributional robustness of {K}-class estimators and the {PULSE},
  2020.
\newblock Preprint arXiv:2005.03353.

\bibitem[Knaus(2020)]{Knaus2020}
M.~C. Knaus.
\newblock Double machine learning based program evaluation under
  unconfoundedness, 2020.
\newblock Preprint arXiv:2003.03191.

\bibitem[Koltchinskii and Yuan(2010)]{Koltchinskii-Yuan2010}
V.~Koltchinskii and M.~Yuan.
\newblock Sparsity in multiple kernel learning.
\newblock \emph{The Annals of Statistics}, 38\penalty0 (6):\penalty0
  3660--3695, 2010.

\bibitem[Kozbur(2020)]{Kozbur2020}
D.~Kozbur.
\newblock Analysis of testing-based forward model selection.
\newblock \emph{Econometrica}, 88\penalty0 (5):\penalty0 2147--2173, 2020.

\bibitem[Lattimore and Szepesv\'{a}ri(2020)]{Lattimore2020}
T.~Lattimore and C.~Szepesv\'{a}ri.
\newblock \emph{Bandit algorithms}.
\newblock Cambridge University Press, Cambridge, 2020.

\bibitem[Lauritzen(1996)]{Lauritzen1996}
S.~L. Lauritzen.
\newblock \emph{Graphical models}.
\newblock Oxford statistical science series. Clarendon Press, Oxford, 1996.

\bibitem[Lewis and Syrgkanis(2020)]{Lewis2020}
G.~Lewis and V.~Syrgkanis.
\newblock Double/debiased machine learning for dynamic treatment effects, 2020.
\newblock Preprint arXiv:2002.07285.

\bibitem[Liu et~al.(2021)Liu, Zhang, and Zhou]{Liu2020}
M.~Liu, Y.~Zhang, and D.~Zhou.
\newblock Double/debiased machine learning for logistic partially linear model.
\newblock \emph{The Econometrics Journal}, 2021.

\bibitem[Lloyd(1975)]{Lloyd1975}
W.~P. Lloyd.
\newblock A note on the use of the two-stage least squares estimator in
  financial models.
\newblock \emph{The Journal of Financial and Quantitative Analysis},
  10\penalty0 (1):\penalty0 143--149, 1975.

\bibitem[Ma and Carroll(2006)]{Ma2006}
Y.~Ma and R.~J. Carroll.
\newblock Locally efficient estimators for semiparametric models with
  measurement error.
\newblock \emph{Journal of the American Statistical Association}, 101\penalty0
  (476):\penalty0 1465--1474, 2006.

\bibitem[Maathuis et~al.(2019)Maathuis, Drton, Lauritzen, and
  Wainwright]{Maathuis2019}
M.~Maathuis, M.~Drton, S.~Lauritzen, and M.~Wainwright, editors.
\newblock \emph{Handbook of graphical models}.
\newblock Handbooks of Modern Statistical Methods. Chapman \& Hall/CRC, Boca
  Raton, FL, 2019.

\bibitem[Mammen and {van de Geer}(1997)]{Geer-Mammen1997}
E.~Mammen and S.~{van de Geer}.
\newblock Penalized quasi-likelihood estimation in partial linear models.
\newblock \emph{The Annals of Statistics}, 25\penalty0 (3):\penalty0
  1014--1035, 1997.

\bibitem[Mariano(1972)]{Mariano1972}
R.~S. Mariano.
\newblock The existence of moments of the ordinary least squares and two-stage
  least squares estimators.
\newblock \emph{Econometrica}, 40\penalty0 (4):\penalty0 643--652, 1972.

\bibitem[Mariano(1982)]{Mariano1982}
R.~S. Mariano.
\newblock Analytical small-sample distribution theory in econometrics: The
  simultaneous-equations case.
\newblock \emph{International Economic Review}, 23\penalty0 (3):\penalty0
  503--533, 1982.

\bibitem[Mariano(2003)]{Mariano2001}
R.~S. Mariano.
\newblock \emph{Simultaneous Equation Model Estimators: Statistical Properties
  and Practical Implications}, chapter~6, pages 122--141.
\newblock John Wiley \& Sons, Ltd, 2003.

\bibitem[Meier et~al.(2009)Meier, van~de Geer, and
  B\"{u}hlmann]{Meier-Geer-Buehlmann2009}
L.~Meier, S.~van~de Geer, and P.~B\"{u}hlmann.
\newblock High-dimensional additive modeling.
\newblock \emph{The Annals of Statistics}, 37\penalty0 (6B):\penalty0
  3779--3821, 2009.

\bibitem[Nagar(1959)]{Nagar1959}
A.~L. Nagar.
\newblock The bias and moment matrix of the general k-class estimators of the
  parameters in simultaneous equations.
\newblock \emph{Econometrica}, 27\penalty0 (4):\penalty0 575--595, 1959.

\bibitem[Nagar(1960)]{Nagar1960}
A.~L. Nagar.
\newblock A monte carlo study of alternative simultaneous equation estimators.
\newblock \emph{Econometrica}, 28\penalty0 (3):\penalty0 573--590, 1960.

\bibitem[Newey and McFadden(1994)]{Newey1994}
W.~K. Newey and D.~McFadden.
\newblock Large sample estimation and hypothesis testing.
\newblock In \emph{Handbook of Econometrics}, volume~4, chapter~36, pages
  2111--2245. Elsevier Science, 1994.

\bibitem[Okui et~al.(2012)Okui, Small, Tan, and Robins]{Okui2012}
R.~Okui, D.~S. Small, Z.~Tan, and J.~M. Robins.
\newblock Doubly robust instrumental variable regression.
\newblock \emph{Statistica Sinica}, 22\penalty0 (1):\penalty0 173--205, 2012.

\bibitem[Pearl(1998)]{Pearl1998}
J.~Pearl.
\newblock Graphs, causality, and structural equation models.
\newblock \emph{Sociological Methods \& Research}, 27\penalty0 (2):\penalty0
  226--284, 1998.

\bibitem[Pearl(2004)]{Pearl2004}
J.~Pearl.
\newblock Robustness of causal claims.
\newblock In \emph{Proceedings of the 20th {Conference on Uncertainty in
  Artificial Intelligence}}, UAI ’04, pages 446--453, Arlington, Virginia,
  USA, 2004. AUAI Press.

\bibitem[Pearl(2009)]{Pearl2009}
J.~Pearl.
\newblock \emph{Causality: Models, reasoning, and inference}.
\newblock Cambridge University Press, Cambridge, 2 edition, 2009.

\bibitem[Pearl(2010)]{Pearl2010}
J.~Pearl.
\newblock An introduction to causal inference.
\newblock \emph{The International Journal of Biostatistics}, 6\penalty0
  (2):\penalty0 \ Article 7, 2010.

\bibitem[Peters et~al.(2017)Peters, Janzing, and Sch\"{o}lkopf]{Peters2017}
J.~Peters, D.~Janzing, and B.~Sch\"{o}lkopf.
\newblock \emph{Elements of causal inference: Foundations and learning
  algorithms}.
\newblock Adaptive computation and machine learning. The MIT Press, Cambridge,
  MA, 2017.

\bibitem[Phillips(1984)]{Phillips1984}
P.~C.~B. Phillips.
\newblock The exact distribution of liml: I.
\newblock \emph{International Economic Review}, 25\penalty0 (1):\penalty0
  249--261, 1984.

\bibitem[Phillips(1985)]{Phillips1985}
P.~C.~B. Phillips.
\newblock The exact distribution of liml: Ii.
\newblock \emph{International Economic Review}, 26\penalty0 (1):\penalty0
  21--36, 1985.

\bibitem[Robinson(1988)]{Robinson1988}
P.~M. Robinson.
\newblock Root-{$N$}-consistent semiparametric regression.
\newblock \emph{Econometrica}, 56\penalty0 (4):\penalty0 931--954, 1988.

\bibitem[Rothenh\"{a}usler et~al.(2021)Rothenh\"{a}usler, Meinshausen,
  B\"{u}hlmann, and Peters]{Rothenhausler2018}
D.~Rothenh\"{a}usler, N.~Meinshausen, P.~B\"{u}hlmann, and J.~Peters.
\newblock Anchor regression: Heterogeneous data meet causality.
\newblock \emph{Journal of the Royal Statistical Society: Series B (Statistical
  Methodology)}, 83\penalty0 (2):\penalty0 215--246, 2021.

\bibitem[Ruppert et~al.(2003)Ruppert, Wand, and Carroll]{Ruppert2003}
D.~Ruppert, M.~P. Wand, and R.~J. Carroll.
\newblock \emph{Semiparametric regression}, volume~12 of \emph{Cambridge series
  in statistical and probabilistic mathematics}.
\newblock Cambridge University Press, Cambridge, 2003.

\bibitem[Smucler et~al.(2019)Smucler, Rotnitzky, and Robins]{Rotnitzky2019}
E.~Smucler, A.~Rotnitzky, and J.~M. Robins.
\newblock A unifying approach for doubly-robust $\ell_1$ regularized estimation
  of causal contrasts, 2019.
\newblock Preprint arXiv:1904.03737.

\bibitem[Speckman(1988)]{Speckman1988}
P.~Speckman.
\newblock Kernel smoothing in partial linear models.
\newblock \emph{Journal of the Royal Statistical Society. Series B
  (Methodological)}, 50\penalty0 (3):\penalty0 413--436, 1988.

\bibitem[Staiger and Stock(1997)]{Staiger1997}
D.~Staiger and J.~H. Stock.
\newblock Instrumental variables regression with weak instruments.
\newblock \emph{Econometrica}, 65\penalty0 (3):\penalty0 557--586, 1997.

\bibitem[Stock et~al.(2002)Stock, Wright, and Yogo]{Stock2002}
J.~H. Stock, J.~H. Wright, and M.~Yogo.
\newblock A survey of weak instruments and weak identification in generalized
  method of moments.
\newblock \emph{Journal of Business and Economic Statistics}, 20:\penalty0
  518--529, 2002.

\bibitem[Su and Zhang(2016)]{Su2016}
L.~Su and Y.~Zhang.
\newblock Semiparametric estimation of partially linear dynamic panel data
  models with fixed effects.
\newblock In G.~Gonz\'{a}lez-Rivera, R.~C. Hill, and T.-H. Lee, editors,
  \emph{Essays in Honor of Aman Ullah}, volume~36 of \emph{Advances in
  Econometrics}, pages 137--204. Emerald Group Publishing Limited, Howard
  House, Wagon Lane, Bingley BD16 1WA, UK, 1 edition, 2016.

\bibitem[Summers(1965)]{Summers1965}
R.~Summers.
\newblock A capital intensive approach to the small sample properties of
  various simultaneous equation estimators.
\newblock \emph{Econometrica}, 33\penalty0 (1):\penalty0 1--41, 1965.

\bibitem[Takeuchi and Morimune(1985)]{Takeuchi-Morimune1985}
K.~Takeuchi and K.~Morimune.
\newblock Third-order efficiency of the extended maximum likelihood estimators
  in a simultaneous equation system.
\newblock \emph{Econometrica}, 53\penalty0 (1):\penalty0 177--200, 1985.

\bibitem[Theil(1953{\natexlab{a}})]{Theil1953a}
H.~Theil.
\newblock Repeated least-squares applied to complete equation systems.
\newblock \emph{Central Planning Bureau, The Hague}, 1953{\natexlab{a}}.
\newblock Mimeographed memorandum.

\bibitem[Theil(1953{\natexlab{b}})]{Theil1953b}
H.~Theil.
\newblock Estimation and simultaneous correlation in complete equation systems.
\newblock \emph{Central Planning Bureau, The Hague}, 1953{\natexlab{b}}.
\newblock Mimeographed memorandum.

\bibitem[Theil(1961)]{Theil1961}
H.~Theil.
\newblock \emph{Economic forecasts and policy}, volume~15 of
  \emph{Contributions to economic analysis}.
\newblock North-Holland Publishing Company, Amsterdam, 2 edition, 1961.

\bibitem[van~der Laan and Robins(2003)]{Laan2003}
M.~J. van~der Laan and J.~M. Robins.
\newblock \emph{Unified methods for censored longitudinal data and causality}.
\newblock Springer series in statistics. Springer, New York, 2003.

\bibitem[Wager and Walther(2016)]{Wager2016}
S.~Wager and G.~Walther.
\newblock Adaptive concentration of regression trees, with application to
  random forests, 2016.
\newblock Preprint arXiv:1503.06388.

\bibitem[Wagner(1958)]{Wagner1958}
H.~M. Wagner.
\newblock A monte carlo study of estimates of simultaneous linear structural
  equations.
\newblock \emph{Econometrica}, 26\penalty0 (1):\penalty0 117--133, 1958.

\bibitem[Wooldridge(2013)]{Wooldridge2013}
J.~M. Wooldridge.
\newblock \emph{Introductory econometrics: A modern approach}.
\newblock South-Western Cengage Learning, Mason, OH, 5 edition, 2013.

\bibitem[Yao(2012)]{Yao2012}
F.~Yao.
\newblock Efficient semiparametric instrumental variable estimation under
  conditional heteroskedasticity.
\newblock \emph{Journal of Quantitative Economics}, 10\penalty0 (1):\penalty0
  32--55, 2012.

\bibitem[Yuan and Zhou(2016)]{Yuan-Zhou2016}
M.~Yuan and D.-X. Zhou.
\newblock Minimax optimal rates of estimation in high-dimensional additive
  models.
\newblock \emph{The Annals of Statistics}, 44\penalty0 (6):\penalty0
  2564--2593, 2016.

\end{thebibliography}

\begin{appendices}

\section{An Example where the Identifiability Condition~\eqref{eq:identificationCondition} holds, but Conditional Moment Requirements do not}\label{sect:strictlyWeakerCondition}

This section presents an SEM where our identifiability condition~\eqref{eq:identificationCondition} holds, but where the conditional moment requirements of~\citet{Chernozhukov2018} do not. 

We assume the model
\begin{displaymath}
	Y \leftarrow X^T\betazero + \gY(W) + \hY(H) + \eps_Y
\end{displaymath}
given 
in~\eqref{eq:SEM} and the identifiability condition $\EP[\Ra(\Ry-\Rx^T\betazero)]=\bo$ given in~\eqref{eq:identificationCondition}. \citet{Chernozhukov2018} assume the model 
\begin{equation}\label{eq:ChernozhukovModel}
	Y=X^T\betazero + \gY(W) + U, \quad A=\gA(W)+V
\end{equation}
for unknown functions $\gY$ and $\gA$ and impose the conditional moment restrictions 
\begin{equation}\label{eq:ChernozhukovMoment}
	\E[U|A,W]=0 \quad\textrm{and}\quad \E[V|W]=\bo
\end{equation}
on the error terms.
Their model is implicitly assumed to be justidentified:  
the dimensions of $A$ and $X$ are implicitly assumed to be  equal. 

Model~\eqref{eq:ChernozhukovModel} and the conditional moment restrictions~\eqref{eq:ChernozhukovMoment} imply the identifiability condition~\eqref{eq:identificationCondition} due to
\begin{displaymath}
		\E\big[\Ra(\Ry-\Rx^T\betazero)\big] = \E\big[\big(A-\gA(W)\big)U\big] 
		= \E\big[\big(A-\gA(W)\big)\E[U|A, W]\big] = \bo. 
\end{displaymath}

However, the reverse direction does not hold. A counterexample is presented in Figure~\ref{fig:ChernozhukovCounterexample} where $W$ directly affects $H$. 
This SEM satisfies the identifiability condition~\eqref{eq:identificationCondition} because $A$ is independent of $H$ conditional on $W$, but it does not satisfy $\E[U|W,A]=0$ because we have
\begin{displaymath}
	\E[U|A,W] = \E[H+\eps_Y|A,W] = \E[H|W] = \E[W+\eps_H|W] = W
\end{displaymath}
due to $A\independent H|W$ and $(\eps_Y, \eps_H)\independent(W,A)$.
We have $A\independent H|W$ because all paths from $A$ to $H$ are blocked by $W$. The path $A\rightarrow X\leftarrow H$ is blocked by the empty set because $X$ is a collider on this path. The path $A\rightarrow X\rightarrow Y\leftarrow H$ is blocked by the empty set because $Y$ is a collider on this path. The path $A\rightarrow X\rightarrow Y\leftarrow W\rightarrow H$ is blocked by $W$. The paths $A\rightarrow X\rightarrow W\rightarrow Y\leftarrow H$ and $A\rightarrow X\rightarrow W\rightarrow H$ are also blocked by $W$. 

\begin{figure}[h!]
	\centering
	\caption[]{\label{fig:ChernozhukovCounterexample}An SEM and its associated causal graph.}
	\begin{tabular}{cc}
	\begin{tabular}{l}
	$\displaystyle 
	\begin{array}{r}
		(\eps_A, \eps_W,\eps_H, \eps_X, \eps_Y)\sim\mathcal{N}_5(\bo,\one)\\	
	\end{array}$
	\\
	$\displaystyle
	\begin{array}{lcl} 
		A &\leftarrow& \eps_A\\
		W &\leftarrow& \eps_W\\
		H &\leftarrow& W + \eps_H\\
		X &\leftarrow& A+W + H+\eps_X\\
		Y &\leftarrow& X+ W +H+\eps_Y
	\end{array}$ 
	\end{tabular}
	& 
	\begin{tabular}{c}
          \includegraphics[width=0.35\textwidth]{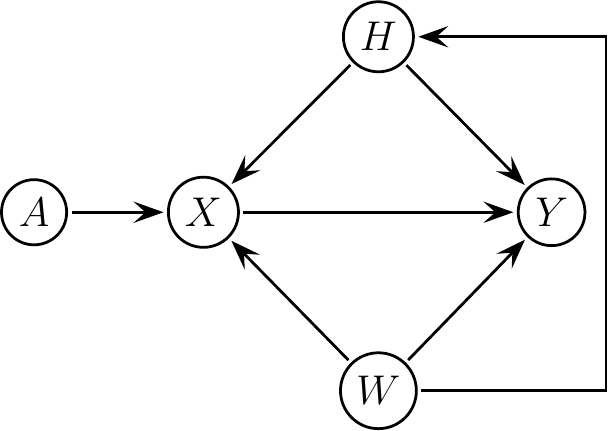}
	\end{tabular}
\end{tabular}
\end{figure}

\section{DML1 Estimators}

The DML1 estimators are less preferred than the DML2 estimators we proposed to use in the main text, but for completeness we provide the definitions in this section.

\subsection{DML1 Estimator of $\betazero$}\label{sect:DML1Binfty}

The \DMLone\ estimator of $\betazero$ is given by
\begin{displaymath}
	\hbetaNDMLone:=\frac{1}{\KK}\sum_{\kk=1}^{\KK}\hbetaNkDMLone, 
\end{displaymath}
where
\begin{equation}\label{eq:hbetaNkDMLone}
		\hbetaNkDMLone := \Big( \big(\hbRxk\big)^T
			\PiIkcIk
			\hbRxk\Big)^{-1}
			\big(\hbRxk\big)^T
			\PiIkcIk
			\hbRyk,
\end{equation}
and where we recall the projection matrix 
$	\PiIkcIk = \hbRak
			\big((\hbRak)^T\hbRak \big)^{-1}
			(\hbRak)^T
$ 
defined in~\eqref{eq:RaProjection}.
The estimator $\hbetaNkDMLone$ 
is  the TSLS estimator of $\hbRyk$ on $\hbRxk$ using the instrument $\hbRak$.

\subsection{DML1 estimator of $\bg$}\label{sect:DML1gamma}

The \DMLone\ estimator of $\bg$ is given by 
\begin{equation}\label{eq:hbg}
	\hbgDMLone := \frac{1}{\KK}\sum_{\kk=1}^{\KK} \hbgkDMLone,
\end{equation}
where 
\begin{displaymath}
	\hbgkDMLone := \argmin_{b\in\R^d} \bigg(\normBig{ \big(\one-\PiIkcIk\big)\big(\hbRyk-\big(\hbRxk\big)^Tb\big)}_2^2+\gamma\normBig{ \PiIkcIk\big(\hbRyk-\big(\hbRxk\big)^Tb\big) }_2^2\bigg).
\end{displaymath} 
This estimator can be expressed in closed form by
\begin{displaymath}
	\hbgkDMLone = \Big(  \big(\hbRxtilk\big)^T\hbRxtilk\Big)^{-1} \big(\hbRxtilk\big)^T\hbRytilk,
\end{displaymath}
where we recall the notation
\begin{displaymath}
	\hbRxtilk = \Big(   \one + (\sqrt{\gamma}-1)\PiIkcIk  \Big)\hbRxk 
	\quad\textrm{and}\quad
	\hbRytilk = \Big(   \one + (\sqrt{\gamma}-1)\PiIkcIk  \Big)\hbRyk
\end{displaymath}
as in~\eqref{eq:OLSresiduals}. 
The computation of $\hbgkDMLone$ is an OLS scheme where $\hbRytilk$ is regressed on $\hbRxtilk$.

\section{SEM of Figure~\ref{fig:AvsRa}}\label{sect:histogramRavsA}

The data from the simulation displayed in Figure~\ref{fig:AvsRa} come from the following SEM. Let the dimension of $W$ be $\sss=20$. Let $R$ be the upper triangular matrix of the  Cholesky decomposition of the Toeplitz matrix whose first row is given by $(1, 0.7, 0.7^2,\ldots,0.7^{19})$. The SEM we consider is given by

\begin{displaymath}
\begin{array}{l}
\begin{array}{r}
		(\eps_A, \eps_W,\eps_H, \eps_X, \eps_Y)\sim\mathcal{N}_{24}(\bo,\one)\\	
	\end{array}
	\\
	\begin{array}{lcl} 
		H &\leftarrow& \eps_H\\
		W &\leftarrow& \eps_W R\\
		A &\leftarrow& \frac{e^{W_1}}{1+e^{W_1}} +W_2 +W_3+ \eps_A\\
		X &\leftarrow& 2A + W_1 + 0.25\cdot\frac{e^{W_3}}{1+e^{W_3}} + H + \eps_X\\
		Y &\leftarrow& X +\frac{e^{W_1}}{1+e^{W_1}}+0.25W_3+H+\eps_Y.
	\end{array}
\end{array}
\end{displaymath}

\section{Additional Numerical Results}\label{sect:additioinalSimulation}

If we say in this section that the nuisance parameters are estimated with additive splines, they are estimated with additive cubic B-splines with $\ceil[\big]{\NN^{\frac{1}{5}}}+2$ degrees of freedom, where $\NN$ denotes the sample size of the data.

If we say in this section that the nuisance parameters are estimated with random forests, they are estimated with random forests consisting of $500$ trees that have a minimal node size of $5$. 
\\

Figure~\ref{fig:COVERsimulationIntroBeta0} and~\ref{fig:COVERsimulation2Beta0} illustrate the simulation results with $\betazero=0$ of the examples presented in Figure~\ref{fig:COVERsimulationIntro} and~\ref{fig:COVERsimulation2} in Sections~\ref{sect:ourContribution} and~\ref{sect:exampleForest}, respectively. The coverage and  length of the scaled confidence intervals are similar to the results obtained for $\betazero\neq 0$. Instead of the power as in Figure~\ref{fig:COVERsimulationIntro} and~\ref{fig:COVERsimulation2}, Figure~\ref{fig:COVERsimulationIntroBeta0} and~\ref{fig:COVERsimulation2Beta0} illustrate the type I error. 
\\

In Figure~\ref{fig:COVERsimulationIntroBeta0}, DML achieves a type I error of $0$ or close to $0$ over all sample sizes considered. The \regsDML\ method achieves a type I error that is closer to the gray line indicating the $5\%$ level. 
The dashed lines represent $95\%$ confidence regions. 
The type I error of \regsDML\ is higher than the type I error of DML because the \regsDML\ confidence intervals are considerably shorter than the DML ones. The right plot in Figure~\ref{fig:COVERsimulationIntroBeta0} indicates that the lengths of the confidence intervals of \regsDML\  is around $10\%-30\%$ the length of DML's.
Although \regsDML\ greatly reduces the confidence interval length, the type I error confidence bands include the $5\%$ level or are below it. 
This means that although \regsDML\ is a regularized version of DML, it does not incur an overlarge bias. 
\\

In Figure~\ref{fig:COVERsimulation2Beta0}, the type I errors of both DML and \regsDML\ are similar. 
The $95\%$ confidence regions of both estimators include the $5\%$ level or are below it.
The $95\%$ confidence regions of the levels are represented by dashed lines. 
These confidence regions of both DML and \regsDML\ contain the $5\%$ level or are below it.
The right plot in Figure~\ref{fig:COVERsimulation2Beta0} illustrates that the \regsDML\ confidence intervals are 
 around $50\%-80\%$ the length of DML's.
Nevertheless, its type I error does not exceed the $95\%$ level.

\begin{figure}[h!]
	\centering
	\caption[]{\label{fig:COVERsimulationIntroBeta0} 
	The results come from $M=1000$ simulation runs each from the SEM in Figure~\ref{fig:introSEM}  with $\betazero=0$ for a range of sample sizes $\NN$ and with $\KK=2$ and $\Salg=100$ in Algorithm~\ref{algo:Summary}. 
	The nuisance functions are estimated with additive splines.
	The figure displays the coverage of two-sided confidence intervals for $\betazero$, type I error for two-sided testing of the 
	hypothesis $H_0:\ \betazero = 0$, and scaled lengths of two-sided confidence intervals of DML (red),  \regDML\ (blue),  \regsDML\  (green), 
	LIML (orange), Fuller(1) (purple), and Fuller(4) (cyan), 
	where all results are at level $95\%$.
		At each sample size $\NN$, the lengths of the confidence intervals are scaled with the median length from DML. 
	The shaded regions in the coverage and the type I error plots represent $95\%$ confidence bands with respect to the $M$ simulation runs. 
	The blue and green lines as 
	well as the red and orange ones 
	are indistinguishable in the left panel.
	}
	\includegraphics[width=\textwidth]{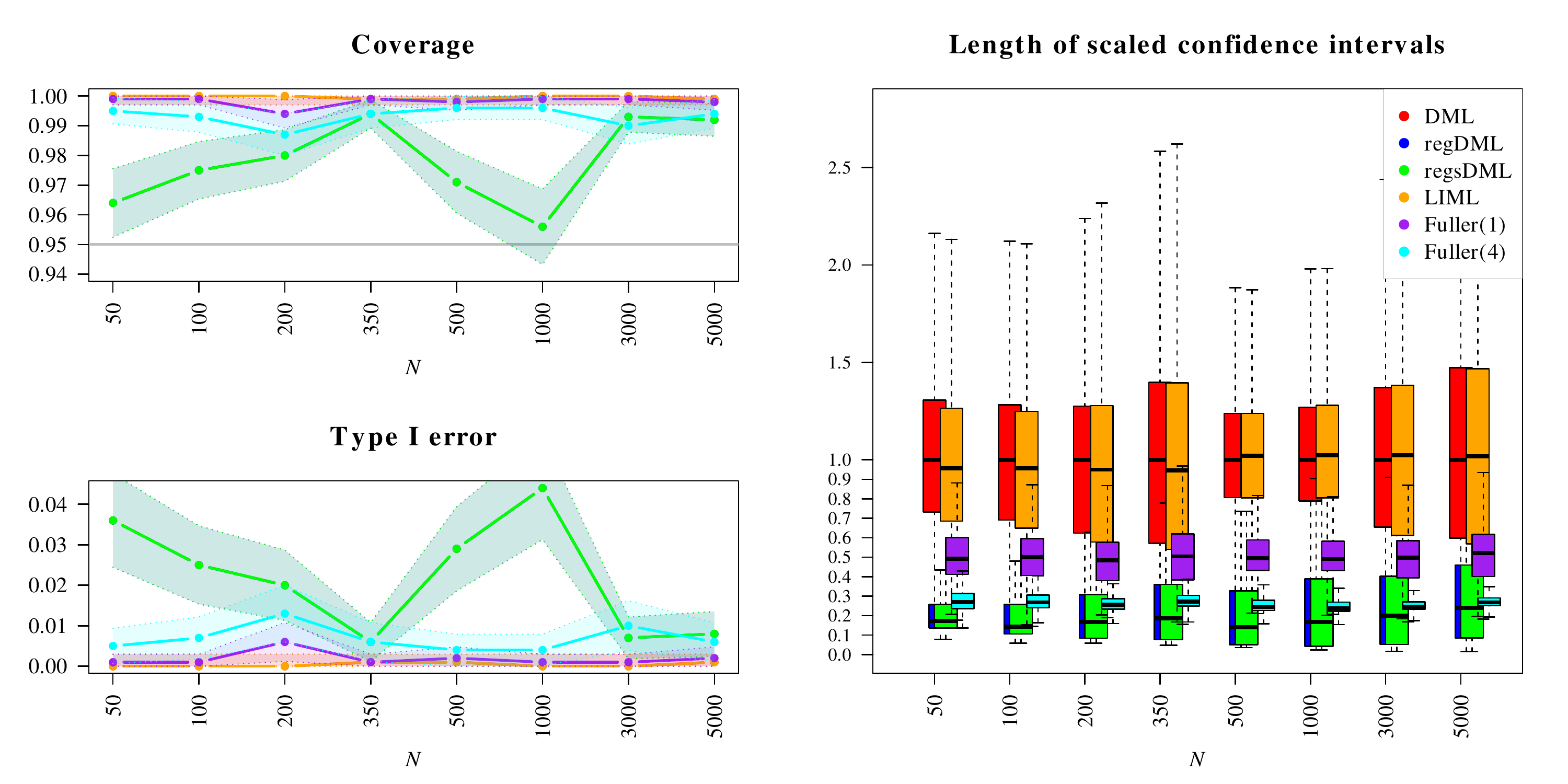}
\end{figure}

\begin{figure}[h!]
	\centering
	\caption[]{\label{fig:COVERsimulation2Beta0} 
	The results come from $M=1000$ simulation runs from the SEM in Figure~\ref{fig:simulation2} with $\betazero=0$ for a range of sample sizes $\NN$ and with $\KK=2$ and $\Salg=100$ in Algorithm~\ref{algo:Summary}. 
	The nuisance functions are estimated with random forests. 
	The figure displays the coverage of two-sided confidence intervals for $\betazero$, type I error for two-sided testing of the 
	hypothesis $H_0:\ \betazero = 0$, and scaled lengths of two-sided confidence intervals of DML (red),  \regDML\ (blue),  \regsDML\  (green), 
	LIML (orange), Fuller(1) (purple), and Fuller(4) (cyan), 
	where all results are at level $95\%$.
	At each sample size $\NN$, the lengths of the confidence intervals are scaled with the median length from DML. 
	The shaded regions in the coverage and the type I error plots represent $95\%$ confidence bands with respect to the $M$ simulation runs.  
	The blue and green lines as 
	well as the red and orange ones 
	are indistinguishable in the left panel.
	}
	\includegraphics[width=\textwidth]{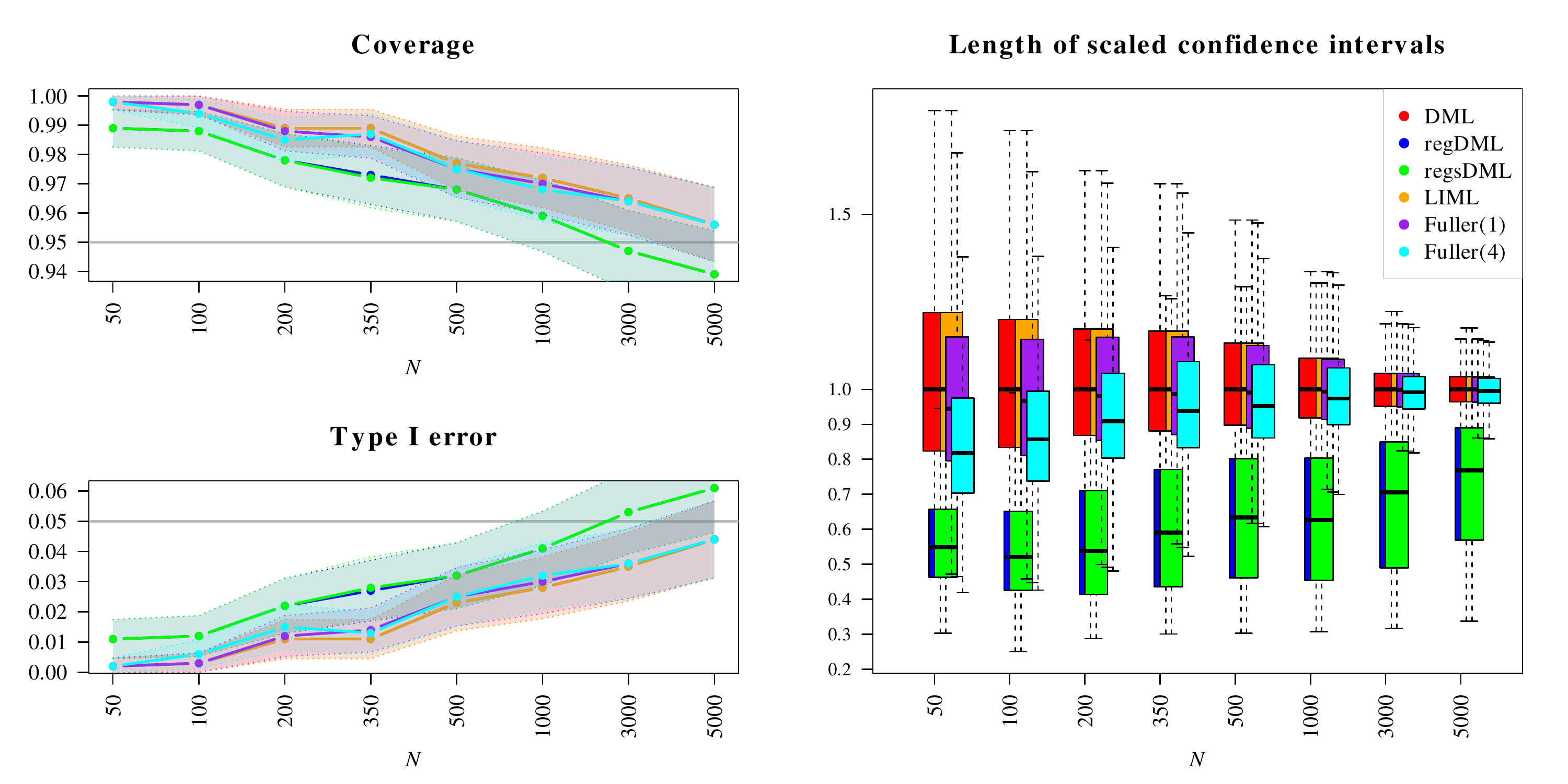}
\end{figure}

\section{Weak $A\rightarrow X$ and Bias-Variance Tradeoff}\label{sect:strong-weak}

First, we analyze the behavior of our methods for varying strength of $A$ on $X$. 
For $N=200$, we consider the coverage and length of the confidence intervals for 
varying strength from $A$ to $X$ for the same settings as in Figure~\ref{fig:COVERsimulationIntro} and~\ref{fig:COVERsimulation2}.

Figure~\ref{fig:weak-spline} illustrates the results for data from the SEM from Figure~\ref{fig:COVERsimulationIntro}. 
We vary the strength of the direct link $A\rightarrow X$ and denote it by $\alpha$ in Figure~\ref{fig:weak-spline}. 
Figure~\ref{fig:weak-forest} illustrates the results for data from the SEM from Figure~\ref{fig:COVERsimulation2}. 
We leave the link $A_2\rightarrow X$ as it is and only vary the strength of the direct link $A_1\rightarrow X$, which we denote by $\alpha$ in Figure~\ref{fig:weak-forest}. 
In both Figure~\ref{fig:weak-spline} and~\ref{fig:weak-forest}, the coverage remains high for all considered methods. If $\alpha$ becomes larger, the confidence intervals become shorter, which leads to a coverage that is closer to the nominal $95\%$ level, especially in Figure~\ref{fig:weak-forest}. The \regsDML\ method yields the shortest confidence intervals in both figures. \\

\begin{figure}[h!]
	\centering
	\caption[]{\label{fig:weak-spline} 
	Same setting as in Figure~\ref{fig:COVERsimulationIntro}, but with $N=200$ only.
	The strength of the direct link $A\rightarrow X$ varies and is denoted by $\alpha$. 
	}
	\includegraphics[width=\textwidth]{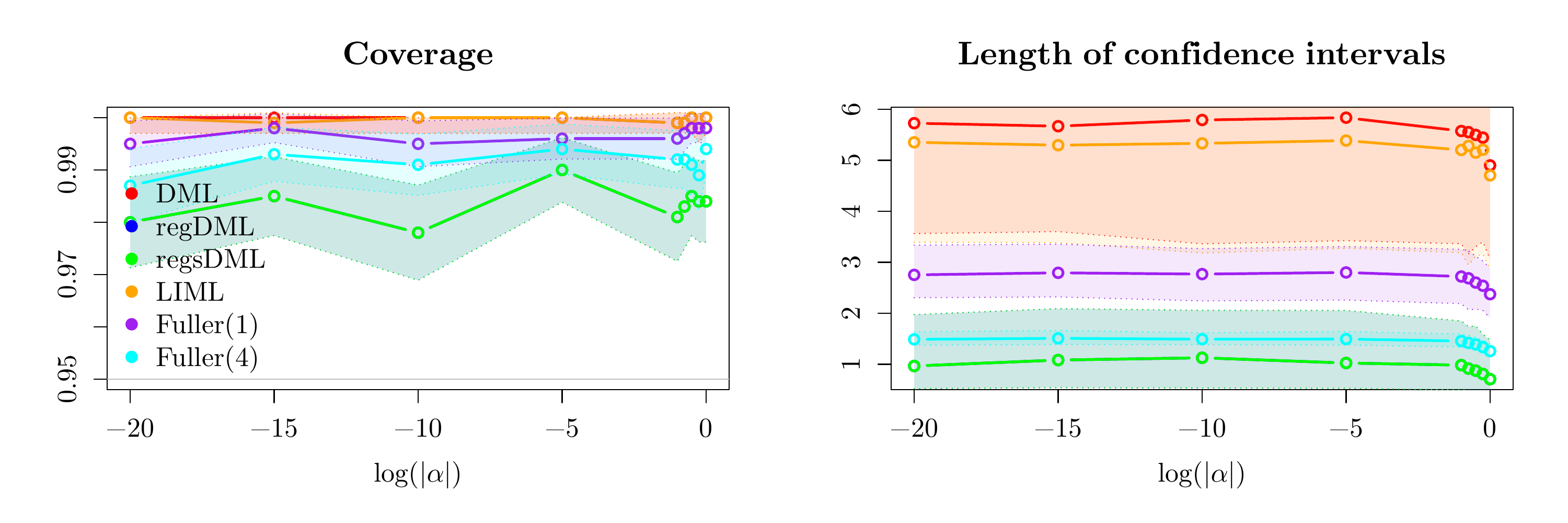}
\end{figure}

\begin{figure}[h!]
	\centering
	\caption[]{\label{fig:weak-forest} 
	Same setting as in Figure~\ref{fig:COVERsimulation2}, but with $N=200$ only.
	The strength of the direct link $A_1\rightarrow X$ varies and is denoted by $\alpha$. 
	}
	\includegraphics[width=\textwidth]{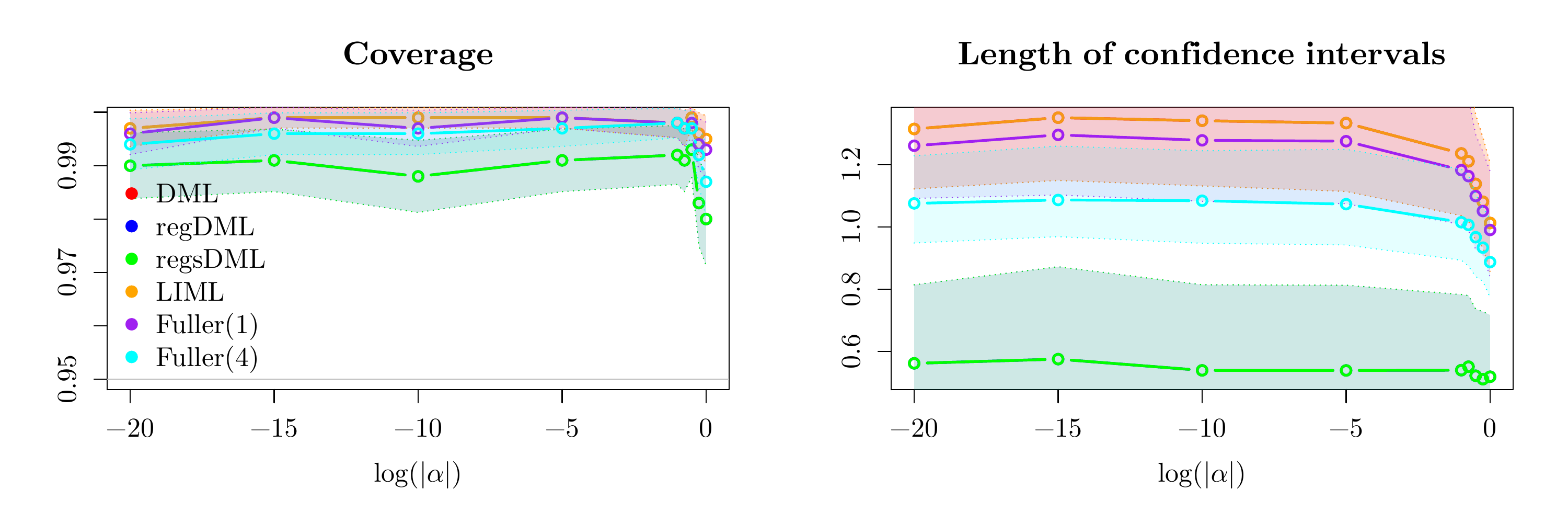}
\end{figure}

Second, we analyze the bias-variance tradeoff of the respective estimated quantities of the regularized methods. 
We again choose the sample size $N=200$ and consider the same settings as in Figure~\ref{fig:COVERsimulationIntro} and~\ref{fig:COVERsimulation2}. 
The results are summarized in Figure~\ref{fig:bias-var-spline} and~\ref{fig:bias-var-forest}
that display the estimated MSE, estimated variance, and estimated squared bias as used in Equation~\eqref{eq:hgamma}.
The MSE in both figures is mainly driven by the variance, and \regsDML\ achieves 
a considerable variance reduction compared to the TSLS-type DML estimator. 

\begin{figure}[h!]
	\centering
	\caption[]{\label{fig:bias-var-spline} 
	Estimated MSE, estimated variance, and estimated squared bias as used in Equation~\eqref{eq:hgamma} for the 
	same setting as in Figure~\ref{fig:COVERsimulationIntro}, but with $N=200$ only.
	 The black solid line displays the median of the respective quantity over the considered range of $\gamma$-values for $\hbg$. The yellow area marks the observed $25\%$ and $75\%$ quantiles. 
	 All methods incorporate an additional variance adjustment from the $\Salg$ 
	 repetitions according to Algorithm~\ref{algo:Summary}. 
	 Boxplots illustrate the performance of the TSLS and the regularized methods. 
	 The position of the boxplots is not linked to the $\gamma$-values on the $x$-axis.  
	}
	\includegraphics[width=\textwidth]{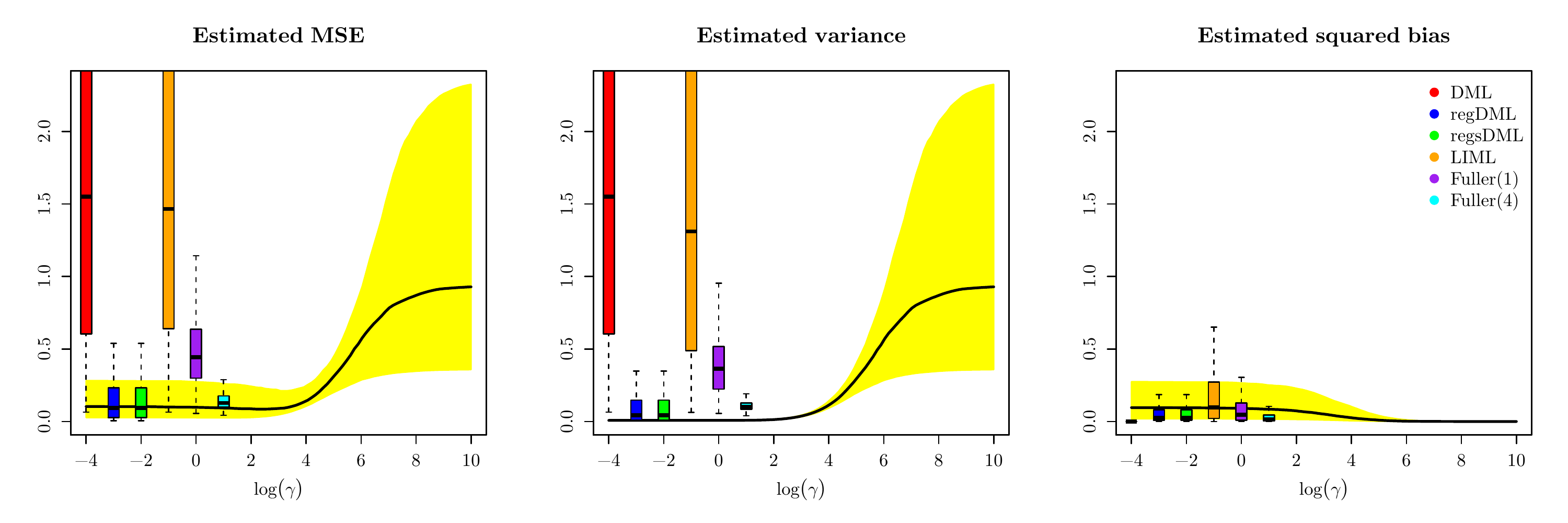}
\end{figure}

\begin{figure}[h!]
	\centering
	\caption[]{\label{fig:bias-var-forest} 
	Estimated MSE, estimated variance, and estimated squared bias as used in Equation~\eqref{eq:hgamma} for the 
	same setting as in Figure~\ref{fig:COVERsimulation2}, but with $N=200$ only.
	The black solid line displays the median of the respective quantity over the considered range of $\gamma$-values for $\hbg$. The yellow area marks the observed $25\%$ and $75\%$ quantiles. 
	All methods incorporate an additional variance adjustment from the $\Salg$ 
	 repetitions according to Algorithm~\ref{algo:Summary}. 
	 Boxplots illustrate the performance of the TSLS and the regularized methods. 
	 The position of the boxplots is not linked to the $\gamma$-values on the $x$-axis. 
	}
	\includegraphics[width=\textwidth]{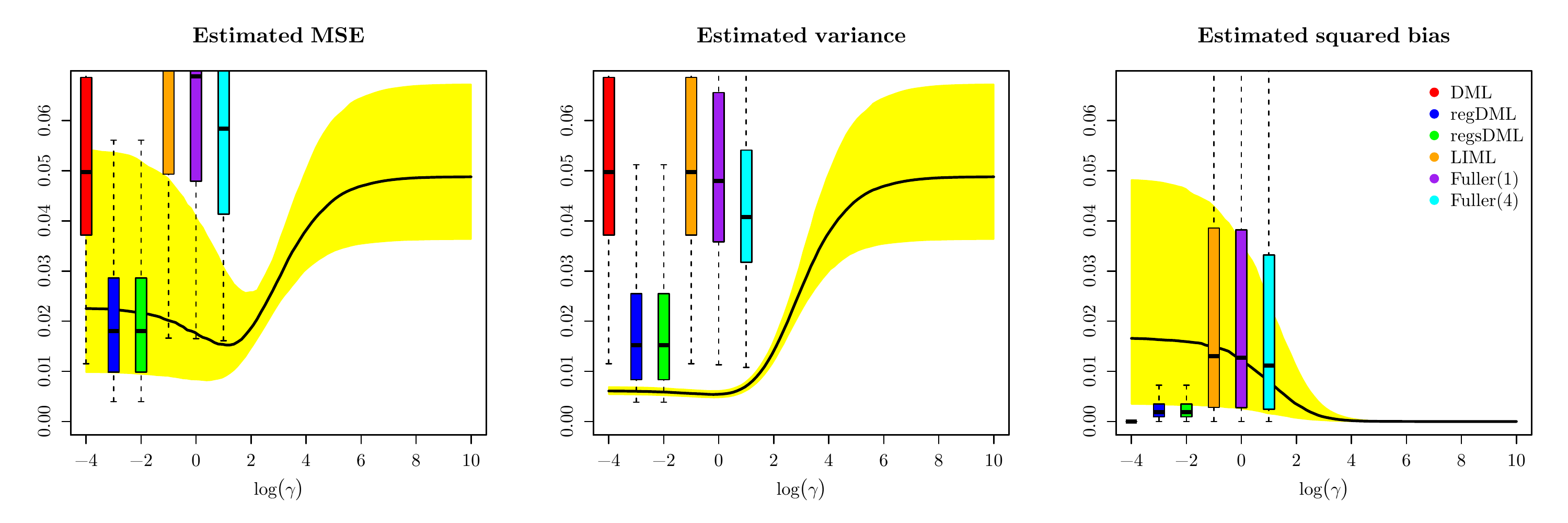}
\end{figure}

\section{Confounding and its Mitigation}\label{sect:counterexamples}
If we say in this section that the nuisance parameters are estimated with additive splines, they are estimated with additive cubic B-splines with $\ceil[\big]{\NN^{\frac{1}{5}}}+2$ degrees of freedom, where $\NN$ denotes the sample size of the data.

If we say in this section that the nuisance parameters are estimated with random forests, they are estimated with random forests consisting of $500$ trees that have a minimal node size of $5$. 
\\

We consider models where the DML and the \regsDML\ methods do not work well in terms of coverage of $\betazero$. We present possible explanations of these failures and illustrate model changes to overcome them. 
The first model in Section~\ref{sect:strongConfounding} features a strong confounding effect $H\rightarrow X$, the second model in Section~\ref{sect:noisyH} features an effect with noise in $W\rightarrow H$, and the third model in Section~\ref{sect:noisyW} features an effect with noise in $H\rightarrow W$.

\subsection{Strong Confounding Effect $H\rightarrow X$}\label{sect:strongConfounding}

If the hidden variable $H$ is strongly confounded with $X$, the resulting TSLS-type DML estimator can be substantially biased depending on the choice of functions in the model. 
 If the estimated variances are not large enough, the coverage of the resulting confidence intervals for $\betazero$ can be too low. This issue is illustrated in Figure~\ref{fig:COVERstrong-weak_strongconfounding}. 
 
 The \regsDML\ estimator mimics the bias behavior of  DML  because the DML estimator is used as a replacement of $\betazero$ in the MSE objective function that defines the estimated regularization parameter of \regDML\ in~\eqref{eq:hgamma}. 
The confidence intervals of \regsDML\  are shorter than the DML ones, but both are computed with a similarly biased coefficient estimate of $\betazero$. Therefore, the coverage of the confidence intervals of \regsDML\  is even worse than the one of DML. 

The coverages of both DML and \regsDML\ are considerably improved if the confounding strength is reduced; see Figure~\ref{fig:COVERstrong-weak_weakconfounding}. 

\begin{figure}[h!]
	\centering
	\caption[]{\label{fig:strongweak}An SEM and its associated causal graph.}
	\begin{tabular}{cc}
	\begin{tabular}{l}
	$\displaystyle 
	\begin{array}{r}
		(\eps_A, \eps_W,\eps_H, \eps_X, \eps_Y)\sim\mathcal{N}_5(\bo,\one)\\	
	\end{array}$
	\\
	$\displaystyle
	\begin{array}{lcl} 
		A &\leftarrow&\eps_A\\
		W &\leftarrow& \eps_W \\
			H &\leftarrow&\eps_H\\
			X &\leftarrow& A + W +\chi H+ 0.25 \eps_X\\
			Y &\leftarrow& \betazero X+W+H+0.25\eps_Y
	\end{array}$ 
	\end{tabular}
	& 
	\begin{tabular}{c}
          \includegraphics[width=0.35\textwidth]{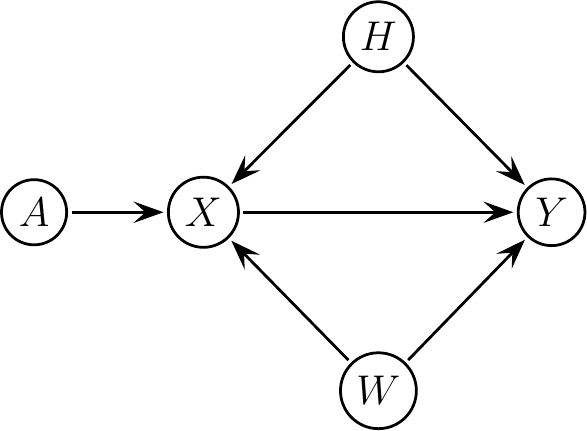}
	\end{tabular}
\end{tabular}
\end{figure}

\begin{figure}[h!]
	\centering
	\caption[]{\label{fig:COVERstrong-weak_strongconfounding} 
	The results come from $M=1000$ simulation runs from the SEM in Figure~\ref{fig:strongweak} with $\chi=15$ and $\betazero=0$ for a range of sample sizes $\NN$ and with $\KK=2$ and $\Salg=100$ in Algorithm~\ref{algo:Summary}. 
	The nuisance functions are estimated with additive splines. 
	The figure displays the coverage of two-sided confidence intervals for $\betazero$, type I error for two-sided testing of the 
	hypothesis $H_0:\ \betazero = 0$, and scaled lengths of two-sided confidence intervals of DML (red),  \regDML\ (blue),  \regsDML\  (green), 
	LIML (orange), Fuller(1) (purple), and Fuller(4) (cyan), 
	where all results are at level $95\%$.
	At each sample size $\NN$, the lengths of the confidence intervals are scaled with the median length from DML. 
	The shaded regions in the coverage and the type I error plots represent $95\%$ confidence bands with respect to the $M$ simulation runs. 
	The blue and green lines are indistinguishable in the left panel.
	}
	\includegraphics[width=\textwidth]{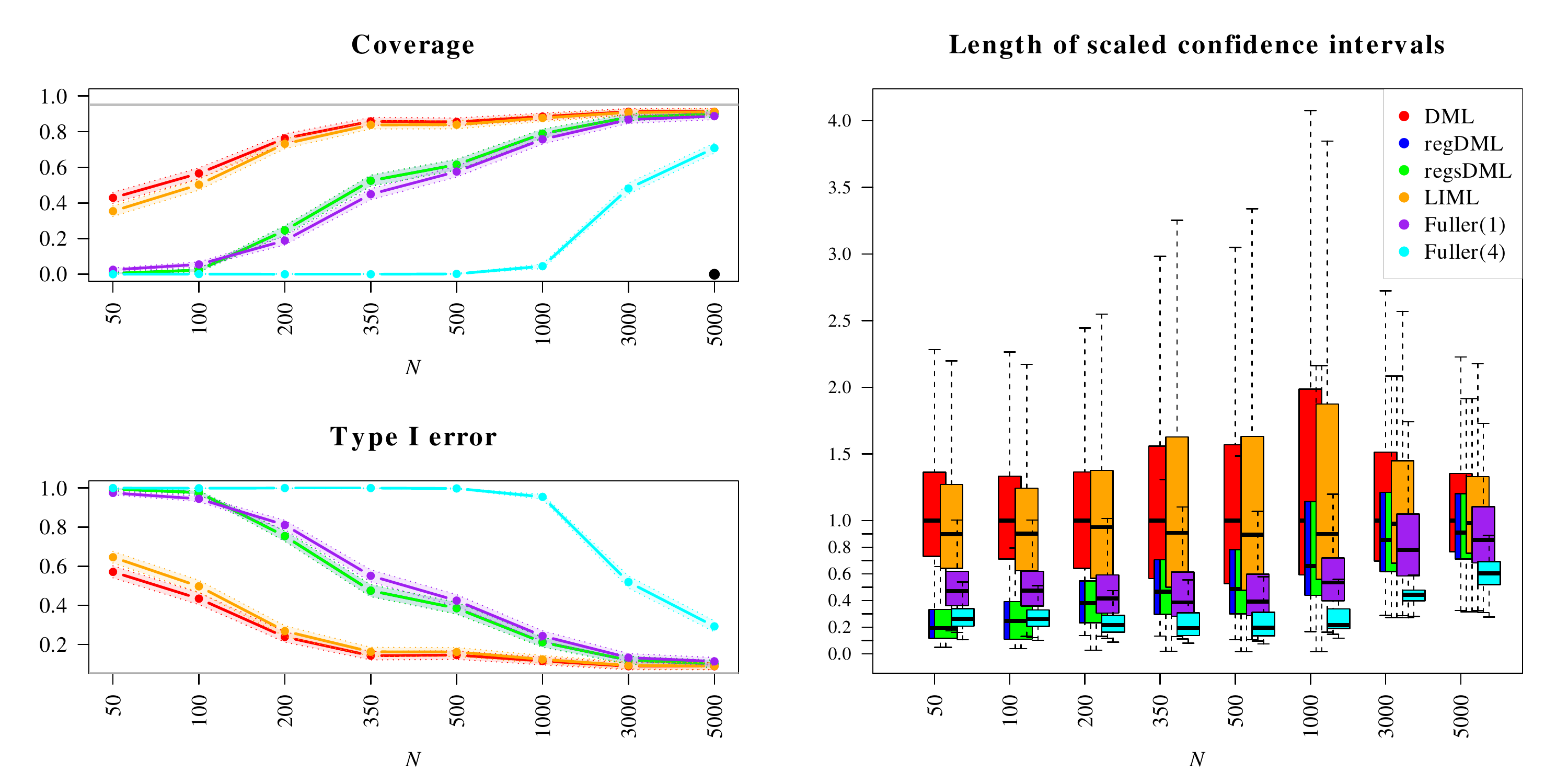}
\end{figure}

\begin{figure}[h!]
	\centering
	\caption[]{\label{fig:COVERstrong-weak_weakconfounding} 
	The results come from $M=1000$ simulation runs from the SEM in Figure~\ref{fig:strongweak} with $\chi=1$ and $\betazero=0$ for a range of sample sizes $\NN$ and with $\KK=2$ and $\Salg=100$ in Algorithm~\ref{algo:Summary}. 
	The nuisance functions are estimated with additive splines. 
	The figure displays the coverage of two-sided confidence intervals for $\betazero$, type I error for two-sided testing of the 
	hypothesis $H_0:\ \betazero = 0$, and scaled lengths of two-sided confidence intervals of DML (red),  \regDML\ (blue),  \regsDML\  (green), 
	LIML (orange), Fuller(1) (purple), and Fuller(4) (cyan), 
	where all results are at level $95\%$.
	At each sample size $\NN$, the lengths of the confidence intervals are scaled with the median length from DML. 
	The shaded regions in the coverage and the type I error plots represent $95\%$ confidence bands with respect to the $M$ simulation runs. 
	The blue and green lines are indistinguishable in the left panel.
	}
	\includegraphics[width=\textwidth]{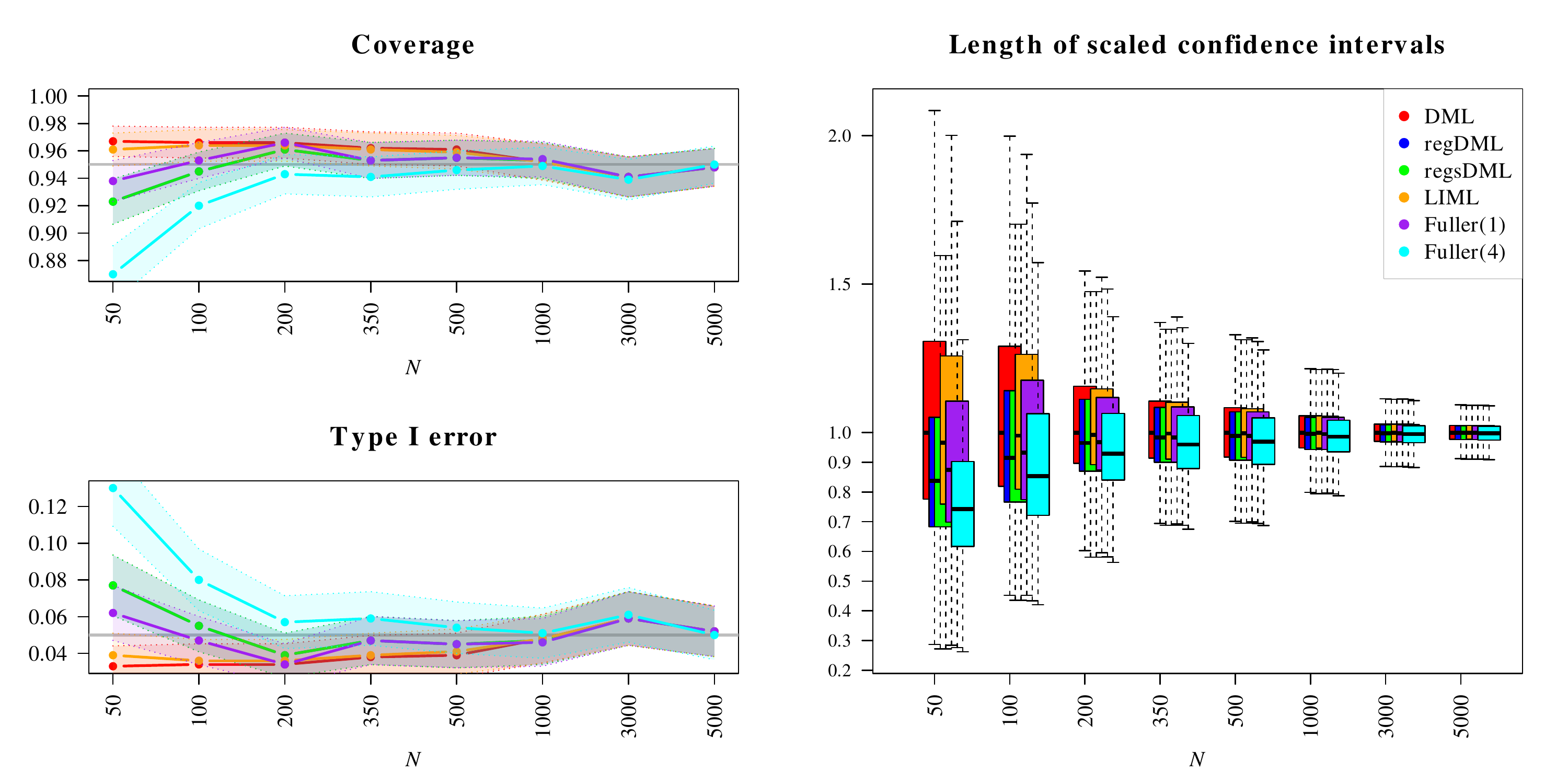}
\end{figure}

\subsection{Noise in  $W\rightarrow H$}\label{sect:noisyH}

The variable $W$ may have a direct effect on $H$. If this link is strong enough with respect to the additional noise $\eps_H$ of $H$, it is possible to obtain some information of $H$ by observing $W$.   This can reduce the overall level of confounding present depending on the choice of functions in the model. 

Simulation results where $W$ explains only part of the variation in $H$ are presented in Figure~\ref{fig:COVERWH-noise}. The confidence intervals of both DML and \regsDML\ do not attain a $95\%$ coverage for small sample sizes $\NN$. The situation can  be considerably improved by reducing the variation of $H$ that is not explained by $W$; see Figure~\ref{fig:COVERWH-nonnoise}.

\begin{figure}[h!]
	\centering
	\caption[]{\label{fig:WH}An SEM and its associated causal graph.}
	\begin{tabular}{cc}
	\begin{tabular}{l}
	$\displaystyle 
	\begin{array}{r}
		(\eps_A, \eps_W,\eps_H, \eps_X, \eps_Y)\sim\mathcal{N}_5(\bo,\one)\\	
	\end{array}$
	\\
	$\displaystyle
	\begin{array}{lcl} 
		A &\leftarrow&\eps_A\\
		W &\leftarrow& \eps_W \\
			H &\leftarrow&W+\kappa\eps_H\\
			X &\leftarrow& 0.5A + 3\tanh(2W)+1.5 H + 0.25 \eps_X\\
			Y &\leftarrow& \betazero X-\tanh(W)+H+0.25\eps_Y
	\end{array}$ 
	\end{tabular}
	& 
	\begin{tabular}{c}
          \includegraphics[width=0.35\textwidth]{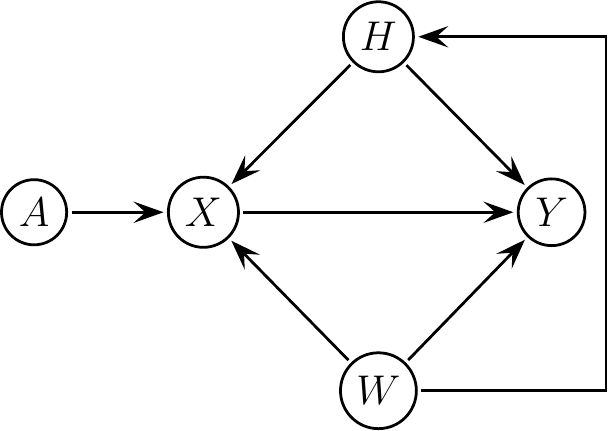}
	\end{tabular}
\end{tabular}
\end{figure}

\begin{figure}[h!]
	\centering
	\caption[]{\label{fig:COVERWH-noise} 
	The results come from $M=1000$ simulation runs from the SEM in Figure~\ref{fig:WH} with $\kappa=2$ and $\betazero=0$ for a range of sample sizes $\NN$ and with $\KK=2$ and $\Salg=100$ in Algorithm~\ref{algo:Summary}. 
	The nuisance functions are estimated with additive splines. 
	The figure displays the coverage of two-sided confidence intervals for $\betazero$, type I error for two-sided testing of the 
	hypothesis $H_0:\ \betazero = 0$, and scaled lengths of two-sided confidence intervals of DML (red),  \regDML\ (blue),  \regsDML\  (green), 
	LIML (orange), Fuller(1) (purple), and Fuller(4) (cyan), 
	where all results are at level $95\%$.
	At each sample size $\NN$, the lengths of the confidence intervals are scaled with the median length from DML. 
		The shaded regions in the coverage and the type I error plots represent $95\%$ confidence bands with respect to the $M$ simulation runs. 
	The blue and green lines are indistinguishable in the left panel.
}
	\includegraphics[width=\textwidth]{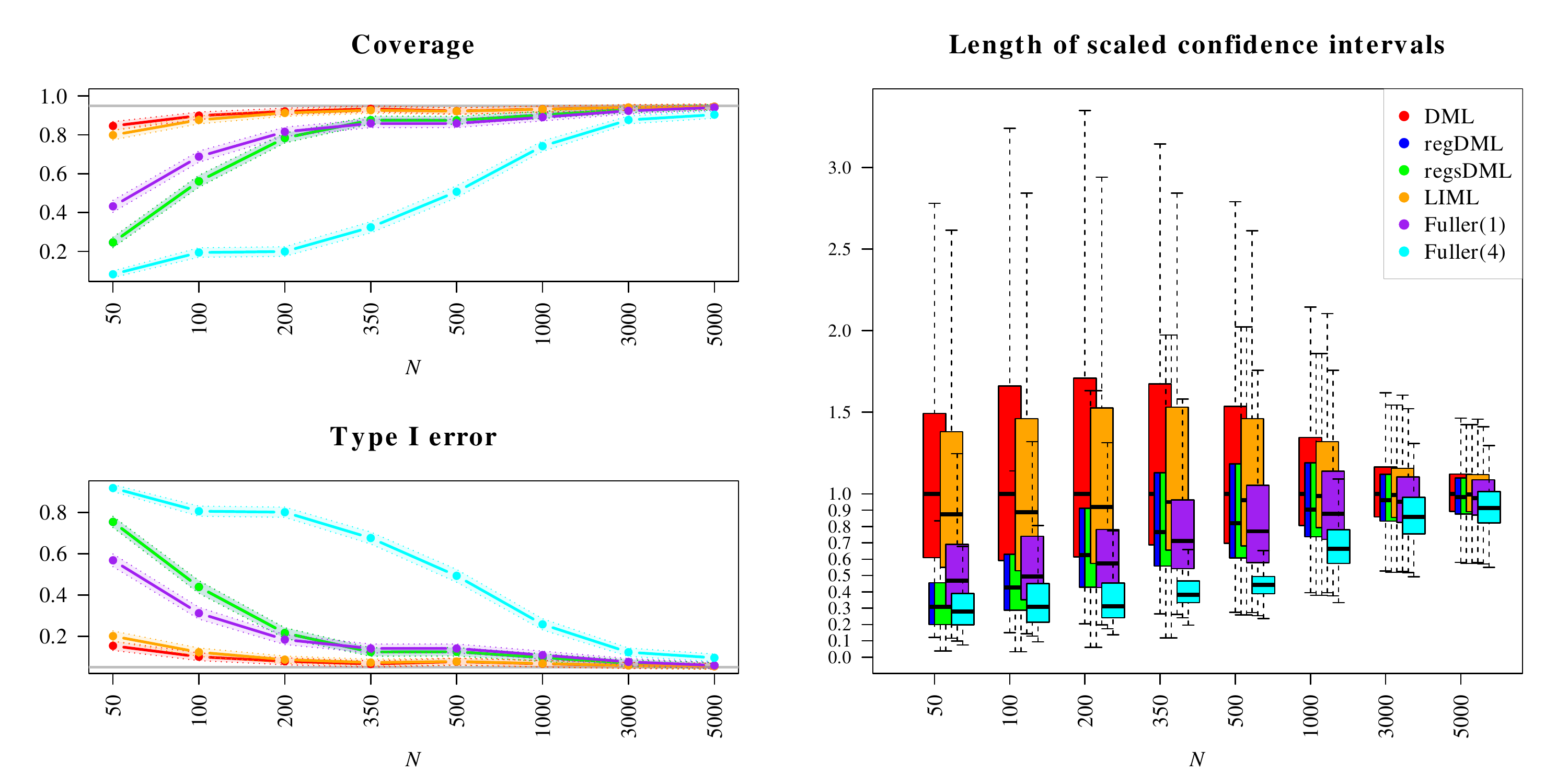}
\end{figure}

\begin{figure}[h!]
	\centering
	\caption[]{\label{fig:COVERWH-nonnoise} 
	The results come from $M=1000$ simulation runs from the SEM in Figure~\ref{fig:WH} with $\kappa=0.25$ and $\betazero=0$ for a range of sample sizes $\NN$ and with $\KK=2$ and $\Salg=100$ in Algorithm~\ref{algo:Summary}.
	The figure displays the coverage of two-sided confidence intervals for $\betazero$, type I error for two-sided testing of the 
	hypothesis $H_0:\ \betazero = 0$, and scaled lengths of two-sided confidence intervals of DML (red),  \regDML\ (blue),  \regsDML\  (green), 
	LIML (orange), Fuller(1) (purple), and Fuller(4) (cyan),
	where all results are at level $95\%$, and
	where the nuisance functions are estimated with additive splines. 
	At each sample size $\NN$, the lengths of the confidence intervals are scaled with the median length from DML. 
		The shaded regions in the coverage and the type I error plots represent $95\%$ confidence bands with respect to the $M$ simulation runs. 	 
	The blue and green lines are  indistinguishable in the left panel.
	}
	\includegraphics[width=\textwidth]{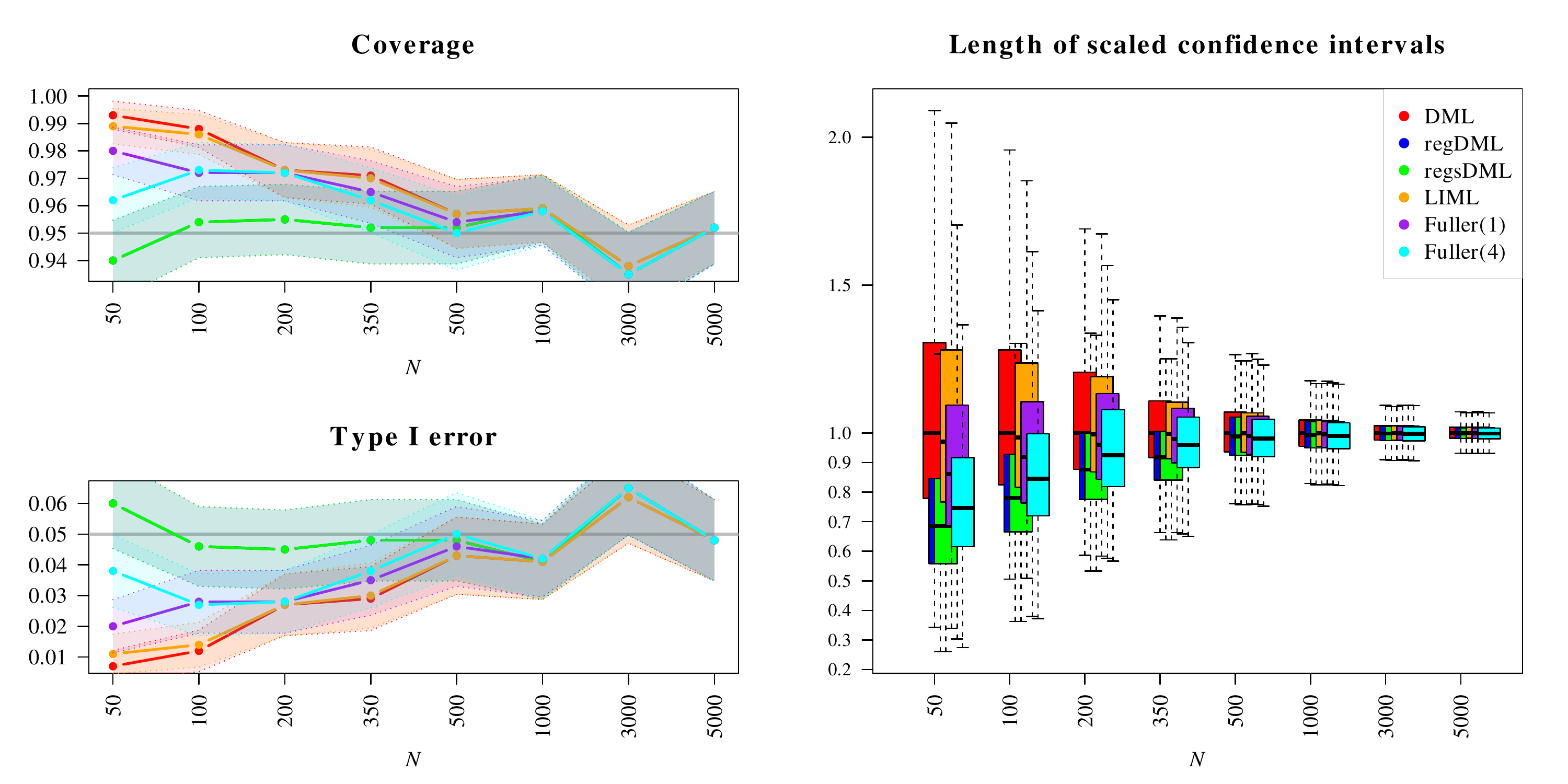}
\end{figure}

\subsection{Noise in $H\rightarrow W$}\label{sect:noisyW}

The variable $H$ may have a direct effect on $W$. If this link is strong enough with respect to the additional noise $\eps_W$ of $W$, it is possible to obtain some information of $H$ by observing $W$ similarly to Section~\ref{sect:noisyH}. 
The results again depend on the choice of functions in the model.

Figure~\ref{fig:COVERHW-bad} presents  simulation results where $H$ explains only little variation of $W$ compared with $\eps_W$. The confidence intervals of \regsDML\ do not attain a $95\%$ coverage for small sample sizes $\NN$ because the estimator inherits additional bias from DML. The situation can be improved by reducing the variation of $W$ that is not explained by $H$; see Figure~\ref{fig:COVERHW-good}. 

\begin{figure}[h!]
	\centering
	\caption[]{\label{fig:HW}An SEM and its associated causal graph.}
	\begin{tabular}{cc}
	\begin{tabular}{l}
	$\displaystyle 
	\begin{array}{r}
		(\eps_H, \eps_W,\eps_A, \eps_X, \eps_Y)\sim\mathcal{N}_5(\bo,\one)\\	
	\end{array}$
	\\
	$\displaystyle
	\begin{array}{lcl} 
			H &\leftarrow&\eps_H\\
			W &\leftarrow& 2H+\kappa \eps_W \\
			A &\leftarrow& e^{-0.5W}+0.5\eps_A\\
			X &\leftarrow& -A -0.1W^3
			-0.2W^2+0.4W\\
			&&\quad +\frac{7}{1+e^{-4H}}+0.25 \eps_X\\
			Y &\leftarrow& \betazero X+0.5W+0.5H+\eps_Y
	\end{array}$ 
	\end{tabular}
	& 
	\begin{tabular}{c}
          \includegraphics[width=0.35\textwidth]{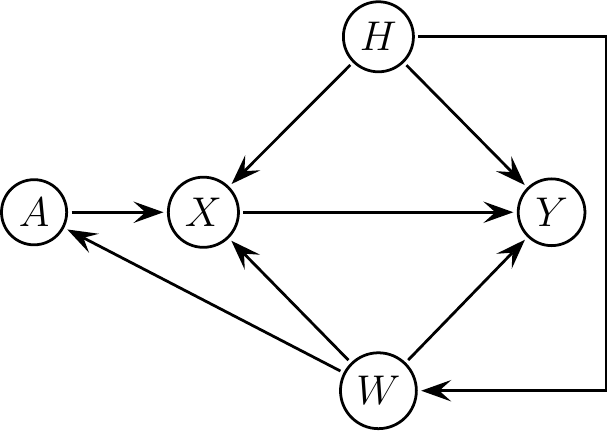}
	\end{tabular}
\end{tabular}
\end{figure}

\begin{figure}[h!]
	\centering
	\caption[]{\label{fig:COVERHW-bad} 
	The results come from $M=1000$ simulation runs from the SEM in Figure~\ref{fig:HW} with $\kappa=1$ and $\betazero=0$ for a range of sample sizes $\NN$ and with $\KK=2$ and $\Salg=100$ in Algorithm~\ref{algo:Summary}. 
	The nuisance functions are estimated with additive splines. 
	The figure displays the coverage of two-sided confidence intervals for $\betazero$, type I error for two-sided testing of the 
	hypothesis $H_0:\ \betazero = 0$, and scaled lengths of two-sided confidence intervals of DML (red),  \regDML\ (blue),  \regsDML\  (green), 
	LIML (orange), Fuller(1) (purple), and Fuller(4) (cyan), 
	where all results are at level $95\%$.
	At each sample size $\NN$, the lengths of the confidence intervals are scaled with the median length from DML. 
	The shaded regions in the coverage and the type I error plots represent $95\%$ confidence bands with respect to the $M$ simulation runs.
	The blue and green lines are indistinguishable in the left panel.
	}
	\includegraphics[width=\textwidth]{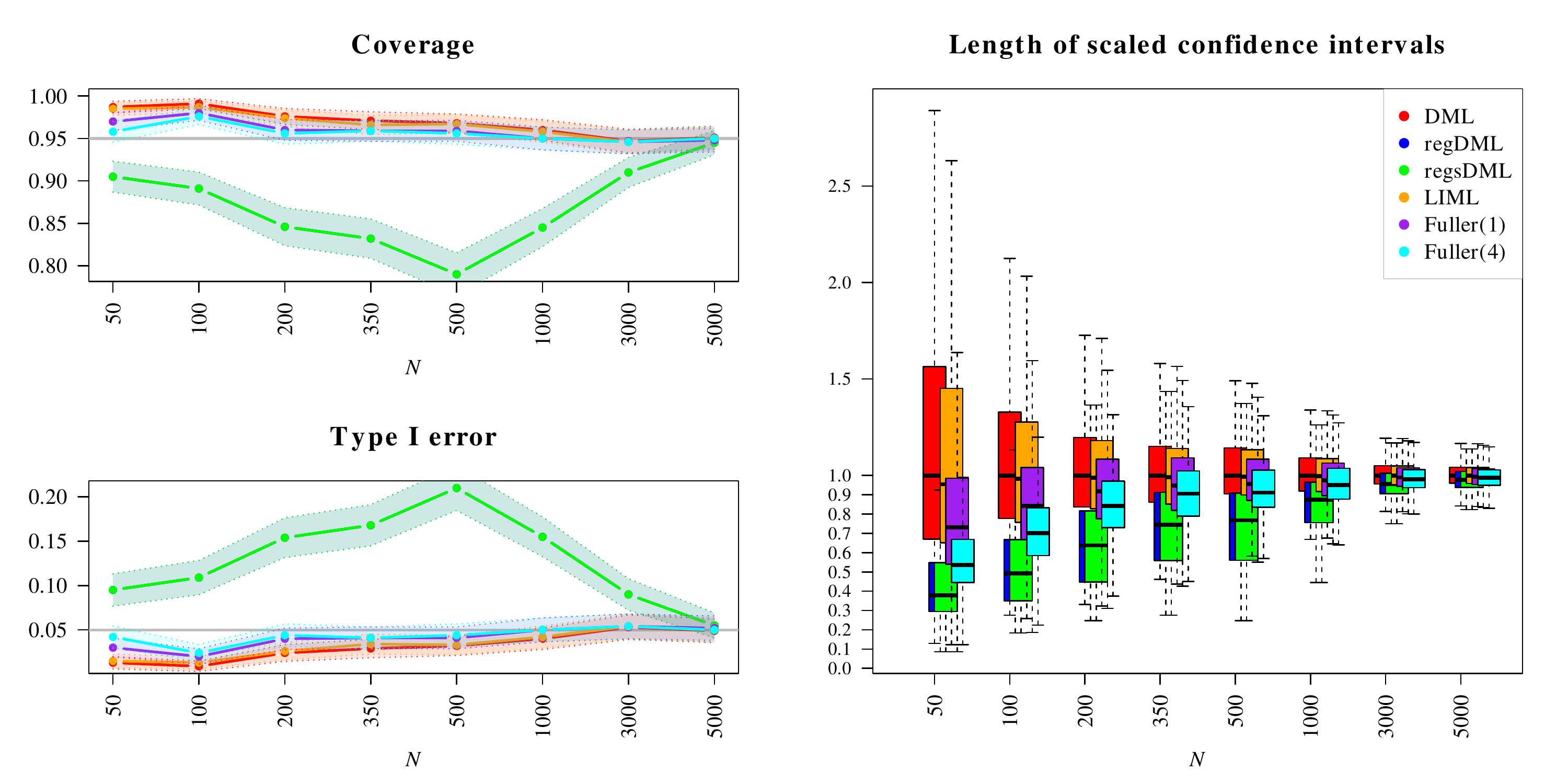}
\end{figure}

\begin{figure}[h!]
	\centering
	\caption[]{\label{fig:COVERHW-good} 
	The results come from $M=1000$ simulation runs from the SEM in Figure~\ref{fig:HW} with $\kappa=0.25$ and $\betazero=0$ for a range of sample sizes $\NN$ and with $\KK=2$ and $\Salg=100$ in Algorithm~\ref{algo:Summary}. 
	The nuisance functions are estimated with additive splines.
	The figure displays the coverage of two-sided confidence intervals for $\betazero$, type I error for two-sided testing of the 
	hypothesis $H_0:\ \betazero = 0$, and scaled lengths of two-sided confidence intervals of DML (red),  \regDML\ (blue),  \regsDML\  (green), 
	LIML (orange), Fuller(1) (purple), and Fuller(4) (cyan), 
	where all results are at level $95\%$.
	  At each sample size $\NN$, the lengths of the confidence intervals are scaled with the median length from DML. 
		The shaded regions in the coverage and the type I error plots represent $95\%$ confidence bands with respect to the $M$ simulation runs.
	The blue and green lines are indistinguishable in the left panel.
	}
	\includegraphics[width=\textwidth]{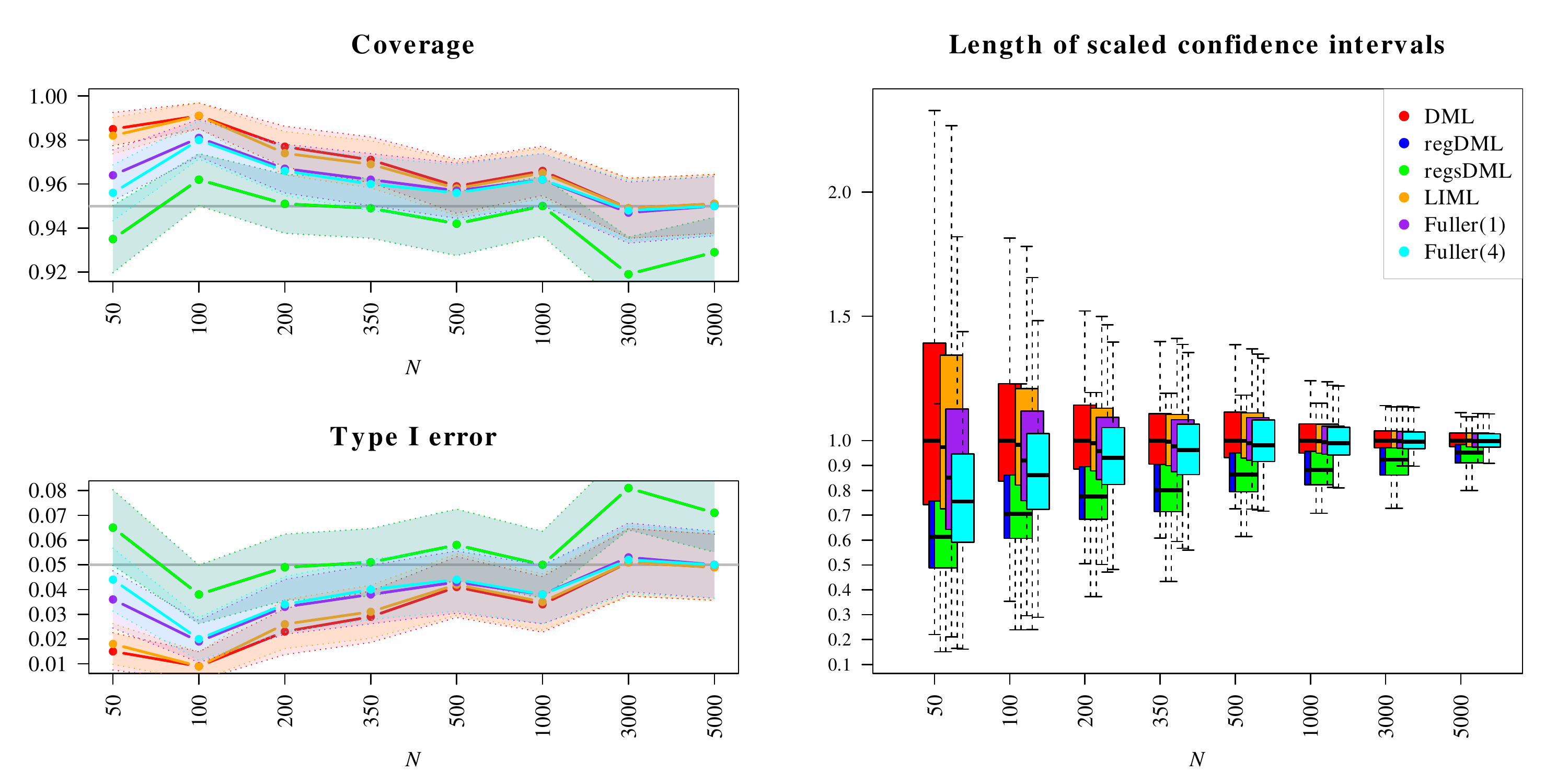}
\end{figure}

\section{Examples where the identifiability condition~\eqref{eq:identificationCondition} does and does not hold}\label{sect:discussionIdentifiabilityCondition}

The following examples illustrate SEMs where the identifiability condition~\eqref{eq:identificationCondition} holds and where it fails to hold. We argue using causal graphs; see~\citet{Lauritzen1996, Pearl1998, Pearl2009, Pearl2010, Peters2017,Maathuis2019}. 
By convention, we omit error variables in a causal graph if they are assumed to be mutually independent~\citep{Pearl2009}. 

\begin{example}\label{example:initialA}
Consider the SEM of the 1-dimensional variables $A$, $W$, $H$, $X$, and $Y$ and its associated causal graph given in Figure~\ref{fig:exampleA}, where $\betazero$ is a fixed unknown parameter, and where $\aW$, $\aX$, $\gY$, $\gH$, $\hX$, and $\hY$ are some appropriate functions. 
The variable $A$ directly influences $W$, and $W$ directly influences the hidden variable $H$. 
The variable $A$ is independent of $H$ given $W$ because every path from $A$ to $H$ is blocked by $W$; a proof is given in the appendix in Section~\ref{appendix:proofsOfSect2}. 

\begin{figure}[ht]
	\centering
	\caption[]{\label{fig:exampleA}An SEM satisfying the identifiability condition~\eqref{eq:identificationCondition} and its associated causal graph as in Example~\ref{example:initialA}. 
	}
	\begin{tabular}{cc}
	\begin{tabular}{l}
	$\displaystyle 
	\begin{array}{r}
		\eps_A, \eps_W,\eps_H, \eps_X, \eps_Y	
	\end{array}$
	\\
	$\displaystyle \begin{array}{lcl} 
		A &\leftarrow& \eps_A\\
		W &\leftarrow& \aW(A) +\eps_W\\
		H &\leftarrow& \gH(W) + \eps_H\\
		X &\leftarrow&  \aX(A) + \hX(H) + \eps_X\\
		Y &\leftarrow& \betazero X + \gY(W) + \hY(H) + \eps_Y 
	\end{array}$ 
	\end{tabular}
	& 
	\begin{tabular}{c}
      \includegraphics[width=0.35\textwidth]{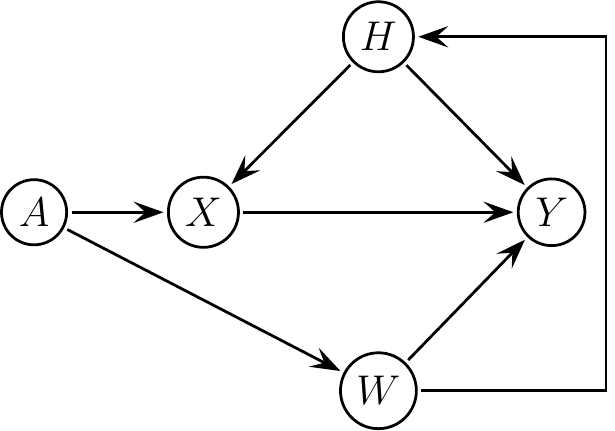}
	\end{tabular}
\end{tabular}
\end{figure}
\end{example}
\begin{proof}[Proof of Example~\ref{example:initialA}]
The path $A\rightarrow X\leftarrow H$ is blocked by the empty set because $X$ is a collider on this path. The paths $A\rightarrow \cdots\rightarrow Y\leftarrow H$ are blocked by the empty set because $Y$ is a collider on these paths. The path $A\rightarrow W\rightarrow H$ is blocked by $W$.
\end{proof}

The variable $A$ is exogenous in  Example~\ref{example:initialA}. In general, this is no requirement; see  Example~\ref{example:initialB}.

\begin{example}\label{example:initialB}
Consider the SEM of the 1-dimensional variables $H$, $W$, $A$, $X$, and $Y$ and its associated causal graph given in Figure~\ref{fig:exampleB}, where $\betazero$ is a fixed unknown parameter, and where $\aX$, $\gA$, $\gX$, $\gY$, $\hX$, $\hW$, and $\hY$ are some appropriate functions. 
The variable $A$ is not a source node. The hidden variable $H$ directly influences $W$, and $W$ directly influences $A$. 
The variable $A$ is independent of $H$ given $W$ because every path from $A$ to $H$ is blocked by $W$; a proof is given in the appendix in Section~\ref{appendix:proofsOfSect2}. 

\begin{figure}[ht]
	\centering
	\caption[]{\label{fig:exampleB}An SEM satisfying the identifiability condition~\eqref{eq:identificationCondition} and its associated causal graph as in Example~\ref{example:initialB}.
	}
	\begin{tabular}{cc}
	\begin{tabular}{l}
	$\displaystyle 
	\begin{array}{r}
		\eps_H,\eps_W,\eps_A,  \eps_X, \eps_Y	
	\end{array}$
	\\
	$\displaystyle \begin{array}{lcl} 
		H &\leftarrow& \eps_H\\
		W &\leftarrow& \hW(H) + \eps_W\\
		A &\leftarrow& \gA(W)+\eps_A\\
		X &\leftarrow& \aX(A) +\gX(W) +  \hX(H) + \eps_X\\
		Y &\leftarrow& \betazero X + \gY(W) + \hY(H) + \eps_Y 
	\end{array}$ 
	\end{tabular}
	& 
	\begin{tabular}{c}
      \includegraphics[width=0.35\textwidth]{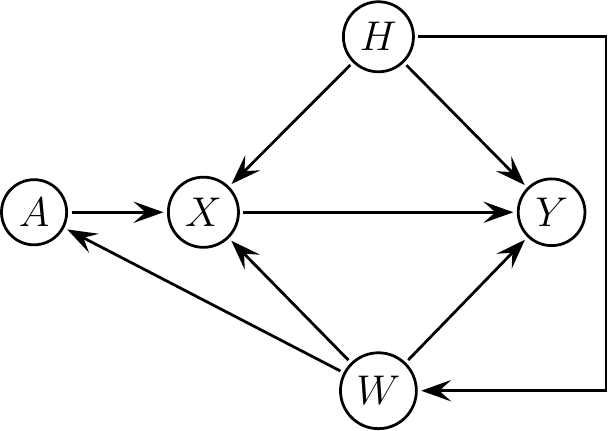}
	\end{tabular}
\end{tabular}
\end{figure}
\end{example}
\begin{proof}[Proof of Example~\ref{example:initialB}]
The path $A\rightarrow X\leftarrow H$ is blocked by the empty set because $X$ is a collider on this path. The paths $A\rightarrow X\rightarrow\cdots\rightarrow Y\leftarrow H$ are blocked by the empty set because $Y$ is a collider on these paths. 
The paths $A\leftarrow W \rightarrow Y \leftarrow X \leftarrow H$,
 $A\leftarrow W\leftarrow H$, and $A\rightarrow X\leftarrow W\leftarrow H$ are blocked by $W$. 
The path $A\leftarrow W\rightarrow Y\leftarrow H$ is blocked by $W$ or alternatively by the empty set because $Y$ is a collider on this path. 
The path $A\leftarrow W\rightarrow X\leftarrow H$ is blocked by $W$ or alternatively by the empty set because $X$ is a collider on this path. 
\end{proof}

Identifiability of $\betazero$ is not guaranteed if  $A$ and $H$ are independent. 
An illustration is given in Example~\ref{example:initialC}. Considering the instrument $A$ instead of $\Ra$ in Theorem~\ref{thm:identifiability} cannot solve the issue. In such a situation, stronger structural assumptions are required. 

\begin{example}\label{example:initialC}
Consider the SEM of the 1-dimensional variables $H$, $A$, $W$, $X$, and $Y$ and its associated causal graph given in Figure~\ref{fig:exampleC}, where $\betazero$ is a fixed unknown parameter. 
Although $A$ and $H$ are independent, the identifiability condition~\eqref{eq:identificationCondition} does not hold; a proof is given in the appendix in Section~\ref{appendix:proofsOfSect2}. 

\begin{figure}[ht]
	\centering
	\caption[]{\label{fig:exampleC}An SEM not satisfying the identifiability condition~\eqref{eq:identificationCondition} together with its associated causal graph as in Example~\ref{example:initialC} 
	}
	\begin{tabular}{cc}
	\begin{tabular}{l}
	$\displaystyle 
	\begin{array}{r}
		(\eps_H, \eps_A, \eps_W, \eps_X, \eps_Y)\sim\mathcal{N}_5(0,\one)
	\end{array}$
	\\
	$\displaystyle \begin{array}{lcl} 
		H &\leftarrow& \eps_H\\
		A &\leftarrow& \eps_A\\
		W &\leftarrow& A+H + \eps_W\\
		X &\leftarrow& A + W + H + \eps_X\\
		Y &\leftarrow& \betazero X + W + H + \eps_Y 
	\end{array}$ 
	\end{tabular}
	& 
	\begin{tabular}{c}
      \includegraphics[width=0.35\textwidth]{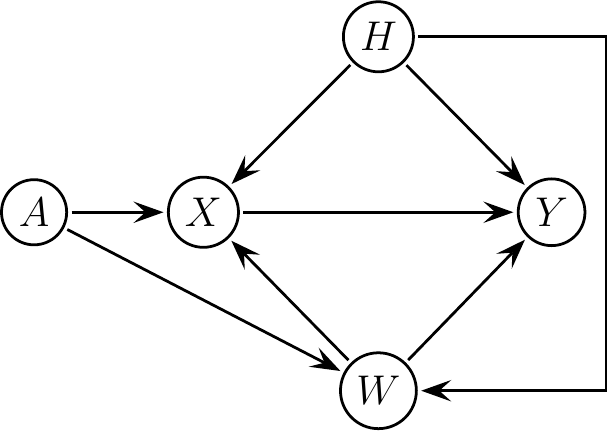}
	\end{tabular}
\end{tabular}
\end{figure}
\end{example}
\begin{proof}[Proof of Example~\ref{example:initialC}]
The two random variables $A$ and $H$ are independent because the path $A\rightarrow W\leftarrow H$ is not blocked by $W$.  Indeed, $W$ is a collider on this path. 

All random variables are 1-dimensional. Therefore, the representation of $\betazero$ in Theorem~\ref{thm:identifiability} is equivalent to the identifiability condition
\begin{displaymath}
	\E[\Ra(\Ry-\Rx\betazero)] =0
\end{displaymath}
in Equation~\eqref{eq:identificationCondition}. However, the identifiability condition does not hold in the present situation. We have
\begin{displaymath}
		\begin{array}{cl}
			&\E[\Ra(\Ry-\Rx\betazero)]\\
			=&\E[\Ra\big(H+ \eps_Y-\E[ H+\eps_Y|W]\big)\big]\\
			=&\E\big[\Ra\big(H-\E[ H|W]\big)\big]
		\end{array}
\end{displaymath}
because $\eps_Y$ is independent of $A$ and $W$ and centered. By the tower property for conditional expectations, we have
\begin{displaymath}
\begin{array}{cl}
	\E[\Ra(\Ry-\Rx\betazero)] = \E\big[AH-A\E[ H|W]\big].
	\end{array}
\end{displaymath}
Because $A$ and $H$ are independent and centered, we have $\E[AH]=0$. Moreover, we have  $H\sim\mathcal{N}(0,1)$, $W\sim\mathcal{N}(0,3)$, and $(W|H=h)\sim\mathcal{N}(h,2)$. The conditional distribution of $H|W=w$ can be obtained by applying Bayes' theorem and is given by $\mathcal{N}(\frac{1}{3}w, \frac{2}{3})$. Hence, we have $\E[H|W]=\frac{1}{3}W$ and 
\begin{displaymath}
	\E\big[ A\E[H|W] \big] = \frac{1}{3}\E[AW] = \frac{1}{3}\E\big[ A^2\big]=\frac{1}{3}\neq 0
\end{displaymath}
because $A$ is independent of $H$ and $\eps_W$. Therefore, we have $\E[\Ra(\Ry-\Rx\betazero)]\neq 0$ and $\betazero$ cannot be represented as in Theorem~\ref{thm:identifiability}.
\end{proof}

	\section{Proofs of Section~\ref{sect:identifiabilityConditionAndDML}}\label{appendix:proofsOfSect2}

\begin{proof}[Proof of Theorem~\ref{thm:identifiability}]
	To prove the theorem, 
	we need to verify 
	\begin{displaymath}
		\betazero= \Big(\E\big[\Rx\Ra^T\big]\E\big[\Ra\Ra^T\big]^{-1}\E\big[\Ra\Rx^T\big] \Big)^{-1}\E\big[\Rx\Ra^T\big]\E\big[\Ra\Ra^T\big]^{-1}\E[\Ra\Ry].	
	\end{displaymath}
	This statement is equivalent to
	\begin{displaymath}
		\bo = \E\big[\Rx\Ra^T\big]\E\big[\Ra\Ra^T\big]^{-1}\E\big[\Ra\big(\Ry-\Rx^T\betazero)\big].
	\end{displaymath}
	This last statement holds because $\E[\Ra(\Ry-\Rx^T\betazero)]$ equals $\bo$ due to the identifiability condition~\eqref{eq:identificationCondition}.
\end{proof}

\section{Proofs of Section~\ref{sect:DML}}\label{sect:proofsOfDML}

We denote by $\norm{\cdot}$ either the Euclidean norm for a vector or the operator norm for a matrix.

\begin{proof}[Proof of Proposition~\ref{prop:Neyman_orth}]
   	We have
	\begin{displaymath}
	\begin{array}{cl}
		&\frac{\partial}{\partial r}\Big|_{r=0} \EP\big[\loss\big(S;\betazero,\etazero+r(\eta-\etazero)\big)\big]\\
			=&\frac{\partial}{\partial r}\Big|_{r=0}  \EP\bigg[ \Big(A-\mA^0(W)-r\big(\mA(W)-\mA^0(W)\big)\Big)\\
			&\quad\quad\quad\cdot
			\bigg(Y-\mY^0(W)-r\big(\mY(W)-\mY^0(W)\big)\\
			&\quad\quad\quad\quad\quad-\Big( X-\mX^0(W)-r\big(\mX(W)-\mX^0(W)\big) \Big)^T\betazero\bigg) \bigg]\\
			=&  \EP\Big[ -\big(m_A(W)-\mA^0(W)\big)\Big(Y-\mY^0(W)-\big( X-\mX^0(W) \big)^T\betazero\Big) \\
			&\quad\quad\quad\quad +
			\big(A-\mA^0(W)\big)\Big(-\big(\mY(W)-\mY^0(W)\big)+\big( \mX(W)-\mX^0(W) \big)^T\betazero\Big)  \Big]. 
	\end{array}
\end{displaymath}
Subsequently, we show that both terms
\begin{equation}\label{eq:Neyman_orth_term1}
	\EP\Big[ \big(m_A(W)-\mA^0(W)\big)\Big(Y-\mY^0(W)-\big( X-\mX^0(W) \big)^T\betazero\Big)\Big]
\end{equation}
and
\begin{equation}\label{eq:Neyman_orth_term2}
	\EP\Big[  \big(A-\mA^0(W)\big)\Big(-\big(\mY(W)-\mY^0(W)\big)+\big( \mX(W)-\mX^0(W) \big)^T\betazero\Big)\Big]
\end{equation}
are equal to $\bo$. 
We first consider the term~\eqref{eq:Neyman_orth_term1}. Recall the notations $\mY^0(W)=\EP[Y|W]$ and $\mX^0(W)=\EP[X|W]$.  
We have
\begin{displaymath}
	\begin{array}{cl}
	&\EP\Big[ \big(m_A(W)-\mA^0(W)\big)\Big(Y-\mY^0(W)-\big( X-\mX^0(W) \big)^T\betazero\Big)\Big]\\
	=& \EP\Big[ \big(m_A(W)-\mA^0(W)\big)\EP\big[Y-\EP[Y|W]-( X-\EP[X|W] )^T\betazero\big|W\big]\Big]\\
	=&\bo. 
	\end{array}
\end{displaymath}
Next, we verify that the term given in~\eqref{eq:Neyman_orth_term2} vanishes. Recall the notation $\mA^0(W)=\EP[A|W]$.  
We have
\begin{displaymath}
	\begin{array}{cl}
		&\EP\Big[  \big(A-\mA^0(W)\big)\Big(-\big(\mY(W)-\mY^0(W)\big)+\big( \mX(W)-\mX^0(W) \big)^T\betazero\Big)\Big]\\
		=& \EP\Big[  \EP\big[A-\E[A|W]\big|W\big]\Big(-\big(\mY(W)-\mY^0(W)\big)+\big( \mX(W)-\mX^0(W) \big)^T\betazero\Big)\Big]\\
		=&\bo.
	\end{array}
\end{displaymath}
	Because both terms~\eqref{eq:Neyman_orth_term1} and~\eqref{eq:Neyman_orth_term2} vanish, we conclude
	\begin{displaymath}
		\frac{\partial}{\partial r}\Big|_{r=0} \EP\big[\loss\big(S;\betazero,\etazero+r(\eta-\etazero)\big)\big]=\bo. 
	\end{displaymath}
\end{proof}	

\begin{definition}\label{def:lossFunc}
Consider a set $\TauN$ of nuisance functions. 
For $S=(A,X,W,Y)$, an element $\eta=(\mA,\mX,\mY)\in \TauN$, and $\beta\in\R^d$, we introduce the score functions
\begin{equation}\label{eq:score_psitilde}
	\losstilde(S,\beta,\eta):= \big(X-\mX(W)\big)\Big(Y-\mY(W) - \big(X-\mX(W)\big)^T\beta\Big),
\end{equation}
and 
\begin{displaymath}
	\begin{array}{rcl}
	\lossone(S,\eta)&:=& \big(X-\mX(W)\big)\big(A-\mA(W)\big)^T,\\
	\losstwo(S,\eta)&:=& \big(A-\mA(W)\big)\big(A-\mA(W)\big)^T,\\
	\lossthree(S,\eta)&:=& \big(X-\mX(W)\big)\big(X-\mX(W)\big)^T. 
	\end{array}
\end{displaymath}
	Furthermore, let the matrices
\begin{displaymath}
	\begin{array}{rcl}
		\matA &:=& \EP[\lossthree(S;\etazero)],\\
		\matB &:=& \EP[\lossone(S;\etazero)] \EP[\losstwo(S;\etazero)]^{-1} \EP\big[\lossone^T(S;\etazero)\big],\\
		\matC &:=& \EP[\lossone(S;\etazero)] \EP[\losstwo(S;\etazero)]^{-1},\\
		\matE&:=& \EP[\losstwo(S;\etazero)]^{-1}\EP[\loss(S;\bg,\etazero)], \\
		\Jzero &:=& \matB^{-1}\matC,\\
		\tilJzero &:=& \EP\big[\loss(S;\betazero,\etazero)\loss^T(S;\betazero,\etazero)\big]=\E\big[\Ra\Ra^T(\Ry-\Rx^T\betazero)^2\big],\\
		\Jzerodp &:=&\EP[\Ra\Ra^T],\\
		\Jzerop &:=& \EP\big[\Rx(\Ra)^T\big](\Jzerodp)^{-1}\EP\big[\Ra(\Rx)^T\big]
	\end{array}
\end{displaymath}
and the variance-covariance matrix $\sigma^2 :=\Jzero\tilJzero\Jzero^T$. Moreover, let the score function 
\begin{displaymath}
	\lossoverline(\cdot;\betazero,\etazero):=\sigma^{-1}\tilJzero^{-\frac{1}{2}}\loss(\cdot;\betazero,\etazero).
\end{displaymath} 
\end{definition}

\begin{definition}\label{def:statisticalRates}
Let $\gamma\ge 0$. 
Consider a realization set $\TauN$ of nuisance functions. Define the statistical rates 
	\begin{displaymath}
		\rNpnumber^4 := \max_{\substack{S=(U,V,W,Z)\in\{A, X, Y\}^2\times\{W\}\times\{A,X,Y\}, \\\bzero\in\{\bg,\betazero, \bo\}}}
		\sup_{\eta\in\TauN} \EP[ \norm{\loss(S;\bzero,\eta)- \loss(S;\bzero,\etazero)} ],
	\end{displaymath}
	\begin{displaymath}
		\lambdaNpnumber := \max_{\substack{\losstest\in\{\loss,\losstilde,\losstwo\}^, \\\bzero\in\{\bg,\betazero, \bo\}}}
		\sup_{r\in(0,1), \eta\in\TauN} \normbig{ \partial_r^2\EP\big[\losstest\big(S;\bzero,\etazero+r(\eta-\etazero)\big)\big] },
	\end{displaymath}
	where we interpret $\losstwo\big(S;\bzero,\etazero+r(\eta-\etazero)\big)$ as $\losstwo\big(S;\etazero+r(\eta-\etazero)\big)$ in the definition of $\lambdaNpnumber$.
\end{definition}

\begin{remark}
We would like to remark that the  respective definition of the statistical rate $\rNpnumber$ given in~\citet{Chernozhukov2018}  involves  the $L_2$-norm of $\loss(S;\bzero,\eta)- \loss(S;\bzero,\etazero)$  instead of its $L_1$-norm. However, it is essential to employ the $L_1$-norm  to ensure that Assumption~\ref{assumpt:DMLboth5} can constrain the $L_2$-norm of the estimation errors incurred by the ML estimators of the nuisance parameters. Thus, we do not have to constrain their higher order errors to employ H{\"o}lder's inequality in Lemma~\ref{lem:boundRN}. 
\end{remark}

\begin{definition}\label{def:asymptNormal}
Let the nonrandom numbers
\begin{displaymath}
\rhoN := \rNpnumber + \NN^{\frac{1}{2}}\lambdaNpnumber \quad\textrm{and}\quad
\rhoNtilde := \NN^{\max\big\{\frac{4}{p}-1, -\frac{1}{2}\big\}}+\rNpnumber. 
\end{displaymath}
\end{definition}

If not stated otherwise, we assume the following Assumption~\ref{assumpt:DMLboth}
in all the results presented in the appendix. 

\begin{assumptions}\label{assumpt:DMLboth}
Let $\gamma\ge 0$. 
	Let $\KK\ge 2$ be a fixed integer independent of $\NN$. We assume that $\NN\ge\KK$ holds. 
Let $\{\deltaN\}_{\NN\ge \KK}$ and $\{\DeltaN\}_{\NN\ge \KK}$ be two sequences of positive numbers that converge to zero, where $\deltaN^{\frac{1}{4}}\ge \NN^{-\frac{1}{2}}$ holds. 
Let $\{\PcalN\}_{\NN\ge 1}$ be a sequence of sets of probability distributions $\PP$ of the quadruple $S=(A,W,X,Y)$.
	
	Let $p>4$. 
	For all $\NN$, for all $\PP\in\PcalN$, 
	consider a nuisance function realization sets $\TauN$ such that the following conditions hold: 
	\begin{enumerate}[label={\theassumptions.\arabic*}]
		\item\label{assumpt:DMLboth1}
			We have an SEM given by~\eqref{eq:SEM} that satisfies the identifiability conditon~\eqref{eq:identificationCondition}. 
		\item\label{assumpt:DMLboth2}
			There exists a finite real constant $\CpnormRV$ satisfying $\normP{A}{p}+\normP{X}{p}+\normP{Y}{p}\le \CpnormRV$.
		\item\label{assumpt:DMLboth3}
			The matrix $\EP[\Rx\Ra^T]\in\R^{d\times q}$ has full rank  $d$. This in particular requires $q\ge d$.  
		The matrices $\matA\in\R^{d\times d}$ and  $\Jzerodp \in\R^{q\times q}$ are invertible.
		Furthermore, the smallest and largest singular values of the symmetric matrices $\Jzerodp$
	and $\Jzerop$ 
	are bounded away from $0$ by $\cone>0$ and are bounded away from $+\infty$ by $\ctwo<\infty$. 	 
		\item\label{assumpt:DMLboth4}
		The symmetric matrices $\tilJzero$, $\matA+(\gamma-1)\matB$, and $\matD$ are invertible, where $\matD$ is introduced in Definition~\ref{def:asymptNormalgamma} in the appendix in Section~\ref{sect:proofsRegularizedDML}.
		The smallest and largest singular values of these matrices are bounded away from $0$ by $\cthree$ and are bounded away from $+\infty$ by $\cfour$. 
		\item\label{assumpt:DMLboth5}
		The set $\TauN$ consists of $\PP$-integrable functions $\eta=(\mA,\mX,\mY)$ whose $p$th moment exists and it contains $\etazero$.  There exists a finite real constant $\CpnormEta$ such that
		\begin{displaymath}
			\begin{array}{l}
			\normP{\etazero-\eta}{p}\le \CpnormEta, 
			\quad
			\normP{\etazero-\eta}{2}\le \deltaNnumber, \quad
			\normP{\mA^0(W)-\mA(W)}{2}^2
			\le\deltaNnumber\NN^{-\frac{1}{2}},\\
			\normP{\mX^0(W)-\mX(W)}{2}\big(\normP{\mY^0(W)-\mY(W)}{2} + \normP{\mX^0(W)-\mX(W)}{2}\big)
			\le\deltaNnumber\NN^{-\frac{1}{2}},\\
			\normP{\mA^0(W)-\mA(W)}{2}\big(\normP{\mY^0(W)-\mY(W)}{2}+\normP{\mX^0(W)-\mX(W)}{2}\big)\le\deltaNnumber\NN^{-\frac{1}{2}}
			\end{array}
		\end{displaymath}
		hold for all elements $\eta$ of $\TauN$. 
		Given  a partition $I_1,\ldots,I_{\KK}$ of $\indset{\NN}$ where each $I_{\kk}$ is of size $\nn=\frac{\NN}{\KK}$, 
		for all $\kk\in\indset{\KK}$, the nuisance parameter estimate $\hetaIkc=\hetaIkc(\SIkc)$ satisfies 
		\begin{displaymath}
			\begin{array}{l}
			\normP{\etazero-\hetaIkc}{p}\le \CpnormEta,
			\quad
			\normP{\etazero-\hetaIkc}{2}\le \deltaNnumber, \quad
			\normP{\mA^0(W)-\hmA^{\Ikc}(W)}{2}^2 
			\le\deltaNnumber\NN^{-\frac{1}{2}},\\
			\normP{\mX^0(W)-\hmX^{\Ikc}(W)}{2}\big(\normP{\mY^0(W)-\hmY^{\Ikc}(W)}{2}+ \normP{\mX^0(W)-\hmX^{\Ikc}(W)}{2} \big)\le\deltaNnumber\NN^{-\frac{1}{2}},\\
			\normP{\mA^0(W)-\hmA^{\Ikc}(W)}{2}\big(\normP{\mY^0(W)-\hmY^{\Ikc}(W)}{2}+\normP{\mX^0(W)-\hmX^{\Ikc}(W)}{2}\big) \le\deltaNnumber\NN^{-\frac{1}{2}}
			\end{array}
		\end{displaymath}
		with $\PP$-probability no less than $1-\DeltaNnumber$.
		Denote by $\EpsN$ the event that $\hetaIkc=\hetaIkc(\SIkc)$ belongs to $\TauN$ and assume that this event  holds with $\PP$-probability no less than $1-\DeltaNnumber$. 
	\end{enumerate}
\end{assumptions}

For instance, 
invertibility of the square matrices $\EP[\Ra\Ra^T]$ and $\tilJzero$ is satisfied if $\eps_Y$ is independent of both $A$ and $W$ and has a strictly positive variance. 

\begin{remark}
	It is possible to drop some of the assumptions in Assumption~\ref{assumpt:DMLboth} if we are interested in proving the results about DML only. 
	The full assumption is required to prove the results about both DML and \regDML. 
\end{remark}

We assume Assumption~\ref{assumpt:DMLboth} throughout. 

	\begin{lemma}\label{lem:boundPnorm1}
	Let $\ttt\ge 1$. 
	Consider a $t$-dimensional random variable $Z$. Denote the joint law of $Z$ and $W$ by $\PP$. Then we have 
	\begin{displaymath}
		\normP{Z-\EP[Z|W]}{{\ttt}}\le 2\normP{Z}{\ttt}. 
	\end{displaymath}
\end{lemma}
\begin{proof}[Proof of Lemma~\ref{lem:boundPnorm1}]
	Because the Euclidean norm to the $\ttt$th power is convex for ${\ttt}\ge 1$, we have
	\begin{displaymath}
		\normP{\EP[Z|W]}{\ttt}^{\ttt} = \EP\big[\norm{\EP[Z|W]}^{\ttt}\big]\le \EP\big[ \EP[ \norm{Z}^{\ttt}|W ] \big] = \EP[\norm{Z}^{\ttt}] = \normP{Z}{{\ttt}}^{\ttt}
	\end{displaymath}
	by Jensen's inequality. We hence have
	\begin{displaymath}
		\begin{array}{rcl}
			\normP{Z-\EP[Z|W]}{{\ttt}}\le\normP{Z}{{\ttt}} + \normP{\EP[Z|W]}{{\ttt}}\le 2\normP{Z}{{\ttt}}
		\end{array}
	\end{displaymath}
	by the triangle inequality. 
\end{proof}

\begin{lemma}\label{lem:boundPnorm2}
	Consider a $t$-dimensional random variable $Z$. 
	Denote the joint law of $Z$ and $W$ by $\PP$. Then we have 
	\begin{displaymath}
		\normbig{\EP\big[ZZ^T-\EP[Z|W]\EP[Z^T|W]\big]}\le 2 \normP{Z}{2}^2. 
	\end{displaymath}
\end{lemma}
\begin{proof}[Proof of Lemma~\ref{lem:boundPnorm2}]
	Because the Euclidean norm is convex, we have 
	\begin{displaymath}
		\begin{array}{rcl}
		\normbig{\EP\big[ZZ^T-\EP[Z|W]\EP[Z^T|W]\big]} 
		&\le& \EP\big[\norm{ZZ^T}+\norm{\EP[Z|W]\EP[Z^T|W]}\big]\\
		&\le& \EP\big[\norm{Z}^2+\norm{\EP[Z|W]}^2 \big]\\
		\end{array}
	\end{displaymath}
	by Jensen's inequality, the triangle inequality and the Cauchy--Schwarz inequality. 
	Because the squared Euclidean norm is convex, we have 
	\begin{displaymath}
		\norm{\EP[Z|W]}^2 \le \EP\big[\norm{Z}^2\big|W\big]
	\end{displaymath}
	by Jensen's inequality. Therefore, we have
	\begin{displaymath}
		\begin{array}{rcl}
		\normbig{\EP\big[ZZ^T-\EP[Z|W]\EP[Z^T|W]\big]} 
		&\le& \EP\big[\norm{Z}^2+\norm{\EP[Z|W]}^2 \big]\\
		&\le& \EP\big[\norm{Z}^2+\EP[\norm{Z}^2|W] \big] \\
		&=& 2\normP{Z}{2}^2. 
		\end{array}
	\end{displaymath}
\end{proof}

\begin{lemma}\label{lem:boundPnorm4}
	Consider a $t_1$-dimensional random variable $Z_1$ and a $t_2$-dimensional  random variable $Z_2$. 
	Denote the joint law of $Z_1$, $Z_2$, and $W$ by $\PP$. Then we have 
	\begin{displaymath}
		\normbig{\EP\big[(Z_1-\EP[Z_1|W])(Z_2-\EP[Z_2|W])^T\big]}^2\le \normP{Z_1}{2}^2\normP{Z_2}{2}^2.
	\end{displaymath}
\end{lemma}
\begin{proof}[Proof of Lemma~\ref{lem:boundPnorm4}]
	By the  Cauchy--Schwarz inequality, we have 
	\begin{displaymath}
	\begin{array}{cl}
		&\normbig{\EP\big[(Z_1-\EP[Z_1|W])(Z_2-\EP[Z_2|W])^T\big]}^2\\
		\le& \EP\big[  \norm{(Z_1-\EP[Z_1|W])}^2\big]\EP\big[\norm{(Z_2-\EP[Z_2|W])}^2 \big]. 
	\end{array}
	\end{displaymath}
	Because the conditional expectation minimizes the mean squared error~\citep[Theorem 5.1.8]{Durrett2010}, we have
	\begin{displaymath}
		 \EP\big[  \norm{(Z_1-\EP[Z_1|W])}^2\big] \le \normP{Z_1}{2}^2
	\end{displaymath}
	and 
	\begin{displaymath}
		 \EP\big[  \norm{(Z_2-\EP[Z_2|W])}^2\big] \le \normP{Z_2}{2}^2. 
	\end{displaymath}
	In total, we thus have
	\begin{displaymath}
		\normbig{\EP\big[(Z_1-\EP[Z_1|W])(Z_2-\EP[Z_2|W])^T\big]}^2
		\le \normP{Z_1}{2}^2\normP{Z_2}{2}^2. 
	\end{displaymath}
\end{proof}

\begin{lemma}\label{lem:boundPnorm3}
	Consider a $t_1$-dimensional random variable $Z_1$ and a $t_2$-dimensional  random variable $Z_2$. 
	Denote the joint law of $Z_1$, $Z_2$, and $W$ by $\PP$. Then we have 
	\begin{displaymath}
		\normbig{\EP\big[(Z_1-\EP[Z_1|W])Z_2^T\big]}^2\le \normP{Z_1}{2}^2\normP{Z_2}{2}^2.
	\end{displaymath}
\end{lemma}
\begin{proof}[Proof of Lemma~\ref{lem:boundPnorm3}]
	By the  Cauchy--Schwarz inequality, we have 
	\begin{displaymath}
		\normbig{\EP\big[(Z_1-\EP[Z_1|W])Z_2^T\big]}^2
		\le \EP\big[  \norm{Z_1-\EP[Z_1|W]}^2\big]\EP\big[\norm{Z_2}^2 \big].
	\end{displaymath}
	Because the conditional expectation minimizes the mean squared error~\citep[Theorem 5.1.8]{Durrett2010}, we have
	\begin{displaymath}
		\EP\big[  \norm{Z_1-\EP[Z_1|W]}^2\big] \le \EP\big[  \norm{Z_1}^2\big] = \normP{Z_1}{2}^2. 
	\end{displaymath}
	Consequently, 
	\begin{displaymath}
		\normbig{\EP\big[(Z_1-\EP[Z_1|W])Z_2^T\big]}^2\le \normP{Z_1}{2}^2\normP{Z_2}{2}^2
	\end{displaymath}
	holds. 	
\end{proof}

\begin{lemma}\label{lem:squareBound}
	Let $a,b\in\R$ be two numbers. We have
	\begin{equation}\label{eq:squareBound}
		(a+b)^2\le 2a^2+2b^2.
	\end{equation}
\end{lemma}
\begin{proof}[Proof of Lemma~\ref{lem:squareBound}]
	The true statement $0\le(a-b)^2$ is equivalent to~\eqref{eq:squareBound}. 
\end{proof}

The following lemma proved in~\citet{Chernozhukov2018} states that conditional convergence in probability implies unconditional convergence in probability. 
\begin{lemma}\label{lem:ChernozhukovLemma}(Based on~\citet[Lemma 6.1]{Chernozhukov2018}.)
Let $\{X_t\}_{t\ge 1}$ and $\{Y_t\}_{t\ge 1}$ be sequences of random vectors and let ${\ttt}\ge 1$.
Consider a deterministic sequence $\{\eps_t\}_{t\ge 1}$ with $\eps_t\rightarrow 0$ as $t\rightarrow\infty$ such that we have $\E[\norm{X_t}^{\ttt} |Y_t]\le\eps_t^{\ttt}$. 
Then we have $\norm{X_t}=O_{\PP}(\eps_t)$ unconditionally, meaning that that for any sequence $\{\ell_t\}_{t\ge 1}$ with $\ell_t\rightarrow\infty$ as $t\rightarrow\infty$ we have $\PP(\norm{X_t}>\ell_t\eps_t)\rightarrow 0$.
\end{lemma}
\begin{proof}[Proof of Lemma~\ref{lem:ChernozhukovLemma}]
	We have
	\begin{displaymath}
		\PP(\norm{X_t}>\ell_t\eps_t) 
		= \E[\PP(\norm{X_t}>\ell_t\eps_t|Y_t)]
		\le\frac{\E\big[\E[\norm{X_t}^{\ttt}|Y_t]\big]}{\ell_t^{\ttt}\eps_t^{\ttt}}
		\le \frac{1}{\ell_t^{\ttt}} \rightarrow 0 \quad (t\rightarrow\infty)
	\end{displaymath}
	by Markov's inequality.
\end{proof}

\begin{lemma}\label{lem:betaBound}
	There exists a finite real constant $\Cbetazero$ satisfying $\norm{\betazero}\le\Cbetazero$. 
\end{lemma}
\begin{proof}[Proof of Lemma~\ref{lem:betaBound}]
Recall the matrices $\Jzerop$ and $\Jzerodp$ in Definition~\ref{def:lossFunc}. 
	We have
	\begin{displaymath}
		\begin{array}{rcl}
			\norm{\betazero} &\le& \normbig{(\Jzerop)^{-1}}\normbig{\EP\big[A(\Rx)^T\big]}\normbig{(\Jzerodp)^{-1}}\normbig{\EP\big[A\Ry\big]}\\
			&\le& \frac{1}{\ctwo^2}\normP{X}{2}\normP{Y}{2}\normP{A}{2}^2 
		\end{array}
	\end{displaymath}
	by submultiplicativity, Assumption~\ref{assumpt:DMLboth3}, 
	and Lemma~\ref{lem:boundPnorm3}. We hence infer
	\begin{displaymath}
		\norm{\betazero} \le \frac{1}{\ctwo^2} \CpnormRV^4 
	\end{displaymath}
	by Assumption~\ref{assumpt:DMLboth2}. 
\end{proof}

\begin{lemma}\label{lem:bgBound}
Let $\gamma\ge 0$.
	There exists a finite real constant $\Cbg$  satisfying $\norm{\bg}\le\Cbg$. 
\end{lemma}
\begin{proof}[Proof of Lemma~\ref{lem:bgBound}]
	We have
	\begin{displaymath}
		\begin{array}{rcl}
		\norm{\bg}&\le&\normBig{\Big(\EP\big[\Rx\Rx^T\big]+(\gamma-1)\EP\big[\Rx\Ra^T\big]\EP\big[\Ra\Ra^T\big]^{-1}\EP\big[\Ra\Rx^T\big]\Big)^{-1}}\\
		&&\quad\cdot
		\normBig{\EP[\Rx\Ry]+(\gamma-1)\EP\big[\Rx\Ra^T\big]\EP\big[\Ra\Ra^T\big]^{-1}\EP[\Ra\Ry]}	
		\end{array}
	\end{displaymath}
	by submultiplicativity. By Assumption~\ref{assumpt:DMLboth4}, the largest singular value of the matrix 
	\begin{displaymath}
		\matA + (\gamma-1)\matB=\EP\big[\Rx\Rx^T\big]+(\gamma-1)\EP\big[\Rx\Ra^T\big]\EP\big[\Ra\Ra^T\big]^{-1} \EP\big[\Ra\Rx^T\big]
	\end{displaymath}
	is upper bounded by $0<\cfour<\infty$.  
	Thus, we have 
	\begin{displaymath}
		\norm{\bg}\le
		\frac{1}{\cfour} \Big( \norm{\EP[\Rx\Ry]} + \normone{\gamma-1}\normbig{\EP\big[\Rx\Ra^T\big]}\normBig{\EP\big[\Ra\Ra^T\big]^{-1}}\normbig{\EP\big[\Ra\Ry^T\big]} \Big)
	\end{displaymath}
	by the triangle inequality and submultiplicativity. 
	By Assumption~\ref{assumpt:DMLboth3}, the largest singular value of $\EP[\Ra\Ra^T]$ is upper bounded by $0<\ctwo<\infty$.
	By Lemma~\ref{lem:boundPnorm4} and Assumption~\ref{assumpt:DMLboth2}, we have
	\begin{displaymath}
		\begin{array}{l}
		\normbig{\EP\big[\Rx\Ry\big]} \le  \normP{X}{2}\normP{Y}{2}\le \CpnormRV^2,\\
		\normbig{\EP\big[\Rx\Ra^T\big]} \le  \normP{X}{2}\normP{A}{2}\le \CpnormRV^2,\\
		\normbig{\EP\big[\Ra\Ry^T\big]} \le  \normP{A}{2}\normP{Y}{2}\le \CpnormRV^2.
		\end{array}
	\end{displaymath}
	In total, we hence have
	\begin{displaymath}
		\norm{\bg}\le
		\frac{1}{\cfour} \bigg(\CpnormRV^2 + \normone{\gamma-1}\frac{\CpnormRV^4}{\ctwo}\bigg).
	\end{displaymath}
\end{proof}

\begin{lemma}\label{lem:statisticalRates}
Let $\gamma\ge 0$
	The statistical rates $\rNpnumber$ and $\lambdaNpnumber$ introduced in Definition~\ref{def:statisticalRates} satisfy
	$\rNpnumber^4  \lesssim \deltaNnumber$ and $\lambdaNpnumber \lesssim \frac{\deltaNnumber}{\sqrt{N}}$. 
\end{lemma}
\begin{proof}[Proof of Lemma~\ref{lem:statisticalRates}]
	This proof is modified from~\citet{Chernozhukov2018}. 
	First, verify the bound on $\rNpnumber$. Let $S=(U,V,W,Z)\in\{A, X, Y\}^2\times\{W\}\times\{A,X,Y\}$, let $\eta=(\mU,\mV,\mZ)\in\TauN$, and let $\bzero\in\{\bg, \betazero, \bo\}$. 
	We have
	\begin{displaymath}
		\begin{array}{cl}
			&\loss(S;\bzero,\eta) - \loss(S;\bzero,\etazero)\\
			=& \big(U-\mU(W)\big)\Big(Z-\mZ(W) - \big(V-\mV(W)\big)^T\bzero\Big)^T\\
			&\quad-\big(U-\mU^0(W)\big)\Big(Z-\mZ^0(W) - \big(V-\mV^0(W)\big)^T\bzero\Big)^T\\
			=& \big(U-\mU^0(W)\big)\Big(\mZ^0(W)-\mZ(W) - \big(\mV^0(W)-\mV(W)\big)^T\bzero\Big)^T\\
			&\quad +\big(\mU^0(W)-\mU(W)\big)\Big(Z-\mZ^0(W) - \big(V-\mV^0(W)\big)^T\bzero\Big)^T \\
			&\quad +\big(\mU^0(W)-\mU(W)\big)\Big(\mZ^0(W)-\mZ(W) - \big(\mV^0(W)-\mV(W)\big)^T\bzero\Big)^T. 
		\end{array}
	\end{displaymath}
	By the triangle inequality and H\"older's inequality, we have 
	\begin{displaymath}
		\begin{array}{cl}
			&\EP[\norm{\loss(S;\bzero,\eta) - \loss(S;\bzero,\etazero)}]\\
			=& \normP{\loss(S;\bzero,\eta) - \loss(S;\bzero,\etazero)}{1}\\
			\le& \normP{U-\mU^0(W)}{2}\normPBig{\mZ^0(W)-\mZ(W) - \big(\mV^0(W)-\mV(W)\big)^T\bzero}{2} \\
			&\quad + \normP{\mU^0(W)-\mU(W)}{2}\normPBig{Z-\mZ^0(W) - \big(V-\mV^0(W)\big)^T\bzero}{2}\\
			&\quad + \normP{\mU^0(W)-\mU(W)}{2}\normPBig{\mZ^0(W)-\mZ(W) - \big(\mV^0(W)-\mV(W)\big)^T\bzero}{2}.
		\end{array}
	\end{displaymath}
	Observe that $\normP{U-\mU^0(W)}{2}\le 2\normP{U}{2}$, and $\normP{V-\mV^0(W)}{2}\le 2\normP{V}{2}$, and  $\normP{Z-\mZ^0(W)}{2}\le 2\normP{Z}{2}$ hold by Lemma~\ref{lem:boundPnorm1}.  
	We have $\normP{\eta-\etazero}{2}\le \deltaNnumber$ by Assumption~\ref{assumpt:DMLboth5}. 
	Therefore, we obtain the upper bound
	\begin{displaymath}
		\begin{array}{cl}
			&\EP[\norm{\loss(S;\bzero,\eta) - \loss(S;\bzero,\etazero)}]\\
			\le& 4\max\{1,\norm{\bzero}\}(\normP{U}{2} + \normP{V}{2} + \normP{Z}{2}) \deltaNnumber + 2\max\{1,\norm{\bzero}\}\deltaNnumber^2\\
			\lesssim& \deltaNnumber
		\end{array}
	\end{displaymath}
	 by the triangle inequality, Lemma~\ref{lem:betaBound}, Lemma~\ref{lem:bgBound}, and Assumptions~\ref{assumpt:DMLboth2} and~\ref{assumpt:DMLboth5}. 
	Because this upper bound is independent of $\eta$, we obtain our claimed bound on  $\rNpnumber^4$.
		
	Subsequently, we verify the bound on $\lambdaNpnumber$. 
	Consider $S=(A,X,W,Y)$, denote by  $U$ either $A$ or $X$, denote by $Z$ either $A$ or $Y$, and let $\losstest\in\{\loss,\losstilde,\losstwo\}$, where we interpret $\losstwo(S;b,\eta)=\losstwo(S;\eta)$. 
	We have
	\begin{displaymath}
	\begin{array}{cl}
		&\partial_r^2\EP\big[\loss\big(S;\bzero,\etazero+r(\eta-\etazero)\big)\big] \\
		=& 2\EP\bigg[\big(\mU(W)-\mU^0(W)\big)\Big(\mZ(W)-\mZ^0(W) - \big(\mX(W)-\mX^0(W)\big)^T\bzero\Big)^T\bigg]. 
	\end{array}
	\end{displaymath}
	Due to the  Cauchy--Schwarz  inequality, we  infer 
	\begin{displaymath}
		\begin{array}{cl}
			&\normbig{\partial_r^2\EP\big[\loss\big(S;\bzero,\etazero+r(\eta-\etazero)\big)\big] }\\
			\le& 2\normP{\mU(W)-\mU^0(W)}{2}\big(\normP{\mZ(W)-\mZ^0(W)}{2}  +\normP{\mX(W)-\mX^0(W)}{2} \norm{\bzero}\big)\\
			\le& 2\max\{1,\norm{\bzero}\}\normP{\mU(W)-\mU^0(W)}{2}\\
			&\quad\cdot\big(\normP{\mZ(W)-\mZ^0(W)}{2}  +\normP{\mX(W)-\mX^0(W)}{2} \big)\\
			\lesssim& \deltaNnumber\NN^{-\frac{1}{2}} 
		\end{array}
	\end{displaymath}
	by Lemma~\ref{lem:betaBound}, Lemma~\ref{lem:bgBound}, and Assumption~\ref{assumpt:DMLboth5}. 
	Consequently, we obtain our claimed bound on $\lambdaNpnumber$. 
\end{proof}

\begin{lemma}\label{lem:boundRN}
Let $\gamma\ge 0$. 
	Let $\kk\in\indset{\KK}$. Let furthermore 
	$\losstest\in\{\loss,\losstilde,\losstwo\}$ and $\bzero\in\{\bg,\betazero, \bo\}$. 
	We have 
	\begin{displaymath}
		\normbigg{\frac{1}{\sqrt{\nn}}\sum_{i\in\Ik}\losstest(S_i;\bzero,\hetaIkc)
		- \frac{1}{\sqrt{\nn}}\sum_{i\in\Ik}\losstest(S_i;\bzero,\etazero)}
		=O_{\PP}(\rhoN),
	\end{displaymath}
	where $\rhoN =\rNpnumber + \NN^{\frac{1}{2}}\lambdaNpnumber$ is as in Definition~\ref{def:asymptNormal} and satisfies $\rhoN\lesssim\deltaNnumber^{\frac{1}{4}}$, and where we interpret $\losstwo(S;b,\eta)=\losstwo(S;\eta)$. 
\end{lemma}
\begin{proof}[Proof of Lemma~\ref{lem:boundRN}]  	This proof is modified from~\citet{Chernozhukov2018}. 
 By the triangle inequality, we have
\begin{displaymath}
	\begin{array}{cl}
		&\normBig{\frac{1}{\sqrt{\nn}}\sum_{i\in\Ik}\losstest(S_i;\bzero,\hetaIkc)
		- \frac{1}{\sqrt{\nn}}\sum_{i\in\Ik}\losstest(S_i;\bzero,\etazero)}\\
		=& \Big\lVert\frac{1}{\sqrt{\nn}}\sum_{i\in\Ik}\big( \losstest(S_i;\bzero, \hetaIkc) -\int \losstest(s; \bzero,\hetaIkc)\mathrm{d}\PP(s) \big) 
		\\
		&\quad\quad - 
		\frac{1}{\sqrt{\nn}}\sum_{i\in\Ik}\big( \losstest(S_i;\bzero, \etazero) -\int \losstest(s;\bzero, \etazero)\mathrm{d}\PP(s) \big) \\
		&\quad\quad+ \sqrt{\nn}\int \big(\losstest(s;\bzero, \hetaIkc) - \losstest(s; \bzero,\etazero)\big)\mathrm{d}\PP(s)\Big\rVert\\
		\le &\Icalone + \sqrt{\nn}\Icaltwo, 
	\end{array}
\end{displaymath}
where $\Icalone:= \norm{M}$
for
\begin{displaymath}
\begin{array}{rcl}
	M &:=&\frac{1}{\sqrt{\nn}}\sum_{i\in\Ik}\bigg( \losstest(S_i;\bzero, \hetaIkc) -\int \losstest(s;\bzero, \hetaIkc)\mathrm{d}\PP(s) \bigg)  \\
	&&\quad- 
		\frac{1}{\sqrt{\nn}}\sum_{i\in\Ik}\bigg( \losstest(S_i;\bzero, \etazero) -\int \losstest(s;\bzero, \etazero)\mathrm{d}\PP(s) \bigg),
		\end{array}
\end{displaymath}
and where
\begin{displaymath}
	\Icaltwo:=\normbigg{\int \big(\losstest(s;\bzero, \hetaIkc) - \losstest(s;\bzero, \etazero)\big)\mathrm{d}\PP(s)}. 
\end{displaymath}
We bound the two terms $\Icalone$ and $\Icaltwo$ individually. 
First, we bound $\Icalone$. Because the dimensions $d$ and $q$ are fixed, it is sufficient to bound one entry of the matrix $M$. Let $l$ index the rows of $M$ and let $t$ index the columns of $M$ (we interpret vectors as matrices with one column). 
On the event $\EpsN$ the that holds with $\PP$-probability $1-\DeltaNnumber$, we have
\begin{equation}\label{eq:rNp1}
	\begin{array}{cl}
		&\EP\big[\norm{M_{l,t}}^2\big|\SIkc\big]\\
		=& \frac{1}{\nn}\sum_{i\in\Ik}\EP\big[\normone{\losstest_{l,t}(S_i;\bzero, \hetaIkc) -  \losstest_{l,t}(S_i;\bzero, \etazero)}^2\big|\SIkc\big] \\
		&\quad+ \frac{1}{\nn}\sum_{i,j\in\Ik,i\neq j}\EP\big[\big( \losstest_{l,t}(S_i;\bzero, \hetaIkc) -  \losstest_{l,t}(S_i;\bzero, \etazero)\big)\\
		&\quad\quad\quad\quad\quad\quad\quad\quad\quad\quad\quad\cdot\big(\losstest_{l,t}(S_j;\bzero, \hetaIkc) -  \losstest_{l,t}(S_j;\bzero, \etazero)\big)\big|\SIkc\big] \\
		&\quad - 2\sum_{i\in\Ik}\EP\big[ \losstest_{l,t}(S_i;\bzero, \hetaIkc) -  \losstest_{l,t}(S_i;\bzero, \etazero)\big|\SIkc\big]\\
		&\quad\quad\quad\quad\quad\quad\cdot\EP\big[\losstest_{l,t}(S;\bzero, \hetaIkc) - \losstest_{l,t}(S;\bzero, \etazero)\big|\SIkc\big]\\
		&\quad + \nn \normonebig{\EP\big[\losstest_{l,t}(S;\bzero, \hetaIkc) - \losstest_{l,t}(S;\bzero, \etazero)\big|\SIkc\big]}^2\\
		=& \EP\big[\normone{ \losstest_{l,t}(S; \bzero,\hetaIkc) -  \losstest_{l,t}(S;\bzero, \etazero)}^2\big|\SIkc\big] \\
		&\quad
		 		+ \big( \frac{\nn(\nn-1)}{\nn} - 2\nn+ \nn   \big) \normonebig{\EP\big[\losstest_{l,t}(S;\bzero, \hetaIkc) - \losstest_{l,t}(S;\bzero, \etazero)\big|\SIkc\big]}^2\\
		\le & \sup_{\eta\in\TauN}\EP\big[\norm{ \losstest(S;\bzero, \eta) -  \losstest(S; \bzero,\etazero)}^2\big].
	\end{array}
\end{equation}
Furthermore, for $\eta\in\TauN$, we have
\begin{equation}\label{eq:rNp2}
\begin{array}{cl}
	&\EP\big[\norm{ \losstest(S;\bzero, \eta) -  \losstest(S; \bzero,\etazero)}^2\big] \\
	\le &\EP[\norm{ \losstest(S;\bzero, \eta) -  \losstest(S; \bzero,\etazero)}] \\
	&\quad+ \EP\big[\norm{ \losstest(S;\bzero, \eta) -  \losstest(S; \bzero,\etazero)}^2\one_{\norm{ \losstest(S;\bzero, \eta) -  \losstest(S; \bzero,\etazero)}\ge 1}\big]
	\end{array}
\end{equation}
and we have
\begin{equation}\label{eq:rNp3}
	\begin{array}{cl}
		&\EP\big[\norm{ \losstest(S;\bzero, \eta) -  \losstest(S; \bzero,\etazero)}^2\one_{\norm{ \losstest(S;\bzero, \eta) -  \losstest(S; \bzero,\etazero)}\ge 1}\big]\\
		\le&\sqrt{ \EP\big[\norm{ \losstest(S;\bzero, \eta) -  \losstest(S; \bzero,\etazero)}^4\big]} \sqrt{\PP(\norm{ \losstest(S;\bzero, \eta) -  \losstest(S; \bzero,\etazero)}\ge 1)}
	\end{array}
\end{equation}
by H\"{o}lder's inequality. 
Observe that the term 
\begin{equation}\label{eq:rNp4}
	\sqrt{ \EP\big[\norm{ \losstest(S;\bzero, \eta) -  \losstest(S; \bzero,\etazero)}^4\big]}
\end{equation}
is upper bounded by Assumption~\ref{assumpt:DMLboth5}, Lemma~\ref{lem:betaBound} and Lemma~\ref{lem:bgBound}. 
By Markov's inequality, 
we have
\begin{equation}\label{eq:rNp5}
	\PP(\norm{ \losstest(S;\bzero, \eta) -  \losstest(S; \bzero,\etazero)}\ge 1) \le \EP[\norm{ \losstest(S;\bzero, \eta) -  \losstest(S; \bzero,\etazero)}] \le\rNpnumber^4.
\end{equation}
Therefore, we have	$\EP[\Icalone^2|\SIkc] \lesssim \rNpnumber^2$ due to~\mbox{\eqref{eq:rNp1}--\eqref{eq:rNp5}}. 
The statistical rate $\rNpnumber$ satisfies $\rNpnumber\lesssim\delta_{\NN}^{\frac{1}{4}}$ by Lemma~\ref{lem:statisticalRates}. 
Thus, we infer $\Icalone = O_{\PP}(\rNpnumber)$
by Lemma~\ref{lem:ChernozhukovLemma}.
Subsequently, we bound $\Icaltwo$. For $r\in [0,1]$, we introduce the function
\begin{displaymath}
	f_k(r) := \EP\big[\losstest\big(S;\bzero, \etazero + r(\hetaIkc-\etazero)\big)\big| \SIkc\big] - \EP[\losstest(S;\bzero, \etazero)].
\end{displaymath}
Observe that $\Icaltwo= \norm{f_k(1)}$ holds. 
We apply a Taylor expansion to this function and obtain
\begin{displaymath}
	f_k(1)=f_k(0)+f_k'(0) + \frac{1}{2}f_k''(\tilde r)
\end{displaymath}
for some $\tilde r\in (0,1)$. We have 
\begin{displaymath}
	f_k(0) = \EP\big[\losstest(S;\bzero, \etazero)\big|\SIkc\big] - \EP[\losstest(S; \bzero,\etazero)] = \bo. 
\end{displaymath}
Furthermore, the score $\losstest$ satisfies the Neyman orthogonality property $f'_k(0)=\bo$.
The proof of this claim is analogous to the proof of Proposition~\ref{prop:Neyman_orth} because the proof of Proposition~\ref{prop:Neyman_orth} does neither depend on the underlying model of the random variables nor on the value of $\beta$. Furthermore, we have
\begin{displaymath}
	f''_k(r) = 2\E\bigg[ \big(\mU(W)-\mU^0(W)\big)\Big( \mZ(W)-\mZ^0(W)-\big(\mX(W)-\mX^0(W)\big)^T\bzero \Big)^T \bigg]
\end{displaymath}
for $U\in\{A,X\}$ and $Z\in\{A,Y\}$. 
On the event $\EpsN$ that holds with $\PP$-probability $1-\DeltaN$, we  have
\begin{displaymath}
	\norm{f_k''(\tilde r)} \le \sup_{r\in (0,1)} \norm{f_k''(r)} \lesssim \lambdaNpnumber.
\end{displaymath}
We thus infer 
\begin{displaymath}
	\normbigg{\frac{1}{\sqrt{\nn}}\sum_{i\in\Ik}\losstest(S_i;\bzero,\hetaIkc)
		- \frac{1}{\sqrt{\nn}}\sum_{i\in\Ik}\losstest(S_i;\bzero,\etazero)}
		\le \Icalone + \sqrt{\nn}\Icaltwo
		= O_{\PP}(\rNpnumber + \NN^{\frac{1}{2}}\lambdaNpnumber). 
\end{displaymath}
Because $\rNpnumber\lesssim\deltaNnumber^{\frac{1}{4}}$ and $\lambdaNpnumber\lesssim\frac{\deltaNnumber}{\sqrt{\NN}}$ hold by Lemma~\ref{lem:statisticalRates} and because $\{\deltaNnumber\}_{\NN\ge\KK}$ converges to 0 by Assumption~\ref{assumpt:DMLboth}, we furthermore have
\begin{displaymath}
	\rhoN = \rNpnumber + \NN^{\frac{1}{2}}\lambdaNpnumber \lesssim\deltaNnumber^{\frac{1}{4}}. 
\end{displaymath}
\end{proof}

\begin{lemma}\label{lem:D1convInProb}
	Let $\kk\in\indset{\KK}$. Let furthermore $U,V\in\{A,X\}$ and $S=(U,V, W, Y)$. Let  $\losstest\in\{\lossone,\losstwo,\lossthree\}$. 
	We have 
	\begin{displaymath}
	\frac{1}{\nn}\sum_{i\in\Ik}\losstest(S_i;\hetaIkc) = \EP[\losstest(S;\etazero)] + O_{\PP}\big(\NN^{-\frac{1}{2}}(1+\rhoN)\big).
\end{displaymath}
\end{lemma}
\begin{proof}[Proof of Lemma~\ref{lem:D1convInProb}]
Consider the decomposition
\begin{displaymath}
\begin{array}{cl}
	&\frac{1}{\nn}\sum_{i\in\Ik}\losstest(S_i;\hetaIkc) - \EP[\losstest(S;\etazero)] \\	
	=& \frac{1}{\nn}\sum_{i\in\Ik}\big(\losstest(S_i;\hetaIkc)-\losstest(S_i;\etazero)\big) 
		+ \frac{1}{\nn}\sum_{i\in\Ik}\big(\losstest(S_i;\etazero) - \EP[\losstest(S;\etazero)] \big) 
\end{array}
\end{displaymath}
The term $\frac{1}{\nn}\sum_{i\in\Ik}\big(\losstest(S_i;\hetaIkc)-\losstest(S_i;\etazero)\big)$ is of order $O_{\PP}(\NN^{-\frac{1}{2}}\rhoN)$
by Lemma~\ref{lem:boundRN}. The term $\frac{1}{\nn}\sum_{i\in\Ik}\big(\losstest(S_i;\etazero) - \EP[\losstest(S;\etazero)] \big)$ is of order $O_{\PP}(\NN^{-\frac{1}{2}})$ by the Lindeberg--Feller CLT and the Cramer--Wold device. Thus, we deduce the statement. 
\end{proof}

\begin{definition}
We denote by $\A^{\Ik}$ the row-wise concatenation of the observations $A_i$ for $i\in\Ik$. We denote similarly by $\X^{\Ik}$, $\W^{\Ik}$, $\Y^{\Ik}$, $\A^{\Ikc}$, $\X^{\Ikc}$, $\W^{\Ikc}$, and $\Y^{\Ikc}$ the row-wise concatenations of the respective observations. 
\end{definition}

\begin{proof}[Proof of Theorem~\ref{thm:asymptNormal}]
This proof is based on~\citet{Chernozhukov2018}. 
We show the stronger statement 
\begin{equation}\label{eq:asymptNormal}
\sqrt{\NN}\sigma^{-1}(\hbetaN-\betazero) = \frac{1}{\sqrt{\NN}}\sum_{i=1}^{\NN}\lossoverline( S_i;\betazero, \etazero) + O_{\PP}(\rhoN) \stackrel{d}{\rightarrow}\mathcal{N}(0,\one_{d\times d})\quad (\NN\rightarrow\infty),
\end{equation}
where $\hbetaN$ denotes the DML1 estimator $\hbetaNDMLone$ or the DML2 estimator $\hbetaNDMLtwo$, and where the rate $\rhoN$ is specified in Definition~\ref{def:asymptNormal}, and we show that this statement holds uniformly over laws $\PP$. 
We first consider $\hbetaNDMLtwo$. 
It suffices to show that~\eqref{eq:asymptNormal} holds uniformly over $\PP\in\PcalN$. 
Fix a sequence $\{\PPN\}_{\NN\ge 1}$ such that $\PPN\in\PcalN$ for all $\NN\ge 1$. 
Because this sequence is chosen arbitrarily, it suffices to show 
\begin{displaymath}
		\sqrt{\NN}\sigma^{-1}(\hbetaNDMLtwo-\betazero) = \frac{1}{\sqrt{\NN}}\sum_{i=1}^{\NN}\lossoverline(S_i;\betazero, \etazero) + O_{\PPN}(\rhoN) \stackrel{d}{\rightarrow}\mathcal{N}(0,\one_{d\times d})\quad (\NN\rightarrow\infty). 
	\end{displaymath}
We have
\begin{equation}\label{eq:betaHatNk} 
	\begin{array}{rcl}
		\hbetaNDMLtwo &=& \Big(\frac{1}{\KK}\sum_{\kk=1}^{\KK}\big(\X^{\Ik}-\hmX^{\Ikc}(\W^{\Ik})\big)^T\PiIkcIk\big(\X^{\Ik}-\hmX^{\Ikc}(\W^{\Ik})\big)\Big)^{-1}\\
		&&\quad\cdot \frac{1}{\KK}\sum_{\kk=1}^{\KK}\big(\X^{\Ik}-\hmX^{\Ikc}(\W^{\Ik})\big)^T\PiIkcIk\big(\Y^{\Ik}-\hmY^{\Ikc}(\W^{\Ik})\big)\\
		&=& \bigg(\frac{1}{\KK}\sum_{\kk=1}^{\KK}\frac{1}{\nn}\big(\X^{\Ik}-\hmX^{\Ikc}(\W^{\Ik})\big)^T\big(\A^{\Ik}-\hmA^{\Ikc}(\W^{\Ik})\big)\\
		&&\quad\quad\quad\quad\quad\cdot\Big(\frac{1}{\nn}\big(\A^{\Ik}-\hmA^{\Ikc}(\W^{\Ik})\big)^T(\A^{\Ik}-\hmA^{\Ikc}(\W^{\Ik})\Big)^{-1}\\
		&&\quad\quad\quad\quad\quad\cdot\frac{1}{\nn}\big(\A^{\Ik}-\hmA^{\Ikc}(\W^{\Ik})\big)^T\big(\X^{\Ik}-\hmX^{\Ikc}(\W^{\Ik})\big)
		\bigg)^{-1}\\
		&&\quad\cdot \frac{1}{\KK}\sum_{\kk=1}^{\KK}\frac{1}{\nn}\big(\X^{\Ik}-\hmX^{\Ikc}(\W^{\Ik})\big)^T\big(\A^{\Ik}-\hmA^{\Ikc}(\W^{\Ik})\big)\\
		&&\quad\quad\quad\quad\quad\cdot\Big(\frac{1}{\nn}\big(\A^{\Ik}-\hmA^{\Ikc}(\W^{\Ik})\big)^T\big(\A^{\Ik}-\hmA^{\Ikc}(\W^{\Ik})\big)\Big)^{-1}\\
		&&\quad\quad\quad\quad\quad\cdot\frac{1}{\nn}\big(\A^{\Ik}-\hmA^{\Ikc}(\W^{\Ik})\big)^T\big(\Y^{\Ik}-\hmY^{\Ikc}(\W^{\Ik})\big)
	\end{array}
\end{equation}
by~\eqref{eq:betaDMLtwo}. 
By Lemma~\ref{lem:D1convInProb}, we have
\begin{equation}\label{eq:D1inftyconvInProb:eq1}
	\begin{array}{cl}
	&\frac{1}{\nn} \big(\X^{\Ik}-\hmX^{\Ikc}(\W^{\Ik})\big)^T\big(\A^{\Ik}-\hmA^{\Ikc}(\W^{\Ik})\big)\\
	=& \EPN\Big[\big(X-\mX^0(W)\big)\big(A-\mA^0(W)\big)^T\Big] + O_{\PPN}\big(\NN^{-\frac{1}{2}}(1+\rhoN)\big)
	\end{array}
\end{equation}
and
\begin{equation}\label{eq:D1inftyconvInProb:eq2}
	\begin{array}{cl}
	&\frac{1}{\nn} \big(\A^{\Ik}-\hmA^{\Ikc}(\W^{\Ik})\big)^T\big(\A^{\Ik}-\hmA^{\Ikc}(\W^{\Ik})\big)\\
		=& \EPN\Big[\big(A-\mA^0(W)\big)\big(A-\mA^0(W)\big)^T\Big] + O_{\PPN}\big(\NN^{-\frac{1}{2}}(1+\rhoN)\big).
	\end{array}
\end{equation}
Recall the matrix $\Jzero$ introduced in Definition~\ref{def:lossFunc}. 
By Weyl's inequality and Slutsky's theorem, combining Equations~\mbox{\eqref{eq:betaHatNk}--\eqref{eq:D1inftyconvInProb:eq2}} gives
\begin{equation}\label{eq:asymptoticsDMLtwo}
	\begin{array}{cl}
		&\sqrt{\NN}(\hbetaNDMLtwo-\betazero)\\
			=& \bigg(\Big( \EPN\Big[\big(X-\mX^0(W)\big)\big(A-\mA^0(W)\big)^T\Big] 
			\EPN\Big[\big(A-\mA^0(W)\big)\big(A-\mA^0(W)\big)^T\Big]^{-1} \\
			&\quad\quad\quad\cdot \EPN\Big[\big(A-\mA^0(W)\big)\big(X-\mX^0(W)\big)^T\Big]  \Big)^{-1}\\
			&\quad\quad\cdot  \EPN\Big[\big(X-\mX^0(W)\big)\big(A-\mA^0(W)\big)^T\Big]  \EPN\Big[\big(A-\mA^0(W)\big)\big(A-\mA^0(W)\big)^T\Big]^{-1} \\
			&\quad\quad+O_{\PPN}\big(\NN^{-\frac{1}{2}}(1+\rhoN)\big)\bigg) \\
			&\quad\cdot \frac{1}{\sqrt{\KK}}\sum_{\kk=1}^{\KK}
			\frac{1}{\sqrt{\nn}}\Big(\big(\A^{\Ik}-\hmA^{\Ikc}(\W^{\Ik})\big)^T
			\big(\Y^{\Ik}-\hmY^{\Ikc}(\W^{\Ik})\big) \\
			& \quad\quad\quad\quad\quad\quad\quad\quad\quad- \big(\A^{\Ik}-\hmA^{\Ikc}(\W^{\Ik})\big)^T
			\big(\X^{\Ik}-\hmX^{\Ikc}(\W^{\Ik})\big)\betazero\Big)\\
			=& \big(\Jzero 
			+O_{\PPN}\big(\NN^{-\frac{1}{2}}(1+\rhoN)\big)\big) \\
			&\quad\cdot \frac{1}{\sqrt{\KK}}\sum_{\kk=1}^{\KK}
			\frac{1}{\sqrt{\nn}}\bigg(\big(\A^{\Ik}-\hmA^{\Ikc}(\W^{\Ik})\big)^T
			\Big(\Y^{\Ik}-\hmY^{\Ikc}(\W^{\Ik})  - 
			\big(\X^{\Ik}-\hmX^{\Ikc}(\W^{\Ik})\big)\betazero\Big)\bigg)
	\end{array}
\end{equation}
because $\KK$ is a constant independent of $\NN$ and because  $\NN=\nn\KK$ holds. Recall the linear score $\loss$ in~\eqref{eq:psi_linear}. 
We have
\begin{equation}\label{eq:inftyAsymptNormal}
		\sqrt{\NN}(\hbetaNDMLtwo-\betazero)
		= \Big(\Jzero + O_{\PPN}\big(\NN^{-\frac{1}{2}}(1+\rhoN)\big)\Big)\frac{1}{\sqrt{\KK}}\sum_{\kk=1}^{\KK} \frac{1}{\sqrt{\nn}}\sum_{i\in\Ik}\loss(S_i;\betazero,\hetaIkc).
\end{equation}
Let $\kk\in\indset{\KK}$. By Lemma~\ref{lem:boundRN}, we have
\begin{equation}\label{eq:inftyAsymptNormalOracle}
	\frac{1}{\sqrt{\nn}}\sum_{i\in\Ik}\loss(S_i;\betazero,\hetaIkc)
	= \frac{1}{\sqrt{\nn}}\sum_{i\in\Ik}\loss(S_i;\betazero,\etazero)+ O_{\PPN}(\rhoN).
\end{equation}
We combine~\eqref{eq:inftyAsymptNormal} and~\eqref{eq:inftyAsymptNormalOracle} to obtain
\begin{displaymath}
	\begin{array}{cl}
		&\sqrt{\NN}(\hbetaNDMLtwo-\betazero)\\
		=& \Big(\Jzero + O_{\PPN}\big(\NN^{-\frac{1}{2}}(1+\rhoN)\big)\Big)\frac{1}{\sqrt{\KK}}\sum_{\kk=1}^{\KK} \frac{1}{\sqrt{\nn}}\sum_{i\in\Ik}\loss(S_i;\betazero,\hetaIkc)\\
		=& \Big(\Jzero + O_{\PPN}\big(\NN^{-\frac{1}{2}}(1+\rhoN)\big)\Big)\frac{1}{\sqrt{\KK}}\sum_{\kk=1}^{\KK} \Big(\frac{1}{\sqrt{\nn}}\sum_{i\in\Ik} \loss(S_i;\betazero,\etazero) + O_{\PPN}(\rhoN)  \Big).
	\end{array}
\end{displaymath}
Recall that we have $\NN=\nn\KK$, that $\KK$ is a constant independent of $\NN$,  that the sets $\Ik$ for $\kk\in\indset{\KK}$ form a partition of $\indset{\NN}$,  that $\rhoN\lesssim\deltaN^{\frac{1}{4}}$ by Lemma~\ref{lem:boundRN}, and that $\deltaN$ converges to $0$ as $\NN\rightarrow\infty$ and that $\deltaN^{\frac{1}{4}}\ge \NN^{-\frac{1}{2}}$ holds by Assumption~\ref{assumpt:DMLboth}. 
Thus, we have
\begin{displaymath}
	\begin{array}{cl}
	&\sqrt{\NN}(\hbetaNDMLtwo-\betazero)\\
	=&\Big(\Jzero + O_{\PPN}\big(\NN^{-\frac{1}{2}}(1+\rhoN)\big)\Big) \frac{1}{\sqrt{\KK}}\sum_{\kk=1}^{\KK}\Big(\frac{1}{\sqrt{\nn}}\sum_{i\in\Ik} \loss(S_i;\betazero,\etazero) + O_{\PPN}(\rhoN)  \Big)\\
	=&\Big(\Jzero +O_{\PPN}\big(\NN^{-\frac{1}{2}}(1+\rhoN)\big)\Big) \frac{1}{\sqrt{\NN}}\sum_{i=1}^{\NN}\big( \loss(S_i;\betazero,\etazero) + O_{\PPN}(\rhoN)  \big)  \\
	=& \Jzero \cdot \frac{1}{\sqrt{\NN}}\sum_{i=1}^{\NN} \loss(S_i;\betazero,\etazero) + O_{\PPN}(\rhoN). 
	\end{array}
\end{displaymath}
We have $\EP[\loss(S;\betazero, \etazero)]=\bo$ due to the identifiability condition~\eqref{eq:identificationCondition}.
Therefore, we conclude the proof concerning the \DMLtwo\ method due to 
the  Lindeberg--Feller CLT and the Cramer--Wold device. 

Subsequently, we consider the \DMLone\ method. 
It suffices to show that~\eqref{eq:asymptNormal} holds uniformly over $\PP\in\PcalN$. 
Fix a sequence $\{\PPN\}_{\NN\ge 1}$ such that $\PPN\in\PcalN$ for all $\NN\ge 1$. 
Because this sequence is chosen arbitrarily, it suffices to show 
\begin{displaymath}
		\sqrt{\NN}\sigma^{-1}(\hbetaNDMLone-\betazero) = \frac{1}{\sqrt{\NN}}\sum_{i=1}^{\NN}\lossoverline(S_i;\betazero, \etazero) + O_{\PPN}(\rhoN) \stackrel{d}{\rightarrow}\mathcal{N}(0,\one_{d\times d})\quad (\NN\rightarrow\infty). 
	\end{displaymath}
We have
\begin{equation}\label{eq:betaOneHatNk} 
\begin{array}{rcl}
		\hbetaNkDMLone &=& \Big( \big(\X^{\Ik}-\hmX^{\Ikc}(\W^{\Ik})\big)^T
			\PiIkcIk
			\big(\X^{\Ik}- \hmX^{\Ikc}(\W^{\Ik})\big)\Big)^{-1}\\
			&&\quad\cdot 
			\big(\X^{\Ik}-\hmX^{\Ikc}(\W^{\Ik})\big)^T
			\PiIkcIk
			\big(\Y^{\Ik}-\hmY^{\Ikc}(\W^{\Ik})\big)\\
			&=& \bigg(\frac{1}{\nn} \big(\X^{\Ik}-\hmX^{\Ikc}(\W^{\Ik})\big)^T
			\big(\A^{\Ik}-\hmA^{\Ikc}(\W^{\Ik})\big) \\
			&&\quad\quad\cdot
			\Big(\frac{1}{\nn}\big(\A^{\Ik}-\hmA^{\Ikc}(\W^{\Ik})\big)^T\big(\A^{\Ik}-\hmA^{\Ikc}(\W^{\Ik})\big) \Big)^{-1}\\
			&&\quad\quad\cdot
			\frac{1}{\nn}\big(\A^{\Ik}-\hmA^{\Ikc}(\W^{\Ik})\big)^T
			\big(\X^{\Ik}- \hmX^{\Ikc}(\W^{\Ik})\big)\bigg)^{-1}\\
			&&\quad\cdot
			\frac{1}{\nn}\big(\X^{\Ik}-\hmX^{\Ikc}(\W^{\Ik})\big)^T
			\big(\A^{\Ik}-\hmA^{\Ikc}(\W^{\Ik})\big) \\
			&&\quad\quad\cdot
			\Big(\frac{1}{\nn}\big(\A^{\Ik}-\hmA^{\Ikc}(\W^{\Ik})\big)^T\big(\A^{\Ik}-\hmA^{\Ikc}(\W^{\Ik})\big) \Big)^{-1}\\
			&&\quad\quad\cdot
			\frac{1}{\nn}\big(\A^{\Ik}-\hmA^{\Ikc}(\W^{\Ik})\big)^T
			\big(\Y^{\Ik}-\hmY^{\Ikc}(\W^{\Ik})\big)
\end{array}
\end{equation}
by~\eqref{eq:hbetaNkDMLone}. 
Due to Weyl's inequality and Slutsky's theorem,~\eqref{eq:D1inftyconvInProb:eq1},~\eqref{eq:D1inftyconvInProb:eq2},  and~\eqref{eq:betaOneHatNk}, we obtain
\begin{equation}\label{eq:asymptoticsDMLone}
	\begin{array}{cl}
		&\sqrt{\NN}(\hbetaNDMLone-\betazero)\\
			=& \bigg(\Big( \EPN\Big[\big(X-\mX^0(W)\big)\big(A-\mA^0(W)\big)^T\Big] \EPN\Big[\big(A-\mA^0(W)\big)\big(A-\mA^0(W)\big)^T\Big]^{-1} \\
			&\quad\quad\quad\cdot \EPN\Big[\big(A-\mA^0(W)\big)\big(X-\mX^0(W)\big)^T\Big]  \Big)^{-1}\\
			&\quad\quad\cdot  \EPN\Big[\big(X-\mX^0(W)\big)\big(A-\mA^0(W)\big)^T\Big]  \EPN\Big[\big(A-\mA^0(W)\big)\big(A-\mA^0(W)\big)^T\Big]^{-1} \\
			&\quad\quad+O_{\PPN}\big(\NN^{-\frac{1}{2}}(1+\rhoN)\big)\bigg) \\
			&\quad\cdot\frac{1}{\sqrt{\KK}}\sum_{\kk=1}^{\KK}\Big( 
			\frac{1}{\sqrt{\nn}}\big(\A^{\Ik}-\hmA^{\Ikc}(\W^{\Ik})\big)^T
			\big(\Y^{\Ik}-\hmY^{\Ikc}(\W^{\Ik})\big)  \\
			&\quad\quad\quad\quad\quad\quad\quad- \frac{1}{\sqrt{\nn}}\big(\A^{\Ik}-\hmA^{\Ikc}(\W^{\Ik})\big)^T
			\big(\X^{\Ik}-\hmX^{\Ikc}(\W^{\Ik})\big)\betazero\Big)\\
			=&\Big(\Jzero +O_{\PPN}\big(\NN^{-\frac{1}{2}}(1+\rhoN)\big)\Big)\\
			&\quad\cdot\frac{1}{\sqrt{\KK}}\sum_{\kk=1}^{\KK}\bigg( 
			\frac{1}{\sqrt{\nn}}\big(\A^{\Ik}-\hmA^{\Ikc}(\W^{\Ik})\big)^T
			\Big(\Y^{\Ik}-\hmY^{\Ikc}(\W^{\Ik})-
			\big(\X^{\Ik}-\hmX^{\Ikc}(\W^{\Ik})\big)\betazero\Big)\bigg). 
	\end{array}
\end{equation}
Observe that the expression for $\sqrt{\NN}(\hbetaNDMLone-\betazero)$ given in~\eqref{eq:asymptoticsDMLone} coincides with the expression for $\sqrt{\NN}(\hbetaNDMLtwo-\betazero)$ given in~\eqref{eq:asymptoticsDMLtwo}. Thus, the asymptotic analysis of $\sqrt{\NN}(\hbetaNDMLone-\betazero)$ coincides with the asymptotic analysis of $\sqrt{\NN}(\hbetaNDMLtwo-\betazero)$ presented above. 
\end{proof}

\begin{lemma}\label{lem:mNboth}
	Let $\gamma\ge 0$. 
	Let $p>4$ be the $p$ from Assumption~\ref{assumpt:DMLboth},
	let $\bzero\in\{\betazero, \bg, \bo\}$, and 
	let  $S=(U,V,Z)\in\{A, X, Y\}^2\times\{W\}\times\{A,X,Y\}$. 	
	There exists a finite real constant $\CnormLossboth$  satisfying
	\begin{displaymath}
		\sup_{\eta\in\TauN} \EP\Big[\norm{\loss(S;\bzero,\eta)}^{\frac{p}{2}}\Big]^{\frac{2}{p}}\le \CnormLossboth.
	\end{displaymath}
\end{lemma}
\begin{proof}[Proof of Lemma~\ref{lem:mNboth}]
	Let $\eta=(\mU,\mV,\mZ)\in\TauN$. By H\"{o}lder's inequality and the triangle inequality, we have
	\begin{equation}\label{eq:lossLp1both}
		\begin{array}{cl}
			&\EP\Big[\norm{\loss(S;\bzero,\eta)}^{\frac{p}{2}}\Big]^{\frac{2}{p}} \\
			=& \normP{(U-\mU(W))\big(Z-\mZ(W)-(V-\mV(W))^T\bzero\big)}{\frac{2}{p}}\\
			\le& \big(\normP{U-\mU^0(W)}{p} + \normP{\mU^0(W)-\mU(W)}{p}\big)\\
			&\quad\cdot
			\big(
				\normP{Z-\mZ^0(W)}{p}+\normP{(V-\mV^0(W))^T\bzero}{p}\\
				&\quad\quad
				+ \normP{\mZ^0(W)-\mZ(W)}{p}+\normP{(\mV^0(W)-\mV(W))^T\bzero}{p}
			\big).
		\end{array}
	\end{equation}
	By the Cauchy--Schwarz inequality, we have
	\begin{equation}\label{eq:lossLp2both}
		\normPBig{\big(V-\mV^0(W)\big)^T\bzero}{p} 
		\le \EP\big[ \norm{V-\mV^0(W)}^p\norm{\bzero}^p \big]^{\frac{1}{p}} = \norm{\bzero} \normP{V-\mV^0(W)}{p}
	\end{equation}
	and analogously
	\begin{equation}\label{eq:lossLp3both}
		\normPBig{\big(\mV^0(W)-\mV(W)\big)^T\bzero}{p}
		\le \norm{\bzero}\normP{\mV^0(W)-\mV(W)}{p}. 
	\end{equation}
	Hence, we infer
	\begin{equation}\label{eq:lossLp4both}
			\EP\Big[\norm{\loss(S;\bzero,\eta)}^{\frac{p}{2}}\Big]^{\frac{2}{p}}
			\le (\normP{U}{p} + \CpnormEta) (\normP{Z}{p} +\normP{V}{p} +2\CpnormEta)\max\{1,\norm{\bzero}\}
	\end{equation}
	by~\eqref{eq:lossLp1both},~\eqref{eq:lossLp2both},~\eqref{eq:lossLp3both},
	Lemma~\ref{lem:boundPnorm1}, and Assumption~\ref{assumpt:DMLboth5}. 
	By Lemma~\ref{lem:betaBound}, 
	there exists a finite real constant $\Cbetazero$ that satisfies $\norm{\betazero}\le\Cbetazero$.  
	By Lemma~\ref{lem:bgBound}, 
	there exists a finite real constant $\Cbg$ that satisfies $\norm{\bg}\le\Cbg$. 
	These two bounds lead to $\norm{\bzero}\le\max\{\Cbetazero, \Cbg\}$. 
	By 
	Assumption~\ref{assumpt:DMLboth2}, 
	we have
	\begin{displaymath}
		\max\{\normP{U}{p},\normP{V}{p},\normP{Z}{p}\}\le \normP{U}{p}+\normP{V}{p}+\normP{Z}{p}\le 3\CpnormRV. 
	\end{displaymath}
	Due to~\eqref{eq:lossLp4both}, we therefore have 
	\begin{displaymath}
		\EP\Big[\norm{\loss(S;\bzero,\eta)}^{\frac{p}{2}}\Big]^{\frac{2}{p}}
		\le(3\CpnormRV + \CpnormEta)(6\CpnormRV +2\CpnormEta)\max\{1,\Cbetazero,\Cbg\}.
	\end{displaymath}
\end{proof}

\begin{lemma}\label{lem:helperVarianceConsistent}
Let $\gamma\ge 0$, 
and let $p$ be as in Assumption~\ref{assumpt:DMLboth}. 
Let the indices $\kk\in\indset{\KK}$ and $(j,l,t,r)\in\indset{\Lone}\times \indset{\Ltwo}\times \indset{\Lthree}\times \indset{\Lfour}$, where $\Lone$, $\Ltwo$, $\Lthree$, and $\Lfour$ are natural numbers representing the intended dimensions. Let $\hb\in\{\hbetaNDMLone, \hbetaNDMLtwo, \hbgDMLone, \hbgDMLtwo\}$ and consider the corresponding true unknown underlying parameter vector $\bzero\in\{\betazero, \bg\}$. Consider the corresponding score function combinations
\begin{displaymath}
	\begin{array}{rcl}
	\hlossA(\cdot)&\in&\{\losstilde_j(\cdot;\hb,\hetaIkc), \loss_j(\cdot;\hb,\hetaIkc), (\lossone(\cdot;\hetaIkc))_{j,l} ,(\losstwo(\cdot;\hetaIkc))_{j,l}\}, \\
	\hlossAfull(\cdot)&\in&\{\losstilde(\cdot;\hb,\hetaIkc), \loss(\cdot;\hb,\hetaIkc), \lossone(\cdot;\hetaIkc), \losstwo(\cdot;\hetaIkc)\}, \\
	\hlossB(\cdot)&\in&\{\losstilde_t(\cdot;\hb,\hetaIkc), \loss_t(\cdot;\hb,\hetaIkc), (\lossone(\cdot;\hetaIkc))_{t,r}, (\losstwo(\cdot;\hetaIkc))_{t,r}\},\\
	\hlossBfull(\cdot)&\in&\{\losstilde(\cdot;\hb,\hetaIkc), \loss(\cdot;\hb,\hetaIkc), \lossone(\cdot;\hetaIkc), \losstwo(\cdot;\hetaIkc)\},
	\end{array}
\end{displaymath}
and their respective nonestimated quantity 
\begin{displaymath}
	\begin{array}{rcl}
	\lossA(\cdot) &\in&\{\losstilde_j(\cdot;\bzero,\etazero), \loss_j(\cdot;\bzero,\etazero), (\lossone(\cdot;\etazero))_{j,l}, (\losstwo(\cdot;\etazero))_{j,l}\},\\
	\lossAfull(\cdot) &\in&\{\losstilde(\cdot;\bzero,\etazero), \loss(\cdot;\bzero,\etazero), \lossone(\cdot;\etazero), \losstwo(\cdot;\etazero)\},\\
	\lossB(\cdot)&\in&\{\losstilde_t(\cdot;\bzero,\etazero), \loss_t(\cdot;\bzero,\etazero), (\lossone(\cdot;\etazero))_{t,r}, (\losstwo(\cdot;\etazero))_{t,r}\},\\
	\lossBfull(\cdot)&\in&\{\losstilde(\cdot;\bzero,\etazero), \loss(\cdot;\bzero,\etazero), \lossone(\cdot;\etazero), \losstwo(\cdot;\etazero)\}. 
	\end{array}
\end{displaymath}
Then we have
\begin{displaymath}
	\Icalk := \normonebigg{\frac{1}{\nn}\sum_{i\in\Ik}
	\hlossA(S_i)\hlossB(S_i) - \EP\big[\lossA(S)\lossB(S)\big]} = O_{\PP}(\rhoNtilde),
\end{displaymath}
 where $\rhoNtilde = \NN^{\max\big\{\frac{4}{p}-1, -\frac{1}{2}\big\}}+\rNpnumber$ is as in Definition~\ref{def:asymptNormal}. 
\end{lemma}
\begin{proof}[Proof of Lemma~\ref{lem:helperVarianceConsistent}]
This proof is modified from~\citet{Chernozhukov2018}. 
By the triangle inequality, we have
\begin{displaymath}
	\Icalk \le \IcalkA+\IcalkB,
\end{displaymath}
where
\begin{displaymath}
	\IcalkA := \normonebigg{\frac{1}{\nn}\sum_{i\in\Ik}\hlossA(S_i)\hlossB(S_i)- \frac{1}{\nn}\sum_{i\in\Ik}\lossA(S_i)\lossB(S_i)}
\end{displaymath}
and 
\begin{displaymath}
	\IcalkB := \normonebigg{ \frac{1}{\nn}\sum_{i\in\Ik}\lossA(S_i)\lossB(S_i) - \EP\big[\lossA(S)\lossB(S)\big]}.
\end{displaymath}
Subsequently, we bound the two terms $\IcalkA$ and $\IcalkB$ individually. 
First, we bound $\IcalkB$. 
We consider the case $p\le 8$. 
The von Bahr--Esseen inequality I~\citep[p. 650]{DasGupta2008} states that  for $1\le \ttt\le 2$ and for independent, real-valued, and mean $0$ variables $Z_1,\ldots, Z_n$, we have
\begin{displaymath}
	\E\bigg[ \normonebigg{\sum_{i=1}^n Z_i}^{\ttt} \bigg] \le \Big( 2-\frac{1}{n} \Big)\sum_{i=1}^n \E[\normone{X_i}^{\ttt}].
\end{displaymath}
The individual summands $\lossA(S_i)\lossB(S_i)-\EP[\lossA(S)\lossB(S)]$ for $i\in\Ik$ are independent and have mean $0$. Therefore,
\begin{displaymath}
	\begin{array}{rcl}
		\EP\Big[\IcalkB^{\frac{p}{4}}\Big] &=& \big(\frac{1}{\nn}\big)^{\frac{p}{4}} \EP\bigg[ \normoneBig{\sum_{i\in\Ik}\big(\lossA(S_i)\lossB(S_i)-\EP\big[\lossA(S)\lossB(S)\big]\big)}^{\frac{p}{4}}  \bigg]\\
		&\le& \big(\frac{1}{\nn}\big)^{-1+\frac{p}{4}} \Big( 2-\frac{1}{n} \Big) \frac{1}{\nn} \sum_{i\in\Ik}\EP\Big[ \normonebig{\lossA(S_i)\lossB(S_i)-\EP\big[\lossA(S)\lossB(S)\big]}^{\frac{p}{4}}  \Big]\\
		&=& \big(\frac{1}{\nn}\big)^{-1+\frac{p}{4}} \Big( 2-\frac{1}{n} \Big)\EP\Big[ \normonebig{\lossA(S)\lossB(S)-\EP\big[\lossA(S)\lossB(S)\big]}^{\frac{p}{4}}  \Big]
	\end{array}
\end{displaymath}
follows due to the von Bahr--Esseen inequality I because $1<\frac{p}{4}\le 2$ holds.
By H\"older's inequality, we have
\begin{displaymath}
	\begin{array}{rcl}
	\Big(\EP\Big[\normonebig{\lossA(S)}^{\frac{p}{4}}\normonebig{\lossB(S)}^{\frac{p}{4}}\Big]\Big)^{\frac{p}{4}}
	 &\le&   \EP\Big[\normonebig{\lossA(S)}^{\frac{p}{2}}\Big]^{\frac{2}{p}}\EP\Big[\normonebig{\lossB(S;\bg,\etazero)}^{\frac{p}{2}}\Big]^{\frac{2}{p}}\\
	 &\le& \normPbig{\lossAfull(S)}{\frac{p}{2}}\normPbig{\lossBfull(S)}{\frac{p}{2}}. 
	 \end{array}
\end{displaymath}
All the terms 
$\normP{\loss(S;\bzero,\etazero)}{\frac{p}{2}}$,
$\normP{\losstilde(S;\bzero,\etazero)}{\frac{p}{2}}$, 
$\normP{\lossone(S;\eta)}{\frac{p}{2}}$, and 
$\normP{\losstwo(S;\eta)}{\frac{p}{2}}$
are upper bounded by the finite real constant $\CnormLossboth$ by Lemma~\ref{lem:mNboth}. Thus, we have $\IcalkB=O_{\PP}(\NN^{\frac{p}{4}-1})$ by Lemma~\ref{lem:ChernozhukovLemma} because we have
\begin{displaymath}
	\begin{array}{cl}
		&\EP\Big[ \normonebig{\lossA(S)\lossB(S)-\EP\big[\lossA(S)\lossB(S)\big]}^{\frac{p}{4}}  \Big]^{\frac{4}{p}}\\
		=& \normP{\lossA(S)\lossB(S)-\EP\big[\lossA(S)\lossB(S)\big]}{\frac{p}{4}}\\
		\le& \normP{\lossA(S)\lossB(S)}{\frac{p}{4}} + \EP\big[\normone{\lossA(S)\lossB(S)}\big]\\
		\le& 2\normP{\lossA(S)\lossB(S)}{\frac{p}{4}} 
	\end{array}
\end{displaymath} 
by the triangle inequality, H\"older's inequality, and due to $\frac{p}{4}>1$. 

Next, consider the case $p>8$. 
Observe that
\begin{displaymath}
	\begin{array}{cl}
		&\EP\bigg[\Big( \frac{1}{\nn}\sum_{i\in\Ik}\lossA(S_i)\lossB(S_i) \Big)^2\bigg]\\
		=& \frac{1}{\nn}\EP\Big[\big(\lossA(S)\big)^2\big(\lossB(S)\big)^2\Big] 
			+ \frac{\nn(\nn-1)}{\nn^2}\EP\big[\lossA(S)\lossB(S)\big]^2
	\end{array}
\end{displaymath}
holds because the data sample is \iid. 
Thus, we infer
\begin{displaymath}
	\begin{array}{rcl}
		\EP[\IcalkB^2] &=& \EP\bigg[\Big( \frac{1}{\nn}\sum_{i\in\Ik}\lossA(S_i)\lossB(S_i) \Big)^2\bigg] + \EP\big[\lossA(S)\lossB(S) \big]^2\\
		&&\quad - 2 \EP\Big[\frac{1}{\nn}\sum_{i\in\Ik}\lossA(S_i)\lossB(S_i)\Big]\EP[\lossA(S)\lossB(S) ]\\
		&\le& \frac{1}{\nn}\EP\big[(\lossA(S))^2(\lossB(S))^2\big]. 
	\end{array}
\end{displaymath}
By the Cauchy--Schwarz inequality, we have
\begin{displaymath}
	\begin{array}{rcl}
	\frac{1}{\nn}\EP\Big[\big(\lossA(S))^2(\lossB(S)\big)^2\Big]
	 &\le&  \frac{1}{\nn} \sqrt{\EP\Big[\big(\lossA(S)\big)^4\Big]\EP\Big[\big(\lossB(S)\big)^4\Big]}\\
	 &\le& \frac{1}{\nn}\normPbig{\lossAfull(S)}{4}^2\normPbig{\lossBfull(S)}{4}^2. 
	 \end{array}
\end{displaymath}
All the terms 
$\normP{\loss(S;\bzero,\etazero)}{4}$ $\normP{\losstilde(S;\bzero,\etazero)}{4}$, 
$\normP{\lossone(S;\eta)}{4}$, and 
$\normP{\losstwo(S;\eta)}{4}$
are upper bounded by $\CnormLossboth$ by Lemma~\ref{lem:mNboth}. 
Thus, we have 
\begin{displaymath}
	\EP[\IcalkB^2] \le \frac{1}{\nn}\normPbig{\lossAfull(S)}{4}^2\normPbig{\lossBfull(S)}{4}^2
	\le  \frac{1}{\nn}(4\CnormLossboth)^4.
\end{displaymath}
We hence infer $\IcalkB=O_{\PP}(\NN^{-\frac{1}{2}})$ by Lemma~\ref{lem:ChernozhukovLemma}.

Second, we bound the term $\IcalkA$. 
For any real numbers $a_1$, $a_2$, $b_1$, and $b_2$ such that real numbers $c$ and $d$ exist that satisfy $\max\{|b_1|, |b_2|\}\le c$ and $\max\{|a_1-b_1|,|a_2-b_2|\}\le d$, we have $\normone{a_1a_2-b_1b_2}\le 2d(c+d)$. Indeed, we have
\begin{displaymath}
\begin{array}{rcl}
	\normone{a_1a_2-b_1b_2} &\le& |a_1-b_1|\cdot |a_2-b_2| +|b_1|\cdot |a_2-b_2| + |a_1-b_1|\cdot |b_2|\\
	&\le& d^2+cd+dc\\
	&\le& 2d(c+d)
	\end{array}
\end{displaymath}
 by the triangle inequality.

We apply this observation 
together with the triangle inequality and the Cauchy--Schwarz inequality to obtain
\begin{displaymath}
	\begin{array}{rcl}
		\IcalkA 
		&\le& \frac{1}{\nn}\sum_{i\in\Ik}\normonebig{\hlossA(S_i)\hlossB(S_i)
		- \lossA(S_i)\lossB(S_i)}\\
		&\le& \frac{2}{\nn}\sum_{i\in\Ik} \max\big\{\normonebig{\hlossA(S_i)-\lossA(S_i)},\normonebig{\hlossB(S_i)-\lossB(S_i)}\big\}\\
		&&\quad\cdot\Big(\max\big\{\normonebig{\lossA(S_i)},\normonebig{\lossB(S_i)}\big\}
		+  \max\big\{\normonebig{\hlossA(S_i)-\lossA(S_i)},\normonebig{\hlossB(S_i)-\lossB(S_i)}\big\}\Big)\\
		&\le& 2\Big(\frac{1}{\nn}\sum_{i\in\Ik} \max\Big\{\normonebig{\hlossA(S_i)-\lossA(S_i)}^2,\normonebig{\hlossB(S_i)-\lossB(S_i)}^2\Big\}\Big)^{\frac{1}{2}}\\
		&&\quad\cdot \Big( \frac{1}{\nn}\sum_{i\in\Ik}\big(\max\big\{\normonebig{\lossA(S_i)},\normonebig{\lossB(S_i)}\big\}\\
		&&\quad\quad\quad\quad\quad\quad\quad\quad
		+ \max\big\{\normonebig{\hlossA(S_i)-\lossA(S_i)},\normonebig{\hlossB(S_i)-\lossB(S_i)}\big\}\big)^2 \Big)^{\frac{1}{2}}. 
	\end{array}
\end{displaymath}
By the triangle inequality, we hence have 
\begin{equation}\label{boundIkl1g}
	\begin{array}{rcl}
		\IcalkA^2 
		&\le& 4 \RN \bigg(\frac{1}{\nn}\sum_{i\in\Ik}\Big(\normbig{\lossAfull(S_i)}^2+\normbig{\lossBfull(S_i)}^2\Big)+\RN\bigg)
	\end{array}
\end{equation}
by Lemma~\ref{lem:squareBound}, 
where
\begin{displaymath}
	\RN := \frac{1}{\nn}\sum_{i\in\Ik}\Big(\normbig{\hlossAfull(S_i) - \lossAfull(S_i)}^2+\normbig{\hlossBfull(S_i) - \lossBfull(S_i)}^2\Big). 
\end{displaymath}
Note that we have 
\begin{displaymath}
	\frac{1}{\nn}\sum_{i\in\Ik}\Big(\normbig{\lossAfull(S_i)}^2+\normbig{\lossBfull(S_i)}^2\Big)=O_{\PP}(1)
\end{displaymath}
by  Markov's inequality because  the terms 
$\normP{\loss(S;\bzero,\etazero)}{4}$ 
$\normP{\losstilde(S;\bzero,\etazero)}{4}$, 
$\normP{\lossone(S;\eta)}{4}$, and 
$\normP{\losstwo(S;\eta)}{4}$
are upper bounded by $\CnormLossboth$ by Lemma~\ref{lem:mNboth}. 
Thus, it suffices to bound the term $\RN$. 
To do this, we need to bound the four terms 
\begin{align}
	 &\frac{1}{\nn}\sum_{i\in\Ik}\norm{\loss(S_i;\hb,\hetaIkc) - \loss(S_i;\bzero,\etazero)}^2\label{eq:firstTermRNk},\\
	 &\frac{1}{\nn}\sum_{i\in\Ik}\norm{\losstilde(S_i;\hb,\hetaIkc)-\losstilde(S_i;\bzero,\etazero)}^2 \label{eq:secondTermRNk},\\
	 &\frac{1}{\nn}\sum_{i\in\Ik}\norm{\lossone(S_i;\hetaIkc)-\lossone(S_i;\etazero)}^2\label{eq:thirdTermRNk},\\
	 &\frac{1}{\nn}\sum_{i\in\Ik}\norm{\losstwo(S_i;\hetaIkc)-\losstwo(S_i;\etazero)}^2\label{eq:fourthTermRNk}.
\end{align}
First, we bound the two terms~\eqref{eq:firstTermRNk} and~\eqref{eq:secondTermRNk} simultaneously. Consider the random variable $U\in\{A,X\}$ and the quadruple $S=(U,X,W,Y)$.  
Because the score $\loss$ is linear in $\beta$, these two terms are upper bounded by
\begin{equation}\label{eq:boundRNboth}
	\begin{array}{cl}
		& \frac{1}{\nn}\sum_{i\in\Ik}\norm{-\loss^a(S_i;\hetaIkc)(\hb-\bzero) + \loss(S_i;\bzero,\hetaIkc)-\loss(S_i;\bzero,\etazero) }^2\\
		\le&\frac{2}{\nn}\sum_{i\in\Ik}\norm{\loss^a(S_i;\hetaIkc)(\hb-\bzero)}^2 +  \frac{2}{\nn}\sum_{i\in\Ik}\norm{\loss(S_i;\bzero,\hetaIkc)-\loss(S_i;\bzero,\etazero) }^2\\
	\end{array}
\end{equation}
due to the triangle inequality and Lemma~\ref{lem:squareBound}. 
Subsequently, we verify that 
\begin{displaymath}
	\frac{1}{\nn}\sum_{i\in\Ik}\norm{\loss^a(S_i;\hetaIkc)}^2=O_{\PP}(1)
\end{displaymath}
holds. 
Indeed, we have
\begin{equation}\label{eq:boundLossA}
	\begin{array}{rcl}
	\frac{1}{\nn}\sum_{i\in\Ik}\norm{\loss^a(S_i;\hetaIkc)}^2
	&=&\frac{1}{\nn}\sum_{i\in\Ik}\normBig{\big(U_i-\hmU^{\Ikc}(W_i)\big)\big(X_i-\hmX^{\Ikc}(W_i)\big)^T}^2\\
	&\le&\sqrt{\frac{1}{\nn}\sum_{i\in\Ik}\norm{U_i-\hmU^{\Ikc}(W_i)}^4}\sqrt{\frac{1}{\nn}\sum_{i\in\Ik}\norm{X_i-\hmX^{\Ikc}(W_i)}^4}
	\end{array}
\end{equation}
by the Cauchy--Schwarz inequality. 
We have
\begin{equation}\label{eq:RaisOone}
	\bigg( \frac{1}{\nn}\sum_{i\in\Ik}\norm{U_i-\mU^0(W_i)}^4 \bigg)^{\frac{1}{4}}
	= O_{\PP}(1)
\end{equation}
by  Markov's inequality because $\EP[\norm{U-\mU^0(W)}^4]$ is upper bounded by Lemma~\ref{lem:boundPnorm1} and Assumption~\ref{assumpt:DMLboth2}. 
 On the event $\EpsN$ that holds with $\PP$-probability $1-\DeltaN$, we have
\begin{equation}\label{eq:boundResEtaZero}
		 \EP\bigg[ \frac{1}{\nn}\sum_{i\in\Ik}\norm{\etazero(W_i)-\hetaIkc(W_i)}^4\Big|\SIkc  \bigg] 
		 = \EP\big[\norm{\etazero(W)-\hetaIkc(W)}^4|\SIkc\big]\\
		\le  \CpnormEta^4
\end{equation}
by Assumption~\ref{assumpt:DMLboth5}. 
We hence have $\frac{1}{\nn}\sum_{i\in\Ik}\norm{\etazero(W_i)-\hetaIkc(W_i)}=O_{\PP}(1)$ 
by Lemma~\ref{lem:ChernozhukovLemma}. 
Let us denote by $\norm{\cdot}_{\FIk, p}$ the $L^p$-norm with the empirical measure on the data indexed by $\Ik$. 
On the event $\EpsN$ that holds with $\PP$-probability $1-\DeltaN$, we have \begin{equation}\label{eq:Rahat}
	\begin{array}{rcl}
		\frac{1}{\nn}\sum_{i\in\Ik}\norm{U_i-\hmU^{\Ikc}(W_i)}^4
		&=&\norm{U-\hmU^{\Ikc}(W)}_{\FIk, 4}^4\\
		&\le& \big(\norm{U-\mU^0(W)}_{\FIk,4} + \norm{\mU^0(W)-\hmU^{\Ikc}(W)}_{\FIk,4}\big)^4\\
		&\le& \big(\norm{U-\mU^0(W)}_{\FIk,4} + \norm{\eta^0(W)-\hetaIkc(W)}_{\FIk,4}\big)^4\\
		&=& O_{\PP}(1)
	\end{array}
\end{equation}
by the triangle inequality,~\eqref{eq:RaisOone}, and~\eqref{eq:boundResEtaZero}. 
Analogous arguments lead to 
\begin{equation}\label{eq:Rxhat}
	\frac{1}{\nn}\sum_{i\in\Ik}\norm{X_i-\hmX^{\Ikc}(W_i)}^4 = O_{\PP}(1). 
\end{equation}
We combine~\eqref{eq:boundLossA},~\eqref{eq:Rahat}, and~\eqref{eq:Rxhat}  to obtain
\begin{equation}\label{eq:boundLossAO}
	\frac{1}{\nn}\sum_{i\in\Ik}\norm{\loss^a(S_i;\hetaIkc)}^2=O_{\PP}(1).
\end{equation}
Because $\norm{\hb-\bzero}^2=O_{\PP}(N^{-1})$ holds by Theorem~\ref{thm:asymptNormal} and Theorem~\ref{thm:asymptNormalgamma}, we can bound the first summand in~\eqref{eq:boundRNboth} by
\begin{equation}\label{eq:RNfirstTermboth}
	\frac{1}{\nn}\sum_{i\in\Ik}\norm{\loss^a(S_i;\hetaIkc)(\hb-\bzero)}^2  
	= O_{\PP}(1)O_{\PP}(N^{-1}) = O_{\PP}(\NN^{-1})
\end{equation}
due to the Cauchy--Schwarz inequality and~\eqref{eq:boundLossAO}. 
On the event $\EpsN$ that holds with $\PP$-probability $1-\DeltaNnumber$, the conditional expectation given $\SIkc$ of the second summand in~\eqref{eq:boundRNboth} is equal to
\begin{displaymath}
	\begin{array}{cl}
		&\EP\Big[ \frac{2}{\nn}\sum_{i\in\Ik}\norm{\loss(S_i;\bzero,\hetaIkc)-\loss(S_i;\bzero,\etazero)}^2 \Big|\SIkc\Big]\\
		=& 2\EP\big[\norm{\loss(S;\bzero,\hetaIkc)-\loss(S;\bzero,\etazero)}^2 \big|\SIkc\big]\\
		\le& 2\sup_{\eta\in\TauN}\EP\big[\norm{\loss(S;\bzero,\eta)-\loss(S;\bzero,\etazero)}^2\big]\\
		\lesssim & \rNpnumber^2
	\end{array}
\end{displaymath}
due to arguments that are analogous to~\mbox{\eqref{eq:rNp1}--\eqref{eq:rNp5}} presented in the proof of Lemma~\ref{lem:boundRN}.
Because the event $\EpsN$ holds with $\PP$-probability $1-\DeltaN=1-o(1)$, we infer 
\begin{displaymath}
	\frac{1}{\nn}\sum_{i\in\Ik}\norm{\loss^a(S_i;\hetaIkc)(\hb-\bzero) + \loss(S_i;\bzero,\hetaIkc)-\loss(S_i;\bzero,\etazero) }^2=O_{\PP}(\NN^{-1}+\rNpnumber^2)
\end{displaymath}
by Lemma~\ref{lem:ChernozhukovLemma}. 
Next, we bound the two terms given in~\eqref{eq:thirdTermRNk} and~\eqref{eq:fourthTermRNk}. We first consider the term given in~\eqref{eq:thirdTermRNk}.
On the event $\EpsN$, we have
\begin{displaymath}
	\begin{array}{cl}
		&\EP\Big[ \frac{1}{\nn}\sum_{i\in\Ik}\norm{\lossone(S_i;\hetaIkc)-\lossone(S_i;\etazero)}^2 \Big|\SIkc\Big]\\
		=& \EP\big[\norm{\lossone(S;\hetaIkc)-\lossone(S;\etazero)}^2 \big|\SIkc\big]\\
		\le& \sup_{\eta\in\TauN}\EP\big[\norm{\lossone(S;\eta)-\lossone(S;\etazero)}^2\big]\\
		\lesssim& \rNpnumber^2
	\end{array}
\end{displaymath}
due to arguments that are analogous to~\mbox{\eqref{eq:rNp1}--\eqref{eq:rNp5}} presented in the proof of Lemma~\ref{lem:boundRN}.
Because the event $\EpsN$ holds with probability $1-\DeltaN=1-o(1)$, we infer 
\begin{displaymath}
	\frac{1}{\nn}\sum_{i\in\Ik}\norm{ \lossone(S_i;\hetaIkc)-\lossone(S_i;\etazero) }^2=O_{\PP}(\rNpnumber^2)
\end{displaymath}
by Lemma~\ref{lem:ChernozhukovLemma}.
On the event $\EpsN$, the conditional expectation given $\SIkc$ of the term~\eqref{eq:fourthTermRNk} is given by
\begin{displaymath}
	\begin{array}{cl}
		&\EP\Big[ \frac{1}{\nn}\sum_{i\in\Ik}\norm{\losstwo(S_i;\hetaIkc)-\losstwo(S_i;\etazero)}^2 \Big|\SIkc\Big]\\
		=& \EP\big[\norm{\losstwo(S;\hetaIkc)-\losstwo(S;\etazero)}^2 \big|\SIkc\big]\\
		\le& \sup_{\eta\in\TauN}\EP\big[\norm{\losstwo(S;\eta)-\losstwo(S;\etazero)}^2\big]\\
		\lesssim& \rNpnumber^2
	\end{array}
\end{displaymath}
due to arguments that are analogous to~\mbox{\eqref{eq:rNp1}--\eqref{eq:rNp5}} presented in the proof of Lemma~\ref{lem:boundRN}.
Because the event $\EpsN$ holds with probability $1-\DeltaN=1-o(1)$, we infer 
\begin{displaymath}
	\frac{1}{\nn}\sum_{i\in\Ik}\norm{ \losstwo(S_i;\hetaIkc)-\losstwo(S_i;\etazero) }^2=O_{\PP}(\rNpnumber^2)
\end{displaymath}
by Lemma~\ref{lem:ChernozhukovLemma}.
Therefore, we have $\IcalkA=O_{\PP}(\NN^{-\frac{1}{2}}+\rNpnumber )$ by~\eqref{boundIkl1g}. In total, we thus have 
\begin{displaymath}
	\Icalk=O_{\PP}\Big(\NN^{\max\big\{\frac{4}{p}-1, -\frac{1}{2}\big\}}\Big)+O_{\PP}\big(\NN^{-\frac{1}{2}}+\rNpnumber\big)=O_{\PP}\Big(\NN^{\max\big\{\frac{4}{p}-1, -\frac{1}{2}\big\}}+\rNpnumber\Big).   
\end{displaymath}
\end{proof}

\begin{theorem}\label{thm:estSD}
Suppose Assumption~\ref{assumpt:DMLboth} holds. 
	Introduce the matrix
	\begin{displaymath}
	\begin{array}{rcl}
	\hJzerok &:=&\bigg(\frac{1}{\nn}\sum_{i\in\Ik}\hbRxki(\hRaki)^T\Big(\frac{1}{\nn}\sum_{i\in\Ik} \hRaki(\hRaki)^T\Big)^{-1} \frac{1}{\nn}\sum_{i\in\Ik}\hRaki(\hbRxki)^T\bigg)^{-1}\\
			&&\quad\cdot\frac{1}{\nn}\sum_{i\in\Ik}\hbRxki(\hRaki)^T \Big(\frac{1}{\nn}\sum_{i\in\Ik} \hRaki(\hRaki)^T\Big)^{-1}.
			\end{array}
	\end{displaymath}
	Let its average over $\kk\in\indset{\KK}$ be
	\begin{displaymath}
		\hJzero := \frac{1}{\KK}\sum_{\kk=1}^{\KK}\hJzerok. 
	\end{displaymath}
	Define further the estimator
	\begin{displaymath}
		\hsigma^2 := 
		\hJzero
		\Big( \frac{1}{\KK}\sum_{\kk=1}^{\KK} \frac{1}{\nn}\sum_{i\in\Ik}\loss(S_i;\hbetaN,\hetaIkc)\loss^T(S_i;\hbetaN,\hetaIkc) \Big)
		\hJzero^T
	\end{displaymath}
	of $\sigma^2$ from Theorem~\ref{thm:asymptNormal}, where $\hbetaN\in\{\hbetaNDMLone,\hbetaNDMLtwo\}$. We then have $\hsigma^2 = \sigma^2+O_{\PP}(\rhoNtilde)$, where $\rhoNtilde= \NN^{\max\big\{\frac{4}{p}-1, -\frac{1}{2}\big\}}+\rNpnumber$ is as in Definition~\ref{def:asymptNormal}.
\end{theorem}
\begin{proof}[Proof of Theorem~\ref{thm:estSD}]
We derived $\hJzerok=\Jzero+O_{\PP}\big(\NN^{-\frac{1}{2}}(1+\rhoN)\big)$ in the  proof of Theorem~\ref{thm:asymptNormal}. Thus,  $\hJzero=\Jzero + O_{\PP}\big(\NN^{-\frac{1}{2}}(1+\rhoN)\big)$ holds because $\KK$ is a fixed number independent of $\NN$. To conclude the proof, it  suffices to verify 
\begin{displaymath}
	\normbigg{\frac{1}{\nn}\sum_{i\in\Ik}\loss(S_i;\hbetaN,\hetaIkc)\loss^T(S_i;\hbetaN,\hetaIkc)
	- \EP\big[\loss(S;\betazero,\etazero)\loss^T(S;\betazero,\etazero)\big]}=O_{\PP}(\rhoNtilde).
\end{displaymath}
But this statement holds by Lemma~\ref{lem:helperVarianceConsistent} because the dimensions of $A$ and $X$ are fixed. 
\end{proof}

	\section{Proofs of Section~\ref{sect:regularizedDML}}\label{sect:proofsRegularizedDML}

\begin{definition}\label{def:asymptNormalgamma}
Let $\gamma\ge 0$ and recall the scalar
$\rhoN =\rNpnumber + \NN^{\frac{1}{2}}\lambdaNpnumber$ in Definition~\ref{def:asymptNormal}. 
Introduce the function
		\begin{displaymath}
			\begin{array}{rcl}
		\lossoverlinep(\cdot;\bg,\etazero)
		&:=&\losstilde(\cdot;\bg,\etazero)
				+ (\gamma-1)
				\matC\loss(\cdot;\bg,\etazero) \\
				&&\quad\quad + (\gamma-1)\big(\lossone(\cdot;\etazero)-\EP[\lossone(S;\etazero)]\big)\matE\\
				&&\quad\quad - (\gamma-1)\matC\big(\losstwo(\cdot;\etazero)-\EP[\losstwo(S;\etazero)]\big)\matE.
	\end{array}
	\end{displaymath}
	Let 
	\begin{displaymath}
		\matD := \EP\big[\lossoverlinep(S;\bg,\etazero)(\lossoverlinep(S;\bg,\etazero))^T  \big],
	\end{displaymath}
	and let the approximate variance
	\begin{displaymath}
		\sigma^2(\gamma) :=\big( \matA + (\gamma-1)\matB \big)^{-1}\matD \big( \matA^T + (\gamma-1)\matB^T \big)^{-1}. 
	\end{displaymath}
	Moreover, define the influence function
	\begin{displaymath}
		\lossoverline(\cdot;\bg,\etazero)
		:=\sigma^{-1}(\gamma) \big( \matA + (\gamma-1)\matB \big)^{-1}\lossoverlinep(\cdot;\bg,\etazero). 
	\end{displaymath}	
\end{definition}
	
\begin{proof}[Proof of Theorem~\ref{thm:asymptNormalgamma}]
This proof is based on~\citet{Chernozhukov2018}. 
The matrices $\matA+(\gamma-1)\matB$ and $\matD$ are invertible by Assumption~\ref{assumpt:DMLboth4}. 
Hence, $\sigma^2(\gamma)$ is invertible. 

Subsequently, we show the stronger statement
\begin{equation}\label{eq:asymptNormalgamma}
		\sqrt{\NN}\sigma^{-1}(\gamma)(\hbg-\bg) = \frac{1}{\sqrt{\NN}}\sum_{i=1}^{\NN}\lossoverline(S_i;\bg, \etazero) + O_{\PP}(\rhoN) \stackrel{d}{\rightarrow}\mathcal{N}(0,\one_{d\times d})\quad (\NN\rightarrow\infty),
	\end{equation}
where $\hbg$ denotes the DML2 estimator  $\hbgDMLtwo$ or its DML1 variant $\hbgDMLone$, and where $\lossoverline$ is as in Definition~\ref{def:asymptNormalgamma}. 
We first consider $\hbgDMLtwo$ and afterwards $\hbgDMLone$. 
Fix a sequence $\{\PPN\}_{\NN\ge 1}$ such that $\PPN\in\PcalN$ for all $\NN\ge 1$. 
Because this sequence is chosen arbitrarily, it suffices to show
\begin{displaymath}
		\sqrt{\NN}\sigma^{-1}(\gamma)(\hbgDMLtwo-\bg) 
		= \frac{1}{\sqrt{\NN}}\sum_{i=1}^{\NN}\lossoverline(S_i;\bg, \etazero) + O_{\PPN}(\rhoN) \stackrel{d}{\rightarrow}\mathcal{N}(0,\one_{d\times d})\quad (\NN\rightarrow\infty). 
	\end{displaymath}
We have
\begin{equation}\label{eq:hbgDMLtwoProof}
	\begin{array}{rcl}
		\hbgDMLtwo
		&=&\Big( \frac{1}{\KK}\sum_{\kk=1}^{\KK}\big(\hbRxk\big)^T\big(\one+(\gamma-1)\PiIkcIk\big)\hbRxk \Big)^{-1}\\
		&&\quad\cdot \frac{1}{\KK}\sum_{\kk=1}^{\KK}\big(\hbRxk\big)^T\big(\one+(\gamma-1)\PiIkcIk\big)\hbRyk\\
		&=&  \Bigg( \frac{1}{\KK}\sum_{\kk=1}^{\KK}\bigg(  
			\frac{1}{\nn}\big(\X^{\Ik}-\hmX^{\Ikc}(\W^{\Ik})\big)^T\big(\X^{\Ik}-\hmX^{\Ikc}(\W^{\Ik})\big) \\
			&&\quad\quad\quad\quad\quad\quad
			+(\gamma-1)\cdot\frac{1}{\nn}\big(\X^{\Ik}-\hmX^{\Ikc}(\W^{\Ik})\big)^T\big(\A^{\Ik}-\hmA^{\Ikc}(\W^{\Ik})\big)\\
		&&\quad\quad\quad\quad\quad\quad\quad\cdot\Big(\frac{1}{\nn}\big(\A^{\Ik}-\hmA^{\Ikc}(\W^{\Ik})\big)^T(\A^{\Ik}-\hmA^{\Ikc}(\W^{\Ik})\Big)^{-1}\\
		&&\quad\quad\quad\quad\quad\quad\quad\cdot\frac{1}{\nn}\big(\A^{\Ik}-\hmA^{\Ikc}(\W^{\Ik})\big)^T\big(\X^{\Ik}-\hmX^{\Ikc}(\W^{\Ik})\big)
			\bigg)\Bigg)^{-1}\\
		&&\quad\cdot \frac{1}{\KK}\sum_{\kk=1}^{\KK}
		\Big(  \frac{1}{\nn}\big(\X^{\Ik}-\hmX^{\Ikc}(\W^{\Ik})\big)^T\big(\Y^{\Ik}-\hmY^{\Ikc}(\W^{\Ik})\big) \\
		&&\quad\quad\quad\quad\quad\quad + 
		(\gamma-1)\cdot\frac{1}{\nn}\big(\X^{\Ik}-\hmX^{\Ikc}(\W^{\Ik})\big)^T(\A^{\Ik}-\hmA^{\Ikc}(\W^{\Ik})\big)\\
		&&\quad\quad\quad\quad\quad\quad\quad\cdot\Big(\frac{1}{\nn}\big(\A^{\Ik}-\hmA^{\Ikc}(\W^{\Ik})\big)^T\big(\A^{\Ik}-\hmA^{\Ikc}(\W^{\Ik})\big)\Big)^{-1}\\
		&&\quad\quad\quad\quad\quad\quad\quad\cdot\frac{1}{\nn}\big(\A^{\Ik}-\hmA^{\Ikc}(\W^{\Ik})\big)^T\big(\Y^{\Ik}-\hmY^{\Ikc}(\W^{\Ik})\big)
		\Big)
	\end{array}
\end{equation}
by~\eqref{eq:regularizedBetaDMLtwo}. 
By Lemma~\ref{lem:D1convInProb}, we have

\begin{displaymath}
\begin{array}{cl}
	&\frac{1}{\nn} \big(\X^{\Ik}-\hmX^{\Ikc}(\W^{\Ik})\big)^T\big(\A^{\Ik}-\hmA^{\Ikc}(\W^{\Ik})\big)\\
	=& \EPN\Big[\big(X-\mX^0(W)\big)\big(A-\mA^0(W)\big)^T\Big] + O_{\PPN}\big(\NN^{-\frac{1}{2}}(1+\rhoN)\big),
\end{array}
\end{displaymath}

\begin{displaymath}
\begin{array}{cl}
	&\frac{1}{\nn} \big(\A^{\Ik}-\hmA^{\Ikc}(\W^{\Ik})\big)^T\big(\A^{\Ik}-\hmA^{\Ikc}(\W^{\Ik})\big)\\
		=& \EPN\Big[\big(A-\mA^0(W)\big)\big(A-\mA^0(W)\big)^T\Big] + O_{\PPN}\big(\NN^{-\frac{1}{2}}(1+\rhoN)\big),
\end{array}
\end{displaymath}

\begin{displaymath}
\begin{array}{cl}
&\frac{1}{\nn} \big(\X^{\Ik}-\hmX^{\Ikc}(\W^{\Ik})\big)^T\big(\X^{\Ik}-\hmX^{\Ikc}(\W^{\Ik})\big)\\
	=& \EPN\Big[\big(X-\mX^0(W)\big)\big((X-\mX^0(W)\big)^T\Big] + O_{\PPN}\big(\NN^{-\frac{1}{2}}(1+\rhoN)\big).
\end{array}
\end{displaymath}
By Weyl's inequality and  Slutsky's theorem, we hence have
\begin{equation}\label{eq:combine1}
	\begin{array}{cl}
		&\sqrt{\NN}(\hbgDMLtwo-\bg)\\
		=& \Big(\big( \matA + (\gamma-1)\matB \big)^{-1} +O_{\PPN}\big(\NN^{-\frac{1}{2}}(1+\rhoN)\big)\Big) \\
		&\quad\cdot \frac{1}{\sqrt{\KK}}\sum_{\kk=1}^{\KK}
			\frac{1}{\sqrt{\nn}}\Big(
				\big(\X^{\Ik}-\hmX^{\Ikc}(\W^{\Ik})\big)^T\Big(\Y^{\Ik}-\hmY^{\Ikc}(\W^{\Ik}) - \big(\X^{\Ik}-\hmX^{\Ikc}(\W^{\Ik})\big)\bg\Big)\\
		&\quad\quad\quad + (\gamma-1)\cdot\frac{1}{\nn}\big(\X^{\Ik}-\hmX^{\Ikc}(\W^{\Ik})\big)^T\big(\A^{\Ik}-\hmA^{\Ikc}(\W^{\Ik})\big)\\
		&\quad\quad\quad\quad\quad\quad\quad\cdot\Big(\frac{1}{\nn}\big(\A^{\Ik}-\hmA^{\Ikc}(\W^{\Ik})\big)^T \big(\A^{\Ik}-\hmA^{\Ikc}(\W^{\Ik})\big)\Big)^{-1}\\
		&\quad\quad\quad\quad\quad\quad\quad\cdot\big(\A^{\Ik}-\hmA^{\Ikc}(\W^{\Ik})\big)^T\Big(\Y^{\Ik}-\hmY^{\Ikc}(\W^{\Ik}) - \big(\X^{\Ik}-\hmX^{\Ikc}(\W^{\Ik})\big)\bg\Big) 
		\Big)\\
		&= \Big(\big( \matA + (\gamma-1)\matB \big)^{-1} +O_{\PPN}\big(\NN^{-\frac{1}{2}}(1+\rhoN)\big)\Big) \\
		&\quad\cdot \frac{1}{\sqrt{\KK}}\sum_{\kk=1}^{\KK}
			\Big(\frac{1}{\sqrt{\nn}}\sum_{i\in\Ik}\losstilde(S_i;\bg,\hetaIkc)\\
			&\quad+ (\gamma-1)\cdot\frac{1}{\sqrt{\nn}}
			\sum_{i\in\Ik}\lossone(S_i;\hetaIkc)\cdot\big(\frac{1}{\nn}\sum_{i\in\Ik}\losstwo(S_i;\hetaIkc)\big)^{-1}\cdot
		\frac{1}{\nn}\sum_{i\in\Ik}\loss(S_i;\bg,\hetaIkc)
		\Big)
	\end{array}
\end{equation}
due to~\eqref{eq:hbgDMLtwoProof} because $\KK$ and $\gamma$ are constants independent of $\NN$ and because  $\NN=\nn\KK$ holds. Let $\kk\in\indset{\KK}$. 
Next, we analyze the individual factors of the last summand  in~\eqref{eq:combine1}. 
By Lemma~\ref{lem:boundRN}, 
we have
\begin{equation}\label{eq:combine1LastSummand1}
	\begin{array}{cl}
		&\frac{1}{\sqrt{\nn}}\sum_{i\in\Ik}\loss(S_i;\bg,\hetaIkc) \\
		=& \frac{1}{\sqrt{\nn}}\sum_{i\in\Ik}\loss(S_i;\bg,\etazero) + \Big(\frac{1}{\sqrt{\nn}}\sum_{i\in\Ik}\loss(S_i;\bg,\hetaIkc) - \frac{1}{\sqrt{\nn}}\sum_{i\in\Ik}\loss(S_i;\bg,\etazero)\Big)\\
		=& \frac{1}{\sqrt{\nn}}\sum_{i\in\Ik}\loss(S_i;\bg,\etazero) + O_{\PPN}(\rhoN),
	\end{array}
\end{equation}
and
\begin{equation}\label{eq:combine1LastSummand2}
	\begin{array}{cl}
		&\frac{1}{\sqrt{\nn}}\sum_{i\in\Ik}\losstilde(S_i;\bg,\hetaIkc) \\
		=& \frac{1}{\sqrt{\nn}}\sum_{i\in\Ik}\losstilde(S_i;\bg,\etazero) + \Big(\frac{1}{\sqrt{\nn}}\sum_{i\in\Ik}\losstilde(S_i;\bg,\hetaIkc) - \frac{1}{\sqrt{\nn}}\sum_{i\in\Ik}\losstilde(S_i;\bg,\etazero)\Big)\\
		=& \frac{1}{\sqrt{\nn}}\sum_{i\in\Ik}\losstilde(S_i;\bg,\etazero) +O_{\PPN}(\rhoN),
	\end{array}
\end{equation}
and
\begin{equation}\label{eq:combine1LastSummand3}
	\begin{array}{cl}
		&\frac{1}{\nn}\sum_{i\in\Ik}\lossone(S_i;\hetaIkc)\\
		=& \frac{1}{\nn}\sum_{i\in\Ik}(\lossone(S_i;\hetaIkc) - \lossone(S_i;\etazero)) 
		+  \frac{1}{\nn}\sum_{i\in\Ik}(\lossone(S_i;\etazero)-\EPN[\lossone(S;\etazero)]) \\
		&\quad+ \EPN[\lossone(S;\etazero)]\\ 
		=& O_{\PPN}(\NN^{-\frac{1}{2}}\rhoN) +  \frac{1}{\nn}\sum_{i\in\Ik}(\lossone(S_i;\etazero)-\EPN[\lossone(S;\etazero)])  + \EPN[\lossone(S;\etazero)].
	\end{array}
\end{equation}
We apply a series expansion to obtain 
\begin{equation}\label{eq:combine1LastSummand4}
	\begin{array}{cl}
		& \Big( \frac{1}{\nn}\sum_{i\in\Ik}\losstwo(S_i;\hetaIkc)\Big)^{-1}\\
		=&\Big( \EPN[\losstwo(S;\etazero)] + \frac{1}{\nn}\sum_{i\in\Ik}\big(\losstwo(S_i;\hetaIkc)- \losstwo(S_i;\etazero)\big) \\
		&\quad\quad+ \frac{1}{\nn}\sum_{i\in\Ik}\big(\losstwo(S_i;\etazero)-\EPN[\losstwo(S;\etazero)]\big) \Big)^{-1}\\
		=& \EPN[\losstwo(S;\etazero)]^{-1} 
		-  \EPN[\losstwo(S;\etazero)]^{-1}\frac{1}{\nn}\sum_{i\in\Ik}\big(\losstwo(S_i;\hetaIkc)- \losstwo(S_i;\etazero)\big)\EPN[\losstwo(S;\etazero)]^{-1} \\
		&\quad - \EPN[\losstwo(S;\etazero)]^{-1}
		\frac{1}{\nn}\sum_{i\in\Ik}\big( \losstwo(S_i;\etazero) - \EPN[\losstwo(S;\etazero)]\big)\EPN[\losstwo(S;\etazero)]^{-1} \\
		&\quad + O_{\PPN}\Big( \normBig{\frac{1}{\nn}\sum_{i\in\Ik}\big(\losstwo(S_i;\hetaIkc)- \losstwo(S_i;\etazero)\big)}^2 \\
		&\quad\quad\quad\quad\quad+ \normBig{\frac{1}{\nn}\sum_{i\in\Ik}\big( \losstwo(S_i;\etazero) - \EPN[\losstwo(S;\etazero)]\big)}^2 \Big)\\
		=& \EPN[\losstwo(S;\etazero)]^{-1} + O_{\PPN}\big(\NN^{-\frac{1}{2}}\rhoN\big) + O_{\PPN}\Big( O_{\PPN}\big(\NN^{-1}\rhoN^2\big) +  O_{\PPN}(\NN^{-1})\Big) \\
		&\quad- \EPN[\losstwo(S;\etazero)]^{-1}\frac{1}{\nn}\sum_{i\in\Ik}\big( \losstwo(S_i;\etazero) - \EPN[\losstwo(S;\etazero)]\big)\EPN[\losstwo(S;\etazero)]^{-1} \\
		=& \EPN[\losstwo(S;\etazero)]^{-1}+ O_{\PPN}\big(\NN^{-\frac{1}{2}}\rhoN\big)\\
		&\quad - \EPN[\losstwo(S;\etazero)]^{-1}\frac{1}{\nn}\sum_{i\in\Ik}\big( \losstwo(S_i;\etazero) - \EPN[\losstwo(S;\etazero)]\big)\EPN[\losstwo(S;\etazero)]^{-1}
	\end{array}
\end{equation}
due to Lemma~\ref{lem:boundRN},  the Lindeberg--Feller CLT,  the Cramer--Wold device, because $\rhoN\lesssim\deltaN^{\frac{1}{4}}$ holds by Lemma~\ref{lem:boundRN}, and because $\deltaN^{\frac{1}{4}}\ge\NN^{-\frac{1}{2}}$ holds by Assumption~\ref{assumpt:DMLboth}. 
Thus, the last summand  in~\eqref{eq:combine1} can be expressed as
\begin{equation}\label{eq:combine1LastSummand}
	\begin{array}{cl}
	&\frac{1}{\sqrt{\nn}}
			\sum_{i\in\Ik}\lossone(S_i;\hetaIkc)\cdot\big(\frac{1}{\nn}\sum_{i\in\Ik}\losstwo(S_i;\hetaIkc)\big)^{-1}\cdot
		\frac{1}{\nn}\sum_{i\in\Ik}\loss(S_i;\bg,\hetaIkc)\\
		=&\sqrt{\nn}\Big(O_{\PPN}\big(\NN^{-\frac{1}{2}}\rhoN\big) +  \frac{1}{\nn}\sum_{i\in\Ik}\big(\lossone(S_i;\etazero)-\EPN[\lossone(S;\etazero)]\big)  + \EPN[\lossone(S;\etazero)]\Big)\\
		&\quad\cdot\Big(  \EPN[\losstwo(S;\etazero)]^{-1}+ O_{\PPN}\big(\NN^{-\frac{1}{2}}\rhoN\big)\\
		&\quad\quad\quad - \EPN[\losstwo(S;\etazero)]^{-1}\frac{1}{\nn}\sum_{i\in\Ik}\big( \losstwo(S_i;\etazero) - \EPN[\losstwo(S;\etazero)]\big)\EPN[\losstwo(S;\etazero)]^{-1} \Big)\\
		&\quad\cdot \Big( \frac{1}{\nn}\sum_{i\in\Ik}\loss(S_i;\bg,\etazero) + O_{\PPN}\big(\NN^{-\frac{1}{2}}\rhoN\big) \Big)\\
		=&\frac{1}{\sqrt{\nn}}\sum_{i\in\Ik}\big(\lossone(S_i;\etazero)-\EPN[\lossone(S;\etazero)]\big)\EPN[\losstwo(S;\etazero)]^{-1}\EPN[\loss(S;\bg,\etazero)]\\
		& \quad + \EPN[\lossone(S;\etazero)]\EPN[\losstwo(S;\etazero)]^{-1}\frac{1}{\sqrt{\nn}}\sum_{i\in\Ik}\loss(S_i;\bg,\etazero)\\
		&\quad - \EPN[\lossone(S;\etazero)]\EPN[\losstwo(S;\etazero)]^{-1}\frac{1}{\sqrt{\nn}}\sum_{i\in\Ik}\big(\losstwo(S_i;\etazero)-\EPN[\losstwo(S;\etazero)]\big)\\
		&\quad\quad\cdot\EPN[\losstwo(S;\etazero)]^{-1}\EPN[\loss(S;\bg,\etazero)] + O_{\PPN}(\rhoN)
	\end{array}
\end{equation}
due to~\mbox{\eqref{eq:combine1LastSummand1}--\eqref{eq:combine1LastSummand4}}, the Lindeberg--Feller CLT and the Cramer--Wold device.

We combine~\eqref{eq:combine1} and~\eqref{eq:combine1LastSummand} and obtain
\begin{equation}\label{eq:finalCLTeq}
	\begin{array}{cl}
		&\sqrt{\NN}(\hbgDMLtwo-\bg)\\
			=& \Big(\big( \matA + (\gamma-1)\matB \big)^{-1} +O_{\PPN}\big(\NN^{-\frac{1}{2}}(1+\rhoN)\big)\Big) \\
			&\quad\cdot \frac{1}{\sqrt{\KK}}\sum_{\kk=1}^{\KK}
			\frac{1}{\sqrt{\nn}}\sum_{i\in\Ik}\Big(\losstilde(S_i;\bg,\etazero)
				+ (\gamma-1)
				\matC\loss(S_i;\bg,\etazero) \\
				& \quad\quad +(\gamma-1) \big(\lossone(S_i;\etazero)-\EPN[\lossone(S;\etazero)]\big)\matE\\
				&\quad\quad - (\gamma-1)\matC\big(\losstwo(S_i;\etazero)-\EPN[\losstwo(S;\etazero)]\big)\matE
			\Big) + O_{\PPN}(\rhoN)\\
		=& \Big(\big( \matA + (\gamma-1)\matB \big)^{-1}\Big) \\
			&\quad\cdot \frac{1}{\sqrt{\NN}}\sum_{i=1}^{\NN}
			\Big(\losstilde(S_i;\bg,\etazero)
				+ (\gamma-1)
				\matC\loss(S_i;\bg,\etazero) \\
				& \quad\quad + (\gamma-1)\big(\lossone(S_i;\etazero)-\EPN[\lossone(S;\etazero)]\big)\matE\\
				&\quad\quad
				- (\gamma-1)\matC\big(\losstwo(S_i;\etazero)-\EPN[\losstwo(S;\etazero)]\big)\matE
			\Big) + O_{\PPN}(\rhoN)
	\end{array}
\end{equation}
by the Lindeberg--Feller CLT and the Cramer--Wold device.
We conclude our proof for the \DMLtwo\ method
by the Lindeberg--Feller CLT and the Cramer--Wold device. 

Subsequently, we consider the \DMLone\ method. 
It suffices to show that~\eqref{eq:asymptNormalgamma} holds uniformly over $\PP\in\PcalN$. 
Fix a sequence $\{\PPN\}_{\NN\ge 1}$ such that $\PPN\in\PcalN$ for all $\NN\ge 1$. 
Because this sequence is chosen arbitrarily, it suffices to show 
\begin{displaymath}
		\sqrt{\NN}\sigma^{-1}(\gamma)(\hbgDMLone-\bg) = \frac{1}{\sqrt{\NN}}\sum_{i=1}^{\NN}\lossoverline( S_i;\bg, \etazero) + O_{\PPN}(\rhoN) \stackrel{d}{\rightarrow}\mathcal{N}(0,\one_{d\times d})\quad (\NN\rightarrow\infty). 
	\end{displaymath}
We have
\begin{equation}\label{eq:hbgDMLoneProof}
	\begin{array}{rcl}
		\hbgDMLone
		&=&\frac{1}{\KK}\sum_{\kk=1}^{\KK}\Big( \big(\hbRxk\big)^T\big(\one+(\gamma-1)\PiIkcIk\big)\hbRxk \Big)^{-1}\\
		&&\quad\quad\quad\quad\quad\cdot (\hbRxk)^T\big(\one+(\gamma-1)\PiIkcIk\big)\hbRyk\\
		&=&  \frac{1}{\KK}\sum_{\kk=1}^{\KK}\bigg( \Big(  
			\frac{1}{\nn}\big(\X^{\Ik}-\hmX^{\Ikc}(\W^{\Ik})\big)^T\big(\X^{\Ik}-\hmX^{\Ikc}(\W^{\Ik})\big) \\
			&&\quad\quad\quad\quad\quad\quad
			+(\gamma-1)\cdot\frac{1}{\nn}\big(\X^{\Ik}-\hmX^{\Ikc}(\W^{\Ik})\big)^T\big(\A^{\Ik}-\hmA^{\Ikc}(\W^{\Ik})\big)\\
		&&\quad\quad\quad\quad\quad\quad\quad\cdot\Big(\frac{1}{\nn}\big(\A^{\Ik}-\hmA^{\Ikc}(\W^{\Ik})\big)^T\big(\A^{\Ik}-\hmA^{\Ikc}(\W^{\Ik})\big)\Big)^{-1}\\
		&&\quad\quad\quad\quad\quad\quad\quad\cdot\frac{1}{\nn}\big(\A^{\Ik}-\hmA^{\Ikc}(\W^{\Ik})\big)^T\big(\X^{\Ik}-\hmX^{\Ikc}(\W^{\Ik})\big)
			\Big)\bigg)^{-1}\\
		&&\quad\quad\quad\quad\quad\cdot 
		\Big(  \frac{1}{\nn}\big(\X^{\Ik}-\hmX^{\Ikc}(\W^{\Ik})\big)^T\big(\Y^{\Ik}-\hmY^{\Ikc}(\W^{\Ik})\big) \\
		&&\quad\quad\quad\quad\quad\quad + 
		(\gamma-1)\cdot\frac{1}{\nn}\big(\X^{\Ik}-\hmX^{\Ikc}(\W^{\Ik})\big)^T\big(\A^{\Ik}-\hmA^{\Ikc}(\W^{\Ik})\big)\\
		&&\quad\quad\quad\quad\quad\quad\quad\cdot\Big(\frac{1}{\nn}\big(\A^{\Ik}-\hmA^{\Ikc}(\W^{\Ik})\big)^T\big(\A^{\Ik}-\hmA^{\Ikc}(\W^{\Ik})\big)\Big)^{-1}\\
		&&\quad\quad\quad\quad\quad\quad\quad\cdot\frac{1}{\nn}\big(\A^{\Ik}-\hmA^{\Ikc}(\W^{\Ik})\big)^T\big(\Y^{\Ik}-\hmY^{\Ikc}(\W^{\Ik})\big) 
		\Big)
	\end{array}
\end{equation}
by~\eqref{eq:hbg}. 
By Slutsky's theorem and Equation~\eqref{eq:hbgDMLoneProof}, we have
\begin{displaymath}
	\begin{array}{cl}
		&\sqrt{\NN}(\hbgDMLone-\bg)\\
			=& \Big(\big( \matA + (\gamma-1)\matB \big)^{-1} +O_{\PPN}\big(\NN^{-\frac{1}{2}}(1+\rhoN)\big)\Big) \\
		&\quad\cdot \frac{1}{\sqrt{\KK}}\sum_{\kk=1}^{\KK}
			\frac{1}{\sqrt{\nn}}\bigg(
				\big(\X^{\Ik}-\hmX^{\Ikc}(\W^{\Ik})\big)^T\Big(\Y^{\Ik}-\hmY^{\Ikc}(\W^{\Ik}) - \big(\X^{\Ik}-\hmX^{\Ikc}(\W^{\Ik})\big)^T\bg\Big)\\
		&\quad\quad\quad + (\gamma-1)\cdot\frac{1}{\nn}\big(\X^{\Ik}-\hmX^{\Ikc}(\W^{\Ik})\big)^T \big(\A^{\Ik}-\hmA^{\Ikc}(\W^{\Ik})\big)\\
		&\quad\quad\quad\quad\quad\quad\quad\cdot\Big(\frac{1}{\nn}\big(\A^{\Ik}-\hmA^{\Ikc}(\W^{\Ik})\big)^T \big(\A^{\Ik}-\hmA^{\Ikc}(\W^{\Ik})\big)\Big)^{-1}\\
		&\quad\quad\quad\quad\quad\quad\quad\cdot\big(\A^{\Ik}-\hmA^{\Ikc}(\W^{\Ik})\big)^T\Big(\Y^{\Ik}-\hmY^{\Ikc}(\W^{\Ik}) - \big(\X^{\Ik}-\hmX^{\Ikc}(\W^{\Ik})\big)^T\bg\Big) 
		\bigg)\\
		=& \Big(\big( \matA + (\gamma-1)\matB \big)^{-1} +O_{\PPN}\big(\NN^{-\frac{1}{2}}(1+\rhoN)\big)\Big) \\
		&\quad\cdot \frac{1}{\sqrt{\KK}}\sum_{\kk=1}^{\KK}
			\sqrt{\nn}\Big(\frac{1}{\nn}\sum_{i\in\Ik}\losstilde(S_i;\bg,\hetaIkc)\\
			&\quad+ (\gamma-1)\cdot\frac{1}{\nn}
			\sum_{i\in\Ik}\lossone(S_i;\hetaIkc)\cdot\big(\frac{1}{\nn}\sum_{i\in\Ik}\losstwo(S_i;\hetaIkc)\big)^{-1}\cdot\frac{1}{\nn}
		\sum_{i\in\Ikc}\loss(S_i;\bg,\hetaIkc)
		\Big)\\
	\end{array}
\end{displaymath}
The last expression above coincides with~\ref{eq:combine1}.
Consequently, the same asymptotic analysis conducted for $\hbgDMLtwo$ can also be employed in this case. 
\end{proof}

\begin{lemma}\label{lem:lossSquareRootBound}
Let $\gamma\ge 0$ and let $\losstest\in\{\loss,\losstilde\}$. We have
\begin{displaymath}
	\frac{1}{\nn}\sum_{i\in\Ik}\losstest(S_i;\hbg,\hetaIkc) = \EP[\losstest(S;\bg,\etazero)] +  O_{\PP}\big(\NN^{-\frac{1}{2}}(1+\rhoN)\big). 
\end{displaymath}
\end{lemma}
\begin{proof}
We consider the case $\losstest=\loss$. 
We decompose
\begin{equation}\label{eq:decomposition}
	\begin{array}{cl}
		&\frac{1}{\nn}\sum_{i\in\Ik}\loss(S_i;\hbg,\hetaIkc) -\EP[\loss(S;\bg,\etazero)]\\
		=&\frac{1}{\nn}\sum_{i\in\Ik}\big(\loss(S_i;\hbg,\hetaIkc) -\loss(S_i;\bg,\hetaIkc) \big) + \frac{1}{\nn}\sum_{i\in\Ik}\big(\loss(S_i;\bg,\hetaIkc)-\loss(S_i;\bg,\etazero)\big)\\
		&\quad +  \frac{1}{\nn}\sum_{i\in\Ik}\big( \loss(S_i;\bg,\etazero)-\EP[\loss(S;\bg,\etazero)]\big).
	\end{array}
\end{equation}
Subsequently, we analyze the three terms in the above decomposition individually. 
We have
\begin{displaymath}
	\begin{array}{cl}
		&\normbig{\frac{1}{\nn}\sum_{i\in\Ik}\loss(S_i;\hbg,\hetaIkc) - \frac{1}{\nn}\sum_{i\in\Ik}\loss(S_i;\bg,\hetaIkc)}\\
		\le&  \normbig{\frac{1}{\nn}\sum_{i\in\Ik}(A_i-\hmA^{\Ikc}(W_i))(X_i-\hmX^{\Ikc}(W_i))^T}\norm{\hbg-\bg}\\
		=& \normbig{\frac{1}{\nn}\sum_{i\in\Ik}\lossone(S_i;\hetaIkc)}\norm{\hbg-\bg}\\
		=& \normbig{\EP[\lossone(S;\etazero)]+O_{\PP}\big(\NN^{-\frac{1}{2}}(1+\rhoN)\big)}\norm{\hbg-\bg}
	\end{array}
\end{displaymath}
by Lemma~\ref{lem:D1convInProb}. Because $\norm{\hbg-\bg}=O_{\PP}(\NN^{-\frac{1}{2}}\rhoN)$ holds by Theorem~\ref{thm:asymptNormalgamma}, we infer 
\begin{equation}\label{eq:decompositionFactor1}
	\normbigg{\frac{1}{\nn}\sum_{i\in\Ik}\loss(S_i;\hbg,\hetaIkc) - \frac{1}{\nn}\sum_{i\in\Ik}\loss(S_i;\bg,\hetaIkc)} = O_{\PP}\big(\NN^{-\frac{1}{2}}\rhoN\big). 
\end{equation}
Due to~\eqref{eq:combine1LastSummand1} that was established in the proof of Theorem~\ref{thm:asymptNormalgamma}, we have
\begin{equation}\label{eq:decompositionFactor2}
	\frac{1}{\nn}\sum_{i\in\Ik}\big(\loss(S_i;\bg,\hetaIkc) - \loss(S_i;\bg,\etazero) \big)=   O_{\PP}\big(\NN^{-\frac{1}{2}}\rhoN\big). 
\end{equation}
Due to the Lindeberg--Feller CLT and the Cramer--Wold device, we have
\begin{equation}\label{eq:decompositionFactor3}
	\frac{1}{\nn}\sum_{i\in\Ik}\big(\loss(S_i;\bg,\etazero) -\EP[\loss(S;\bg,\etazero)]\big) =O_{\PP}(\NN^{-\frac{1}{2}}).
\end{equation}
We combine~\eqref{eq:decomposition} and~\mbox{\eqref{eq:decompositionFactor1}--\eqref{eq:decompositionFactor3}} to infer the claim for $\losstest=\loss$. The case $\losstest=\losstilde$ can be analyzed analogously. 
\end{proof}

\begin{theorem}\label{thm:estSDgamma}
Suppose Assumption~\ref{assumpt:DMLboth} holds. 
	Recall the  score functions introduced in Definition~\ref{def:lossFunc}, and
	let $\hbg\in\{\hbgDMLone,\hbgDMLtwo\}$.
	Introduce the matrices
	\begin{displaymath}
		\begin{array}{rcl}
			\hmatAk &:=& \frac{1}{\nn}\sum_{i\in\Ik}\lossthree(S_i;\hetaIkc),\\ 			\hmatBk &:=& \frac{1}{\nn}\sum_{i\in\Ik}\lossone(S;\hetaIkc) \Big(\frac{1}{\nn}\sum_{i\in\Ik}\losstwo(S_i;\hetaIkc) \Big)^{-1}\frac{1}{\nn}\sum_{i\in\Ik} \lossone^T(S_i;\hetaIkc),\\ 
			\hmatCk &:=& \frac{1}{\nn}\sum_{i\in\Ik}\lossone(S_i;\hetaIkc)  \Big(\frac{1}{\nn}\sum_{i\in\Ik}\losstwo(S_i;\hetaIkc) \Big)^{-1},\\ 			
			\hmatEk&:=& \Big(\frac{1}{\nn}\sum_{i\in\Ik}\losstwo(S_i;\hetaIkc) \Big)^{-1} \frac{1}{\nn}\sum_{i\in\Ik} \loss(S_i;\hbg,\hetaIkc).\\ 		
			\end{array}
	\end{displaymath}
	 Let furthermore
	\begin{displaymath}
			\begin{array}{rcl}
		\hlossoverlinep(\cdot;\hbg,\hetaIkc)
		&:=&\losstilde(\cdot;\hbg,\hetaIkc)
				+ (\gamma-1)
				\hmatCk\loss(\cdot;\hbg,\hetaIkc) \\
				&&\quad\quad + (\gamma-1)\big(\lossone(\cdot;\hetaIkc)-\frac{1}{\nn}\sum_{i\in\Ik}\lossone(S_i;\hetaIkc)\big)\hmatEk\\
				&&\quad\quad - (\gamma-1)\hmatCk\big(\losstwo(\cdot;\hetaIkc)-\frac{1}{\nn}\sum_{i\in\Ik} \losstwo(S_i;\hetaIkc) \big)\hmatEk
	\end{array}
		\end{displaymath}
	and
	\begin{displaymath}
		\hmatDk := \frac{1}{\nn}\sum_{i\in\Ik} \hlossoverlinep(S_i;\hbg,\hetaIkc)\big(\hlossoverlinep(S_i;\hbg,\hetaIkc)\big)^T.
	\end{displaymath}
	Define the estimators 
		\begin{displaymath}
		\hmatA := \frac{1}{\KK}\sum_{\kk=1}^{\KK}\hmatAk,
		\quad
		\hmatB :=  \frac{1}{\KK}\sum_{\kk=1}^{\KK}\hmatBk,
		\quad\textrm{and}\quad
		\hmatD :=  \frac{1}{\KK}\sum_{\kk=1}^{\KK}\hmatDk.
	\end{displaymath}
	We estimate the asymptotic variance covariance matrix $\sigma^2(\gamma)$    in Theorem~\ref{thm:asymptNormalgamma} by
	\begin{displaymath}
		\hsigma^2(\gamma) := 
		\big(\hmatA + (\gamma-1)\hmatB\big)^{-1}
		\hmatD
		\big(\hmatA^T + (\gamma-1)\hmatB^T\big)^{-1}.
	\end{displaymath}
Then we have $\hsigma^2(\gamma) = \sigma^2(\gamma)+O_{\PP}\big(\rhoNtilde + \NN^{-\frac{1}{2}}(1+\rhoN)\big)$, where $\rhoNtilde= \NN^{\max\big\{\frac{4}{p}-1, -\frac{1}{2}\big\}}+\rNpnumber$ is as in Definition~\ref{def:asymptNormal}.  
\end{theorem}
\begin{proof}[Proof of Theorem~\ref{thm:estSDgamma}]
This proof is based on~\citet{Chernozhukov2018}.
We already verified 
\begin{displaymath}
	\hmatA = \matA + O_{\PPN}\big(\NN^{-\frac{1}{2}}(1+\rhoN)\big)
	\quad\textrm{and}\quad
	\hmatB = \matB + O_{\PPN}\big(\NN^{-\frac{1}{2}}(1+\rhoN)\big)
\end{displaymath}
in the proof of Theorem~\ref{thm:asymptNormalgamma}
because $\KK$ is a fixed number independent of $\NN$. 
Thus, we have
\begin{displaymath}
	\big(\hmatA + (\gamma-1)\hmatB\big)^{-1} = \big(\matA + (\gamma-1)\matB\big)^{-1} + O_{\PPN}\big(\NN^{-\frac{1}{2}}(1+\rhoN)\big)
\end{displaymath}
by Weyl's inequality. 
Moreover, we have $\hmatCk=\matC+O_{\PP}\big(\NN^{-\frac{1}{2}}(1+\rhoN)\big)$ by Lemma~\ref{lem:D1convInProb}. 

Subsequently, we argue that  $\hmatEk=\matE+O_{\PP}\big(\NN^{-\frac{1}{2}}(1+\rhoN)\big)$ holds. 
By Lemma~\ref{lem:D1convInProb} and Weyl's inequality, we have
\begin{displaymath}
	\frac{1}{\nn}\sum_{i\in\Ik}\lossone(S_i;\hetaIkc)
	=\EP[\lossone(S;\etazero)] + O_{\PP}\big(\NN^{-\frac{1}{2}}(1+\rhoN)\big)
\end{displaymath}
and 
\begin{equation}\label{eq:hmatEk1}
		 \Big( \frac{1}{\nn}\sum_{i\in\Ik}\losstwo(S_i;\hetaIkc)\Big)^{-1}
		= \EP[\losstwo(S;\etazero)]^{-1}+ O_{\PP}\big(\NN^{-\frac{1}{2}}(1+\rhoN)\big). 
\end{equation}
Due to~\eqref{eq:hmatEk1}, it suffices to show 
\begin{equation}\label{eq:hmatEk2}
	\frac{1}{\nn}\sum_{i\in\Ik}\loss(S_i;\hbg,\hetaIkc)
	=\EP[\loss(S;\bg,\etazero)] + O_{\PP}\big(\NN^{-\frac{1}{2}}(1+\rhoN)\big)
\end{equation}
to infer $\hmatEk=\matE+O_{\PP}\big(\NN^{-\frac{1}{2}}(1+\rhoN)\big)$. But~\eqref{eq:hmatEk2} holds due to Lemma~\ref{lem:lossSquareRootBound}. 
To conclude the theorem, it remains verify $\hmatDk=\matD+O_{\PP}(\rhoNtilde)$. 
We have
{\allowdisplaybreaks 
\begin{align*}
		&\norm{\hmatDk-\matD}\\
		\le& \normbigg{\frac{1}{\nn}\sum_{i\in\Ik} \losstilde(S_i;\hbg,\hetaIkc) \losstilde^T(S_i;\hbg,\hetaIkc)-\EP\big[\losstilde(S;\bg,\etazero) \losstilde^T(S;\bg,\etazero)\big] }\\
		& +(\gamma-1) \normbigg{\frac{1}{\nn}\sum_{i\in\Ik}\losstilde(S_i;\hbg,\hetaIkc)\loss^T(S_i;\hbg,\hetaIkc)\matC^T
			- \EP\big[\losstilde(S;\bg,\etazero)\loss^T(S;\bg,\etazero)\big]\matC^T}\\
		& + (\gamma-1)\normbigg{\frac{1}{\nn}\sum_{i\in\Ik} \matC\loss(S_i;\hbg,\hetaIkc)\losstilde^T(S_i;\hbg,\hetaIkc) 
			- \matC\EP\big[\loss(S;\bg,\etazero)\losstilde^T(S;\bg,\etazero)\big]}\\
		& + (\gamma-1)^2\normbigg{\frac{1}{\nn}\sum_{i\in\Ik}\matC\loss(S_i;\hbg,\hetaIkc)\loss^T(S_i;\hbg,\hetaIkc)\matC^T
			- \matC\EP\big[\loss(S;\bg,\etazero)\loss^T(S;\bg,\etazero)\big]\matC^T}\\
		& +  (\gamma-1)\bigg\lVert\frac{1}{\nn}\sum_{i\in\Ik}\losstilde(S_i;\hbg,\hetaIkc)\matE^T\big(\lossone(S_i;\hetaIkc) - \EP[\lossone(S;\etazero)]\big)^T\\
			&\quad\quad
			- \EP\Big[\losstilde(S;\bg,\etazero)\matE^T\big(\lossone(S;\etazero) - \EP[\lossone(S;\etazero)]\big)^T\Big]\bigg\rVert\\
		& + (\gamma-1) \bigg\lVert\frac{1}{\nn}\sum_{i\in\Ik} \big(\lossone(S_i;\hetaIkc)-\EP[\lossone(S;\etazero)]\big)\matE\losstilde^T(S_i;\hbg,\hetaIkc)\\
			&\quad\quad
			- \EP\Big[\big(\lossone(S;\etazero)-\EP[\lossone(S;\etazero)]\big)\matE\losstilde^T(S;\bg,\etazero)\Big]\bigg\rVert\\
		& +  (\gamma-1)^2\bigg\lVert \frac{1}{\nn}\sum_{i\in\Ik}\big(\lossone(S_i;\hetaIkc)-\EP[\lossone(S;\etazero)]\big)\matE\matE^T\big(\lossone(S_i;\hetaIkc)-\EP[\lossone(S;\etazero)]\big)^T\\
			&	\quad\quad- \EP\Big[\big(\lossone(S;\etazero)-\EP[\lossone(S;\etazero)]\big)\matE\matE^T\big(\lossone(S;\etazero)-\EP[\lossone(S;\etazero)]\big)^T\Big]
			\bigg\rVert\\
		&+(\gamma-1) \bigg\lVert \frac{1}{\nn}\sum_{i\in\Ik}\matC\big(\losstwo(S_i;\hetaIkc)-\EP[\losstwo(S;\etazero)]\big)\matE\losstilde^T(S_i;\hbg,\hetaIkc) \\
			& 	\quad\quad - \matC\EP\Big[\big(\losstwo(S;\etazero)-\EP[\losstwo(S;\etazero)]\big)\matE\losstilde^T(S;\bg,\etazero) \Big] \bigg\rVert\\
		& +(\gamma-1) \bigg\lVert \frac{1}{\nn}\sum_{i\in\Ik} \losstilde(S_i;\hbg,\hetaIkc)\matE^T \big(\losstwo(S_i;\hetaIkc)-\EP[\losstwo(S;\etazero)]\big)^T\matC^T\\
			&\quad\quad - \EP\Big[\losstilde(S;\bg,\etazero)\matE^T \big(\losstwo(S;\etazero)-\EP[\losstwo(S;\etazero)]\big)^T\Big]\matC^T\bigg\lVert\\
		& + (\gamma-1)^2 \bigg\lVert  \frac{1}{\nn}\sum_{i\in\Ik} \matC\loss(S_i;\hbg,\hetaIkc)\matE^T\big(\lossone(S_i;\hetaIkc)-\EP[\lossone(S;\etazero)]\big)^T\\
			&\quad\quad - \matC\EP\Big[\loss(S;\bg,\etazero)\matE^T\big(\lossone(S;\etazero)-\EP[\lossone(S;\etazero)]\big)^T\Big] \bigg\rVert\\
		& + (\gamma-1)^2\bigg\lVert  \frac{1}{\nn}\sum_{i\in\Ik}\big(\lossone(S_i;\hetaIkc)-\EP[\lossone(S;\etazero)]\big)\matE \loss^T(S_i;\hbg,\hetaIkc)\matC^T\\
			&\quad\quad - \EP\big[\big(\lossone(S;\etazero)-\EP[\lossone(S;\etazero)]\big)\matE \loss^T(S;\bg,\etazero)\big]\matC^T\bigg\lVert\\
		& +  (\gamma-1)^2\bigg\lVert  \frac{1}{\nn}\sum_{i\in\Ik} \big(\lossone(S_i;\hetaIkc)-\EP[\lossone(S;\etazero)]\big)\matE\matE^T\big(\losstwo(S_i;\hetaIkc)-\EP[\losstwo(S;\etazero)]\big)^T \matC^T \\
			&\quad\quad -\EP\Big[\big(\lossone(S;\etazero)-\EP[\lossone(S;\etazero)]\big)\matE\matE^T\big(\losstwo(S;\etazero)-\EP[\losstwo(S;\etazero)]\big)^T \Big]\matC^T\bigg\rVert\\
		& +  (\gamma-1)^2\bigg\lVert  \frac{1}{\nn}\sum_{i\in\Ik} \matC\loss(S_i;\hbg,\hetaIkc)\matE^T \big(\losstwo(S_i;\hetaIkc)-\EP[\losstwo(S;\etazero)]\big)^T\matC^T \\
			&\quad\quad - \matC\EP\Big[\loss(S_i;\bg,\etazero)\matE^T\big(\losstwo(S_i;\etazero)-\EP[\losstwo(S;\etazero)]\big)^T\Big]\matC^T \bigg\rVert\\
		& +(\gamma-1)^2 \bigg\lVert  \frac{1}{\nn}\sum_{i\in\Ik} \matC\big(\losstwo(S_i;\hetaIkc)-\EP[\losstwo(S;\etazero)]\big)\matE\loss^T(S_i;\hbg,\hetaIkc)\matC^T \\
			&\quad\quad - \matC\EP\big[\big(\losstwo(S;\etazero)-\EP[\losstwo(S;\etazero)]\big)\matE\loss^T(S;\bg,\etazero)\big]\matC^T \bigg\rVert\\
		& + (\gamma-1)^2 \bigg\lVert  \frac{1}{\nn}\sum_{i\in\Ik} \matC\big(\losstwo(S_i;\hetaIkc)-\EP[\losstwo(S;\etazero)]\big)\matE\matE^T\big(\lossone(S_i;\hetaIkc)-\EP[\lossone(S;\etazero)]\big)^T \\
			&\quad\quad - \matC\EP\Big[\big(\losstwo(S;\etazero)-\EP[\losstwo(S;\etazero)]\big)\matE\matE^T\big(\lossone(S;\etazero)-\EP[\lossone(S;\etazero)]\big)^T\Big] \bigg\rVert\\
		& +  (\gamma-1)^2\bigg\lVert  \frac{1}{\nn}\sum_{i\in\Ik}  \matC\big(\losstwo(S_i;\hetaIkc)-\EP[\losstwo(S;\etazero)]\big)\matE\matE^T\big(\losstwo(S_i;\hetaIkc)-\EP[\losstwo(S;\etazero)]\big)^T\matC^T\\
			&\quad\quad -  \matC\EP\Big[\big(\losstwo(S;\etazero)-\EP[\losstwo(S;\etazero)]\big)\matE\matE^T\big(\losstwo(S;\etazero)-\EP[\losstwo(S;\etazero)]\big)^T\Big]\matC^T \bigg\rVert \\
			& + O_{\PP}\big(\NN^{-\frac{1}{2}}(1+\rhoN)\big)\\
		=:& \sum_{i=1}^{16}\mathcal{I}_i+ O_{\PP}\big(\NN^{-\frac{1}{2}}(1+\rhoN)\big)
\end{align*}}

by the triangle inequality and the results derived so far. Subsequently, we bound the terms $\Ione,\ldots,\Isixteen$ individually. Because all these terms consist of norms of matrices of fixed size, it suffices to bound the individual matrix entries. 
Let $j, l, t, r$ be natural numbers not exceeding the dimensions of the respective object they index. 
By  Lemma~\ref{lem:helperVarianceConsistent}, we have
\begin{displaymath}
	\normonebigg{\frac{1}{\nn}\sum_{i\in\Ik} \losstilde_j(S_i;\hbg,\hetaIkc)\losstilde_l(S_i;\hbg,\hetaIkc) - \EP\big[\losstilde_j(S;\bg,\etazero)\losstilde_l(S;\bg,\etazero)\big]}
		=O_{\PP}(\rhoNtilde), 
\end{displaymath}
which implies $\Ione=O_{\PP}(\rhoNtilde)$. 
By  Lemma~\ref{lem:helperVarianceConsistent}, we have
\begin{displaymath}
	\normonebigg{\frac{1}{\nn}\sum_{i\in\Ik}\losstilde_j(S_i;\hbg,\hetaIkc)\loss_l(S_i;\hbg,\hetaIkc) - \EP\big[\losstilde_j(S;\bg\etazero)\loss_l(S;\betazero,\etazero)\big]}
		=O_{\PP}(\rhoNtilde),
\end{displaymath}
which implies $\Itwo=O_{\PP}(\rhoNtilde)=\Ithree$
due to
\begin{displaymath}
	\begin{array}{cl}
		&\normbig{\losstilde(S_i;\hbg,\hetaIkc)\loss^T(S_i;\hbg,\hetaIkc)\matC^T
			- \EP\big[\losstilde(S;\bg,\etazero)\loss^T(S;\bg,\etazero)\big]\matC^T}\\
		\le& \normbig{\frac{1}{\nn}\sum_{i\in\Ik}\losstilde(S_i;\hbg,\hetaIkc)\loss^T(S_i;\hbg,\hetaIkc)
			- \EP\big[\losstilde(S;\bg,\etazero)\loss^T(S;\bg,\etazero)\big]}\norm{\matC}
	\end{array}
\end{displaymath}
and 
\begin{displaymath}
	\begin{array}{cl}
		&\normbig{\frac{1}{\nn}\sum_{i\in\Ik} \matC\loss(S_i;\hbg,\hetaIkc)\losstilde^T(S_i;\hbg,\hetaIkc) 
			- \matC\EP\big[\loss(S;\bg,\etazero)\losstilde^T(S;\bg,\etazero)\big]}\\
		\le&\norm{\matC}\normBig{\frac{1}{\nn}\sum_{i\in\Ik} \loss(S_i;\hbg,\hetaIkc)\losstilde^T(S_i;\hbg,\hetaIkc) 
			- \EP\big[\loss(S;\bg,\etazero)\losstilde^T(S;\bg,\etazero)\big]}.
	\end{array}
\end{displaymath}
By  Lemma~\ref{lem:helperVarianceConsistent}, we have
\begin{displaymath}
	\normonebigg{\frac{1}{\nn}\sum_{i\in\Ik}\loss_j(S_i;\hbg,\hetaIkc)\loss_l(S_i;\hbg,\hetaIkc) - \EP[\loss_j(S;\betazero,\etazero)\loss_l(S;\betazero,\etazero)]}=O_{\PP}(\rhoNtilde),
\end{displaymath}
which implies $\Ifour=O_{\PP}(\rhoNtilde)$ due to
\begin{displaymath}
	\begin{array}{cl}
		&\normbig{\frac{1}{\nn}\sum_{i\in\Ik}\matC\loss(S_i;\hbg,\hetaIkc)\loss^T(S_i;\hbg,\hetaIkc)\matC^T
			- \matC\EP\big[\loss(S;\bg,\etazero)\loss^T(S;\bg,\etazero)\big]\matC^T}\\
		\le&\norm{\matC}^2\normbig{\frac{1}{\nn}\sum_{i\in\Ik}\loss(S_i;\hbg,\hetaIkc)\loss^T(S_i;\hbg,\hetaIkc)
			- \EP\big[\loss(S;\bg,\etazero)\loss^T(S;\bg,\etazero)\big]}.
	\end{array}
\end{displaymath}
By  Lemma~\ref{lem:helperVarianceConsistent}, we have
\begin{displaymath}
	\normonebigg{\frac{1}{\nn}\sum_{i\in\Ik}\losstilde_j(S_i;\hbg,\hetaIkc)\big(\lossone(S_i;\hetaIkc)\big)_{l,t} - \EP\Big[\losstilde_j(S;\bg,\etazero)\big(\lossone(S;\etazero)\big)_{l,t}\Big]}=O_{\PP}(\rhoNtilde),
\end{displaymath}
which implies $\Ifive=O_{\PP}(\rhoNtilde)$ because we have
\begin{displaymath}
	\begin{array}{cl}
		&\Big\lVert\frac{1}{\nn}\sum_{i\in\Ik}\losstilde(S_i;\hbg,\hetaIkc)\matE^T\big(\lossone(S_i;\hetaIkc) - \EP[\lossone(S;\etazero)]\big)^T\\
			&\quad\quad
			- \EP\Big[\losstilde(S;\bg,\etazero)\matE^T\big(\lossone(S;\etazero) - \EP[\lossone(S;\etazero)]\big)^T\Big]\Big\rVert\\
		\le& \normbig{\frac{1}{\nn}\sum_{i\in\Ik}\losstilde(S_i;\hbg,\hetaIkc)\matE^T\lossone^T(S_i;\hetaIkc)
				- \EP\big[\losstilde(S;\bg,\etazero)\matE^T\lossone^T(S;\etazero)\big]}\\
				&\quad + \normbig{\frac{1}{\nn}\sum_{i\in\Ik}\losstilde(S_i;\hbg,\hetaIkc)
				- \EP\big[\losstilde(S;\bg,\etazero)\big]}\norm{\matE}\norm{\EP[\lossone(S;\etazero)]}, 
	\end{array}
\end{displaymath}
where the last summand is $O_{P}\big(\NN^{-\frac{1}{2}}(1+\rhoN)\big)$ by Lemma~\ref{lem:lossSquareRootBound}, 
and we have
\begin{displaymath}
	\begin{array}{cl}
		&\normonebig{\frac{1}{\nn}\sum_{i\in\Ik}\big(\losstilde(S_i;\hbg,\hetaIkc)\matE^T\lossone^T(S_i;\hetaIkc)\big)_{j,l}
				- \big(\EP[\losstilde(S;\bg,\etazero)\matE^T\lossone^T(S;\etazero)]\big)_{j,l}}\\
		=&\normonebig{\frac{1}{\nn}\sum_{i\in\Ik}\matE^T \big(\lossone(S_i;\hetaIkc)\big)_{\cdot,l} \losstilde_j(S_i;\hbg,\hetaIkc) - \matE^T\EP\big[\big(\lossone(S;\etazero)\big)_{\cdot,l}\losstilde_j(S;\bg,\etazero)\big]}\\
		\le& \norm{\matE}\normbig{\frac{1}{\nn}\sum_{i\in\Ik} \big(\lossone(S_i;\hetaIkc)\big)_{\cdot,l} \losstilde_j(S_i;\hbg,\hetaIkc) - \EP\big[\big(\lossone(S;\etazero)\big)_{\cdot,l}\losstilde_j(S;\bg,\etazero)\big]}.
	\end{array}
\end{displaymath}
The term $\Isix$ can be bounded analogously to $\Ifive$. 
By  Lemma~\ref{lem:helperVarianceConsistent}, we have
\begin{displaymath}
	\normonebigg{\frac{1}{\nn}\sum_{i\in\Ik} \big(\lossone(S_i;\hetaIkc)\big)_{j,l} \big(\lossone(S_i;\hetaIkc)\big)_{t,r} - \EP\Big[\big(\lossone(S;\etazero)\big)_{j,l} \big(\lossone(S;\etazero)\big)_{l,t}\Big]}=O_{\PP}(\rhoNtilde),
\end{displaymath}
which implies $\Iseven=O_{\PP}(\rhoNtilde)$. Indeed, we have
\begin{displaymath}
	\begin{array}{cl}
		&\Big\lVert \frac{1}{\nn}\sum_{i\in\Ik}\big(\lossone(S_i;\hetaIkc)-\EP[\lossone(S;\etazero)]\big)\matE\matE^T\big(\lossone(S_i;\hetaIkc)-\EP[\lossone(S;\etazero)]\big)^T\\
			&	\quad\quad- \EP\Big[\big(\lossone(S;\etazero)-\EP[\lossone(S;\etazero)]\big)\matE\matE^T\big(\lossone(S;\etazero)-\EP[\lossone(S;\etazero)]\big)^T\Big]
			\Big\rVert\\
		\le& \normbig{ \frac{1}{\nn}\sum_{i\in\Ik}\lossone(S_i;\hetaIkc)\matE\matE^T\lossone^T(S_i;\hetaIkc) - \EP\big[\lossone(S;\etazero)\matE\matE^T\lossone^T(S;\etazero)\big]}\\
		&\quad + 2\normbig{ \frac{1}{\nn}\sum_{i\in\Ik}\lossone(S_i;\hetaIkc) - \EP[\lossone(S;\etazero)]}\norm{\matE}^2\norm{\EP[\lossone(S;\etazero)]}\\
		=& \normbig{ \frac{1}{\nn}\sum_{i\in\Ik}\lossone(S_i;\hetaIkc)\matE\matE^T\lossone^T(S_i;\hetaIkc) - \EP\big[\lossone(S;\etazero)\matE\matE^T\lossone^T(S;\etazero)\big]} \\
		&\quad+ O_{\PP}\big(\NN^{-\frac{1}{2}}(1+\rhoN)\big)
	\end{array}
\end{displaymath}
by Lemma~\ref{lem:D1convInProb}, and we have
\begin{displaymath}
	\begin{array}{cl}
		& \normoneBig{ \frac{1}{\nn}\sum_{i\in\Ik}\big(\lossone(S_i;\hetaIkc)\matE\matE^T\lossone^T(S_i;\hetaIkc)\big)_{j,r} - \big(\EP\big[\lossone(S;\etazero)\matE\matE^T\lossone^T(S;\etazero)\big]\big)_{j,r}}\\
		=&  \normoneBig{ \frac{1}{\nn}\sum_{i\in\Ik}\big(\lossone(S_i;\hetaIkc)\big)_{j,\cdot}\matE\matE^T(\lossone^T(S_i;\hetaIkc))_{\cdot,r} - \EP\Big[\big(\lossone(S;\etazero)\big)_{j,\cdot}\matE\matE^T\big(\lossone^T(S;\etazero)\big)_{\cdot,r}\Big]}\\
		=&  \normonebig{ \frac{1}{\nn}\sum_{i\in\Ik}\matE^T\big(\lossone^T(S_i;\hetaIkc)\big)_{\cdot,r}\big(\lossone(S_i;\hetaIkc)\big)_{j,\cdot}\matE - \EP\big[\matE^T\big(\lossone^T(S;\etazero)\big)_{\cdot,r}\big(\lossone(S;\etazero)\big)_{j,\cdot}\matE\big]}\\
		\le& \normBig{ \frac{1}{\nn}\sum_{i\in\Ik}\big(\lossone^T(S_i;\hetaIkc)\big)_{\cdot,r}\big(\lossone(S_i;\hetaIkc)\big)_{j,\cdot} - \EP\Big[\big(\lossone^T(S;\etazero)\big)_{\cdot,r}\big(\lossone(S;\etazero)\big)_{j,\cdot}\Big]}\norm{\matE}^2. 
	\end{array}
\end{displaymath}

Next, we bound $\Ieight$. 
By  Lemma~\ref{lem:helperVarianceConsistent}, we have
\begin{displaymath}
	\normonebigg{\frac{1}{\nn}\sum_{i\in\Ik}\losstilde_j(S_i;\hbg,\hetaIkc)\big(\losstwo(S_i;\hetaIkc)\big)_{l,t}
	- \EP\Big[\losstilde_j(S_i;\bg,\etazero)\big(\losstwo(S;\etazero)\big)_{l,t}\Big]}=O_{\PP}(\rhoNtilde),
\end{displaymath}
which implies $\Ieight=O_{\PPN}(\rhoNtilde)$. Indeed, we have
\begin{displaymath}
	\begin{array}{cl}
		&\big\lVert \frac{1}{\nn}\sum_{i\in\Ik}\matC\big(\losstwo(S_i;\hetaIkc)-\EP[\losstwo(S;\etazero)]\big)\matE\losstilde^T(S_i;\hbg,\hetaIkc) \\
			& 	\quad\quad - \matC\EP\big[\big(\losstwo(S;\etazero)-\EP[\losstwo(S;\etazero)]\big)\matE\losstilde^T(S;\bg,\etazero) \big] \big\rVert\\
		\le& \normBig{ \frac{1}{\nn}\sum_{i\in\Ik} \matC\losstwo(S_i;\hetaIkc)\matE\losstilde^T(S_i;\hbg,\hetaIkc) - \matC\EP\big[\losstwo(S;\etazero)\matE\losstilde^T(S;\bg,\etazero)\big]}\\
		&\quad + \normbig{ \frac{1}{\nn}\sum_{i\in\Ik}\matC\EP[\losstwo(S;\etazero)]\matE\losstilde^T(S_i;\hbg,\hetaIkc) - \matC\EP[\losstwo(S;\etazero)]\matE\EP\big[\losstilde^T(S;\bg,\etazero)\big]}\\
		\le& \norm{\matC}\normbig{ \frac{1}{\nn}\sum_{i\in\Ik} \losstwo(S_i;\hetaIkc)\matE\losstilde^T(S_i;\hbg,\hetaIkc) - \EP\big[\losstwo(S;\etazero)\matE\losstilde^T(S;\bg,\etazero)\big]}\\
		&\quad +  \norm{\matC}\norm{\EPN[\losstwo(S;\etazero)]}\norm{\matE}\normbig{ \frac{1}{\nn}\sum_{i\in\Ik}\losstilde^T(S_i;\hbg,\hetaIkc) - \EPN\big[\losstilde^T(S;\bg,\etazero)\big]}\\
		\le& \norm{\matC}\normbig{ \frac{1}{\nn}\sum_{i\in\Ik} \losstwo(S_i;\hetaIkc)\matE\losstilde^T(S_i;\hbg,\hetaIkc) - \EP\big[\losstwo(S;\etazero)\matE\losstilde^T(S;\bg,\etazero)\big]} \\
		&\quad + O_{\PP}\big(\NN^{-\frac{1}{2}}(1+\rhoN)\big)
	\end{array}
\end{displaymath}
by Lemma~\ref{lem:lossSquareRootBound}, and we have
\begin{displaymath}
	\begin{array}{cl}
		&\normonebig{ \frac{1}{\nn}\sum_{i\in\Ik} \big(\losstwo(S_i;\hetaIkc)\matE\losstilde^T(S_i;\hbg,\hetaIkc)\big)_{j,t} - \big(\EP\big[\losstwo(S;\etazero)\matE\losstilde^T(S;\bg,\etazero)\big]\big)_{j,t}}\\
		=& \normonebig{ \frac{1}{\nn}\sum_{i\in\Ik} \big(\losstwo(S_i;\hetaIkc)\big)_{j,\cdot}\matE\losstilde_t(S_i;\hbg,\hetaIkc) - \EP\big[\big(\losstwo(S;\etazero)\big)_{j,\cdot}\matE\losstilde_t(S;\bg,\etazero)\big]}\\
		=&  \normonebig{ \frac{1}{\nn}\sum_{i\in\Ik} \losstilde_t(S_i;\hbg,\hetaIkc)\big(\losstwo(S_i;\hetaIkc)\big)_{j,\cdot}\matE - \EP\big[\losstilde_t(S;\bg,\etazero)\big(\losstwo(S;\etazero)\big)_{j,\cdot}\matE\big]}\\
		\le& \normBig{ \frac{1}{\nn}\sum_{i\in\Ik} \losstilde_t(S_i;\hbg,\hetaIkc)\big(\losstwo(S_i;\hetaIkc)\big)_{j,\cdot} - \EP\Big[\losstilde_t(S;\bg,\etazero)\big(\losstwo(S;\etazero)\big)_{j,\cdot}\Big]}\norm{\matE}. 
	\end{array}
\end{displaymath}
The term $\Inine$ can be bounded analogously to $\Ieight$.
Next, we bound $\Iten$. 
By  Lemma~\ref{lem:helperVarianceConsistent}, we have
\begin{displaymath}
	\normonebigg{\frac{1}{\nn}\sum_{i\in\Ik}\loss_j(S_i;\hbg,\hetaIkc)\big(\lossone(S_i;\hetaIkc) \big)_{l,t} - \EP\Big[\loss_j(S;\bg,\etazero)\big(\lossone(S;\etazero)\big)_{l,t}\Big]}=O_{\PP}(\rhoNtilde),
\end{displaymath}
which implies $\Iten=O_{\PPN}(\rhoNtilde)$. Indeed, we have
\begin{displaymath}
	\begin{array}{cl}
		&\Big\lVert  \frac{1}{\nn}\sum_{i\in\Ik} \matC\loss(S_i;\hbg,\hetaIkc)\matE^T\big(\lossone(S_i;\hetaIkc)-\EP[\lossone(S;\etazero)]\big)^T\\
			&\quad\quad - \matC\EP\Big[\loss(S;\bg,\etazero)\matE^T\big(\lossone(S;\etazero)-\EP[\lossone(S;\etazero)]\big)^T\Big] \Big\rVert\\
		\le& \big\lVert  \frac{1}{\nn}\sum_{i\in\Ik} \matC\loss(S_i;\hbg,\hetaIkc)\matE^T\lossone^T(S_i;\hetaIkc) - \matC\EP\big[\loss(S;\bg,\etazero)\matE^T\lossone^T(S;\etazero)\big] \big\rVert\\
		&\quad
		 + \big\lVert  \frac{1}{\nn}\sum_{i\in\Ik} \matC\loss(S_i;\hbg,\hetaIkc)\matE^T\EPN[\lossone^T(S;\etazero)] - \matC\EP\big[\loss(S;\bg,\etazero)\matE^T\EP[\lossone^T(S;\etazero)\big] \big\rVert\\
		 \le& \norm{\matC}\big\lVert  \frac{1}{\nn}\sum_{i\in\Ik} \loss(S_i;\hbg,\hetaIkc)\matE^T\lossone^T(S_i;\hetaIkc) - \EP\big[\loss(S;\bg,\etazero)\matE^T\lossone^T(S;\etazero)\big] \big\rVert\\
		&\quad
		 + \norm{\matC}\big\lVert  \frac{1}{\nn}\sum_{i\in\Ik} \loss(S_i;\hbg,\hetaIkc) - \EP[\loss(S;\bg,\etazero) \big\rVert\norm{\matE}\norm{\EPN[\lossone(S;\etazero)]}\\
		 \le& \norm{\matC}\big\lVert  \frac{1}{\nn}\sum_{i\in\Ik} \loss(S_i;\hbg,\hetaIkc)\matE^T\lossone^T(S_i;\hetaIkc) - \EP\big[\loss(S;\bg,\etazero)\matE^T\lossone^T(S;\etazero)\big] \big\rVert \\
		 &\quad+ O_{\PP}\big(\NN^{-\frac{1}{2}}(1+\rhoN)\big)
	\end{array}
\end{displaymath}
by Lemma~\ref{lem:lossSquareRootBound}, and we have
\begin{displaymath}
	\begin{array}{cl}
		& \normonebig{  \frac{1}{\nn}\sum_{i\in\Ik} \big(\loss(S_i;\hbg,\hetaIkc)\matE^T\lossone^T(S_i;\hetaIkc)\big)_{j,t} - \big(\EP\big[\loss(S;\bg,\etazero)\matE^T\lossone^T(S;\etazero)\big]\big)_{j,t} }\\
		=& \normoneBig{ \frac{1}{\nn}\sum_{i\in\Ik} \loss_j(S_i;\hbg,\hetaIkc)\matE^T\big(\lossone^T(S_i;\hetaIkc)\big)_{\cdot,t} - \EP\Big[\loss_j(S;\bg,\etazero)\matE^T\big(\lossone^T(S;\etazero)\big)_{\cdot,t}\Big]}\\
		=& \normonebig{  \frac{1}{\nn}\sum_{i\in\Ik}\matE^T\big(\lossone^T(S_i;\hetaIkc)\big)_{\cdot,t} \loss_j(S_i;\hbg,\hetaIkc) - \EP\big[\matE^T\big(\lossone^T(S;\etazero)\big)_{\cdot,t}\loss_j(S;\bg,\etazero)\big] }\\
		\le& \norm{\matE}\lVert  \frac{1}{\nn}\sum_{i\in\Ik}\big(\lossone^T(S_i;\hetaIkc)\big)_{\cdot,t} \loss_j(S_i;\hbg,\hetaIkc) - \EP\big[\big(\lossone^T(S;\etazero)\big)_{\cdot,t}\loss_j(S;\bg,\etazero)\big] \big\rVert.
	\end{array}
\end{displaymath}

The term $\Ieleven$ can be bounded analogously to $\Iten$. 
Next, we bound $\Itwelve$. 
By  Lemma~\ref{lem:helperVarianceConsistent}, we have
\begin{displaymath}
	\normonebigg{\frac{1}{\nn}\sum_{i\in\Ik}\big(\lossone(S_i;\hetaIkc))_{j,l} (\losstwo(S_i;\hetaIkc)\big)_{t,r} - \EP\Big[\big(\lossone(S;\etazero)\big)_{j,l} \big(\losstwo(S;\etazero)\big)_{t,r}\Big]}=O_{\PP}(\rhoNtilde),
\end{displaymath}
which implies $\Itwelve=O_{\PPN}(\rhoNtilde)$. Indeed, we have
\begin{displaymath}
	\begin{array}{cl}
		&\big\lVert  \frac{1}{\nn}\sum_{i\in\Ik} \big(\lossone(S_i;\hetaIkc)-\EP[\lossone(S;\etazero)]\big)\matE\matE^T\big(\losstwo(S_i;\hetaIkc)-\EP[\losstwo(S;\etazero)]\big)\matC^T \\
			&\quad\quad -\EP\big[\big(\lossone(S;\etazero)-\EP[\lossone(S;\etazero)]\big)\matE\matE^T\big(\losstwo(S;\etazero)-\EP[\losstwo(S;\etazero)]\big)\big]\matC^T\big\rVert\\
		\le& \big\lVert  \frac{1}{\nn}\sum_{i\in\Ik} \lossone(S_i;\hetaIkc)\matE\matE^T\losstwo^T(S_i;\hetaIkc)\matC^T  -\EP\big[\lossone(S;\etazero)\matE\matE^T\losstwo^T(S;\etazero)\big]\matC^T\big\rVert\\
			&\quad + \big\lVert  \frac{1}{\nn}\sum_{i\in\Ik} \lossone(S_i;\hetaIkc)\matE\matE^T\EP[\losstwo^T(S;\etazero)]\matC^T -\EP\big[\lossone(S;\etazero)\matE\matE^T\EP[\losstwo^T(S;\etazero)]\big]\matC^T\big\rVert\\
			&\quad + \big\lVert  \frac{1}{\nn}\sum_{i\in\Ik} \EP[\lossone(S;\etazero)]\matE\matE^T\losstwo^T(S_i;\hetaIkc)\matC^T  -\EP\big[\EP[\lossone(S;\etazero)]\matE\matE^T\losstwo^T(S;\etazero)\big]\matC^T\big\rVert\\
	\le& \big\lVert  \frac{1}{\nn}\sum_{i\in\Ik} \lossone(S_i;\hetaIkc)\matE\matE^T\losstwo^T(S_i;\hetaIkc)  -\EP\big[\lossone(S;\etazero)\matE\matE^T\losstwo^T(S;\etazero)\big]\big\rVert\norm{\matC}\\
			&\quad + \big\lVert  \frac{1}{\nn}\sum_{i\in\Ik} \lossone(S_i;\hetaIkc) -\EP[\lossone(S;\etazero)]\big\rVert \norm{\matE}^2\norm{\EP[\losstwo(S;\etazero)]}\norm{\matC} \\
			&\quad + \norm{\EP[\lossone(S;\etazero)]}\norm{\matE}^2\norm{\matC}\big\lVert  \frac{1}{\nn}\sum_{i\in\Ik} \losstwo(S_i;\hetaIkc)  -\EP[\losstwo(S;\etazero)]\big\rVert\\
		\le& \big\lVert  \frac{1}{\nn}\sum_{i\in\Ik} \lossone(S_i;\hetaIkc)\matE\matE^T\losstwo^T(S_i;\hetaIkc)  -\EP[\lossone(S;\etazero)\matE\matE^T\losstwo^T(S;\etazero)]\big\rVert\norm{\matC} \\
		&\quad + O_{\PP}\big(\NN^{-\frac{1}{2}}(1+\rhoN)\big)
	\end{array}
\end{displaymath}
by Lemma~\ref{lem:D1convInProb}, and we have
\begin{displaymath}
	\begin{array}{cl}
		&\normonebig{  \frac{1}{\nn}\sum_{i\in\Ik} \big(\lossone(S_i;\hetaIkc)\matE\matE^T\losstwo^T(S_i;\hetaIkc)\big)_{j,r}  -\big(\EP\big[\lossone(S;\etazero)\matE\matE^T\losstwo^T(S;\etazero)\big]\big)_{j,r}}\\
		=& \normonebig{ \frac{1}{\nn}\sum_{i\in\Ik} \big(\lossone(S_i;\hetaIkc)\big)_{j,\cdot}\matE\matE^T\big(\losstwo^T(S_i;\hetaIkc)\big)_{\cdot,r}  -\EP\Big[\big(\lossone(S;\etazero)\big)_{j,\cdot}\matE\matE^T\big(\losstwo^T(S;\etazero)\big)_{\cdot,r}\Big]}\\
		=& \normonebig{  \frac{1}{\nn}\sum_{i\in\Ik} \matE^T\big(\losstwo^T(S_i;\hetaIkc)\big)_{\cdot,r}\big(\lossone(S_i;\hetaIkc)\big)_{j,\cdot}\matE  -\EP\big[\matE^T\big(\losstwo^T(S;\etazero)\big)_{\cdot,r}\big(\lossone(S;\etazero)\big)_{j,\cdot}\matE\big]}\\
		\le& \norm{\matE}^2\Big\lVert  \frac{1}{\nn}\sum_{i\in\Ik} \big(\losstwo^T(S_i;\hetaIkc)\big)_{\cdot,r}\big(\lossone(S_i;\hetaIkc)\big)_{j,\cdot}  -\EP\Big[\big(\losstwo^T(S;\etazero)\big)_{\cdot,r}\big(\lossone(S;\etazero)\big)_{j,\cdot}\Big]\Big\rVert. 
	\end{array}
\end{displaymath}

Next, we bound $\Ithirteen$. 
By  Lemma~\ref{lem:helperVarianceConsistent}, we have
\begin{displaymath}
	\normonebigg{\frac{1}{\nn}\sum_{i\in\Ik}\loss_j(S_i;\hbg, \hetaIkc) \big(\losstwo(S_i;\hetaIkc)\big)_{t,r} - \EP\Big[\loss_j(S;\bg, \etazero)\big(\losstwo(S;\etazero)\big)_{t,r}\Big]}=O_{\PP}(\rhoNtilde),
\end{displaymath}
which implies $\Ithirteen=O_{\PP}(\rhoNtilde)$. Indeed, we have
\begin{displaymath}
	\begin{array}{cl}
		&\big\lVert  \frac{1}{\nn}\sum_{i\in\Ik} \matC\loss(S_i;\hbg,\hetaIkc)\matE^T\big(\losstwo(S_i;\hetaIkc)-\EP[\losstwo(S;\etazero)]\big)^T\matC^T \\
			&\quad\quad - \matC\EP\Big[\loss(S;\bg,\etazero)\matE^T\big(\losstwo(S;\etazero)-\EP[\losstwo(S;\etazero)]\big)^T\Big]\matC^T \big\rVert\\
		\le& \norm{\matC}^2\big\lVert  \frac{1}{\nn}\sum_{i\in\Ik} \loss(S_i;\hbg,\hetaIkc)\matE^T\losstwo^T(S_i;\hetaIkc) - \EP\big[\loss(S;\bg,\etazero)\matE^T\losstwo^T(S;\etazero)\big] \big\rVert\\
		&\quad + \norm{\matC}^2\norm{\matE}\norm{\EP[\losstwo(S;\etazero)]}\big\lVert  \frac{1}{\nn}\sum_{i\in\Ik} \loss(S_i;\hbg,\hetaIkc)- \EP[\loss(S;\bg,\etazero)]\big\rVert\\
		=& \norm{\matC}^2\big\lVert  \frac{1}{\nn}\sum_{i\in\Ik} \loss(S_i;\hbg,\hetaIkc)\matE^T\losstwo^T(S_i;\hetaIkc) - \EP\big[\loss(S;\bg,\etazero)\matE^T\losstwo^T(S;\etazero)\big] \big\rVert\\
		&\quad
		+ O_{\PP}\big(\NN^{-\frac{1}{2}}(1+\rhoN)\big)
	\end{array}
\end{displaymath}
by Lemma~\ref{lem:lossSquareRootBound}, and we have
\begin{displaymath}
	\begin{array}{cl}
		&\normoneBig{ \frac{1}{\nn}\sum_{i\in\Ik} \big(\loss(S_i;\hbg,\hetaIkc)\matE^T\losstwo^T(S_i;\hetaIkc)\big)_{j,r} - \EP\Big[\big(\loss(S;\bg,\etazero)\matE^T\losstwo^T(S;\etazero)\big)_{j,r} \Big]}\\
		=&\normoneBig{ \frac{1}{\nn}\sum_{i\in\Ik} \loss_j(S_i;\hbg,\hetaIkc)\matE^T\big(\losstwo^T(S_i;\hetaIkc)\big)_{\cdot,r} - \EP\Big[\loss_j(S;\bg,\etazero)\matE^T\big(\losstwo^T(S;\etazero)\big)_{\cdot,r}\Big]}\\
		=&  \normonebig{\frac{1}{\nn}\sum_{i\in\Ik} \matE^T(\losstwo^T(S_i;\hetaIkc))_{\cdot,r}\loss_j(S_i;\hbg,\hetaIkc) - \EP\big[\matE^T(\losstwo^T\big(S;\etazero)\big)_{\cdot,r}\loss_j(S;\bg,\etazero)\big] }\\
		\le & \norm{\matE}\big\lVert  \frac{1}{\nn}\sum_{i\in\Ik} \big(\losstwo^T(S_i;\hetaIkc)\big)_{\cdot,r}\loss_j(S_i;\hbg,\hetaIkc) - \EP\big[\big(\losstwo^T(S;\etazero)\big)_{\cdot,r}\loss_j(S;\bg,\etazero)\big] \big\rVert.
	\end{array}
\end{displaymath}

The term $\Ifourteen$ can be bounded analogously to $\Ithirteen$. 
The term $\Ififteen$ can be bounded analogously to $\Itwelve$. 
Last, we bound the term $\Isixteen$. 
By  Lemma~\ref{lem:helperVarianceConsistent}, we have
\begin{displaymath}
	\normonebigg{\frac{1}{\nn}\sum_{i\in\Ik} \big(\losstwo^T(S_i;\hetaIkc)\big)_{t,r}\big(\losstwo(S_i;\hetaIkc)\big)_{j,l} - \EP\Big[\big(\losstwo^T(S;\etazero)\big)_{t,r}\big(\losstwo(S;\etazero)\big)_{j,l}\Big]}=O_{\PP}(\rhoNtilde),
\end{displaymath}
which implies $\Isixteen=O_{\PP}(\rhoNtilde)$. Indeed, we have
\begin{displaymath}
	\begin{array}{cl}
		&\big\lVert  \frac{1}{\nn}\sum_{i\in\Ik}  \matC\big(\losstwo(S_i;\hetaIkc)-\EP[\losstwo(S;\etazero)]\big)\matE\matE^T\big(\losstwo(S_i;\hetaIkc)-\EP[\losstwo(S;\etazero)]\big)^T\matC^T\\
			&\quad\quad -  \matC\EP\Big[\big(\losstwo(S;\etazero)-\EPN[\losstwo(S;\etazero)]\big)\matE\matE^T\big(\losstwo(S;\etazero)-\EP[\losstwo(S;\etazero)]\big)^T\Big]\matC^T \big\rVert \\
		\le& \norm{\matC}^2\big\lVert  \frac{1}{\nn}\sum_{i\in\Ik}  \losstwo(S_i;\hetaIkc)\matE\matE^T\losstwo^T(S_i;\hetaIkc) - \EP\big[\losstwo(S;\etazero)\matE\matE^T\losstwo^T(S;\etazero)\big] \big\rVert \\
		&\quad + 2\norm{\matC}^2\big\lVert  \frac{1}{\nn}\sum_{i\in\Ik}  \losstwo(S_i;\hetaIkc)\matE\matE^T\EPN\big[\losstwo^T(S;\etazero)\big] - \EP\big[\losstwo(S;\etazero)\matE\matE^T\EP[\losstwo^T(S;\etazero)]\big] \big\rVert \\
		\le& \norm{\matC}^2\big\lVert  \frac{1}{\nn}\sum_{i\in\Ik}  \losstwo(S_i;\hetaIkc)\matE\matE^T\losstwo^T(S_i;\hetaIkc) - \EP\big[\losstwo(S;\etazero)\matE\matE^T\losstwo^T(S;\etazero)\big] \big\rVert \\
		&\quad + 2\norm{\matC}^2\norm{\matE}^2\norm{\EP[\losstwo(S;\etazero)]}\big\lVert  \frac{1}{\nn}\sum_{i\in\Ik}  \losstwo(S_i;\hetaIkc)- \EP[\losstwo(S;\etazero)] \big\rVert \\
		=& \norm{\matC}^2\big\lVert  \frac{1}{\nn}\sum_{i\in\Ik}  \losstwo(S_i;\hetaIkc)\matE\matE^T\losstwo^T(S_i;\hetaIkc) - \EP\big[\losstwo(S;\etazero)\matE\matE^T\losstwo^T(S;\etazero)\big] \big\rVert\\
		&\quad
		+ O_{\PP}\big(\NN^{-\frac{1}{2}}(1+\rhoN)\big)
	\end{array}
\end{displaymath}
by Lemma~\ref{lem:D1convInProb}, and we have
\begin{displaymath}
	\begin{array}{cl}
		&\normonebig{ \frac{1}{\nn}\sum_{i\in\Ik} \big( \losstwo(S_i;\hetaIkc)\matE\matE^T\losstwo^T(S_i;\hetaIkc)\big)_{j,r} - \big(\EP\big[\losstwo(S;\etazero)\matE\matE^T\losstwo^T(S;\etazero)\big]\big)_{j,r}}\\
		=&\normoneBig{ \frac{1}{\nn}\sum_{i\in\Ik}  \big(\losstwo(S_i;\hetaIkc)\big)_{j,\cdot}\matE\matE^T\big(\losstwo^T(S_i;\hetaIkc)\big)_{\cdot,r} - \EP\Big[\big(\losstwo(S;\etazero)\big)_{j,\cdot}\matE\matE^T\big(\losstwo^T(S;\etazero)\big)_{\cdot,r}\Big] }\\
		=& \normoneBig{ \frac{1}{\nn}\sum_{i\in\Ik}  \matE^T\big(\losstwo^T(S_i;\hetaIkc)\big)_{\cdot,r}\big(\losstwo(S_i;\hetaIkc)\big)_{j,\cdot}\matE - \matE^T\EP\Big[\big(\losstwo^T(S;\etazero)\big)_{\cdot,r}(\losstwo(S;\etazero))_{j,\cdot}\Big]\matE }\\
		\le& \norm{\matE}^2\Big\lVert  \frac{1}{\nn}\sum_{i\in\Ik} \big(\losstwo^T(S_i;\hetaIkc)\big)_{\cdot,r} (\losstwo(S_i;\hetaIkc))_{j,\cdot} - \EP\Big[(\losstwo^T(S;\etazero)\big)_{\cdot,r}\big(\losstwo(S;\etazero)\big)_{j,\cdot} \Big]\Big\rVert. 
	\end{array}
\end{displaymath}
\end{proof}

\begin{proof}[Proof of Proposition~\ref{prop:populationBias}]
The statement of Proposition~\ref{prop:populationBias} can be reformulated as 
\begin{displaymath}
		\sqrt{\NN}\normone{\bgN-\betazero} \rightarrow \begin{cases}0, & \mbox{if }\gammaN = \Omega(\sqrt{\NN}) \mbox{ and } \gammaN\not \in \Theta(\sqrt{\NN})\\
		                                                                          C, & \mbox{if }\gammaN=\Theta(\sqrt{\NN})\\
		                                                                          \infty, &\mbox{if } \gammaN = o(\sqrt{\NN})\end{cases}
	\end{displaymath}
using the Bachmann--Landau notation. For instance, the Bachmann--Landau notation  is presented in~\citet{Lattimore2020}.

	Introduce the matrices
	\begin{eqnarray*}
			\Fmatone &:=& \EP[\Rx\Ry],\\
			\Fmattwo &:=& \EP\big[\Rx\Rx^T\big],\\
			\Gmatone &:=& \EP\big[\Rx\Ra^T\big]\EP\big[\Ra\Ra^T\big]^{-1}\EP[\Ra\Ry],\\
			\Gmattwo &:=& \EP\big[\Rx\Ra^T\big]\EP\big[\Ra\Ra^T\big]^{-1}\EP\big[\Ra\Rx^T\big].
	\end{eqnarray*}
	We have
	\begin{displaymath}
			\sqrt{\NN}\normone{\bgN-\betazero}
			= \sqrt{\NN}\normoneBig{\big(\Fmattwo + (\gammaN-1)\Gmattwo\big)^{-1}\big(\Fmatone + (\gammaN-1)\Gmatone\big) - \Gmattwo^{-1}\Gmatone}.
	\end{displaymath}
	First, we assume that the sequence  $\{\gammaN\}_{\NN\ge 1}$ diverges to $+\infty$ 	
	as $\NN\rightarrow\infty$, so that $\gammaN-1$ is bounded away from $0$ for $\NN$ large enough.
	By~\citet[Section 3]{Henderson1981}, we have
	\begin{displaymath}
		\big(\Fmattwo + (\gammaN-1)\Gmattwo\big)^{-1}
		= \frac{1}{\gammaN-1}\Gmattwo^{-1} -  \Big(\one +  \frac{1}{\gammaN-1}\Gmattwo^{-1}\Fmattwo\Big)^{-1}\frac{1}{\gammaN-1}\Gmattwo^{-1}\Fmattwo\frac{1}{\gammaN-1}\Gmattwo^{-1}. 
	\end{displaymath}
	Hence, we have
	\begin{displaymath}
		\begin{array}{rcl}
			\sqrt{\NN}\normone{\bgN-\betazero}
			&=& \frac{\sqrt{\NN}}{\gammaN-1}\Big| \Gmattwo^{-1}\Fmatone 
			- \big(\one +  \frac{1}{\gammaN-1}\Gmattwo^{-1}\Fmattwo\big)^{-1}\frac{1}{\gammaN-1}\Gmattwo^{-1}\Fmattwo\Gmattwo^{-1}\Fmatone \\
			&&\quad\quad
			- \big(\one +  \frac{1}{\gammaN-1}\Gmattwo^{-1}\Fmattwo\big)^{-1}\Gmattwo^{-1}\Fmattwo\Gmattwo^{-1}\Gmatone \Big|
		\end{array}
	\end{displaymath}
	and infer our claim because we have
	\begin{displaymath}
		\begin{array}{cl}
			&\Gmattwo^{-1}\Fmatone 
			- \big(\one +  \frac{1}{\gammaN-1}\Gmattwo^{-1}\Fmattwo\big)^{-1}\frac{1}{\gammaN-1}\Gmattwo^{-1}\Fmattwo\Gmattwo^{-1}\Fmatone \\
			&\quad\quad
			- \big(\one +  \frac{1}{\gammaN-1}\Gmattwo^{-1}\Fmattwo\big)^{-1}\Gmattwo^{-1}\Fmattwo\Gmattwo^{-1}\Gmatone \\
			=& O(1).
		\end{array}
	\end{displaymath} 
	Next, we assume that the sequence $\{\gammaN\}_{\NN\ge 1}$ 
	is bounded. We have 
	\begin{displaymath}
	    \normone{\bgN-\betazero}=
		\normoneBig{\big(\Fmattwo + (\gammaN-1)\Gmattwo\big)^{-1}\big(\Fmatone + (\gammaN-1)\Gmatone\big) - \Gmattwo^{-1}\Gmatone}
		= O(1),
    \end{displaymath}
	which concludes the proof.
\end{proof}

\begin{proof}[Proof of Theorem~\ref{thm:stochasticOrderGammaN}]
	We show that
	\begin{displaymath}
		\PP\big(\hsigma^2(\gammaN) + \NN(\hbgN-\hbetaN)^2\le \hsigma^2\big)\le \PP(|\GammaN|\ge \CN)
	\end{displaymath}
	holds for some random variable $\GammaN$ satisfying $\GammaN=O_{\PP}(1)$ and for some sequence $\{\CN\}_{\NN\ge 1}$  of non-negative numbers diverging to $+\infty$ as $\NN\rightarrow\infty$. 
	
	For real numbers $a$ and $b$, observe that we have
	\begin{displaymath}
		\sqrt{\normone{a}^2+\normone{b}^2}\ge\frac{1}{2}\normone{a}+\frac{1}{2}\normone{b}
	\end{displaymath}
	due to
	\begin{displaymath}
		\frac{3}{4}\Big( \normone{a}^2+\normone{b}^2-\frac{2}{3}\normone{a}\normone{b}\Big)
		\ge \frac{3}{4}(\normone{a}-\normone{b})^2\ge 0.
	\end{displaymath}
	Thus, we have
	\begin{displaymath}
		\begin{array}{rcl}
			\PP\big(\hsigma^2(\gammaN) + \NN(\hbgN-\hbetaN)^2\le \hsigma^2\big)
			&=& \PP\Big(\sqrt{\hsigma^2(\gammaN) + \NN(\hbgN-\hbetaN)^2}\le \hsigma\Big)\\
			&\le& \PP\big(\hsigma(\gammaN) + \sqrt{\NN}|\hbgN-\hbetaN|\le 2\hsigma\big).
		\end{array}
	\end{displaymath}
	By the reverse triangle inequality, we have
	\begin{displaymath}
		\begin{array}{rcl}
			|\hbgN-\hbetaN| &=& |\hbgN-\bgN + \bgN-\betazero+\betazero-\hbetaN|\\
			&\ge& |\bgN-\betazero| - |\hbgN-\bgN| - |\betazero-\hbetaN|. 
		\end{array}
	\end{displaymath}
	Thus, we have
	\begin{displaymath}
		\begin{array}{cl}
			&\PP\big(\hsigma^2(\gammaN) + \NN(\hbgN-\hbetaN)^2\le 2\hsigma^2\big)\\
			\le& \PP\big(\hsigma(\gammaN) + \sqrt{\NN} |\bgN-\betazero| - \sqrt{\NN} |\hbgN-\bgN| - \sqrt{\NN} |\betazero-\hbetaN|\le 2\hsigma\big)\\
			=& \PP\big( \sqrt{\NN} |\bgN-\betazero| \le 2\hsigma-\hsigma(\gammaN)+ \sqrt{\NN} |\hbgN-\bgN| + \sqrt{\NN} |\betazero-\hbetaN|\big)\\
			\le& \PP\big( \big|\hsigma(\gammaN) -2\hsigma- \sqrt{\NN} |\hbgN-\bgN| - \sqrt{\NN} |\betazero-\hbetaN|  \big|  \ge \sqrt{\NN} |\bgN-\betazero|  \big)\\
			\le&  \PP\big( |\hsigma(\gammaN) -2\hsigma- \sqrt{\NN} (\hbgN-\bgN) - \sqrt{\NN} (\betazero-\hbetaN)  |  \ge \sqrt{\NN} |\bgN-\betazero|  \big)\\
		\end{array}
	\end{displaymath}
	by the reverse triangle inequality. 
	Let us introduce the random variable
	\begin{displaymath}
		\GammaN := \hsigma(\gammaN) -2\hsigma- \sqrt{\NN} (\hbgN-\bgN) - \sqrt{\NN} (\betazero-\hbetaN)
	\end{displaymath}
	and the deterministic number $\CN:=\sqrt{\NN} |\bgN-\betazero|$. By Lemma~\ref{lem:XiO1}, we have $\GammaN=O_{\PP}(1)$. 
	Let $\eps>0$, and choose $\Ceps$ and $\Neps$ such that for all $\NN\ge\Neps$ the statement $\PP(|\GammaN|>\Ceps)<\eps$ holds. 
	By Proposition~\ref{prop:populationBias}, $\CN$ tends to infinity as $\NN\rightarrow\infty$ due to $\gammaN=o(\sqrt{\NN})$. 
	Hence, there exists  some $\Ntilde=\Ntilde(\Ceps)$ such that we have $\CN>\Ceps$ for all $\NN\ge\Ntilde$. This implies $\PP(|\GammaN|>\CN)\le\PP(|\GammaN|>\Ceps)$ for all $\NN\ge\Ntilde$. 
	
	Let $\Nbar:=\max\{\Neps, \Ntilde\}$. For all $\NN\ge\Nbar$, we therefore have $\PP(|\GammaN|>\CN)<\eps$. We conclude $\lim_{\NN\rightarrow\infty}\PP(|\GammaN|>\CN)=0$. 
\end{proof} 

\begin{lemma}\label{lemma:hbgammaNbounded}
	Let $\gammaN = o(\sqrt{\NN})$. We have $\sqrt{\NN}(\hbgN - \bgN)=O_{\PP}(1)$. 
\end{lemma}
\begin{proof}[Proof of Lemma~\ref{lemma:hbgammaNbounded}]
We already verified 
	$\hmatA = \matA + o_{\PP}(1)$
and	$\hmatB = \matB + o_{\PP}(1)$
in the proof of Theorem~\ref{thm:asymptNormalgamma}.
	Let us assume that $\gammaN$ diverges to $+\infty$ 
	as $\NN\rightarrow\infty$. 
	We then have
	\begin{displaymath}
		\begin{array}{rcl}
		\big(\hmatA + (\gammaN-1)\hmatB\big)^{-1}&=& \frac{1}{\gammaN-1}\Big(\frac{1}{\gammaN-1}\matA + \matB  + o_{\PP}(1) +\frac{1}{\gammaN-1}o_{\PP}(1) \Big)^{-1}\\
			&=& \frac{1}{\gammaN-1}\Big(  \Big(\frac{1}{\gammaN-1}\matA + \matB\Big)^{-1}+o_{\PP}(1) \Big)\\
			&=& \big(\matA + (\gammaN-1)\matB\big)^{-1}+o_{\PP}\big(\frac{1}{\gammaN-1}\big)
		\end{array}
	\end{displaymath}
	because $\frac{1}{\gammaN-1}=O(1)$ holds. Furthermore, we have
	\begin{displaymath}
		\begin{array}{cl}
			&\sqrt{\NN}(\hbgN-\bgN)\\
			=&\Big(\big(\matA + (\gammaN-1)\matB\big)^{-1}+o_{\PP}\big(\frac{1}{\gammaN-1}\big)\Big)\\
			&\quad\cdot \frac{1}{\sqrt{\KK}}\sum_{\kk=1}^{\KK}\frac{1}{\sqrt{\nn}}\sum_{i\in\Ik}\Big( \losstilde(S_i;\bgN,\hetaIkc)\\
			&\quad\quad\quad+(\gammaN-1)\frac{1}{\nn}\sum_{i\in\Ik}\lossone(S_i;\hetaIkc)\Big(\frac{1}{\nn}\sum_{i\in\Ik}\losstwo(S_i;\hetaIkc)\Big)^{-1}\loss(S_i;\bgN,\hetaIkc)\Big)
		\end{array}
	\end{displaymath}
	by~\eqref{eq:regularizedBetaDMLtwo}. 
	Lemma~\ref{lem:boundRN} states that  
	\begin{displaymath}
		\normbigg{\frac{1}{\sqrt{\nn}}\sum_{i\in\Ik}\losstest(S_i;\bzero,\hetaIkc)
		- \frac{1}{\sqrt{\nn}}\sum_{i\in\Ik}\losstest(S_i;\bzero,\etazero)}
		=O_{\PP}(\rhoN)
	\end{displaymath}
	holds for 
	$\kk\in\indset{\KK}$, $\losstest\in\{\loss,\losstilde,\losstwo\}$, and $\bzero\in\{\bg,\betazero, \bo\}$, and where $\rhoN =\rNpnumber + \NN^{\frac{1}{2}}\lambdaNpnumber$ is as in Definition~\ref{def:asymptNormal} and satisfies $\rhoN\lesssim\deltaNnumber^{\frac{1}{4}}$, and where we interpret $\losstwo(S;b,\eta)=\losstwo(S;\eta)$.
	 	 This statement remains valid in the present setting because there exists some finite real constant $C$ such that  we have $\normone{\bgN}\le C$ for $\NN$ large enough. Hence, we have
	\begin{displaymath}
		\begin{array}{cl}
			&\sqrt{\NN}(\hbgN-\bgN)\\
			=&\bigg(\Big(\frac{1}{\gammaN-1}\matA + \matB\Big)^{-1}+o_{\PP}(1)\bigg)\\
			&\quad\cdot \frac{1}{\sqrt{\KK}}\sum_{\kk=1}^{\KK}\bigg(\frac{1}{\sqrt{\nn}}\sum_{i\in\Ik}\Big( \frac{1}{\gammaN-1}\losstilde(S_i;\bgN,\etazero) + \matC\loss(S_i;\bgN,\etazero)\\
			&\quad\quad\quad+ \big(\lossone(S_i;\etazero)-\EP[\lossone(S;\etazero)]\big)\matE  \\
			&\quad\quad\quad - \matC\big(\losstwo(S_i;\etazero)-\EP[\losstwo(S;\etazero)]\big)\matE\Big) + o_{\PP}(1)\bigg)
		\end{array}
	\end{displaymath}
	by~\eqref{eq:finalCLTeq}. 
	Consider the random variables 
	\begin{displaymath}
		\begin{array}{rcl}
			\widetilde X_i &:=& \frac{1}{\gammaN-1}\losstilde(S_i;\bgN,\etazero) + \matC\loss(S_i;\bgN,\etazero)\\
			&&\quad + \big(\lossone(S_i;\etazero)-\EP[\lossone(S;\etazero)]\big)\matE - \matC\big(\losstwo(S_i;\etazero)-\EP[\losstwo(S;\etazero)]\big)\matE
		\end{array}
	\end{displaymath} 
	for $i\in\indset{\NN}$,
	and $S_{\nn}:=\sum_{i\in\Ik}\widetilde X_i$, and $V_{\nn}:=\sum_{i\in\Ik}\EP[\widetilde X_i^2]$, where $\nn=\frac{\NN}{\KK}$ denotes the size of $\Ik$. 
	The Lyapunov condition is satisfied for  $\delta=2>0$ because  
	\begin{displaymath}
		\frac{1}{\big(\sum_{i\in\Ik}\EP[\widetilde X_i^2]\big)^{2+\delta}}\sum_{i\in\Ik}\EP\big[|\widetilde X_i|^{2+\delta}\big] 
		= \frac{1}{(\EP[\widetilde X_1^2])^{2+\delta}}\cdot\frac{1}{\nn^{1+\delta}}\EP\big[|\widetilde X_1|^{2+\delta}\big] \rightarrow 0
	\end{displaymath}
	holds 
	as $\nn\rightarrow\infty$. Therefore, the Lindeberg--Feller condition is satisfied that implies $\frac{S_{\nn}}{V_{\nn}}\rightarrow\mathcal{N}(0,1)$ as $\nn\rightarrow\infty$. 
	
	The case where the sequence $\gammaN$ is bounded can be analyzed analogously.
\end{proof}

\begin{lemma}\label{lemma:hSigmaGammaNbounded}
	Let $\gammaN = o(\sqrt{\NN})$. We then have $\hsigma^2(\gammaN)=O_{\PP}(1)$. 
\end{lemma}
\begin{proof}[Proof of Lemma~\ref{lemma:hSigmaGammaNbounded}]
	We have
	\begin{displaymath}
			\hsigma^2(\gammaN)
			= \big(\hmatA + (\gammaN-1)\hmatB\big)^{-1}
		\hmatD
		\big(\hmatA^T + (\gammaN-1)\hmatB^T\big)^{-1}. 
	\end{displaymath}
	As verified in the proof of Theorem~\ref{thm:asymptNormalgamma}, we have 
	$\hmatA = \matA + o_{\PP}(1)$
and	$\hmatB = \matB + o_{\PP}(1)$.
We established $\hmatDk=\matD+o_{\PP}(1)$ 
in the proof of Theorem~\ref{thm:estSDgamma} for fixed $\gamma$. 
Consequently, the claim follows if the sequence $\{\gammaN\}_{\NN\ge 1}$ is bounded. Next, assume that $\gammaN$ diverges to $+\infty$
as $\NN\rightarrow\infty$. 
We verified
\begin{displaymath}
		\big(\hmatA + (\gammaN-1)\hmatB\big)^{-1}= \big(\matA + (\gammaN-1)\matB\big)^{-1}+o_{\PP}\Big(\frac{1}{\gammaN-1}\Big)
	\end{displaymath}
in the proof of Lemma~\ref{lemma:hbgammaNbounded}. 
It can be shown that $\frac{1}{(\gammaN-1)^2}\hmatD$ is bounded in $\PP$-probability by adapting the arguments presented in the prof of Theorem~\ref{thm:estSDgamma} because there exists some finite real constant $C$ such that we have $\normone{\bgN}\le C$ for $\NN$ large enough. 
Therefore, 
\begin{displaymath}
\begin{array}{rcl}
			\hsigma^2(\gammaN)
			&=& \Big(\frac{1}{\gammaN-1}\matA + \matB+o_{\PP}(1)\Big)^{-1}
		\frac{1}{(\gammaN-1)^2}\hmatD
		\Big(\frac{1}{\gammaN-1}\matA^T + \matB^T +o_{\PP}(1)\Big)^{-1}
		\end{array}
	\end{displaymath}
	is bounded in $\PP$-probability. 
\end{proof}

\begin{lemma}\label{lem:XiO1}
	Let $\gamma = o(\sqrt{\NN})$. We then have
	\begin{displaymath}
		\GammaN := \hsigma(\gammaN) -2\hsigma- \sqrt{\NN} (\hbgN-\bgN) - \sqrt{\NN} (\betazero-\hbetaN)=O_{\PP}(1). 
	\end{displaymath}
\end{lemma}
\begin{proof}[Proof of Lemma~\ref{lem:XiO1}]
	By Theorem~\ref{thm:asymptNormal}, the term $\sqrt{\NN} (\betazero-\hbetaN)$ asymptotically follows a Gaussian distribution and is hence bounded in $\PP$-probability.
	By Theorem~\ref{thm:estSD}, the term $\hsigma^2$ converges in $\PP$-probability. Thus, $2\hsigma$  is  bounded in $\PP$-probability as well. 
	By Lemma~\ref{lemma:hbgammaNbounded}, we have $\sqrt{\NN}(\hbgN - \bgN)=O_{\PP}(1)$. By
	 Lemma~\ref{lemma:hSigmaGammaNbounded}, we have $\hsigma^2(\gammaN)=O_{\PP}(1)$. 
\end{proof}

\begin{proof}[Proof of Theorem~\ref{thm:gammaHatNormal}]
The fact the the statement holds uniformly for $P \in \PcalN$ can be derived using analogous arguments as used to prove Theorem~\ref{thm:asymptNormal} and~\ref{thm:asymptNormalgamma}. 
Theorem~\ref{thm:estSDgamma} in the appendix shows that $\hsigma(\gamma)$ consistently estimates $\sigma(\gamma)$ for fixed $\gamma$. Analogous arguments show that $\hsigma(\hgammap)$ consistently estimates $\sigma$ from Theorem~\ref{thm:asymptNormal}.
Let $\hmu:=\hgammap - 1$. 
We have 
\begin{displaymath}
	\begin{array}{cl}
		& \sqrt{\NN}(\hbhgp-\bhgp)\\
		=& \Big( \frac{1}{\KK}\sum_{\kk=1}^{\KK}\big(\hbRxk\big)^T\big(\frac{1}{\hmu}\one+\PiIkcIk\big)\hbRxk \Big)^{-1}
		\frac{1}{\KK}\sum_{\kk=1}^{\KK}\big(\hbRxk\big)^T\big(\frac{1}{\hmu}\one+\PiIkcIk\big)\big(\hbRyk - \hbRxk\bhgp \big).
	\end{array}
\end{displaymath}
Due to Theorem~\ref{thm:stochasticOrderGammaN}, we have $\frac{1}{\hmu}=\frac{1}{\sqrt{\NN}}o_{\PP}(1)$. Due to Proposition~\ref{prop:populationBias}, whose statements also hold stochastically for random $\gamma$, we have $\bhgp = \betazero + \frac{1}{\sqrt{\NN}}o_{\PP}(1)$. Therefore, we have
\begin{displaymath}
	\begin{array}{cl}
		& \sqrt{\NN}(\hbhgp-\bhgp)\\
		=& \Big( \frac{1}{\KK}\sum_{\kk=1}^{\KK}\big(\hbRxk\big)^T\PiIkcIk\hbRxk \Big)^{-1}
		\frac{1}{\KK}\sum_{\kk=1}^{\KK}\big(\hbRxk\big)^T\PiIkcIk\big(\hbRyk - \hbRxk\betazero \big) + o_{\PP}(1)\\
		=& \sqrt{\NN}(\hbetaN-\betazero) + o_{\PP}(1)
	\end{array}
\end{displaymath}
due to Slutsky's theorem and similar arguments as presented in the proofs of Theorem~\ref{thm:asymptNormal} and~\ref{thm:asymptNormalgamma}. 
\end{proof}

\section{Proof of Section~\ref{sect:exampleForest}}\label{appendix:proofRandomForestSEM}

We argue that $A_1$ and $A_2$ are independent of $H$ conditional on $W_1$ and $W_2$ in the SEM in Figure~\ref{fig:simulation2}. First, we consider $A_1$. All paths from $A_1$ to $H$ through $X$ or $Y$ are blocked by the empty set because either $X$ or $Y$ is a collider on these paths. The path $A_1\rightarrow A_2\rightarrow W_1\rightarrow H$ is blocked by $W_1$. 
Second, we consider $A_2$. All paths from $A_2$ to $H$ through $X$ or $Y$ are blocked by the empty set because either $X$ or $Y$ is a collider on these paths. The path $A_2\rightarrow W_1\rightarrow H$ is blocked by $W_1$.

\end{appendices}

\end{document}